\def\version{November 9, 2017}

 \documentclass[12pt]{article}
 \def\macrosPb{}% include general packages
\usepackage{amsfonts}
\usepackage{amsmath,amssymb,amsthm}
\usepackage{appendix}
\usepackage{bbm} % used in \newcommand{\1}{\mathbbm{1}}
\usepackage{amsbsy}
\usepackage{enumerate}
\usepackage{cite}
\usepackage{enumerate}

 \def\macrosHarxiv

\InputIfFileExists{./macros_local.tex}{}{}

\ifdefined\macrosPa
  \usepackage[textwidth=465pt,textheight=650pt,centering]{geometry} % 11pt
\else\ifdefined\macrosPb
  \usepackage[textwidth=500pt,textheight=650pt,centering]{geometry} % 12pt
\fi\fi

\ifdefined\macrosS

  \usepackage{mathptmx}
  \DeclareMathAlphabet{\mathcal}{OMS}{cmsy}{m}{n}
\fi

\ifdefined\macrosBirk
\else
\usepackage[dvips]{graphicx}
\fi

\ifdefined\macrosSB
\def\UseSection{%%
        \numberwithin{equation}{section}
	\theoremstyle{plain}% default theorem style 
        \newtheorem{theorem}    {Theorem}[section]
        \DefineTheorems % Use this to define other environments to be 
}

\def\DefineTheorems{%%
	
	\newtheorem{lemma}      [theorem] {Lemma}
	
	\newtheorem{prop}       [theorem] {Proposition}
	
	\newtheorem{cor}        [theorem] {Corollary}

	\theoremstyle{definition}% ``defn'' theorem style 
	\newtheorem{defn}       [theorem] {Definition}

	\newtheorem{rk} 	[theorem] {Remark}
	\theoremstyle{definition}% ``remark'' theorem style 

}

\newcommand{\bt}   {\begin{theorem}}
\newcommand{\et}   {\end  {theorem}}
\newcommand{\bl}   {\begin{lemma}}
\newcommand{\el}   {\end  {lemma}}
\newcommand{\bp}   {\begin{prop}}
\newcommand{\ep}   {\end  {prop}}
\newcommand{\bc}   {\begin{cor}}
\newcommand{\ec}   {\end  {cor}}
\newcommand{\bd}   {\begin{defn}}
\newcommand{\ed}   {\end  {defn}}

\newcommand{\ba}   {\begin{array}}
\newcommand{\ea}   {\end  {array}}
\newcommand{\be}   {\begin{enumerate}}
\newcommand{\ee}   {\end  {enumerate}}
\newcommand{\bi}   {\begin{itemize}}
\newcommand{\ei}   {\end  {itemize}}

\def\eq#1\en{\begin{equation}#1\end{equation}}  
\def\eqsplit#1\ensplit{
	\begin{equation}\begin{split}#1\end{split}\end{equation}
	}
\def\eqalign#1\enalign{
	\begin{align}#1\end{align}
	}
\def\eqmul#1\enmul{
	\begin{multline}#1\end{multline}
	}
\newcommand{\eqarrstar} {\begin{eqnarray*}} 
\newcommand{\enarrstar} {\end{eqnarray*}} 
\newcommand{\eqarray}   {\begin{eqnarray}} 
\newcommand{\enarray}   {\end{eqnarray}} 
\newcommand{\nnb}	{\nonumber \\} 

\newcommand{\lbeq}[1]  {\label{e:#1}}
\newcommand{\refeq}[1] {\eqref{e:#1}}    % AMS-LaTeX trick!
\newcommand{\lbfg}[1]  {\label{fg: #1}}

\makeatletter
\newcommand{\labelcounter}[2]{{%
	\stepcounter{#1}%	First, increase the ``countC'' by one.
	\protected@write\@auxout{}%
	{\string\newlabel{#2}{{\csname the#1\endcsname}{\thepage}}}%
	{\ref{#2}}%	Finally, make sure to refer to this label, 
	}}
\makeatother
\newcommand{\Ebold} {{\mathbb E}}

\newcommand{\Nbold} {{\mathbb N}}

\newcommand{\Rbold} {{\mathbb R}}

\newcommand{\Zbold} {{\mathbb Z}}

\newcommand{\Bcal}   {\mathcal{B}} 
\newcommand{\Ccal}   {\mathcal{C}} 
\newcommand{\Dcal}   {\mathcal{D}} 
 
\newcommand{\Fcal}   {\mathcal{F}} 
\newcommand{\Gcal}   {\mathcal{G}}

\newcommand{\Kcal}   {\mathcal{K}} 
\newcommand{\Lcal}   {\mathcal{L}} 
 
\newcommand{\Ncal}   {\mathcal{N}} 
 
\newcommand{\Pcal}   {\mathcal{P}}

\newcommand{\Scal}   {\mathcal{S}} 
 
\newcommand{\Ucal}   {\mathcal{U}} 
\newcommand{\Vcal}   {\mathcal{V}} 
\newcommand{\Wcal}   {\mathcal{W}}

\newcommand{\ghat}  {{ \hat{g}  }}

\newcommand{\Rd}    {{ {\Rbold}^d}}
\newcommand{\Zd}    {{ {\Zbold}^d }}

\newcommand{\spose}[1] {{\hbox to 0pt{#1\hss}} }
\newcommand{\ltapprox} {\mathrel{\spose{\lower 3pt\hbox{$\mathchar"218$}}
 \raise 2.0pt\hbox{$\mathchar"13C$}}}
\newcommand{\gtapprox} {\mathrel{\spose{\lower 3pt\hbox{$\mathchar"218$}}
 \raise 2.0pt\hbox{$\mathchar"13E$}}}

\else
\fi

\UseSection   %Necessary to define Numbering scheme for Theorem, etc.
\usepackage[usenames]{color}

\definecolor{bw}{RGB}{240, 120, 0}
\definecolor{at}{rgb}{0.0, 0.5, 0.0} % green, darker than the standard green.

\newcommand{\LT}{{\rm Loc}  }

\newcommand{\DV}{\Dcal}
\newcommand{\DVa}{\alpha}

\renewcommand{\to} {\rightarrow}

\newcommand{\R}{\Rbold}
\newcommand{\Z}{\Zbold}

\newcommand{\N}{\Nbold}
\newcommand{\C}{\mathbb{C}}
\newcommand{\volume}{\mathbb{V}}

\newcommand{\1}{\mathbbm{1}}

\newcommand{\la}{\langle}
\newcommand{\ra}{\rangle}

\newcommand{\psib}{\bar\psi}

\newcommand{\w}{{\sf w}}
\newcommand{\Mext}{M_\mathrm{ext}}

\newcommand{\Ex}{\mathbb{E}}

\newcommand{\chicCov}{{\chi}}
\newcommand{\ellconst}{\mathfrak{c}}

\newcommand{\Gtilp}{\gamma}

\newcommand{\pair}[1]{\langle #1 \rangle}

\newcommand{\cgam}{\gamma}

\newcommand{\pt}{{\rm pt}}

\newcommand{\Upt}{U_{\rm pt}}
\newcommand{\Vpt}{V_{\rm pt}}
\newcommand{\gpt}{g_{\mathrm{pt}}}
\newcommand{\nupt}{\nu_{\mathrm{pt}}}

\newcommand{\Kch}{\check{K}}

\newcommand{\h}{\mathfrak{h}}
\ifdefined\macrosSB \else

\fi

\renewcommand{\ghat}{\hat{g}}
\newcommand{\ggen}{\tilde{g}}

\newcommand{\mgen}{\tilde{m}}
\newcommand{\Iint}{\mathbb{I}}

\newcommand{\mubar}{\bar{\mu}}
\newcommand{\muhat}{\hat{\mu}}

\newcommand{\domRG}{\mathbb{D}}

\newcommand{\half}{\textstyle{\frac 12}}

\newcommand{\ddp}[2]{\frac{\partial #1}{\partial #2}}

\newcommand{\epV}{\epsilon_{V}}

\newcommand{\epVbar}{\epsilon_{g\tau^{2}}}

\newcommand{\epdV}{\bar{\epsilon}}

\newcommand{\phib}{\bar\phi}

\newcommand{\Kspace}{\Kcal}

\DeclareMathOperator{\Loc}{Loc} %\LT

\ifdefined\macrosH
  \usepackage{xr-hyper}
  \usepackage{hyperref}
  \hypersetup{hypertexnames=false}
  \hypersetup{colorlinks,citecolor=blue,linkcolor=blue}  % Comment out for printing!

\else\ifdefined\macrosHarxiv
  \usepackage{xr-hyper}
  \usepackage{hyperref}
  \hypersetup{hypertexnames=false, hidelinks}

\else
  \newcommand{\texorpdfstring}[2]{#1}
  \usepackage{xr}
\fi\fi

 \usepackage{url}
 \usepackage{xr-hyper}
 \usepackage{hyperref}

 \title {
  Critical exponents for long-range $O(n)$
  \\
  models below the upper critical dimension
 }

 \author{
   Gordon Slade\thanks{Department of Mathematics,
     University of British Columbia,
     Vancouver BC, Canada V6T 1Z2.
     Email:  {\tt slade@math.ubc.ca}}}

\date\version

\renewcommand{\chicCov}{{\vartheta}}
\newcommand{\chiL}{{\vartheta}}

\newcommand{\gL}{s}
\newcommand{\dgL}{y}

\newcommand{\gLfix}{\bar{s}}
\newcommand{\mufix}{\bar{\mu}}

\newcommand{\gamhat}{\bar{\gamma}}
\newcommand{\alphahat}{\hat{\alpha}}
\DeclareMathOperator{\PT}{PT}
\newcommand{\upt}{u_{\pt}}
\newcommand{\multia}{a}

\newcommand{\jgen}{j_{\mgen}}
\newcommand{\newxi}{\pi}
\newcommand{\CRG}{C_{\rm RG}}

\renewcommand{\DVa}{t}
\newcommand{\Kweight}{\lambda}
\renewcommand{\Gtilp}{{\sf t}}

\begin{document}

\maketitle

\begin{abstract}
  We consider the critical behaviour of long-range $O(n)$ models ($n \ge 0$)
  on $\Zd$, with interaction that decays with distance $r$ as $r^{-(d+\alpha)}$,
  for $\alpha \in (0,2)$. For $n \ge 1$, we study   the $n$-component $|\varphi|^4$
  lattice spin model.  For $n =0$, we study the weakly self-avoiding
  walk via an exact representation as a supersymmetric spin model.
  These models have upper critical dimension $d_c=2\alpha$.
  For dimensions $d=1,2,3$ and small $\epsilon>0$, we choose $\alpha = \frac 12 (d+\epsilon)$,
  so that $d=d_c-\epsilon$ is below the upper critical dimension.
  For small $\epsilon$ and weak coupling, to order $\epsilon$ we
  prove existence of and compute the values of the critical exponent $\gamma$
  for the susceptibility (for $n \ge 0$) and  the critical exponent
  $\alpha_H$ for the specific heat (for $n \ge 1$).  For the susceptibility,
  $\gamma = 1 +  \frac{n+2}{n+8} \frac \epsilon\alpha + O(\epsilon^2)$,
  and a similar result is proved for the specific heat.
  Expansion  in $\epsilon$ for such long-range models was first carried out
  in the physics literature in 1972.  Our proof adapts and
  applies a rigorous renormalisation group method developed in previous
  papers with Bauerschmidt and Brydges for the nearest-neighbour models in
  the critical dimension $d=4$, and is based on the construction of a
  non-Gaussian renormalisation group fixed point.  Some aspects of the method
  simplify below the upper critical dimension, while some require different treatment,
  and new ideas and techniques with potential future application are introduced.
  \end{abstract}

%\keywords{renormalisation group \and critical exponents
%  \and susceptibility \and specific heat \and long-range model}
%\subclass{82B28 \and 82B27 \and 82B20 \and 60K35}

\setcounter{tocdepth}{2}
\tableofcontents

\section{Introduction and main results}

\subsection{Introduction}

The understanding of critical phenomena via the renormalisation group is
one of the great achievements of physics in the twentieth century, as
it simultaneously provides an explanation of universality,
as well as a systematic method for the computation of universal quantities such
as critical exponents.
It remains a challenge to place these methods on a firm mathematical foundation.

For short-range Ising or
$n$-component $|\varphi|^4$ spin systems ($n \ge 1$), the upper critical
dimension is $d_c=4$, meaning that mean-field theory applies in dimensions $d>4$.
Renormalisation group methods have been applied in a mathematically rigorous
manner  to study the critical behaviour of the
$|\varphi|^4$ model in the upper critical dimension $d=4$, using block
spin renormalisation in \cite{GK85,GK86,Hara87,HT87} (for $n=1$), phase space expansion
methods in \cite{FMRS87} (for $n=1$), and using the methods that
we apply and further develop in this paper in
\cite{BBS-phi4-log,ST-phi4,BSTW-clp} (for $n \ge 1$).
The low-temperature phase has been studied, e.g., in \cite{Bala95,BO99}.
For $n=0$, a supersymmetric version of the $|\varphi|^4$ model corresponds exactly to
the weakly self-avoiding walk, and has been analysed in detail for
$d= 4$ \cite{BBS-saw4-log,BBS-saw4,ST-phi4,BSTW-clp}.
A model related to the 4-dimensional weakly self-avoiding walk is studied in \cite{IM94}.
Renormalisation group methods have recently been applied to gradient field models in \cite{AKM16},
to the Coloumb gas in \cite{Falc13}, to interacting dimers in \cite{GMT17}, and to
symmetry breaking in low temperature many-boson systems in \cite{BFKT16}.
For \emph{hierarchical}
models, the critical behaviour of spin systems was studied in \cite{CE78,GK83a,GK84a,HHW01},
and for weakly self-avoiding walk in \cite{BEI92,BI03c,BI03d}.
An introductory account of a renormalisation group analysis of the
4-dimensional hierarchical $|\varphi|^4$ model,
using methods closely related to those used in the present paper,
is given in \cite{BBS-brief}.

In a 1972 paper entitled ``Critical exponents in 3.99 dimensions''
\cite{WF72}, Wilson and Fisher explained how to apply the renormalisation group
method in dimension $d=4-\epsilon$ for small $\epsilon >0$.
This has long been physics textbook material, e.g., in
\cite[p.236]{Amit84} the values of the critical exponents for
the susceptibility ($\gamma$), the specific heat ($\alpha_H$), the correlation length ($\nu$),
and the critical
two-point function ($\eta$) can be found:
\begin{align}
\lbeq{ep1}
    \gamma & = 1 +  \frac{n+2}{n+8} \frac{\epsilon}{2}  + \cdots,
    \quad
    \alpha_H  =  \frac{4-n}{n+8}\frac{\epsilon}{2} + \cdots,
    \\
\lbeq{ep2}
    \nu & = \frac 12 +  \frac{n+2}{n+8}\frac{\epsilon}{4} + \cdots,
    \quad
    \eta  =  \frac{n+2}{(n+8)^2} \frac{\epsilon^2}{2} + \cdots.
\end{align}
Quadratic terms in $\epsilon$ are also
given in \cite{Amit84}, and terms up to order $\epsilon^6$
are known in the physics literature \cite{GZ85,GZ98,KP17}.   These $\epsilon$-expansions
are believed to be asymptotic, but they must be divergent
since analyticity at $\epsilon=0$ would
be inconsistent with mean-field exponents for $\epsilon<0$ (which
obey \refeq{ep1}--\refeq{ep2} with $\epsilon=0$).
Critical exponents for dimension $d=3$ (corresponding to $\epsilon=4-d =1$)
have been computed from the $\epsilon$ expansions via Borel resummation,
and the results are consistent with those obtained via other methods \cite{GZ85,GZ98,KP17}.

The $\epsilon$ expansion is not mathematically rigorous---in particular the
spin models are not directly defined in non-integer dimensions.
This particular issue can be circumvented by considering long-range models with
interaction decaying with distance as $|x|^{-(d+\alpha)}$, for
$\alpha \in (0,2)$.  It is known that these models have upper critical dimension
$d_c=2\alpha$ \cite{FMN72,AF88}; this is the dimension above which the \emph{bubble
diagram} converges (see Section~\ref{sec:bubble}).
A hint that the long-range model may have an upper critical dimension that is lower
than its short-range counterpart
can be seen already from the fact that random walk on $\Zd$ with
step distribution decaying as $|x|^{-(d+\alpha)}$ is transient if and only if $d>\alpha$,
as opposed to $d >2$ in the short-range case.
That the upper critical dimension should be $2\alpha$ can be anticipated from the fact
that the range of an $\alpha$-stable process has dimension $\alpha$ \cite{BG60},
so two independent processes generically do not intersect in dimensions above $2\alpha$.
Several mathematical papers establish
mean-field behaviour for long-range models in dimensions $d>d_c=2\alpha$, including
\cite{AF88,HHS08,Heyd11,CS08,CS11,CS15}.

In a 1972 paper, Fisher, Ma and Nickel \cite{FMN72} carried out the $\epsilon$
expansion to compute critical exponents for long-range $O(n)$ models in dimension
$d=d_c-\epsilon = 2\alpha -\epsilon$; see also \cite{Sak73}.
The work of Suzuki, Yamazaki and Igarashi \cite{SYI72} is roughly contemporaneous
with that of Fisher, Ma and Nickel, and reaches similar conclusions.
The results of \cite{FMN72,SYI72}, which are
not mathematically rigorous, include (see \cite[(5),(9)]{FMN72})
\begin{align}
    \gamma & = 1 +  \frac{n+2}{n+8} \frac{\epsilon}{\alpha} + \cdots
    ,
    \quad
    \eta  = 2-\alpha
    ,
\lbeq{gamep}
\end{align}
where we omit terms of order $\epsilon^2$ present in \cite{FMN72,SYI72}.
If we assume
the scaling relations $\gamma = (2-\eta)\nu$ and hyperscaling relation $\alpha_H = 2-d\nu$,
then we obtain
\begin{align}
\lbeq{nuep}
    \nu &=
    \frac{\gamma}{\alpha}
    ,
    \quad
    \alpha_H
    =
       \frac{4-n}{n+8} \frac{\epsilon}{\alpha} + \cdots.
\end{align}
Interestingly, the critical
exponent $\eta$ was predicted to ``stick'' at the mean-field value $\eta = 2-\alpha$ to all orders
in $\epsilon$ \cite{FMN72,Sak73} (we are interested here only
in $d=1,2,3$ and small $\epsilon$ and not in the range corresponding
to crossover to short-range behaviour \cite{HN89,BRRZ17,BPR-T14}).
This has been proved very recently \cite{LSW17} for all $n \ge 0$, using an extension
of the methods we develop here.  Earlier,
a proof that $\eta = 2-\alpha$ for small $\epsilon$ was announced by Mitter
for a 1-component continuum model \cite{Mitt13}.
The long-range model has also recently been studied in connection with conformal invariance
of the critical theory
for $d=d_c-\epsilon$ with $\epsilon$ small \cite{Abde15,PRRZ16}.

From a mathematical point of view, the long-range $O(n)$ model
has the advantage that it can be defined in integer dimension $d$ with
$\alpha = \frac 12 (d+\epsilon)$ chosen so that $d$ is just slightly below
the upper critical dimension: $d=d_c-\epsilon = 2\alpha -\epsilon$.
This approach has been adopted in the mathematical physics literature
\cite{BMS03,MS08,Abde07}, where the emphasis has been on
the construction of a non-Gaussian renormalisation group fixed point,
including a construction of a renormalisation group trajectory between
the Gaussian and non-Gaussian fixed points in \cite{Abde07}.
An earlier paper in a related direction is \cite{BDH98}.
The papers \cite{BDH98,BMS03,Abde07} consider continuum models, whereas
\cite{MS08} considers a supersymmetric model on $\Z^3$ that is essentially
the same as the $n=0$ model we consider here.
None of these papers address the computation of critical exponents.
Critical correlation functions were studied in a
hierarchical version of the model in \cite{GK83a,GK84a}, and the recent paper
\cite{ACG13} carries out a computation of critical
exponents in a different hierarchical setting; see also \cite{Abde15}.

In this paper, we apply a rigorous renormalisation group method to
the long-range $O(n)$ model on $\Zd$, for $d=1,2,3$.
To order $\epsilon$, we prove the existence of and compute
the values of the critical exponent
$\gamma$ for the susceptibility (for all $n \ge 0$)
and the exponent $\alpha_H$ for the specific heat (for all $n \ge 1$).
The case $n=0$ is treated
exactly as a supersymmetric version of the $|\varphi|^4$ model, with most
of the analysis carried out simultaneously and in a unified manner
for the spin model ($n \ge 1$)
and the weakly self-avoiding walk ($n=0$).
This unification has grown out of work on the 4-dimensional case in
\cite{BBS-saw4-log,BBS-phi4-log,ST-phi4,BSTW-clp}.

The proof adapts and
applies a rigorous renormalisation group method that was developed in a series of
papers with Bauerschmidt and Brydges for the nearest-neighbour models in
the critical dimension $d=4$.
Some aspects of the method require extension to deal with the fact that
the renormalisation group fixed point is non-Gaussian for $d=d_c-\epsilon
=2\alpha -\epsilon$.
On the other hand, some aspects of the method
simplify significantly compared to the critical
dimension.
We also adapt and simplify some ideas from the construction of the non-Gaussian fixed point
in \cite{BMS03,MS08}.
We use the term ``fixed point'' loosely in this paper, as the notion itself is
faulty here because the renormalisation group map does not act autonomously due
to lattice effects.   Nevertheless, our analysis is based on what would be a fixed
point if the lattice effects were absent, and we persist in using the terminology.

The $|\varphi|^4$ model has been studied in the mathematical literature for many
decades \cite{GJ87}.  Recently its dynamical version and the connection with
renormalisation and stochastic partial differential equations have received
renewed interest \cite{Hair15,Kupi16}.  Our topic here is the equilibrium setting
of the model,
and we do not consider dynamics.

\subsection{The \texorpdfstring{$|\varphi|^4$}{phi4} model}

We now give a precise definition of the
long-range $n$-component $|\varphi|^4$ model, for $n \ge 1$.
As usual, it is defined first in finite volume, followed by an infinite volume limit.

Let $L,N >1$ be integers, and let $\Lambda = \Lambda_N = \Z^d/L^N\Z^d$ be the
$d$-dimensional discrete torus of side length $L^N$.
Let $n \ge 1$.  The spin field $\varphi$ is  a function $\varphi : \Lambda \to \R^n$,
denoted $x \mapsto \varphi_x$,
and we sometimes write $\varphi \in (\R^n)^\Lambda$.
The Euclidean norm of $v =(v^1,\ldots,v^n)\in \R^n$ is $|v| = [\sum_{i=1}^n (v^i)^2]^{1/2}$,
with inner product $v \cdot w = \sum_{i=1}^n v^i w^i$.

We fix a $\Lambda \times \Lambda$ real symmetric matrix
$M$, and for $x \in \Lambda$ we define $(M \varphi)_x \in \R^n$ by the component-wise action
$(M \varphi)_x^i = \sum_{y\in \Lambda} M_{xy} \varphi^i_y$.
Given $g>0$ and $\nu \in \R$, we define a function $V: (\R^n)^\Lambda \to \R$ by
\begin{equation} \label{e:VdefM}
  V(\varphi)
  =
  \sum_{x\in\Lambda}
  \big(\tfrac{1}{4} g |\varphi_x|^4 + \half \nu |\varphi_x|^2
  +
  \half
   \varphi_x \cdot (M \varphi)_x \big)
  .
\end{equation}
By definition, the quartic term is $|\varphi_x|^4 = (\varphi_x \cdot \varphi_x)^2$.
The \emph{partition function} is defined by
\begin{equation}
\lbeq{pf}
    Z_{g,\nu,N} = \int_{(\R^n)^{\Lambda}}   e^{-V(\varphi)} d\varphi,
\end{equation}
where $d\varphi$ is the Lebesgue measure on
$(\R^n)^{\Lambda}$.
The expectation of a random variable $F:(\R^n)^{\Lambda} \to \R$ is
\begin{equation}
  \label{e:Pdef}
  \langle F \rangle_{g,\nu,N}
  = \frac{1}{Z_{g,\nu,N}} \int_{(\R^n)^{\Lambda}} F(\varphi) e^{-V(\varphi)} d\varphi.
\end{equation}
Thus $\varphi$ is a classical continuous unbounded $n$-component spin field
on the torus $\Lambda$, i.e., with periodic boundary conditions.

For $f:\Zd \to \R$ and $e\in\Z^d$ with $|e|_1=1$, the discrete gradient is defined by
$(\nabla^e f)_x =f_{x+e}-f_x$.  The gradient acts component-wise on $\varphi$, and has a natural
interpretation for functions $f:\Lambda \to \R$.
The discrete Laplacian is
$\Delta = -\frac{1}{2}\sum_{e\in\Z^d:|e|_1=1}\nabla^{-e} \nabla^{e}$.
The Laplacian has versions on both $\Zd$ and the torus $\Lambda$, which we
distinguish when necessary by writing $\Delta_{\Zd}$ or $\Delta_\Lambda$.
Let $\alpha \in (0,2)$.
We choose $M$ to be the lattice fractional Laplacian $M=(-\Delta_\Lambda)^{\alpha/2}$.
Then \refeq{VdefM} becomes
\begin{equation} \label{e:Vdef1}
  V(\varphi)
  = \sum_{x\in\Lambda}
  \Big(\tfrac{1}{4} g |\varphi_x|^4 + \half \nu |\varphi_x|^2
  + \half \varphi_x\cdot ((-\Delta_\Lambda)^{\alpha/2} \varphi_x)  \Big)
  .
\end{equation}
The definition and properties of the positive semi-definite operator
$(-\Delta_\Lambda)^{\alpha/2}$ are
discussed in Section~\ref{sec:fracLap}.  Since
$(-\Delta_\Lambda)^{\alpha/2}_{x,y} \le 0$
for  $x \neq y$, $V$ is a \emph{ferromagnetic} interaction
which prefers spins to align.
It is long-range, and on $\Zd$
decays at large distance as $-(-\Delta_{\Zd})^{\alpha/2}_{xy} \asymp |x-y|^{-(d+\alpha)}$.
Here, and in the following, we write $a \asymp b$ to denote the existence of $c>0$
such that $c^{-1} b \leq a \leq c b$.

The \emph{susceptibility} is defined by
\begin{equation}
  \label{e:susceptdef}
  \chi(g, \nu; n)
  = \lim_{N \to \infty} \sum_{x\in\Lambda_N}
  \pair{\varphi_0^1\varphi_x^1}_{g,\nu,N}
  = n^{-1} \lim_{N\to\infty} \sum_{x\in\Lambda_N}
  \pair{\varphi_0 \cdot \varphi_x}_{g,\nu,N}
  ,
\end{equation}
assuming the limit exists.
We prove the existence of
the infinite volume limit directly, with periodic boundary conditions and large $L$,
in the situations covered by our theorems.
The general theory of such infinite volume limits is well developed for $n=1,2$,
but not for $n>2$  \cite{FFS92}.
Even monotonicity of $\chi$ in $\nu$ is not known for all $n$, but it is to be
expected that $\chi$ is monotone decreasing in $\nu$ and that there is a \emph{critical
value} $\nu_c=\nu_c(g;n)< 0$
(depending also on $d$)
such that $\chi(g,\nu;n) \uparrow \infty$ as $\nu \downarrow
\nu_c$.  We are interested in the nature of this divergence.
For $g=0$, \refeq{Vdef1} is quadratic, \refeq{Pdef} is a Gaussian expectation,
$\nu_c(0;n)=0$, and $\chi(0,\nu;n) = (\nu-\nu_c)^{-1}$ for $\nu > \nu_c=0$ (cf.\ \refeq{C1}).

The \emph{pressure} is defined by
\begin{align}
  \label{e:pressuredef}
  p(g,\nu) &= \lim_{N \to \infty} \frac{1}{|\Lambda_N|} \log Z_{g,\nu,N}
  ,
\end{align}
and the \emph{specific heat} is defined by
\begin{equation} \label{e:cHdef}
  c_H(g,\nu) =  \ddp{^2p}{\nu^2}(g,\nu).
\end{equation}
Assuming the second derivative exists and commutes with the infinite volume limit,
\begin{align}
  \ddp{^2p}{\nu^2}(g,\nu)
  &= \lim_{N\to\infty}
  \tfrac 14 \sum_{x\in\Lambda_N}
  \langle |\varphi_0|^2 ; |\varphi_x|^2 \rangle_{g,\nu,N},
\end{align}
where we write $\pair{A;B} = \pair{AB}-\pair{A}\pair{B}$ for
the \emph{covariance} or \emph{truncated expectation} of random variables $A,B$.
We are interested in the behaviour of the specific heat as $\nu \downarrow \nu_c$.
Similar to the susceptibility, we prove the existence of the relevant
infinite volume limits directly
in the situations covered by our theorems.

\subsection{Weakly self-avoiding walk}
\label{sec:wsaw}

A continuous-time Markov chain $X$ with state space $\Zd$ can be defined via specification
of a $Q$ matrix \cite{Norr97}, namely a
$\Zd\times\Zd$ matrix $(Q_{xy})$ with $Q_{xx}<0$, $Q_{xy}\ge 0$
for $x\neq y$, and $\sum_y Q_{xy}=0$.  Such a Markov chain takes steps from $x$ at
rate $-Q_{xx}$, and jumps to $y$ with probability $-\frac{Q_{xy}}{Q_{xx}}$.
The matrix $Q$ is called the infinitesimal generator of the Markov chain, and, for $t \ge 0$,
\begin{equation}
\lbeq{Pxy}
    P_x(X(t)=y) = E_x(\1_{X(t)=y}) = (e^{tQ})_{xy},
\end{equation}
where the subscripts on $P_x$ and $E_x$ specify $X(0)=x$.
Here $P$ is the probability measure associated with $X$, and $E$ is the corresponding expectation.

The Laplacian $\Delta_{\Zd}$ is a $Q$ matrix and generates the familiar nearest-neighbour
continuous-time simple random walk.  We fix instead $Q=-(-\Delta_{\Zd})^{\alpha/2}$ with
$\alpha \in (0,2)$.
In Section~\ref{sec:MC}, we verify the standard fact that this is indeed a $Q$ matrix
as defined above.  The Markov chain with this generator takes long-range steps, with the probability
of a step from $x$ to $y$ decaying like $|x-y|^{-(d+\alpha)}$.
The Green function $(-Q^{-1})_{xy}$ is finite for $\alpha <d$, and
decays at large distance as $|x-y|^{-(d-\alpha)}$.
We define the finite positive number $\tau^{(\alpha)}$ as the diagonal of the Green function:
\begin{equation}
\lbeq{taualphadef}
    \tau^{(\alpha)} = (-Q^{-1})_{00} = ((-\Delta_{\Zd})^{\alpha/2})^{-1}_{00}.
\end{equation}

The \emph{local time} of $X$ at
$x$ up to time $T$ is the random variable $L_T^x = \int_0^T \1_{X(t)=x} \; dt$.
The \emph{self-intersection local time} up to time $T$ is the random variable
\begin{equation} \label{e:ITdef}
  I_T
  =
  \sum_{x\in\Z^d} \big(L_T^x\big)^2
  = \int_0^T \!\! \int_0^T \1_{X(t_1) = X(t_2)} \; dt_1 \, dt_2
  .
\end{equation}
Given $g>0$ and $\nu \in \R$,
the continuous-time weakly self-avoiding walk
\emph{susceptibility} is defined by
\begin{equation}
  \label{e:suscept-def}
  \chi(g,\nu;0)
  =
  \int_0^\infty E_{0}(e^{-gI_T}) e^{- \nu T} dT.
\end{equation}
The name ``weakly self-avoiding walk'' arises from the fact that the
factor $e^{-gI_T}$ serves to discount trajectories with large self-intersection
local time.
A standard subadditivity argument (a slight adaptation of
\cite[Lemma~A.1]{BBS-saw4-log})
shows that for all dimensions $d \ge 1$
there exists a $d$-dependent \emph{critical value}
$\nu_c = \nu_c(g; 0) \in [-2g\tau^{(\alpha)}, 0]$
such that
\begin{equation}
  \label{e:chi-nuc}
  \text{$\chi(g,\nu;0) < \infty$ \; if and only if \; $\nu > \nu_c$}.
\end{equation}
We are interested in the nature of the divergence of $\chi$ as $\nu \downarrow \nu_c$.

Our notation above reflects the fact that the weakly self-avoiding walk corresponds
to the $n=0$ case of the $n$-component $|\varphi|^4$ model.  Our methods treat
both cases $n \ge 1$ (spins) and $n=0$ (self-avoiding walk) simultaneously, by
using a supersymmetric spin representation for the weakly self-avoiding walk.
This aspect is reviewed in Section~\ref{sec:saw}.

\subsection{Main results}

We consider dimensions $d=1,2,3$; fixed $\epsilon >0$ (small); and
\begin{equation}
    \alpha = \half (d+\epsilon)  .
\end{equation}
In particular, $\alpha$ lies in the interval $(0,2)$.
The upper critical dimension is $d_c= 2\alpha$,
and $d=d_c-\epsilon$ is below the upper critical dimension.
Our main results are given by the following two theorems, which provide statements
consistent with the values of $\gamma,\alpha_H$ in \refeq{gamep}--\refeq{nuep}.
The first theorem
applies to both the spin and self-avoiding walk models, whereas the second applies
only to the spin models.
In the statements of the theorems, and throughout their proofs, the order of choice
of $L$ and $\epsilon$ is that first $L$ is chosen large, and then $\epsilon$ is chosen
small depending on $L$.

\begin{theorem}
\label{thm:suscept}
Let  $n \ge 0$, let $L$ be sufficiently large, and let $\epsilon>0$ be sufficiently
small.
There exists $\gLfix \asymp \epsilon$
such that, for
$g \in [\frac{63}{64}\gLfix,\frac{65}{64}\gLfix]$,
there exist $\nu_c=\nu_c(g;n)$ and $C>0$ such that
for $\nu=\nu_c+t$ with $t \downarrow 0$,
\begin{equation}
\lbeq{chigam}
    C^{-1} t^{-(1 +  \frac{n+2}{n+8} \frac{\epsilon}{\alpha} -C\epsilon^2)}
    \le
    \chi (g,\nu ;n)
    \le
    C t^{-(1 +  \frac{n+2}{n+8} \frac{\epsilon}{\alpha} +C\epsilon^2)}.
\end{equation}
This is a statement that the critical exponent $\gamma$ exists to order $\epsilon$, and
\begin{equation}
    \gamma = 1 +  \frac{n+2}{n+8} \frac{\epsilon}{\alpha} +O(\epsilon^2).
\end{equation}
The critical point obeys (recall \refeq{taualphadef})
\begin{equation}
\lbeq{nucasy}
    \nu_c(g;n)
    = -(n+2) \tau^{(\alpha)} g(1+O(g)).
\end{equation}
\end{theorem}

\begin{theorem}
\label{thm:sh}
Let  $n \ge 1$, let $L$ be sufficiently large, and let $\epsilon>0$ be sufficiently
small.
For $g  \in [\frac{63}{64}\gLfix,\frac{65}{64}\gLfix]$,
and  for $\nu= \nu_c+t$ with $t \downarrow 0$,
\begin{alignat}{2}
\lbeq{cHasy}
    c_H(g,\nu ;n)
    & \asymp
    t^{-  \frac{4-n}{n+8}\frac{\epsilon}{\alpha} +O(\epsilon^2)}
    \quad
    && (n <4),
    \nnb
    c_H(g,\nu ;n)
    & \le O(t^{-O(\epsilon^2)})  && (n=4),
    \\ \nonumber
    c_H(g,\nu ;n)
    & \asymp 1 && (n>4).
\end{alignat}
More explicitly, for $n<4$, \refeq{cHasy} is shorthand for the existence of
$C>0$ such that
\begin{equation}
    C^{-1}t^{-  (\frac{4-n}{n+8}\frac{\epsilon}{\alpha} -C\epsilon^2)}
    \le
    c_H(g,\nu ;n)
    \le
    Ct^{- ( \frac{4-n}{n+8}\frac{\epsilon}{\alpha} +C\epsilon^2)}
    \quad (n<4)
    .
\end{equation}
This is a statement that the critical exponent $\alpha_H$ is
\begin{equation}
    \alpha_H =    \frac{4-n}{n+8}\frac{\epsilon}{\alpha} +O(\epsilon^2) \quad (n<4),
\end{equation}
whereas
the specific heat is at most $Ct^{-C\epsilon^2}$ for $n=4$ and is not divergent
for $n>4$.
\end{theorem}

Mean-field behaviour
has been proved for $d>4$ for the
nearest-neighbour $\varphi^4$ model \cite{Aize82,Froh82,Saka15,FFS92,Soka79}
(e.g., $\gamma=1$, $\alpha_H=0$),
and
for the nearest-neighbour strictly self-avoiding walk \cite{BS85,HS92a}.
For long-range self-avoiding walk (\emph{spread-out} via a small parameter)
in dimensions $d>2\alpha$, it has been proved that $\gamma=1$, that the scaling limit
is an $\alpha$-stable process, in addition to other results \cite{Heyd11,CS11}.
For the nearest-neighbour model in dimension $d=d_c=4$, logarithmic corrections to
mean-field scaling are proved in \cite{BBS-saw4-log,BBS-phi4-log,BSTW-clp,ST-phi4};
the first such result was obtained for the case $n=1$ in \cite{HT87}.
In contrast, Theorems~\ref{thm:suscept}--\ref{thm:sh} study critical behaviour \emph{below}
the upper critical dimension.

\subsection{Organisation}

The proof of Theorems~\ref{thm:suscept}--\ref{thm:sh}
involves several components, some of which are
closely related to components used to analyse the nearest-neighbour model in
dimension 4 \cite{BBS-saw4-log,BBS-phi4-log}, and some of which are new or are
adaptations of methods of \cite{BMS03,MS08}.
We now describe the organisation of the paper, and comment on aspects of the proof.

We begin in Section~\ref{sec:fL} with a review of elementary facts about the fractional
Laplacian, both on $\Zd$ and on the torus $\Lambda_N$.  The renormalisation group
method we apply is based on a finite-range decomposition of the resolvent of the
fractional Laplacian.  Such a decomposition was recently provided in \cite{Mitt16},
and in Section~\ref{sec:frd} we introduce the aspects we need.  Some detailed proofs
of results needed for the finite-range decomposition are deferred to Section~\ref{sec:Cbound},
where in particular an ingredient in \cite{Mitt16} is corrected.

The finite-range decomposition allows expectations such as \refeq{Pdef} to be
evaluated progressively, in a multi-scale analysis.  This is described in
Section~\ref{sec:rg1}, where the first aspects of the renormalisation group method
are explained.  We concentrate our exposition on the case $n \ge 1$, as the case $n=0$
can be handled via minor notational changes using the supersymmetric representation
of the weakly self-avoiding walk outlined in Section~\ref{sec:saw}.
In Section~\ref{sec:rg1}, we note a major simplification here compared to
the nearest-neighbour model for $d=4$:  the monomial
$|\nabla \varphi|^2$ is irrelevant for the renormalisation group flow.
This means that
the coupling constants $z_j,y_j$ used in \cite{BBS-saw4-log,BBS-phi4-log} are
unnecessary, and that there is no need to tune the
wave function renormalisation $z_0$.
Also, the monomial $|\varphi|^4$ is relevant for the renormalisation group flow
in our current setting, whereas it was marginal for $d=4$.  This requires changes
to the analysis for $d=4$.

In Section~\ref{sec:pt}, we develop perturbation theory
and state the second-order perturbative flow equations; these can be
taken from \cite{BBS-phi4-log}.
We also state estimates on
the coefficients appearing in those flow equations, and defer proofs of these estimates
to Section~\ref{sec:Cbound}.  We identify the perturbative value
of the nonzero fixed point $\gLfix$ for the flow of the coupling constant $g_j$
for $|\varphi|^4$.  This $\gLfix$ is the number appearing in the statements of
Theorems~\ref{thm:suscept}--\ref{thm:sh}.  As in \cite{Abde07,BMS03,MS08}, we
must study the deviation of the flow of the coupling constant $g_j$
(coefficient of $|\varphi|^4$) from the fixed point.  This is
a feature that differs from $d=4$, where the fixed point is the Gaussian
one and the analogue of $\gLfix$ is $0$.

In Section~\ref{sec:rg2}, we recall aspects of the nonperturbative renormalisation
group analysis applied in \cite{BBS-saw4-log,BBS-phi4-log}.  We apply
the main result of \cite{BS-rg-step} to handle the nonperturbative analysis,
with adaptation to take into account the new scaling in our present setting.
The norms we use simplify compared to \cite{BBS-phi4-log,BBS-saw4-log},
because it is no longer necessary to include the running coupling constant $g_j$ as
a norm parameter.  This was a serious technical difficulty for $d=4$
because in that case $g_j \to 0$.  Our treatment of scales beyond the so-called
\emph{mass scale} differs from that in \cite{BBS-saw4-log,BBS-phi4-log}
and is inspired by, but is not identical to, the treatment in \cite{BSTW-clp}.

In Section~\ref{sec:rgflow}, we analyse the dynamical system arising from the
renormalisation group.
Our analysis is inspired in part by the corresponding analysis in \cite{BDH98,BMS03,MS08},
but it is done differently and in some aspects more simply, and
it must account for the fact that we work slightly away
from the critical point unlike in those references.
A simplification compared to $d=4$ is that
the dynamical system is hyperbolic, rather than non-hyperbolic as in \cite{BBS-rg-flow}.
On the other hand, the flow now converges to a non-Gaussian fixed point.
It is in Section~\ref{sec:rgflow} that we take the main step in identifying the critical
point $\nu_c$.
The methods of Section~\ref{sec:rgflow} constitute one of the main novelties in the paper.

We transfer the conclusions
from Section~\ref{sec:rgflow} concerning the dynamical system to analyse the flow
equations  and
use this analysis to prove Theorems~\ref{thm:suscept}--\ref{thm:sh}
in Sections~\ref{sec:pfsuscept}--\ref{sec:pfsh}, respectively.
The ideas in Sections~\ref{sec:pfsuscept}--\ref{sec:pfsh} share
many features with \cite{BBS-phi4-log,BBS-saw4-log}, but here it is more delicate.

\subsection{Discussion}

\subsubsection{Speculative extensions}

In the following discussion,
we use ``$\approx$'' to denote uncontrolled approximation in arguments
whose rigorous justification
is not within the current scope of the methods in this paper, but which
nevertheless provide two interpretations of our main results.

Firstly, for $n =1,2,3$,
consider the critical correlation function
$\langle|\varphi_0|^2;|\varphi_x|^2 \rangle_{\nu_c}$.
We argue now that Theorem~\ref{thm:sh} is consistent with
\begin{equation}
\lbeq{eecorr}
    \langle|\varphi_0|^2;|\varphi_x|^2 \rangle_{\nu_c}
    \approx |x|^{-(d-\epsilon \frac{4-n}{n+8})}.
\end{equation}
For $n=1$, this agrees with the scaling in \cite[Conjecture~6]{Abde15},
as the exponent $d-\frac{\epsilon}{3}$ is equal to $2[\varphi^2]$
with $[\varphi^2]=2[\varphi]+\frac{\epsilon}{3}$ and $[\varphi]=\frac 12 (d-\alpha)$.
To obtain \refeq{eecorr}, suppose that
$\langle|\varphi_0|^2;|\varphi_x|^2 \rangle_{\nu_c} \approx |x|^{-d+q}$,
with $q$ to be determined.  Write $\nu_t = \nu_c+t$ with $t>0$, so $\nu_t>\nu_c$.
Then, with $\xi_t = \xi(\nu_t)$ the correlation length, we expect that
\begin{align}
    c_H(\nu_t) &
    = \frac 14 \sum_{x\in\Zd} \langle|\varphi_0|^2;|\varphi_x|^2 \rangle_{\nu_t}
    \approx
    \sum_{|x|\le \xi_t} \langle|\varphi_0|^2;|\varphi_x|^2 \rangle_{\nu_t}
    \nnb
    &
    \approx
    \sum_{|x|\le \xi_t} \langle|\varphi_0|^2;|\varphi_x|^2 \rangle_{\nu_c}
    \approx
    \sum_{1\le |x|\le \xi_t} |x|^{-d+q}
    \approx
    \xi_t^q \approx t^{-\nu q} = t^{-\gamma q/\alpha},
\end{align}
where we inserted $\nu=\gamma/\alpha$ from \refeq{nuep} in
the last step.  This gives $\alpha_H = \gamma q/\alpha$.
With the values of $\gamma$ and $\alpha_H$ from
Theorems~\ref{thm:sh} and \ref{thm:suscept},
this gives, as claimed above,
\begin{equation}
\lbeq{phi2corr}
    q= \frac{4-n}{n+8}\epsilon +O(\epsilon^2)\quad (n =1,2,3).
\end{equation}

Secondly, for $n=0$,
assuming the applicability of Tauberian theory, Theorem~\ref{thm:suscept} is consistent with
\begin{equation}
\lbeq{EgI}
    E_0(e^{-gI_T})
    \approx
    e^{\nu_c T} T^{\frac 14 \frac{\epsilon}{\alpha}+O(\epsilon^2)}
    .
\end{equation}
In addition, assuming again that $\nu=\gamma/\alpha$, we
expect the typical end-to-end distance of the weakly self-avoiding walk to be
given, for $0<p<\alpha$, by
\begin{equation}
\lbeq{distance}
   \left[ \frac{E_0(|X(T)|^p e^{-gI_T})}{E_0( e^{-gI_T})}\right]^{1/p}
    \approx T^\nu
    =
    T^{\frac 1\alpha(1+ \frac 14 \frac{\epsilon}{\alpha})+O(\epsilon^2)}
    .
\end{equation}
Also, assuming that \refeq{eecorr} and \refeq{phi2corr} apply also to $n=0$
leads to the prediction that
\begin{equation}
\lbeq{watermelon}
    \int_0^\infty\int_0^\infty
    E_0[e^{-gI_2(T_1,T_2)}\1_{X^1(T_1)=X^2(T_2)=x}] e^{-\nu_c(T_1+T_2)}dT_1dT_2
    \approx
    |x|^{-d+\frac 12 \epsilon},
\end{equation}
where $X^1$ and $X^2$ are independent Markov chains as in Section~\ref{sec:wsaw}
and $I_2(T_1,T_2) = \sum_x (L_{T_1}^{x}(X_1) +L_{T_2}^{x}(X_2))^2$.
In \cite{ST-phi4},
a detailed analysis of such critical ``watermelon diagrams''
and their relation to critical correlations of field powers like \refeq{eecorr}
is given for the
nearest-neighbour
case when $d=4$.

\subsubsection{Open problems}

It would be of interest to attempt to extend the methods applied here to
the following problems:
\begin{enumerate}
\item
Very recently the $|x|^{-(d-\alpha)}$ decay of the critical two-point function
$\langle \varphi_0 \cdot \varphi_x \rangle_{\nu_c}$ has been proved for $n \ge 1$,
as well as for its $n=0$ counterpart \cite{LSW17}.
This required the introduction of
observables into the renormalisation group analysis presented here.
It would be of interest to
extend this to other critical correlation functions
including $\langle |\varphi_0|^2;|\varphi_x|^2\rangle_{\nu_c}$
discussed in \refeq{eecorr}, and also \refeq{watermelon}.
Such quantities are analysed for $d=4$ in \cite{BBS-saw4,ST-phi4}.

In \cite{PRRZ16}, for $n=1$,
it is argued that appropriately adapted critical correlations
$\langle \varphi_0;\varphi_x^3\rangle_{\nu_c}$ and $\langle \varphi_0^2;\varphi_x^4\rangle_{\nu_c}$
vanish in the long-range model
due to conformal invariance.  These correlations
go beyond what was studied in \cite{ST-phi4} for $d=4$, and the work of \cite{PRRZ16} provides
additional motivation
to investigate such matters rigorously.
\item
Extend the methods of \cite{BSTW-clp} to analyse the
correlation length and confirm the scaling relation $\nu = \gamma/\alpha$.
For the long-range model, there can be no exponential decay
of correlations,
so the correlation length should be studied in terms of $\xi_p$,
the \emph{correlation length of order} $p$ ($0<p<\alpha$), as in \cite{BSTW-clp}.
\item
Study scaling limits of the spin field
for $n \ge 1$.  Work in this direction
was initiated for the nearest-neighbour model
with $d=4$ in \cite{BBS-phi4-log}, but for the long-range
model with $\epsilon>0$ there will be non-Gaussian scaling limits.
\item
Prove \refeq{EgI}--\refeq{distance}.
This needs new ideas even for the nearest-neighbour model
on $\Z^4$, but the difficulties have been overcome for
the 4-dimensional
hierarchical model \cite{BEI92,BI03c,BI03d}.
\item
Study the upper critical dimension, with $\alpha = \frac d2$,
i.e., $\epsilon =0$, for $d=1,2,3$.  The analysis should have much in common with that
used in \cite{BBS-saw4-log,BBS-phi4-log} for the short-range model with $d=4$,
with the simplification that wave function renormalisation will not be required (i.e., $z_0=0$).
\end{enumerate}

\section{Fractional Laplacian}
\label{sec:fL}

The fractional Laplacian is a much-studied object \cite{LSSW16},
particularly in the continuum setting.
Our focus is the discrete setting, and we
review relevant aspects here for arbitrary $d\ge 1$ and $\alpha \in (0,2)$.
We often write $\beta = \frac\alpha 2 \in (0,1)$.

\subsection{Definition and basic properties}
\label{sec:fracLap}

\subsubsection{Definition of fractional Laplacian}

Let $d \ge 1$.
Let $J$ be the $\Zd \times \Zd$ matrix with $J_{xy}=1$ if $|x-y|_1=1$,
and otherwise $J_{xy}=0$.  Let $I$ denote the identity matrix.  The lattice
Laplacian on $\Zd$, with our normalisation, is
\begin{equation}
\lbeq{DeltaJ}
    \Delta = J-2dI.
\end{equation}
There are various equivalent ways to define
the $\Zd\times\Zd$ matrix $(-\Delta)^\beta_{x,y}$, as follows.

\smallskip \noindent
\emph{Fourier transform.}
The matrix element $-\Delta_{xy}$ can be written as a Fourier integral
\begin{equation}
    -\Delta_{x,y} = \frac{1}{(2\pi)^d}\int_{[-\pi,\pi]^d} \lambda(k) e^{ik\cdot (x-y)}dk
\end{equation}
with
\begin{equation}
\lbeq{lambdak}
    \lambda(k) = 4 \sum_{j=1}^d \sin^2 (k_j/2) = 2 \sum_{j=1}^d (1-\cos k_j)
    .
\end{equation}
The matrix $(-\Delta)^\beta$ is defined by
\begin{equation}
\lbeq{QbetaFT}
    (-\Delta)^{\beta}_{x,y}
    =
    \frac{1}{(2\pi)^d}\int_{[-\pi,\pi]^d} \lambda(k)^{\beta} e^{ik\cdot ( x-y)}dk.
\end{equation}

\smallskip \noindent
\emph{Taylor expansion.}  Let $D=\frac{1}{2d}J$.  Then
\begin{equation}
\lbeq{QbetaTaylor}
    (-\Delta)^\beta = (2d)^\beta (I-D)^{\beta}
    =
    (2d)^\beta
    \sum_{n=0}^\infty (-1)^n \binom{\beta}{n} D^n.
\end{equation}
The coefficient $(-1)^{n} \binom{\beta}{n}=\binom{n-1-\beta}{n}$ is negative for $n \ge 1$, and
equals $1$ for $n=0$.  By Stirling's formula,
\begin{equation}
\lbeq{Stir}
    (-1)^{n} \binom{\beta}{n}  \sim -\frac{\beta}{\Gamma(1-\beta)} \frac{1}{n^{1+\beta}}
    \quad \text{as $n \to \infty$}.
\end{equation}
The matrix elements $(D^n)_{x,y}$ are the $n$-step transition probabilities for
discrete-time nearest-neighbour simple random walk on $\Zd$.

\smallskip \noindent
\emph{Stable subordinator.}
Via the change of variables $s=u/t$,
it is immediately seen (apart from the value of the constant) that
\begin{equation}
    t^{\beta}
    =
    \frac{\beta}{\Gamma(1-\beta)} \int_0^\infty (1-e^{-st})s^{-1-\beta} ds
    .
\end{equation}
This explicitly exhibits the L\'evy measure $\frac{\beta}{\Gamma(1-\beta)}s^{-1-\beta} ds$
for the Laplace exponent of
the stable subordinator, i.e., for the Bernstein function
$t \mapsto t^{\beta}$ \cite {SSV12}.
Now put $t=-\Delta$ to get
\cite[p.260 (5)]{Yosi80}
\begin{equation}
\lbeq{heatint}
    (-\Delta)^{\beta} =
    \frac{\beta}{\Gamma(1-\beta)}
    \int_0^\infty (I-e^{s\Delta}) s^{-1-\beta} ds.
\end{equation}
A related formula  \cite[p.260 (4)]{Yosi80}
is
\begin{equation}
\lbeq{Yosi4}
    (-\Delta)^{\beta} =
    \frac{\sin \beta \pi}{\pi}
    \int_0^\infty (-\Delta + s)^{-1}(-\Delta) s^{-1+\beta} ds.
\end{equation}
We do not make use of \refeq{heatint}--\refeq{Yosi4}, though Proposition~\ref{prop:covint}
below bears relation to \refeq{Yosi4}.

\smallskip
The following lemma shows that $-(-\Delta)^\beta_{0,x}$ has
$|x|^{-d-2\beta}$ decay ($|x|$ denotes the Euclidean norm $|x|_2$).
A much more general result can be found in
\cite[Theorem~5.3]{BCT15}, including an asymptotic formula with
precise constant.  We provide a simple proof based on
an estimate for simple random walk.
Let $D=\frac{1}{2d}J$ as above.
For $n$ and $x$ of the same parity, with $0<|x|_\infty \le n$,
the heat kernel estimate
\begin{equation}
\lbeq{hk}
    \frac{c}{n^{d/2}} e^{-|x|^2/cn}
    \le
    (D^n)_{0,x}
    \le \frac{C}{n^{d/2}} e^{-|x|^2/Cn}
\end{equation}
is proved in \cite[Theorem~3.1]{GT02}.
Unlike standard local central limit theorems
which give precise constants
(e.g., \cite{Lawl91}), \refeq{hk} includes
exponential upper and lower bounds
for $x$ well beyond the diffusive scale, e.g., for $|x|\asymp n$.

\begin{lemma}
For $d \ge 1$ and $\beta \in (0,1)$, as $|x| \to \infty$,
\label{lem:fracLapdecay}
    $-(-\Delta)^{\beta}_{0,x} \asymp |x|^{-d-2\beta}$.
\end{lemma}

\begin{proof}
By \refeq{QbetaTaylor}--\refeq{Stir} and \refeq{hk}, it suffices to prove that
\begin{equation}
\lbeq{pnsum}
    \sum_{n=|x|}^\infty \frac{1}{n^{1+\beta+d/2}} e^{-c|x|^2/ n}
    \asymp \frac{1}{|x|^{d+2\beta}}.
\end{equation}
For the lower bound of \refeq{pnsum}, we bound the left-hand side below by
\begin{align}
    e^{-c}\sum_{n=|x|^2}^{2|x|^2} \frac{1}{n^{1+\beta+d/2}}
    &\ge
    c' \frac{1}{|x|^{d+2\beta}}.
\end{align}
For the upper bound, we use (with change of variables $t=s|x|^2$)
\begin{align}
    \sum_{n=|x|}^{\infty} \frac{1}{n^{1+\beta+d/2}} e^{-c|x|^2/n}
    & \le
    C \int_0^\infty \frac{1}{t^{1+\beta + d/2}} e^{-c|x|^2/t}dt
    =
    \frac{1}{|x|^{d+2\beta}}
    C
    \int_0^\infty \frac{1}{s^{1+\beta + d/2}} e^{-c/s}ds.
\end{align}
The integral on the right-hand side is a positive constant, and this completes the proof.
\end{proof}

\subsubsection{Resolvent of fractional Laplacian}

For $m^2\ge 0$, the resolvent of $-\Delta$ is given by
\begin{align}
\lbeq{resol}
    (-\Delta + m^2)_{0,x}^{-1} & =
    \frac{1}{(2\pi)^d}\int_{[-\pi,\pi]^d} \frac{e^{ik\cdot x}}{ \lambda(k) +m^2}dk.
\end{align}
The integral converges for $d \ge 1$ if $m^2>0$, and also for $m^2=0$ when $d>2$.
The resolvent of $(-\Delta)^\beta$ is
\begin{align}
\lbeq{resolvent}
    ((-\Delta)^{\beta} + m^2)_{0,x}^{-1} & =
    \frac{1}{(2\pi)^d}\int_{[-\pi,\pi]^d}
    \frac{e^{ik\cdot x}}{( \lambda(k))^{\beta}+m^2}dk
    ,
\end{align}
where now convergence requires $d>2\beta$ if $m^2=0$.
For the massless case, an asymptotic formula
$((-\Delta)^{\beta})_{0x}^{-1}\sim a_\beta |x|^{-(d-2\beta)}$
is proven in \cite[Theorem~2.4]{BC15}, with precise constant $a_\beta$.
For $m^2>0$, an upper bound
\begin{equation}
        ((-\Delta)^{\beta}+m^2)^{-1}_{0x}
        \le
        c_\beta \frac{1}{|x|^{d-2\beta}} \frac{1}{1+m^4 |x|^{4\beta}}
\end{equation}
is proven in Lemma~\ref{lem:covbd} below.

The next proposition is due to \cite{Kato60} (see also \cite[p.260 (6)]{Yosi80}),
and was rediscovered in \cite{Mitt16}.
Because it plays an essential role in our analysis,
we provide a simple direct proof based
on the following lemma.
For $\beta \in(0,1)$, $a \ge 0$ and $s > 0$, let
\begin{equation}
\lbeq{rhodef}
    \rho^{(\beta)}(s,a) = \frac{\sin \pi\beta}{\pi} \frac{s^\beta}{s^{2\beta}+a^2+2as^\beta \cos\pi\beta}.
\end{equation}
An elementary proof that $\rho^{(\beta)}(s,a) \ge 0$ is given in
\cite[Proposition~2.1]{Mitt16}.

\begin{lemma}
\label{lem:tbetaint}
Let $\beta \in(0,1)$, $t\ge 0$ and $a\ge 0$, excepting $t= a=0$.
Then
\begin{equation}
\lbeq{trho}
    \frac{1}{t^\beta + a} = \int_0^\infty \frac{1}{s+t} \; \rho^{(\beta)}(s,a)  ds.
\end{equation}
\end{lemma}

\begin{proof}
We first consider $t>0$ and $a \ge 0$.
Let $C$ be a simple closed contour  that
encloses $t$ in the cut plane $\C \setminus (-\infty,0]$,
oriented counterclockwise.  We define $z^\beta$ to be the branch given by
$z^\beta= r^\beta e^{i\beta\theta}$, for $z=re^{i\theta}$
with $\theta \in [-\pi,\pi)$.  By the
Cauchy integral formula,
\begin{equation}
    \frac{1}{t^\beta + a}
    =
    \frac{1}{2\pi i} \oint_C  \frac{1}{z-t} \frac{1}{z^\beta + a} dz,
\end{equation}
since $z^\beta + a$ has no zero inside $C$ (for $a>0$ and $z=re^{i\theta}$,
a zero requires
$e^{i\theta\beta}=-1$ which cannot happen for $\beta \in (0,1)$
and $\theta \in [-\pi,\pi)$).

Now we deform the contour to a keyhole contour around the branch cut.
We shrink the small circle at the origin, and send the big circle to infinity;
the contributions from both circles vanish in the limit since $\beta \in(0,1)$,
$t>0$, and $a\ge 0$.
The contributions from the branch cut give (after change of sign in the integrals)
\begin{align}
    \frac{1}{t^\beta + a}
    & =
    \frac{1}{2\pi i} \int_0^\infty ds \frac{1}{s+t}
    \left( \frac{1}{e^{-i\pi\beta}s^\beta + a} -  \frac{1}{e^{i\pi\beta}s^\beta + a}  \right)
    .
\end{align}
After algebraic manipulation this gives \refeq{trho},
and the proof is complete for $t>0$.

Finally, \refeq{trho} follows immediately for $t=0$, when $a>0$, by
letting $t \downarrow 0$ in \refeq{trho} and applying monotone convergence.
\end{proof}

\begin{prop}
\label{prop:covint}
For $d\ge 1$, if $m^2 > 0$ and $\beta \in (0,1)$, or if
$m^2 = 0$ and $\beta \in (0,1 \wedge \frac d2)$, then
\begin{equation}
\lbeq{resint-bis}
    ((-\Delta)^{\beta}+m^2)^{-1}_{0,x}
    =
    \int_0^\infty  (-\Delta+s)^{-1}_{0,x} \; \rho^{(\beta)}(s,m^2)\, ds .
\end{equation}
\end{prop}

\begin{proof}
Note that the right-hand side of \refeq{resint-bis} only involves $s>0$,
for which $(-\Delta+s)^{-1}$ is well-defined in all dimensions.
By \refeq{resolvent} and Lemma~\ref{lem:tbetaint},
\begin{equation}
    ((-\Delta)^{\beta} + m^2)_{0,x}^{-1} =
    \frac{1}{(2\pi)^d}\int_{[-\pi,\pi]^d}
    dk \;
    e^{ik\cdot x}
    \int_0^\infty ds \frac{1}{\lambda(k)+s} \; \rho^{(\beta)}(s,m^2)  .
\end{equation}
Then we apply Fubini's Theorem and
\refeq{resol} to obtain \refeq{resint-bis}.
\end{proof}

Let $\1: \Zd \to \R$ denote the constant function $\1_x=1$.
For future reference, we observe that it follows from
Proposition~\ref{prop:covint} and Lemma~\ref{lem:tbetaint} that
\begin{align}
\lbeq{C1}
     ((-\Delta)^{\beta} + m^2)^{-1}\1
     & =
     \int_0^\infty  (-\Delta+s)^{-1}\1 \; \rho^{(\beta)}(s,m^2) \, ds
    \nnb & =  \int_0^\infty  s^{-1}\1\; \rho^{(\beta)}(s,m^2)\, ds = m^{-2}\1.
\end{align}

\subsubsection{The bubble diagram}
\label{sec:bubble}

Let $\alpha \in (0,2)$.
The (free) \emph{bubble diagram} is defined by
\begin{equation}
\lbeq{bubdef}
    B_{m^2}=\sum_{x \in \Zd}  \Big[((-\Delta)^{\alpha/2}+m^2)^{-1}_{0,x} \Big]^2
    .
\end{equation}
By the Parseval relation and \refeq{resolvent}, the bubble diagram is also given by
\begin{equation}
\lbeq{bubk}
    B_{m^2}=
    \frac{1}{(2\pi)^d}\int_{[-\pi,\pi]^d}
    \frac{1}{[(\lambda(k))^{\alpha/2}+m^2]^2}dk
    .
\end{equation}
The bubble diagram is finite in all dimensions when $m^2>0$.  It is infinite
for $d \le  2\alpha$ when $m^2=0$, due to the singularity $|k|^{-2\alpha}$ of
the integrand.

It is the divergence of the massless bubble diagram
that identifies $d_c=2\alpha$ as the upper critical dimension
\cite{AF88,HHS08,Heyd11,CS08,CS11,CS15}, and the rate of divergence of the bubble
diagram for $d=d_c-\epsilon$ plays a role in the determination
of the critical exponents in
Theorems~\ref{thm:suscept}--\ref{thm:sh}.
Since the singularity at $k=0$ determines the leading behaviour, for $d<2\alpha$
we have (using $r=tm^{2/\alpha}$)
\begin{align}
    B_{m^2} & \sim
     \frac{1}{(2\pi)^d}
    \int_{|k|<1} \frac{1}{(|k|^\alpha + m^2)^2} d k
    =
    C_d
    \int_0^1 \frac{1}{(r^\alpha + m^2)^2}r^{d-1}dr
    \nnb & =
    C_d
    m^{-4+2d/\alpha} \int_0^{m^{-2/\alpha}} \frac{1}{(t^\alpha + 1)^2}t^{d-1}dt
    \nnb &
    \sim b_\epsilon
    m^{-2\epsilon /\alpha}
    \quad
    \text{as $m^2 \downarrow 0$}
    ,
\lbeq{bubm}
\end{align}
with
\begin{equation}
    b_\epsilon = C_d \int_0^\infty \frac{1}{(t^\alpha + 1)^2}t^{d-1}dt.
\end{equation}
Note that $b_\epsilon \asymp \epsilon^{-1}$ as $\epsilon \downarrow 0$,
due to the decay $t^{d-1-2\alpha}=t^{-1-\epsilon}$ of the integrand as $t \to \infty$.

\subsection{Continuous-time Markov chains}
\label{sec:MC}

We now prove that
$-(-\Delta)^\beta$ has the properties required of a generator of a Markov chain
on $\Zd$.  We also consider related issues on the torus $\Lambda_N = \Zd/L^N\Zd$.

\subsubsection{Markov chain on \texorpdfstring{$\Zd$}{Zd}}
\label{sec:rwrep}

Recall \refeq{DeltaJ}.
The matrix $\Delta=\Delta_{\Zd}=J-2dI$
obeys $\Delta_{x,x}<0$, $\Delta_{x,y}\ge 0$ if $x \neq y$,
and $\sum_y \Delta_{x,y}=0$.  Thus $\Delta$ is the generator of a
continuous-time Markov chain, namely the continuous-time nearest-neighbour
simple random walk on $\Zd$.  The following lemma shows that $-(-\Delta)^\beta$ also generates
a Markov chain on $\Zd$.  By Lemma~\ref{lem:fracLapdecay}, this Markov chain takes long-range steps.

\begin{lemma}
\label{lem:DJ}
For $d \ge 1$ and $\beta \in (0,1)$, $(-\Delta)^\beta_{x,x}>0$, $(-\Delta)^\beta_{x,y} < 0$
if $x \neq y$, and $\sum_y (-\Delta)^\beta_{x,y}=0$.
\end{lemma}

\begin{proof}
It is clear from \refeq{QbetaFT}
and the nonnegativity of $\lambda(k)$ that $(-\Delta)^\beta_{x,x}>0$.
To see that $(-\Delta)^\beta_{x,y}< 0$
if $x \neq y$,
we evaluate \refeq{QbetaTaylor} at $x,y$ and note that only terms with $n \geq 1$ contribute,
and these terms are all nonpositive and not all are zero.
Finally, again from \refeq{QbetaTaylor} we obtain
\begin{equation}
    \sum_{y\in\Zd} (-\Delta)^\beta_{x,y}
    = \sum_{n=0}^\infty (-1)^n \binom{\beta}{n} = (1-1)^\beta =0.
\end{equation}
This completes the proof.
\end{proof}

\subsubsection{Markov chain on torus}
\label{sec:rwtorus}

We approximate $\Z^d$ by a sequence of finite
tori of period $L^N$.  The torus $\Lambda=\Lambda_N = \Z^d/L^N\Z^d$ is defined
as a quotient space, with canonical projection $\Z^d   \rightarrow \Lambda_N$.
The torus Laplacian
$\Delta_{\Lambda_N}$ is defined by
\begin{equation}
    (\Delta_{\Lambda_N})_{x,y} = \sum_{z \in \Zd} (\Delta_{\Zd})_{x,y+zL^N}
  \quad
  (x,y\in\Lambda_N),
\end{equation}
where on the right-hand side $x,y$ are any fixed representatives in $\Zd$ of the torus points.
The torus Laplacian is the generator for simple random walk on the torus.

Similarly, the canonical projection induces a Markov chain on
$\Lambda_N$ with generator given by
\begin{equation}
\lbeq{betaN}
    -(-\Delta_{\Lambda_N})^\beta_{x,y}
    =
     -\sum_{z \in \Zd} (-\Delta_{\Zd})^\beta_{x,y+zL^N}
  \quad
  (x,y\in\Lambda_N).
\end{equation}
Summability of the right-hand side is guaranteed by Lemma~\ref{lem:fracLapdecay}.
The fact that $-(-\Delta_{\Lambda_N})^\beta$ is indeed a generator can be concluded
from \refeq{betaN} and Lemma~\ref{lem:DJ}.

Let $E^{N}_x$ denote expectation for this
Markov chain $X^N$ on $\Lambda_N$, started from $x \in \Lambda_N$.
A coupling of the Markov chains $X^N$ on $\Lambda_N$ for all $N$ is provided by the
Markov chain $X$ on $\Zd$ with generator $-(-\Delta_{\Zd})^\beta$:
the image $X^N$ of $X$ under the canonical projection $\Z^d
\rightarrow \Lambda_N$ has the distribution of the torus chain.
This fact is used in our discussion of the supersymmetric representation
for $n=0$, in Section~\ref{sec:ir}.

\subsubsection{Torus resolvents}

By Lemma~\ref{lem:matrix} below (with $T=-\Delta+m^2$ and $T=(-\Delta)^\beta +m^2$),
the torus resolvents for the Laplacian and fractional Laplacian are the inverse matrices
given by
\begin{align}
\lbeq{torcov}
  (-\Delta_{\Lambda_N} + m^2)^{-1}_{x,y}
  &=
  \sum_{z \in \Z^d} (-\Delta_{\Z^d}+m^2)^{-1}_{x,y+zL^N}
  \quad
  (x,y\in\Lambda_N),
  \\
  ((-\Delta_{\Lambda_N})^\beta + m^2)^{-1}_{x,y}
  &=
  \sum_{z \in \Z^d} ((-\Delta_{\Z^d})^\beta+m^2)^{-1}_{x,y+zL^N}
  \quad
  (x,y\in\Lambda_N)  .
\end{align}
By Proposition~\ref{prop:covint},
for $d\ge 1$, $m^2 \ge 0$ and $\beta \in (0,1\wedge \frac d2)$, it then follows that
\begin{equation}
\lbeq{resint-Lambda}
    ((-\Delta_\Lambda)^{\beta}+m^2)^{-1}_{0,x}
    =
    \int_0^\infty (-\Delta_\Lambda+s)^{-1}_{0,x} \; \rho^{(\beta)}(s,m^2) \, ds.
\end{equation}

In the statement and proof of
the following elementary lemma, we write $x \sim y$ for $x,y\in\Zd$ with $y-x \in L^N\Zd$.

\begin{lemma}
\label{lem:matrix}
Let $T=(T_{x',y'})_{x',y'\in \Zd}$ be a matrix $T:\ell^\infty(\Zd) \to \ell^\infty(\Zd)$
satisfying $T_{x'+z',y'+z'}=T_{x',y'}$
for all $x',y',z' \in \Zd$, with inverse matrix $T^{-1}:\ell^\infty(\Zd) \to \ell^\infty(\Zd)$.
Define $(\hat T_{x,y})_{x,y\in \Lambda}$
by $\hat T_{x,y} = \sum_{y' \sim y} T_{x,y'}$ (on the right-hand side we choose representatives
in $\Zd$ for $x,y \in \Lambda$).  Then $\hat T$ has inverse matrix
$\hat T^{-1}_{x,y} = \sum_{y' \sim y} T^{-1}_{x,y'}$.
\end{lemma}

\begin{proof}
The assumed translation invariance for $T$ implies the same for $T^{-1}$.
Let $\hat S_{y,z} = \sum_{z' \sim z} T^{-1}_{y,z'}$.
By definition, and by translation invariance (in second equality),
for $x,z \in \Lambda$ we have
\begin{align}
    \sum_{y\in\Lambda} \hat T_{x,y} \hat S_{y,z}
    & =
    \sum_{y\in\Lambda} \sum_{y'\sim y} T_{x,y'}
    \sum_{z' \sim z}  T_{y,z'}^{-1}
    =
    \sum_{y\in\Lambda} \sum_{y'\sim y} T_{x,y'}
    \sum_{z'' \sim z}  T_{y',z''}^{-1}
    =
    \sum_{z'' \sim z} \delta_{x,z''}
    =\delta_{x,z},
\end{align}
which verifies that $\hat S_{y,z}$
is indeed the inverse matrix for $\hat T$.
\end{proof}

\section{Finite-range covariance decomposition}
\label{sec:frd}

In this section, we recall the covariance decomposition for the fractional
Laplacian from \cite{Mitt16}.  We use this to identify which monomials are
\emph{relevant} in the sense of the renormalisation group, and define the field's
scaling dimension.

\subsection{Covariance decomposition for  Laplacian}

We begin with the finite-range decomposition
\begin{equation}
\lbeq{Gammadecomp}
    \Gamma = (-\Delta_\Zd +s)^{-1} = \sum_{j=1}^\infty \Gamma_j
\end{equation}
obtained in \cite{Baue13a} (see also \cite{BBS-brief};
an alternate decomposition is given in \cite{BGM04}).
We review some aspects of the decomposition in Section~\ref{sec:Cbound}.
Each $\Gamma_j$ is a positive semi-definite $\Zd \times \Zd$ matrix,
has the finite-range property
\begin{equation}
\lbeq{Cjfinran}
    \Gamma_{j;x,y} = 0 \quad \text{if $|x-y| \geq \frac{1}{2} L^j$},
\end{equation}
and obeys certain regularity properties.
The decomposition is valid for $d>2$ when $s\ge 0$, but requires $s>0$ for $d\le 2$.
We refer to $j$ as the \emph{scale}.

As in \refeq{torcov}, the torus covariance is
\begin{equation}
  (-\Delta_\Lambda + s)^{-1}_{x,y}
  =
  \sum_{z \in \Z^d} (-\Delta_{\Z^d}+s)^{-1}_{x,y+zL^N}
  \quad
  (x,y\in\Lambda).
\end{equation}
By \refeq{Cjfinran}, $\Gamma_{j;x,y+L^Nz} = 0$ if $j<N$, $|x-y| < L^N$,
and if $z \in \Z^d$ is nonzero,
and thus
\begin{equation}
     \Gamma_{j;x,y} = \sum_{z \in \Z^d} \Gamma_{j;x,y+zL^N} \quad \text{for $j<N$}.
\end{equation}
We can therefore regard $\Gamma_{j}$ as either a $\Zd\times\Zd$
or a $\Lambda_N \times \Lambda_N$ matrix if $j<N$.
We also define
\begin{equation} \label{e:CNNdef}
     \Gamma_{N,N;x,y} = \sum_{z \in \Z^d} \sum_{j=N}^\infty \Gamma_{j;x,y+zL^N}.
\end{equation}
It follows that
\begin{equation}
\lbeq{torusdecomp}
    (-\Delta_\Lambda + s)^{-1}
    = \sum_{j=1}^{N-1} \Gamma_{j} + \Gamma_{N,N}.
\end{equation}
Since $\Gamma_j$ serves as a term
in the decomposition of the $\Zd$ covariance as well as in the torus covariance
when $j<N$, the effect of the torus in
the finite-range decomposition of $(-\Delta_\Lambda + s)^{-1}$
is concentrated in the term $\Gamma_{N,N}$.

The matrices $\Gamma_j$ and $\Gamma_{N,N}$ are Euclidean invariant on $\Lambda$, i.e.,
obey $\Gamma_{Ex,Ey}=\Gamma_{x,y}$ for every graph automorphism $E:\Lambda \to \Lambda$
(with $\Lambda$ considered as a graph with nearest-neighbour edges).

\subsection{Covariance decomposition for fractional Laplacian}

For $d\ge 1$, for $\alpha \in (0,2\wedge d)$, and for $m^2 \ge 0$,
we consider the covariance on $\Lambda=\Lambda_N$ given by
\begin{equation}
    C = ( (-\Delta_\Lambda)^{\alpha/2} +m^2)^{-1}.
\end{equation}
By
\refeq{resint-Lambda},
\begin{equation}
\lbeq{resint}
    C_{0,x}
    =
    \int_0^\infty (-\Delta_\Lambda+s)^{-1}_{0x} \; \rho^{(\alpha/2)}(s,m^2) \, ds.
\end{equation}
As in \cite{Mitt16}, we obtain
a finite-range positive-definite covariance decomposition
by inserting \refeq{torusdecomp} into \refeq{resint}, namely
\begin{equation}
\lbeq{Calphadecomp}
    ((-\Delta_\Lambda)^{\alpha/2} + m^2)^{-1} = \sum_{j=1}^{N-1} C_{j} + C_{N,N},
\end{equation}
with
\begin{equation}
\lbeq{CGamint}
    C_{j;0,x} =
    \int_0^\infty  \Gamma_{j;0,x}(s) \; \rho^{(\alpha/2)}(s,m^2) \, ds,
    \quad
    C_{N,N;0,x} =
    \int_0^\infty  \Gamma_{N,N;0,x}(s) \; \rho^{(\alpha/2)}(s,m^2) \, ds.
\end{equation}
This is valid whenever we have a decomposition \refeq{Gammadecomp} for strictly
positive $s>0$, i.e., for all $d \ge 1$.
Again it is the case that $C_j$ serves as a term
in the decomposition of the $\Zd$ covariance as well as in the torus covariance
when $j<N$, and again the effect of the torus is concentrated in the term $C_{N,N}$.
To simplify the notation, we sometimes write
$C_N$ instead of the more careful $C_{N,N}$.

\subsection{Estimates on decomposition for fractional Laplacian}

The following proposition provides estimates on the terms in the covariance
decomposition \refeq{Calphadecomp}.
A version of \refeq{scaling-estimate} is stated in \cite{Mitt16}.
We defer the proof to Section~\ref{sec:Cbound}, where a somewhat
stronger statement than Proposition~\ref{prop:Cbound} is proved.

Derivatives estimates use
multi-indices $\multia$ which record the number of forward and backward discrete gradients
applied in each component of $x$ and $y$, and
we write $|\multia|$ for the total number of derivatives.

\begin{prop}
\label{prop:Cbound}
Let $d \ge 1$, $\alpha \in (0,2 \wedge d)$,
$L \ge 2$, $\bar m^2>0$, $m^2 \in [0,\bar m^2]$,
and let $a$ be a multi-index with $|\multia| \le \bar a$.
Let $j \ge 1$ for $\Zd$, and let
$1 \le j <N$ for $\Lambda_N$.
The covariance
$C_j=C_j(m^2)$ has range $\frac 12 L^j$, i.e.,
$C_{j;x,y}=0$ if $|x-y|\ge \frac 12 L^j$; $C_{j;x,y}$ is continuous in
$m^2 \in [0,\bar m^2]$; and
\begin{equation}
\lbeq{scaling-estimate}
    | \nabla^\multia  C_{j;x,y}|
    \le
    c L^{-(d-\alpha+|\multia|)(j-1)}
    \frac{1}{1+m^4 L^{2\alpha (j-1)}}
    ,
\end{equation}
where $\nabla^\multia$ can act on either $x$ or $y$ or both.
For $m^2 \in (0,\bar m^2]$,
\begin{align}
\lbeq{CNNbd}
    |\nabla^\multia  C_{N,N;x,y}|
    &\le
    c
    L^{-(d-\alpha+|\multia|)(N-1)}
    \frac{1}{(m^2L^{\alpha (N-1)})^{2}}.
\end{align}
The constant $c$ may depend on $\bar m^2,\bar a$, but does not depend on
$m^2, L,j,N$.
\end{prop}

From \refeq{resint} and Proposition~\ref{prop:Cbound},
we obtain a bound on the full covariance on $\Zd$, in the following lemma.
For fixed $m>0$, the lemma implies an upper bound
$O(m^{-4}|x|^{-(d+\alpha)})$, which has best possible power of $|x|$.

\begin{lemma}
For $d \ge 1$, $\alpha \in (0,2\wedge d)$, $\bar m^2>0$,
 $m^2 \in [0,\bar m^2]$, and $x \neq 0$,
\label{lem:covbd}
\begin{equation}
        ((-\Delta_{\Zd})^{\alpha/2}+m^2)^{-1}_{0,x}
        \le
        c \frac{1}{|x|^{d-\alpha}} \frac{1}{1+m^4|x|^{2\alpha}},
\end{equation}
with $c$ depending on $\bar m^2$.
\end{lemma}

\begin{proof}
For the proof, we take $L=3$ (this arbitrary choice shows that $c$ is
independent of $L$).
Given $x \in \Zd$, let $j_x$ be the
nonnegative integer
for which $\frac 12 L^{j_x} \le |x| < \frac 12 L^{j_x+1}$.
By Proposition~\ref{prop:Cbound}, $C_{j;0,x}=0$ when $j \le j_x$.
By \refeq{scaling-estimate}, with a constant $c$ that may change from one occurrence to the next,
\begin{align}
        ((-\Delta_{\Zd})^{\alpha/2}+m^2)^{-1}_{0,x}
        & \le
        c \sum_{j = j_x+1}^\infty
        (1+m^4 L^{2\alpha (j-1)})^{-1} L^{-(d-\alpha)(j-1)}
        \nnb &
        \le c
        \frac{1}{1+m^4 L^{2\alpha j_x}} L^{-(d-\alpha)j_x}
        \le c
         \frac{1}{|x|^{d-\alpha}} \frac{1}{1+m^4|x|^{2\alpha} }.
\end{align}
This completes the proof.
\end{proof}

\subsection{Field dimension and relevant monomials}

As a guideline, the typical size of a Gaussian field $\varphi_x$, where $\varphi$
has covariance $C_j$,
can be regarded as the square root of $C_{j;x,x}$.  In view of \refeq{scaling-estimate},
we therefore roughly expect
\begin{equation}
\lbeq{phisize0}
    |\varphi_x|\approx \frac{1}{1+m^2L^{\alpha (j-1)}} L^{-\frac 12 (d-\alpha)(j-1)}.
\end{equation}
The first factor on the right-hand side is insignificant for scales $j$ that are small
enough that $m^2L^{\alpha j}$ is small,
but acquires importance for large scales.
We define the
\emph{mass scale}
as the smallest scale $j_m=j_m(L)$ for which $m^2L^{\alpha (j_m-1)} \ge 1$,
namely,
\begin{equation}
\lbeq{jmdef}
    j_m = \lceil f_m \rceil, \quad\quad
    f_m = 1+ \frac{1}{\alpha} \log_L m^{-2} .
\end{equation}

For scales $j > j_m$, we have
(with $x_+=\min\{x,0\}$)
\begin{align}
\lbeq{cbetabd}
    \frac{1}{1+m^2L^{\alpha (j-1)}}
    &
    =
    \frac{1}{1+m^2L^{\alpha (j_m-1)}L^{\alpha (j- j_m)}}
    \le
    L^{-\alpha (j- j_m)_+}
    ,
\end{align}
and the same bound holds trivially for $j \le j_m$ since the left-hand side is
bounded above by $1$.
In several recent papers, e.g., \cite{BBS-saw4-log,BBS-phi4-log}, the additional
decay beyond the mass scale has been utilised only to a lesser extent than \refeq{cbetabd},
with $L$ on the right-hand side replaced by $2$.
We follow the insight raised in
\cite{BSTW-clp} that there is value in retaining more of this decay.
We reserve a portion of the additional decay beyond the mass scale, and use as
guiding principle that
\begin{equation}
\lbeq{phisize}
    |\varphi_x| \lesssim
    L^{-\frac 12 (d-\alpha)(j-1)} L^{-\alphahat (j- j_m)_+}
    ,
\end{equation}
where we are free to choose $\alphahat \in [0,\alpha]$.
We define
$\alpha'$ by
\begin{equation}
\lbeq{alphapdef}
    \alpha' = 2\alphahat - \alpha,
    \quad
    \alphahat = \half (\alpha+\alpha').
\end{equation}
We then define
\begin{align}
\lbeq{elldefa}
    \ell_j &
    = \ell_0 L^{-\frac 12 (d-\alpha)j} L^{- \alphahat (j-j_m)_+}
    =
    \begin{cases}
    \ell_0 L^{-\frac 12 (d-\alpha)j} & (j \le j_m)
    \\
    \ell_0 L^{-\frac 12 (d-\alpha)j_m}L^{-\frac 12 (d+\alpha')(j-j_m)} & (j > j_m),
    \end{cases}
\end{align}
where $\ell_0$ can be chosen (large depending on $L$).  We consider $\ell_j$ as
an approximate measure of (an upper bound on) the size of a typical Gaussian field with covariance $C_j$.

\begin{rk}
\label{rk:alphas}
We will require the restrictions
\begin{equation}
\lbeq{alphapbds}
     \alpha' \in (0, \half \alpha),
     \quad
     \alphahat \in (\textstyle{\frac 12} \alpha   , \textstyle{\frac 34} \alpha ).
\end{equation}
In particular, $\alphahat > \half \alpha > \alpha'$.
\end{rk}

\begin{defn}
(i)
We define the \emph{scaling dimension} or \emph{engineering dimension} $[\varphi]$
of the field as
the power of $L$ gained in $\ell_j$ when the scale
is advanced from $j-1$ to $j$, namely
\begin{align}
\lbeq{dimphi}
    [\varphi] =[\varphi]_j & =
    \begin{cases}
        \frac{d-\alpha}{2} & (j \le j_m)
    \\
    \frac{d+\alpha'}{2}    & (j >j_m).
    \end{cases}
\end{align}
\\
(ii)
A \emph{local field monomial} (located at $x$) has the form
\begin{equation}
\lbeq{Mx}
  M_x = \prod_{k=1}^m \nabla^{\multia_k} \varphi_{x}^{i_k}
\end{equation}
for some integer $m$, where $\multia_k$ are multi-indices and $i_k\in \{1,\ldots,n\}$ indicates
a component of $\varphi_x \in \R^n$.
The \emph{dimension} of $M_x$
is defined to be
$[M_x] =[M_x]_j = \sum_{k=1}^m ([\varphi]_j + |\multia_k|)$, with $[\varphi]_j$ given by \refeq{dimphi}.
We include the case of the empty product in \refeq{Mx}, which defines the constant
monomial $1$, of dimension zero.
\\
(iii)
A local field monomial is said to be
\emph{relevant} if $[M_x]_j < d$,
\emph{marginal} if $[M_x]_j = d$, and
\emph{irrelevant} if $[M_x]_j > d$.
\end{defn}

\begin{table}
\caption{Dimensions of monomials.}
\label{tab:mondim}
\vspace{0.1cm}
\begin{center}
\begin{tabular}{c|ccccc} & \multicolumn{1}{c}{$[1]$} &  \multicolumn{1}{c}{$[\varphi^2]$}
&  \multicolumn{1}{c}{$[\varphi^4]$} &  \multicolumn{1}{c}{$[\varphi^6]$} &  \multicolumn{1}{c}{$[\nabla^2\varphi^2]$}
\\ \hline
$j \le j_m$ & $0$& $d-\alpha$& $d-\epsilon$& $\tfrac{3}{2}(d-\epsilon)$& $d+2-\alpha$\\
$j>j_m$  & $0$& $d+\alpha'$ & $2(d+\alpha')$& $3(d+\alpha')$& $d+2+\alpha'$\\
\end{tabular} \end{center} \normalsize \end{table}

Symmetry considerations preclude the occurrence of monomials with an odd number of
fields or an odd number of gradients.  For the symmetric cases, the dimensions are
given in Table~\ref{tab:mondim}.
The monomials $\varphi^2$ and $\varphi^4$ are relevant below the mass scale and irrelevant above the
mass scale.  Higher powers of $\varphi$ are irrelevant at all scales, and the constant
monomial $1$ is relevant at all scales.  In summary:
\begin{align*}
    & 1, \; |\varphi|^2, \; |\varphi|^4 \;\;\text{are relevant for $j \le j_m$},
    \\
    & 1 \;\;\text{is relevant for $j > j_m$}.
\end{align*}
The monomial $\nabla^2\varphi^2$ is irrelevant; this is a major simplification
compared to the nearest-neighbour model for
$d=4$, where it is marginal \cite{BBS-phi4-log,BBS-saw4-log}.

The effect of relevant monomials is best measured via a sum over a block of side $L^j$,
consisting of $L^{dj}$ points.  With the field regarded as having typical size $\ell_j$
given by \refeq{elldefa},
below the mass scale the relevant monomials $1$, $|\varphi|^2$, $|\varphi|^4$
on a block have size
given by $L^{dj}\ell_j^{p}$, for $p=0,2,4$.
For the monomial $1$ this grows like $L^{dj}$, for $|\varphi|^2$
it is $L^{\alpha j}$, and for $|\varphi|^4$ it is $L^{\epsilon j}$.
The growth of the monomial $1$ is not problematic.  The growth of $|\varphi|^2$
and $|\varphi|^4$ is however potentially problematic, and
will be shown to be compensated by multiplication by coupling
constants $\nu_j$ and $g_j$ (respectively) which behave
as $\nu_j \approx \epsilon L^{-\alpha(j\wedge j_m)}$ and
$g_j \approx \epsilon L^{-\epsilon (j\wedge j_m)}$.  This cancels the growth of $|\varphi|^2$
and $|\varphi|^4$ up to the mass scale.  After the mass scale, the coupling constants stabilise,
which is connected with the fact that the renormalisation group fixed point is non-Gaussian.
Their products with the monomials are then controlled instead by the
additional decay in $\ell_j$ for $j>j_m$.  It is for this purpose that we exploit the
additional decay in $\ell_j$.

\section{First aspects of the renormalisation group method}
\label{sec:rg1}

In this section, we introduce some of the basic ingredients of the renormalisation
group analysis, including perturbation theory.

Some preparation is required in order to formulate the weakly self-avoiding walk
model as the infinite volume limit of a supersymmetric
version of the $|\varphi|^4$ spin model, which involves a complex boson field
$\phi,\phib$ and
a fermion field given by the 1-forms $\psi_x = \frac{1}{\sqrt{2\pi i}}d\phi_x$,
$\psib_x = \frac{1}{\sqrt{2\pi i}}d\phib_x$.
This is discussed in Section~\ref{sec:saw}, and for the nearest-neighbour model
it is addressed in detail in \cite{BBS-saw4-log}.
Our analysis applies
equally well to the supersymmetric model with minor notational changes,
with $n$ interpreted as $n=0$, and with Gaussian expectations replaced by
superexpectations; see \refeq{supersym}.
For notational simplicity, we focus our presentation
on the case $n \ge 1$.
We only consider fields on the torus $\Lambda=\Lambda_N$, and ultimately
we will be interested in the limit $N \to \infty$.

\subsection{Progressive integration}
\label{sec:pi}

For the $n$-component $|\varphi|^4$  model with $n \ge 1$, or for the weakly
self-avoiding walk ($n=0$), we define
\begin{equation}
\lbeq{tauxdef}
    \tau_x =
    \begin{cases}
        \half |\varphi_x|^2 & (n\ge 1)
        \\
        \phi_x\phib_x + \psi_x \wedge \psib_x & (n=0).
    \end{cases}
\end{equation}
The general Euclidean- and $O(n)$-invariant local
polynomial consisting of relevant monomials is, for $j \le j_m$,
\begin{equation}
    \label{e:Vdef-bis}
    U(\varphi_x) = g\tau_x^2 + \nu\tau_x +  u
    .
\end{equation}
There are no marginal monomials.
Above the mass scale $j_m$, the monomials $\tau$ and $\tau^2$ become irrelevant, but
we nevertheless retain $\tau$ and reduce  to $U$ of the form $\nu\tau_x +  u$
(a reason for retaining $\tau$ is given in Remark~\ref{rk:tauLoc}).
For $n=0$ and for all scales $j$, we can take $u=0$ due to supersymmetry (see \cite{BBS-rg-pt}).
For $U$ as in \refeq{Vdef-bis}, and for $X \subset \Lambda$, we write
\begin{equation}
  \label{e:Vdef-bisbis}
  U(X,\varphi) = \sum_{x \in X}U(\varphi_x)
  .
\end{equation}
For notational simplicity, we often write $U(X)$ instead of $U(X,\varphi)$.

Given $m^2>0$, let
\begin{equation} \label{e:g0gnu0nu}
  g_0 = g , \quad \nu_0 =  \nu-m^2
  ,
\end{equation}
and define
\begin{equation}
\lbeq{V0Z0}
    V_0(\varphi_x) =  g_0\tau_x^2+  \nu_0\tau_x ,
     \quad
  Z_0(\varphi) = e^{-V_{0}( \Lambda_N)}
  .
\end{equation}
For $n \ge 1$, given a $\Lambda \times \Lambda$ covariance matrix $C$,
let $P_C$ denote the Gaussian probability measure on $(\R^n)^\Lambda$
with covariance $C$. This means that $P_C$ is proportional to
$e^{-\frac 12 \sum_{i=1}^n \sum_{x,y\in \Lambda} \varphi^i_x C^{-1}_{xy}\varphi^i_y}
\prod_{x\in \Lambda}d\varphi_x$ (properly interpreted when $C$ is only positive semi-definite
rather than positive-definite).
Let $\Ex_C$ denote the corresponding expectation.
For $n=0$, $\Ex_C$ denotes the superexpectation \refeq{superex}.
With $m^2>0$ and $C=((-\Delta_{\Lambda_N})^{\alpha/2} +m^2)^{-1}$, we can
rewrite the expectation \refeq{Pdef} (with $M$ the fractional Laplacian) as
\begin{equation}
\lbeq{ExF}
    \pair{F}_{g,\nu,N}
    = \frac{\Ex_C  F Z_0}{\Ex_C Z_0} .
\end{equation}
In the right-hand side, part of the $\tau$ term has been shifted into the Gaussian
measure, because otherwise the massless torus covariance
$(-\Delta_{\Lambda_N})^{\alpha/2}$ is not invertible.
We evaluate \refeq{ExF} by separate evaluation of the numerator
and denominator on the right-hand side.
For $n=0$, the denominator equals $1$ due to supersymmetry
(see \cite[Proposition~4.4]{BIS09}).

For $n \ge 1$,
we write $\Ex_C\theta F$  for the convolution of $F$ with $P_C$.
Explicitly, for $n \ge 1$, given $F \in L^1(P_C)$,
$\theta$ is the shift operator $\theta F(\varphi,\zeta)=F(\varphi + \zeta)$, and
\begin{equation}
\lbeq{thetadef}
  (\Ex_C\theta F)(\varphi) = \Ex_C F(\varphi+\zeta),
\end{equation}
where the expectation $\Ex_C$ acts on $\zeta$ and leaves $\varphi$ fixed.
We define a generalisation of the denominator of \refeq{ExF} by
\begin{equation}
\label{e:ZNdef}
  Z_N(\varphi) = (\Ex_{C}\theta Z_0)(\varphi) = \Ex_C Z_0(\varphi +\zeta).
\end{equation}
Then $Z_N(0)=\Ex_CZ_0$.
It is a basic property of Gaussian integrals
(see \cite[Proposition~2.6]{BS-rg-norm}) that, given covariances $C',C''$,
\begin{equation}
  \lbeq{CC1}
  \Ex_{C''+C'} \theta F = (\Ex_{C''}\theta \circ \Ex_{C'}\theta) F.
\end{equation}
In terms of the decomposition \refeq{Calphadecomp}, this implies that
\begin{equation}
    \label{e:progressive}
    \Ex_{C}\theta F
    =
    \big( \Ex_{C_{N,N}}\theta \circ \Ex_{C_{N-1}}\theta \circ \cdots
    \circ \Ex_{C_{1}}\theta\big) F
    .
\end{equation}
To compute the expectations on the right-hand side of \eqref{e:ExF}, we use
\refeq{progressive}
to integrate progressively.  Namely, if we set $Z_0 = e^{-V_0(\Lambda_N)}$ as in \refeq{V0Z0},
and define
\begin{equation}
\label{e:Zjdef}
  Z_{j+1} = \Ex_{C_{j+1}}\theta Z_j \quad\quad
  (j<N),
\end{equation}
then, consistent with \eqref{e:ZNdef},
\begin{equation}
\label{e:ZN}
    Z_N = \Ex_C\theta Z_0.
\end{equation}
This leads us to study the recursion $Z_j \mapsto Z_{j+1}$.
To simplify the notation, we write $\Ex_{j} = \Ex_{C_j}$,
and leave implicit the dependence of the covariance $C_j$
on the mass $m$.  The formula \refeq{ZN} has a supersymmetric counterpart for $n=0$,
exactly as in \cite{BBS-saw4-log}.

The introduction of $m^2$ allows for a change in perspective, which is
that the right-hand side of \refeq{ExF} makes sense as a function of
\emph{independent} variables $m^2,\nu_0$.
We adopt this perspective until Section~\ref{sec:pfsuscept}, when the variable
$\nu$ will recover its prominence and $m^2,\nu_0$ will be required to
satisfy $\nu=\nu_0 +m^2$.
With this in mind, we define
\begin{equation}
\lbeq{chihatdef}
     \hat\chi_N (g,m^2,\nu_0)
     =
     n^{-1}
     \sum_{x\in \Lambda_N}
     \frac{\Ex_C  \left( (\varphi_0 \cdot \varphi_x) Z_0 \right)}{\Ex_C Z_0}
     .
\end{equation}
The right-hand side of \refeq{chiNn0} gives the analogous formula for $n=0$.
The finite-volume susceptibility is defined by (recall \refeq{susceptdef})
\begin{equation}
    \chi_N(g,\nu) = n^{-1}
     \sum_{x\in \Lambda_N}
     \langle \varphi_0 \cdot \varphi_x \rangle_{g,\nu,N}.
\end{equation}
By definition,
\begin{equation}
\lbeq{chichihat}
     \chi_N(g,\nu_0 + m ^2)
     =
     \hat\chi_N (g,m^2,\nu_0)
     .
\end{equation}

\subsection{Localisation}
\label{sec:Loc}

We use the localisation operator $\LT$ defined and studied in \cite{BS-rg-loc}.
This operator maps a function of the field $\varphi$ to a local polynomial.
For $n \ge 1$, we take as its domain the space
\begin{equation}
\label{e:Ncaldef}
    \Ncal = \Ncal(\Lambda) = C^{p_\Ncal}((\R^n)^\Lambda,\R)
\end{equation}
of real-valued functions of $\varphi\in (\R^n)^\Lambda$, having at least
$p_\Ncal$ continuous derivatives, with a fixed value $p_\Ncal \ge 10$.
For $n = 0$, the space $\Ncal$ is instead a space of differential forms;
see Section~\ref{sec:ir}.  It is useful at times to permit elements of $\Ncal$
to be complex-valued functions, as this allows analyticity techniques such as
the Cauchy estimates employed in \cite[Section~2.2]{BS-rg-step}.

We define the 3-dimensional linear space $\Ucal\cong \C^3$ to consist
of the local polynomials
of the form \refeq{Vdef-bis}.
We make the identification $U=(g,\nu,u)$
for elements of $\Ucal$.
We often write $V$ for elements of $\Ucal$ with $u=0$,
and we write $\Vcal \subset \Ucal$ for the subspace of such elements.
For $n=0$, the distinction between $\Vcal$ and $\Ucal$ is unimportant, since,
as mentioned previously,
the constant monomial $1$ plays no role due to supersymmetry.

Given $X \subset \Lambda$,
the localisation operator is a linear projection map
$\LT_X:\Ncal \to \Ucal(X)$  to
the subspace $\Ucal(X) = \{\sum_{x \in X} U(\varphi_x) : U \in \Ucal\}$
of $\Ncal$.
For scales $j \ge j_m$, we instead define $\LT_X$ to have range $\sum_{x\in X}(\nu\tau_x+u)$,
i.e., we no longer retain $\tau^2$ in the range.
Thus the range of $\LT_X$ depends on the scale at which the operator is applied, and
\begin{equation}
\lbeq{Locrange}
    \text{range of $\LT$ is spanned by}
    \;\;
    \begin{cases}
    \{1,\tau,\tau^2\} & ( j < j_m)
    \\
    \{1,\tau \} & ( j \ge j_m).
    \end{cases}
\end{equation}
The precise definition and properties of $\LT$ are developed in detail in \cite{BS-rg-loc}
and applied in \cite{BS-rg-IE,BS-rg-step}.
(There is a caveat of little significance here, discussed in
\cite{BS-rg-loc}: $X$ cannot be so large that it ``wraps around'' the torus.)

\subsection{Definition of the map \texorpdfstring{$\PT$}{PT}}

In this section, we define a quadratic map $\PT_j: \Vcal \to \Ucal$.
The notation ``$\PT$'' stands for ``perturbation theory.''
We base the discussion here on
$n \ge 1$; the case of $n=0$ is a small extension (see \cite{BBS-rg-pt}
and set $y=z=\lambda=q=0$ there).

Given a covariance matrix $C$, we define an operator on $\Ncal$ by
\begin{equation}
\label{e:LapC}
    \Lcal_C =
    \sum_{i=1}^n
    \sum_{u,v \in \Lambda}
    C_{u,v}
    \frac{\partial}{\partial \varphi_{u}^i}
    \frac{\partial}{\partial \varphi_{v}^i}.
\end{equation}
For polynomials $A,B$ in the field $\varphi$, we define
\begin{equation}
  \label{e:FCAB}
    F_{C}(A,B)
    = e^{\Lcal_C}
    \big(e^{-\Lcal_C}A\big)
    \big(e^{-\Lcal_C}B\big) - AB
    ,
\end{equation}
where the exponential is defined by power series expansion, which terminates when applied
to a polynomial.
With $C_j$ given by \refeq{CGamint}, let $w_j=\sum_{i=1}^j C_i$ and $w_0=0$.
The range of $w_j$ is that of $C_j$, namely $\frac 12 L^j$.
For $V \in \Vcal$ and $X \subset \Lambda$, we set
\begin{equation}
  \label{e:WLTF}
  W_j(V,X) = \frac 12 \sum_{x\in X} (1-\LT_{x}) F_{w_j}(V_x,V(\Lambda)).
\end{equation}
The map $\LT_x$ on the right-hand side is
the map $\LT_X$ discussed above with $X=\{x\}$, and $V_x$ is shorthand for $V(\varphi_x)$.
The definition \refeq{WLTF} cannot be applied when $j=N$
due to torus effects; an appropriate
alternate definition for the final scale is
provided in
\cite[Section~1.1.5]{BS-rg-IE}.

The map $\PT_j:\Vcal \to \Ucal$ is defined by
\begin{equation}
  \lbeq{Vptdef}
  \Upt = \PT_j(V) = e^{\Lcal_{C_{j+1}}} V - P_j(V)
  ,
\end{equation}
where
\begin{align}
\label{e:PdefF}
    P_j(V)_x
    &=
    \LT_x \left(
    e^{\Lcal_{C_{j+1}}} W_j(V,x)
    + \frac 12
    F_{C_{j+1}}
    (e^{\Lcal_{C_{j+1}}} V_x,e^{\Lcal_{C_{j+1}}} V(\Lambda))
    \right).
\end{align}
By translation invariance,
$P_j(V)_x$ does define a local polynomial with coefficients independent of $x$.

The motivation for the above definition is explained in
\cite{BBS-rg-pt}.
The basic idea is that if $Z_j$ is represented perturbatively as $Z_j \approx
e^{-V_j(\Lambda)}$ for a polynomial
$V_j \in \Ucal$, then the map $Z_j \mapsto Z_{j+1}$ can be approximated by
the map $V_j \mapsto \PT_j(V_j)$.
A nonperturbative analysis is also needed,
and this is the crux of the difficulty, to which we return in Section~\ref{sec:rg2}.

\section{Perturbative flow equations}
\label{sec:pt}

In this section, we study the perturbative flow equations.
The map $\PT$ is computed explicitly in Section~\ref{sec:PTmap}, and the
coefficients arising in this computation are estimated in Section~\ref{sec:PTcoeff}.
A change of variables to simplify the perturbative flow equations is presented in
Section~\ref{sec:rescalept}, where we define the map
$\overline{\PT}$.  The map $\overline{\PT}$
determines the perturbative fixed point, as discussed
in Section~\ref{sec:ptfp}.

\subsection{Computation of \texorpdfstring{$\PT$}{PT}}
\label{sec:PTmap}

The evaluation of the map $\PT$ is mechanical enough to be done via symbolic computation
on a computer.  This has been discussed already in \cite{BBS-rg-pt,BBS-phi4-log},
and the results reported there apply also here once simplified
due to irrelevance of $|\nabla \varphi|^2$; in particular we can set $z=y=0$
in the results of \cite{BBS-rg-pt,BBS-phi4-log}.
To state these results, we need some definitions.
Throughout Section~\ref{sec:pt}, the covariance decomposition is for $\Zd$ rather than
for the torus $\Lambda$, and formulas including \refeq{xidef}--\refeq{ugreeks2} are
computed with the $\Zd$ decomposition.

We write $C=C_{j+1;0,0}$, $w=w_j$, and
\begin{align}
   w_+& =w+C_{j+1}, \quad\quad
\lbeq{etadef}
   \eta' = (n+2)C .
\end{align}
Given $g,\nu \in \C$, and given a function $f=f(\nu,w)$, let
\begin{equation}
\label{e:delta-def}
    \nu^+ = \nu + \eta' g ,
   \quad\quad
    \delta[f (\nu ,w)]
=
    f (\nu^+ ,w_{+}) -  f (\nu  ,w )
.
\end{equation}
For $q : \Zd \to \C$ with finite support, we define
\begin{equation} \label{e:wndef}
  q^{(n)} = \sum_{x\in\Lambda} q_{x}^n
  .
\end{equation}
For integers $n \ge 0$, we also define the rational number
\begin{equation}
\lbeq{gamhatdef}
    \gamhat = \frac{n+2}{n+8},
\end{equation}
which appears in \refeq{xidef} and \refeq{nupta}, and which ultimately
appears in the determination of the
order $\epsilon$ terms in the critical exponents $\gamma,\alpha_H$
in Theorems~\ref{thm:suscept}--\ref{thm:sh}.
Let
\begin{align}
  \beta_j' &= (n+8) \delta[w^{(2)}]
  ,
  \quad
  \xi_j' =
  2(n+2)
  \big(
  \delta[w^{(3)}]- 3 w^{(2)}C
  \big)
  + \gamhat \beta_j'   \eta_j'
  ,
\lbeq{xidef}
\\
 \kappa_{g,j}' &= \tfrac{1}{4} n(n+2) C^2 ,
 \quad
 \kappa_{\nu,j}' = \half n C,
 \quad
 \kappa_{g\nu,j}' = \tfrac{1}{2} n(n+2)C ( \delta[w^{(2)}] - 2 Cw^{(1)} ),
 \nnb
\lbeq{ugreeks2}
 \kappa_{gg,j}'
 &=
 \tfrac{1}{4}n(n+2) (\delta[w^{(4)}] - 4 Cw^{(3)}
 - 6 C^2 w^{(2)} + (n+2) C^2 \delta[w^{(2)}]),
 \\ \nonumber
 \kappa_{\nu\nu,j}'
 &= \tfrac{1}{4}n (\delta[w^{(2)}] - 2C w^{(1)}
 )
.
\end{align}

\begin{prop}
\label{prop:PT}
Let $n \ge 0$.
The map $V=(g,\nu) \mapsto \Upt=\PT_j(V)=(\gpt,\nupt,\delta\upt)$ is given by
\begin{align}
  \gpt
  &
  =
  \begin{cases}
  g - \beta_j' g^{2} - 4g \delta[\nu w^{(1)}] & ( j < j_m)
  \\
  g & ( j\ge j_m)  ,
  \end{cases}
\label{e:gpt2a}
  \\
  \nu_\pt
  &=
  \nu
  +  \eta_j' (g + 4g\nu w^{(1)})
  - \xi_j' g^{2}
  -\gamhat \beta_j' \nu g
  - \delta[\nu^{2} w^{(1)}]
  ,
\label{e:nupta}
\\
\lbeq{uptlong}
  \delta u_\pt &=
  \begin{cases}
   \kappa_{g,j}' g + \kappa_{\nu,j}' \nu
  - \kappa_{g\nu,j}' g\nu
  - \kappa_{gg,j}' g^2 - \kappa_{\nu\nu,j}' \nu^2
  & ( n \ge 1)
  \\
  0 & (n=0).
  \end{cases}
\end{align}
\end{prop}

\begin{proof}
This follows from explicit calculation using \refeq{Vptdef}--\refeq{PdefF},
and the result for $n \ge 1$
is taken from \cite{BBS-phi4-log}, and for $n=0$ from \cite{BBS-rg-pt}.
Compared to \cite{BBS-phi4-log,BBS-rg-pt}, we omit $z,y$ terms
here, as well as terms with $w^{(**)}$ that appear  in $\kappa_{gg}',\kappa_{\nu\nu}'$ for $d=4$ but
that do not occur here because $\nabla^2\varphi^2$ is not in the range of $\LT$.
The $j>j_m$ case of \refeq{gpt2a} is due to the fact that the range of $\LT$ no longer includes
$\tau^2$ after the mass scale.  (The term $\kappa_{g\nu}'g\nu$
in \refeq{uptlong} was erroneously omitted in
\cite{BBS-phi4-log}, but this omission does not affect the conclusions in \cite{BBS-phi4-log}.)
The simplification that $\delta u_\pt=0$ for $n=0$ is a consequence of supersymmetry,
as explained in \cite{BBS-rg-pt}.
\end{proof}

\subsection{Estimates on coefficients}
\label{sec:PTcoeff}

Typically we use primes for coefficients that scale with $L$ below the mass scale, and
remove the primes for rescaled versions.
Thus, we define rescaled coefficients
\begin{equation}
\label{e:Greeknoprime}
    \beta_j = L^{-\epsilon (j\wedge j_m)} \beta_j', \quad
    \eta_j = L^{(d-\alpha)j} \eta_j'
    ,
    \quad
    \xi_j = L^{(\alpha-2\epsilon)j}\xi_j' ,
    \quad
  \bar{w}_j^{(1)} = L^{-\alpha (j\wedge j_m)}  w_j^{(1)}
  ,
\end{equation}
\begin{equation}
\begin{aligned}
 &\kappa_g = L^{-\epsilon (j\wedge j_m)}\kappa_g',
 \quad
 \kappa_\nu = L^{-\alpha j} \kappa_\nu',
 \quad
 \kappa_{g\nu}= L^{-\alpha j-\epsilon(j\wedge j_m)}\kappa_{g\nu}',
 \\
 &\quad
 \kappa_{gg}
 =
 L^{-2\epsilon j}\kappa_{gg}',
 \quad
 \kappa_{\nu\nu}
 = L^{-2\alpha(j\wedge j_m)}\kappa_{\nu\nu}'
.
\lbeq{kappadef}
\end{aligned}
\end{equation}
In Section~\ref{sec:rescalept}, we analyse transformed flow equations, which require
the additional definitions:
\begin{equation}
\lbeq{etagedef}
        \eta'_{\ge j} = \sum_{k=j}^{\infty}\eta_k', \quad\quad
        \eta_{\ge j} = L^{(d-\alpha)j}\eta'_{\ge j}
        = \sum_{k=j}^{\infty}L^{-(d-\alpha)(k-j)}\eta_k,
\end{equation}
\begin{equation}
\lbeq{betaW}
    \beta^:_j = \beta_j
    +4(\eta_{\ge j} \bar w^{(1)}_j - \eta_{\ge j+1} \bar{w}_{j+1}^{(1)})
    ,
\end{equation}
\begin{equation}
\lbeq{newxidef}
    \newxi_j = \xi_j -\gamhat\beta_j\eta_{\ge j} +  L^{-(d-\alpha)}\eta_{\ge j+1}\beta_j
    .
\end{equation}

The following four lemmas provide estimates for the above coefficients.
The proofs of Lemmas~\ref{lem:wlims}--\ref{lem:betadiff}
are deferred to Section~\ref{sec:Cbound}.
The first lemma is an
adaptation and extension of
\cite[Lemma~6.2]{BBS-rg-pt}
and \cite[Lemma~A.1]{BBS-phi4-log}.
In its statement, we use the notation
\begin{equation}
\lbeq{Mdef}
    M_j = (1+m^2L^{\alpha (j-1)})^{-2}.
\end{equation}
By \refeq{cbetabd},
\begin{equation}
\lbeq{Mjbd}
    M_j \le L^{-2\alpha (j-j_m)_+},
\end{equation}
so beyond the mass scale $M_j$ decays exponentially with base $L$.
The hypothesis $\alpha > \frac d2$ ensures that $\epsilon =2\alpha -d>0$.
Equation~\refeq{betabd}
shows that the scaling introduced in \refeq{Greeknoprime}--\refeq{kappadef}
is natural.

\begin{lemma}
  \label{lem:wlims}
  Let $d = 1,2,3$; $\alpha \in (\frac d2,2\wedge d)$; $j \ge 1$; $\bar m^2 >0$.
  The following bounds hold uniformly in $m^2 \in [0,\bar m^2]$:
    \begin{equation}
  \lbeq{betabd}
    \eta_j, \eta_{\ge j}, \beta_j, \beta^:_j, \xi_j , \newxi_j=
     O(M_j),
  \quad
    \bar w_j^{(1)} = O(1)
    ,
    \quad
    \kappa_{*,j} = O(M_j L^{-dj}).
  \end{equation}
Constants in \refeq{betabd} may depend on $L,\bar m^2$ but not on $j$, except in
the bound on $\eta_j$ where the constant is also independent of $L$.
Each of the left-hand sides in \refeq{betabd} is
continuous in $m^2 \in [0,\bar m^2]$.
  Moreover, with $c_{\partial\beta}$ dependent on $L$,
  and assuming
$m^2 \in (0,\bar m^2]$, and $j \le j_m$,
\begin{equation}
\lbeq{dbetam}
    \Big|\frac{\partial\beta_j}{\partial m^2}  \Big|
    \leq c_{\partial\beta}
        L^{\alpha j}
    \frac{1+\1_{d=2}|\log(m^2L^{\alpha j})|}{(m^2L^{\alpha j})^r}
    \quad \text{with} \;\; r =
    \begin{cases}
    2-1/\alpha & (d=1)
    \\
    2-2/\alpha & (d=2,3).
    \end{cases}
\end{equation}
\end{lemma}

As $\epsilon \downarrow 0$, the values of $r$ in \refeq{dbetam} obey
$r \sim 2\epsilon$ for $d=1$,
$r \sim \epsilon$ for $d=2$,
 and $r \sim \frac 23$ for $d=3$.
The next lemma controls the rate of convergence of the sequence $\beta_j$ to
its limiting value in the massless case.

\begin{lemma}
\label{lem:beta-a0}
Let $d=1,2,3$ and $\alpha \in (\frac d2,2\wedge d)$.
There exists $a >0$ (possibly depending on $L$),
and an $L$-dependent constant $\bar b_L$, such that for all $j \ge 1$,
\begin{align}
\lbeq{w2b}
    |\beta_j(0)-a|
    &\le \bar b_L L^{-(\alpha\wedge 1)j}
        \quad
    (m^2=0).
\end{align}
\end{lemma}

The next lemma controls the difference between $\beta_j$ and $\beta_j^:$, below the mass
scale  $j_m$ defined in \refeq{jmdef}.
Its upper bound is not small for small $j$ due to lattice effects (the constant may
be a large function of $L$).  For large $j$,
but not so large as to be near the mass scale, the difference is small because
of cancellation within $\eta_{\ge j}\bar w_j^{(1)}-\eta_{\ge j+1}\bar w_{j+1}^{(1)}$
in \refeq{betaW}.  This cancellation
breaks down near the mass scale.

\begin{lemma}
\label{lem:betadiff}
Let $d = 1,2,3$; $\alpha \in (\frac d2,2\wedge d)$; $\bar m^2 >0$.
There exists $z>0$ such that, uniformly  in $m^2 \in [0,\bar m^2]$ and $1 \le j \le j_m$,
with a possibly $L$-dependent constant,
\begin{equation}
\lbeq{betadiff}
    |\beta^:_j -\beta_j| \le O(L^{-zj} + L^{-z(j_m-j)}).
\end{equation}
\end{lemma}

The next lemma controls the difference between the
(possibly) massive $\beta_j^:$ and the limit of the
massless $\beta_j$, below the mass scale.
Lattice effects cause the estimate to be degraded
at small scales, and near the mass scale the estimate is degraded because
the $m^2$-dependence of $\beta_j^:(m^2)$
begins to take effect.  For the intermediate scales, which form the vast majority
for small $m^2$, the difference between $\beta_j^:(m^2)$ and $a$ is well controlled by the
lemma.

\begin{lemma}
\label{lem:beta-am}
Let $d=1,2,3$; $\alpha \in (\frac d2,2\wedge d)$;
$\bar m^2 > 0$. There exist $J_L$ and $b_L$ such that,
uniformly in $m^2 \in [0,\bar m^2]$ and $j \le j_m$, and
with
the constant $a$ of Lemma~\ref{lem:beta-a0},
\begin{equation}
\lbeq{betajmdiff}
    |\beta_j^:(m^2)-a | \le
        \begin{cases}
     b_L & (j \le J_L)
    \\
    \frac{a}{64}  & (J_L \le j \le j_m-J_L)
    \\
    b_L  & (j_m - J_L \le j \le j_m).
    \end{cases}
\end{equation}
\end{lemma}

\begin{proof}
Let $j \le j_m$.
By the triangle inequality,
and by Lemmas~\ref{lem:beta-a0} and \ref{lem:betadiff},
\begin{align}
    |\beta_j^:(m^2)-a| & \le |\beta^:_j(m^2)-\beta_j(m^2)|
    +|\beta_j(m^2)-\beta_j(0)| + |\beta_j(0)-a|
    \nnb & \le
    |\beta_j(m^2)-\beta_j(0)| + \bar b_L L^{-(\alpha\wedge 1)j} + \bar{\bar b}_L (L^{-zj}+L^{-z(j_m-j)}).
\end{align}
We choose $J_L$ to be large enough that
$\bar b_L L^{-(\alpha\wedge 1)j} + \bar{\bar b}_L (L^{-zj}+L^{-z(j_m-j)})
\le \frac{a}{128}$, for $J_L \le j \le j_m-J_L$.
Then
\begin{align}
\lbeq{betam0pf3}
    |\beta_j^:(m^2)-a| &
    \le
    |\beta_j(m^2)-\beta_j(0)| +
    \begin{cases}
    \bar b_L + \bar{\bar b}_L & ( j \le J_L)
    \\
    \frac{a}{128} & ( J_L \le j \le j_m-J_L)
    \\
    \bar b_L + \bar{\bar b}_L & ( j_m-J_L \le j \le J_L).
    \end{cases}
\end{align}
To deal with the logarithmic factor in \refeq{dbetam} for $d=2$, we increase $r$ slightly
to absorb it.  Then integration of this modification of
\refeq{dbetam} gives (note that $1-r >0$)
\begin{align}
\lbeq{betam0pf}
    |\beta_j(m^2)-\beta_j(0)| & \le
    \frac{1}{1-r} c_{\partial\beta}   (  m^2L^{\alpha j})^{1-r}.
\end{align}
We write $m^2L^{\alpha j} = m^2L^{\alpha j_m} L^{-\alpha(j_m-j)}$ and use the definition
of $j_m$ to see that \refeq{betam0pf} implies that there exists $\hat b_L$ such
that
\begin{equation}
    |\beta_j(m^2)-\beta_j(0)| \le \hat b_L L^{-\alpha(1-r)(j_m-j)}.
\end{equation}
By increasing $J_L$ if necessary, the right-hand side is at most $\frac{a}{128}$
for $j \le j_m-J_L$, and in any case is at most $\hat b_L$.
This gives the desired result, with $b_L=\hat b_L +\bar b_L +\bar{\bar b}_L$.
\end{proof}

\subsection{Change of variables}
\label{sec:rescalept}

For $j \le j_m$, we define rescaled coupling constants
\begin{equation}
\label{e:munu}
  \ghat_j = L^{\epsilon j}g_j,
  \quad\quad
  \muhat_j = L^{\alpha j}\nu_j
  .
\end{equation}
With \refeq{Greeknoprime}--\refeq{kappadef},
the flow equations \refeq{gpt2a}--\refeq{nupta} can be rewritten as
\begin{align}
  \ghat_{\pt}
  &
  =
  L^\epsilon \ghat
  \Big( 1 - \beta  \ghat -4\delta[\muhat\bar w^{(1)}] \Big)
  ,
\label{e:gL2az5}
  \\
  \muhat_{\pt}
  &=
  L^\alpha \Big( \muhat
  +   \eta  (\ghat +4\ghat \muhat\bar w^{(1)})
  - \gamhat \beta \muhat  \ghat
  -  \xi  \ghat ^{2}
   -
  \delta[\muhat^2 \bar w^{(1)}]
  \Big)
  ,
\label{e:muaz5}
\end{align}
where $\ghat_\pt = L^{\epsilon (j+1)}g_\pt$, $\muhat_\pt = L^{\alpha (j+1)}\nu_\pt$,
and
\begin{align}
    \delta[\muhat\bar w^{(1)}]
    & =
    (\muhat +\eta \ghat ) L^\alpha\bar w_{+}^{(1)}
  -
  \muhat \bar w^{(1)}
  ,
    \\
    \delta[\muhat^2 \bar w^{(1)}]
   & = (\muhat +\eta \ghat )^2 L^{\alpha}\bar w_{+}^{(1)}
  -
  \muhat^2 \bar w^{(1)}
  .
\end{align}

For scales $j \le j_m$, we analyse transformed perturbative flow equations.
The transformation eliminates the $\delta$ terms in \refeq{gL2az5}--\refeq{muaz5} as in
\cite[Proposition~4.3]{BBS-rg-pt}, but additionally
removes the $\eta$ term in \refeq{muaz5} by a version of Wick ordering.
The transformation uses the quadratic
map $T=T_j: \C^2 \to \C^2$, denoted $T(\ghat,\muhat) = (\gL,\mu)$,
and defined by
\begin{align}
  \label{e:gch-def}
  \gL   & = \ghat  + 4 \ghat (\muhat + \eta_{\ge j}\ghat)\bar w_j^{(1)},
  \\
  \lbeq{nuhatdef}
  \mu  &
   = \muhat
   + \eta_{\ge j} (\ghat + 4 \ghat \muhat \bar w_j^{(1)})
   + \muhat^2 \bar w_j^{(1)}
  .
\end{align}
The transformation $T_j$ has an inverse $T_j^{-1}$ defined on a $j$-independent ball
$B$ centred at the origin of $\C^2$.  By definition, the linear parts of $T_j$ and $T_j^{-1}$
are given by
\begin{align}
\lbeq{Tlin}
    T_j(\ghat,\muhat) &=(\ghat,\muhat+ \eta_{\ge j} \ghat) +O(|\ghat|^2+|\muhat|^2),
\\
\lbeq{Tinvlin}
    T_j^{-1}(s,\mu) &= (s,\mu-\eta_{\ge j}s) +O (|s|^2+|\mu|^2).
\end{align}
Finally, we define a map
$\overline\PT_j: \R^2 \to \R^2$, denoted
$(\bar\gL_j,\bar\mu_j) \mapsto (\bar\gL_{j+1},\bar\mu_{j+1})$, by
\begin{align}
\lbeq{gLbar}
    \bar\gL_{j+1} & = L^\epsilon \bar\gL_j (1 - \beta^:_j \bar\gL_j)
    ,
    \\
\lbeq{mubar}
    \bar \mu_{j+1} &
    = L^\alpha
    \left(
    \bar\mu_j -\gamhat \beta_j \mubar_j\bar\gL_j - \newxi_j \bar\gL_j^2
    \right)
\end{align}
($\newxi_j$ is defined in \refeq{newxidef}).
Note that $\beta^:_j$ appears in \refeq{gLbar} and $\beta_j$ appears in \refeq{mubar}.
Although these coefficients are not identical, they differ only by an amount
that is insignificant except for a few scales.
Equations~\refeq{gLbar}--\refeq{mubar} have the advantage, compared to
\refeq{gL2az5}--\refeq{muaz5}, that $\mubar$ does not appear in the $\bar\gL$ equation,
and no linear $\bar\gL$ term appears in the $\mubar$ equation.
In \refeq{Tcomp}, we write $\PT_{j}^{(0)}$ for the map $\PT_j$ with the $u$ component
suppressed.
The following proposition shows that, below the mass scale
and up to a third-order error, the map $\PT^{(0)}$
for the variables $(\ghat,\muhat)$
is equivalent to the map $\overline{\PT}$ for the variables $(s,\mu)$.

\begin{prop}
\label{prop:transformation}
Let $d=1,2,3$, $\bar m^2 >0$, $m^2 \in [0,\bar m^2]$, and $j \le j_m$.
On the open ball $B$ mentioned below \refeq{nuhatdef},
there exists
an analytic map $e_{\pt,j} : B \to \R^2$ such that
\begin{equation}
\lbeq{Tcomp}
  T_{j+1} \circ \PT_{j}^{(0)} \circ \, T_j^{-1}
  =
  \overline\PT_j   + e_{\pt,j}
  ,
\end{equation}
where  $e_{\pt,j}(s,\mu) = O(|s|^3+ |s|^2\epsilon + |\mu|^3)$ with
constant uniform in $m^2 \in [0,\bar m^2]$ and $j\le j_m$.
\end{prop}

\begin{proof}
We write the components of the map $T$ as $T=(T^{(s)},T^{(\mu)})$.
By definition, $\eta_j = \eta_{\ge j} - L^{-(d-\alpha)} \eta_{\ge j+1}$.
Using this, and $\epsilon = 2\alpha -d$, \refeq{muaz5} can be rewritten as
\begin{align}
\lbeq{tran1}
  & \muhat_{\pt} + \eta_{\ge j+1}L^{\epsilon}(\ghat + 4\ghat\muhat \bar w^{(1)})
  + (L^{\alpha}(\muhat+ \eta \ghat))^2 \bar w_+^{(1)}
  \nnb & \quad =
  L^\alpha \Big(\Big[ \muhat
  +   \eta_{\ge j}  (\ghat +4\ghat \muhat\bar w^{(1)})
  +\muhat^2 \bar w^{(1)}
   \Big]
  - \gamhat \beta \muhat  \ghat
 -  \xi  \ghat ^{2} \Big)
 \nnb & \quad =
  L^\alpha (\mu
  -  \gamhat \beta \muhat  \ghat
 -  \xi  \ghat ^{2}).
\end{align}
Also, \refeq{gL2az5} can be rewritten as
\begin{align}
\lbeq{tran2}
    \ghat_\pt + 4L^\epsilon \ghat L^\alpha (\muhat + \eta\ghat)\bar w_+^{(1)}
    & =
    L^\epsilon \Big(\ghat + 4\ghat\muhat \bar w^{(1)} - \beta \ghat^2 \Big)
    .
\end{align}
We solve \refeq{tran2} for $L^\epsilon (\ghat + 4\ghat\muhat \bar w^{(1)})$,
insert the result into the left-hand side of \refeq{tran1}, and then use
\refeq{gL2az5}--\refeq{muaz5}, to see that
the left-hand side of \refeq{tran1} is
equal to
\begin{align}
    & \muhat_{\pt} + \eta_{\ge j+1}(\ghat_\pt + 4L^\epsilon \ghat L^\alpha (\muhat + \eta\ghat)\bar w_+^{(1)})
    + (L^{\alpha}(\muhat+ \eta \ghat))^2 \bar w_+^{(1)}  + \eta_{\ge j+1}L^\epsilon\beta\ghat^2
    \nnb &
    \quad
    =
    \muhat_{\pt} + \eta_{\ge j+1}(\ghat_\pt
    + 4\ghat_\pt \muhat_\pt \bar w_+^{(1)})
    + \muhat_\pt^2 \bar w_+^{(1)}  + \eta_{\ge j+1}L^\epsilon\beta\ghat^2
    + O(\hat{x}^3)
    \nnb &
    \quad
    =
    T^{(\mu)}_{j+1}(\ghat_\pt,\muhat_\pt) + \eta_{\ge j+1}L^\epsilon\beta\ghat^2 + O(\hat{x}^3),
\lbeq{tran3}
\end{align}
with $O(\hat{x}^3)$ meaning $O(|\ghat|^3+|\muhat|^3)$.
We use the equality of the right-hand sides of \refeq{tran1} and \refeq{tran3},
together with \refeq{gch-def}--\refeq{nuhatdef}, \refeq{tran1}, and the
definition of $\newxi$ in \refeq{newxidef}, to obtain
\begin{align}
    T^{(\mu)}_{j+1}(\ghat_\pt,\muhat_\pt) & = L^\alpha (\mu
    -  \gamhat \beta \muhat  \ghat
    -  \xi  \ghat ^{2})
    -
    \eta_{\ge j+1}L^\epsilon\beta\ghat^2 + O(\hat{x}^3)
    \nnb & =
     L^\alpha (\mu
    -  \gamhat \beta (\mu -\eta_{\ge j} \gL) \gL
    -  \xi  \gL^{2})
    -
    \eta_{\ge j+1}L^\epsilon\beta\gL^2 + O(\hat{x}^3)
    \nnb & =
    L^\alpha (\mu
    -  \gamhat \beta \mu  \gL
    -  \newxi
    \gL^{2})
    + O(\hat{x}^3)
    =
    \overline{\PT}_j^{(\mu)}(s,\mu) + O(\hat{x}^3),
\end{align}
as required.

For the $\ghat$ equation, by \refeq{tran2} and \refeq{gL2az5}--\refeq{muaz5}, we have
\begin{align}
    \ghat_\pt + 4 \ghat_\pt \muhat_\pt \bar w_+^{(1)} + O(\hat x^3)
    & =
    L^\epsilon (\ghat - \beta \ghat^2 + 4\ghat\muhat \bar w^{(1)} ).
\end{align}
This leads to
\begin{align}
    \ghat_\pt + 4  \ghat_\pt (\muhat_\pt + \eta_{\ge j+1} \ghat_\pt )\bar w_+^{(1)} + O(\hat x^3)
    & =
    L^\epsilon (\ghat - \beta \ghat^2 + 4\ghat\muhat \bar w^{(1)} )
    + 4  \eta_{\ge j+1} \ghat_\pt^2\bar w_+^{(1)}
    \nnb & =
    L^\epsilon (\ghat - \beta \ghat^2 + 4\ghat(\muhat+ \eta_{\ge j}\ghat) \bar w^{(1)} )
    \nnb & \quad\quad
    -
    L^\epsilon 4\eta_{\ge j}\ghat^2 \bar w^{(1)}
    + 4  \eta_{\ge j+1} \ghat_\pt^2\bar w_+^{(1)}
    .
\end{align}
With \refeq{gch-def}, and using $\hat g_\pt = L^\epsilon \ghat + O(\hat x^2)$ on the right-hand
side, this gives
\begin{align}
    T^{(s)}_{j+1}(\ghat_\pt,\muhat_\pt) = L^\epsilon(s - \tilde\beta s^2)
    +O(\hat x^3),
\end{align}
with
\begin{align}
    \tilde\beta_j
    &=
    \beta_j +4L^{\epsilon}\eta_{\ge j} \bar w^{(1)}_j -
      4\eta_{\ge j+1} \bar{w}_{j+1}^{(1)}
    =
    \beta^:_j + (L^{\epsilon}-1) 4\eta_{\ge j} \bar{w}_{j}^{(1)}.
\end{align}
The last term on the right-hand side is $O(\epsilon)$, and
hence, as required,
\begin{align}
    T^{(s)}_{j+1}(\ghat_\pt,\muhat_\pt) =
    \overline{\PT}_j(s,\mu)
    +O(\hat x^3 + |s|^2\epsilon).
\end{align}
This completes the proof.
\end{proof}

\subsection{Perturbative fixed point}
\label{sec:ptfp}

The perturbative
fixed point equation arises by replacing $\beta^:_j=\beta^:_j(m^2)$ in \refeq{gLbar} by
its limiting value in the massless case.
By Lemmas~\ref{lem:beta-a0}--\ref{lem:betadiff}, this limiting value is the number
\begin{equation}
    a=\lim_{j\to\infty} \beta_j(0).
\end{equation}
The nonzero solution of
\begin{equation}
\lbeq{sbar}
    \bar\gL  = L^\epsilon \bar\gL (1 -   a  \bar\gL )
\end{equation}
is
\begin{equation}
\lbeq{sfix}
    \gLfix
    =\frac 1a (1-L^{-\epsilon})
    = O(\epsilon).
\end{equation}
With $L$ fixed, we have
\begin{equation}
\lbeq{asfix}
    a\gLfix \sim \epsilon \log L \quad \text{as $\epsilon \downarrow 0$.}
\end{equation}

Let
\begin{equation}
\lbeq{ydef}
    \bar\dgL_j = \gLfix - \bar\gL_j.
\end{equation}
A calculation using \refeq{sbar} with \refeq{gLbar}--\refeq{mubar} gives
\begin{align}
\lbeq{sigsm0}
    \bar\dgL_{j+1} &= c_\epsilon\bar\dgL_j
    + L^{\epsilon} \left( a\bar\dgL_j^2 +
    (\beta^:_j-a) (\gLfix-\bar\dgL_j)^2 \right)
    ,
    \\
\lbeq{musm0}
    \mubar_{j+1} &= L^{\alpha}
    \left( \mubar_j
    - \gamhat \beta_j \mubar_j(\gLfix - \bar\dgL_j) - \newxi_j (\gLfix - \bar\dgL_j)^2 \right)
    ,
    \end{align}
where
\begin{equation}
    c_\epsilon=2-L^\epsilon = 1 - x < 1 \quad
    \text{with} \quad
    x=L^\epsilon -1 \sim \epsilon \log L \; \text {as $\epsilon \downarrow 0$}.
\end{equation}

Although we do not use it, for completeness we note that,
assuming $\newxi_j$ with $m^2=0$ approaches
a limiting value $\newxi$, the fixed point equation corresponding to \refeq{mubar} is
\begin{equation}
    \mufix = L^\alpha (\mufix - \gamhat a \mufix \gLfix - \newxi \gLfix^2).
\end{equation}
Solving this to second order in $\gLfix$ gives
\begin{equation}
    \mufix = \frac{\newxi \gLfix^2}{1-L^{-\alpha} - \gamhat a \gLfix}
    \sim
    \frac{\newxi \gLfix^2}{1-L^{-\alpha} }.
\end{equation}
This is second order in $\epsilon$, consistent with the choice of weight we make in
\refeq{Xweights}.

\section{Nonperturbative analysis}
\label{sec:rg2}

This section concerns the nonperturbative analysis, and provides a solution
to the \emph{large-field problem}.
We define the necessary norms and regulators, as well as domains and small
parameters for the renormalisation group map.
The main result is Theorem~\ref{thm:step-mr}, whose proof involves
adaptation of some details in the proof
of the main result of \cite{BS-rg-step}.

\subsection{Nonperturbative coordinate}
\label{sec:circ}

\begin{figure}[ht]
\begin{center}
  \input{alpha_hier1.tex}
\end{center}
\caption{\lbfg{reblock} Blocks in ${\cal B}_j$ for
$j=0,1,2,3$ when $d=2$, $L=2$, $N=3$.}
\label{fig:RG_hierarchy1}
\end{figure}

For each $j=0,1,\ldots,N$, the torus $\Lambda_N$ partitions  into
$L^{N-j}$ disjoint $d$-dimensional cubes of side $L^j$, as in
Figure~\ref{fig:RG_hierarchy1}.
We call these cubes {\em blocks}, or $j$-{\em blocks}.
The block that contains the origin is
    $\{x\in \Lambda:  0 \le x_{i} < L^j \,  (i=1,\dots ,d)\}$,
and other blocks are translates of this one by vectors in
$L^{j} \Zd$.
We denote the set of $j$-blocks by ${\cal B}_j$.
A union of $j$-blocks (possibly empty) is called a {\em polymer} or $j$-\emph{polymer},
and the set of $j$-polymers is denoted ${\cal P}_j$.
The set of blocks that comprise a polymer $X \in \Pcal_j$ is denoted $\Bcal_j(X)$.
The unique $N$-block is $\Lambda_N$ itself.

A nonempty subset $X\subset \Lambda $ is
said to be \emph{connected} if for any $x, x'\in X$ there
exist $ x_0,x_1,\ldots,x_n \in X$ with
$|x_{{i+1}}-x_{i}|_\infty =1$, $x_{0} = x$ and $x_{n}=x'$.  The set
of connected polymers in $\Pcal_j$ is denoted $\Ccal_j$.
We write ${\rm Comp}_j( X) \subset \Ccal_j$ for the set of connected components
of $X\in \Pcal_j$.

A \emph{small set} is a
connected polymer $X \in \Ccal_j$ consisting of at most $2^d$ blocks
(the specific number
$2^d$ is important in \cite{BS-rg-step} but its role is not apparent here).
Let $\Scal_j \subset \Ccal_j$ denote the set of small sets.
The \emph{small-set neighbourhood} of $X \subset \Lambda $ is
the enlargement of $X$ defined by
$
    X^{\Box}
=
    \bigcup_{Y\in \Scal_{j}:X\cap Y \not =\varnothing } Y$.

Given $F_1, F_2 :{\cal P}_j \to \Ncal$ (with $\Ncal$ defined by
\refeq{Ncaldef}), the \emph{circle product}
$F_1\circ F_2 :{\cal P}_j \to \Ncal$ is defined by
\begin{equation} \label{e:circprod}
    (F_1 \circ F_2)(Y) = \sum_{X\in {\cal P}_j: X \subset Y} F_1(X) F_2(Y \setminus X)
    \quad\quad (Y \in \Pcal_j).
\end{equation}
The circle product depends on the scale $j$, but
we do not record this in the notation.
The terms corresponding to $X=\varnothing$ and $X=Y$ are included
in the summation on the right-hand side, and we only consider $F: \Pcal_j \to \Ncal$ with
$F(\varnothing)=1$.
The circle product is associative and commutative, since the product
on $\Ncal$ has these properties.
The identity element is
$\1_\varnothing(X) = \1_{X=\varnothing}$, i.e.,
$(F\circ \1_{\varnothing})(Y) = F(Y)$ for all $F$ and $Y$.

For $V\in \Ucal$ and $X \in \Pcal_j$, we set
\begin{equation}
  \lbeq{Idef}
  I_j(V,X) = e^{-V(X)}\prod_{B \in \Bcal_j(X)}(1+W_j(V,B)),
\end{equation}
with $W_j$ defined by \refeq{WLTF}.
For $j=0$, we have $W_0=0$ and $I_0(V,X)=e^{-V(X)}$.
Let $K_0 : \Pcal_0 \to \Ncal$ be the identity element $K_0 = \1_{\varnothing}$.
Then $Z_{0}=I_0(V_0,\Lambda)$ defined in \eqref{e:V0Z0} is
also given by
\begin{equation}
\label{e:Zinit}
    Z_0 =  (I_0 \circ K_0)(\Lambda)
    .
\end{equation}

In the recursion $Z_j \mapsto Z_{j+1} = \Ex_{j+1}\theta Z_j$ of
\eqref{e:Zjdef},
we maintain the form \refeq{Zinit} over all scales, as
\begin{gather}\label{e:ZIK}
    Z_j =  e^{-u_j|\Lambda|}(I_j \circ K_j)(\Lambda),
\end{gather}
with
\begin{equation}
\lbeq{Vjdef}
    V_j = \tfrac 14 g_j |\varphi|^4 + \tfrac 12 \nu_j |\varphi|^2,
\end{equation}
$I_j=I_j(V_j)$, and $K_{j}:\Pcal_{j} \rightarrow \Ncal$.
The initial condition given by \eqref{e:Zinit} has $u_0=0$, and the value of $\nu_0$
must be tuned carefully, depending on $m^2$,
in order to maintain \refeq{ZIK} with control of $K_j$ as $j$ becomes increasingly larger.
The action of $\Ebold_{j+1}\theta$ on $Z_{j}$ is then expressed as a map:
\begin{equation}
\label{e:RGmap}
     ( V_j, K_j) \mapsto ( U_{j+1},K_{j+1}) = (\delta u_{j+1},V_{j+1},K_{j+1}).
\end{equation}
To achieve this, given $u_j\in \R$ and $(V_j,K_j)$ in a suitable domain,
it is necessary to produce $U_{j+1}= (\delta u_{j+1},V_{j+1})\in \Ucal$ and
$K_{j+1} : \Pcal_{j+1} \to \Ncal $ such that, with $I_{j+1}=I_{j+1}(V_{j+1})$ and
$u_{j+1}=u_j+\delta u_{j+1}$,
\begin{equation}
    \label{e:Kspace-objective}
    Z_{j+1}
    =
    \Ex_{j+1}\theta Z_j
    =
    e^{-u_j|\Lambda|}
    \Ex_{j+1}\theta (I_j \circ K_j)(\Lambda)
    =
    e^{-u_{j+1}|\Lambda|}(I_{j+1} \circ K_{j+1})(\Lambda)
    .
\end{equation}
Then $Z_{j}$ retains its form under progressive integration.
The construction of the map \refeq{RGmap} occurs in Theorem~\ref{thm:step-mr} below.

The nonperturbative coordinate $K_j$ is an element of the space
$\Kspace_j$ defined in  Definition~\ref{def:Kspace} (recalled  from
\cite[Definition~1.7]{BS-rg-step}).
There are two versions of the space $\Kspace_j$, one for the torus $\Lambda_N$ for
scales $j \le N$, and one for the infinite volume $\Zd$ for all scales $j < \infty$.
We write $\volume$ to denote either $\Lambda_N$ or $\Zd$, and write $j \le N(\volume)$
as shorthand for the above two restrictions on $j$.
Given a subset $X\subset \volume$, let $\Ncal(X)=\Ncal(X,\volume)$ denote the set of elements of
$\Ncal$ (functions of $\varphi$) which depend on the values of $\varphi_x$ only for
$x \in X$.

\begin{rk}
We use the case $\volume=\Zd$ to tune $\nu_0$, in Section~\ref{sec:PCMI},
in a manner independent of the size of the torus $\Lambda_N$.
\end{rk}

\begin{defn} \label{def:Kspace}
For $\volume = \Lambda_N$ or $\volume =\Zd$, and
for $j \le N(\volume)$, let $\Kspace_{j} = \Kspace_{j} (\volume)$ be the
complex vector space of functions $K :
\Pcal_j (\volume) \to \Ncal (\volume)$ with the properties:
\begin{itemize}
\item Field locality: $K(X) \in \Ncal (X^{\Box},\volume)$ for each $X\in\Ccal_j$.
\item Symmetry: $K$ is Euclidean covariant, is supersymmetric if $n=0$, and is $O(n)$
invariant if $n\ge 1$.
\item Component Factorisation: $K (X) = \prod_{Y
\in {\rm Comp}_j( X)}K (Y)$ for all $X\in\Pcal_j$.
\end{itemize}
Let $\Ccal\Kcal$ denote the real vector space of functions  $K :
\Ccal_j (\volume) \to \Ncal (\volume)$ with the above properties.
\end{defn}

The symmetries mentioned in Definition~\ref{def:Kspace} are discussed
in
\cite[Section~1.6]{BS-rg-step} and \cite[Section~2.3]{BBS-phi4-log}.
They do not play an explicit
role for us now, but they are needed in results applied from
\cite{BBS-rg-pt,BS-rg-loc,BS-rg-IE,BS-rg-step}.  We do not discuss
them further here.  We have no need for the \emph{observables} discussed, e.g., in \cite{BS-rg-step}.

\subsection{Norms and regulators}
\label{sec:nr}

We recall the definitions of several norms from \cite{BS-rg-norm,BS-rg-step}.
Ultimately, we define a norm on the space $\Kcal$.  Elements of $\Kcal$ are collections
of maps $K(X)$ defined on field configurations on $X^\Box$, and the norm is designed
to control the dependence of $K(X)$ on the field (in particular, on large fields)
as well as the dependence of $K(X)$ on large polymers $X$.

\subsubsection{Norm on test functions}
\label{sec:tfnorm}

Let $\Lambda^*$ consist of sequences $z=((x_1,i_1),\ldots ,(x_p,i_p))$, with $x_k \in \Lambda$,
$i_k \in \{1,\ldots, n\}$, and $p \ge 0$ (the case $p=0$ is the empty sequence).
The set of sequences of fixed length $p$ is denoted $\Lambda^*_p$.
Fix $p_\Ncal \ge 10$.
A \emph{test function} is a function $g: \Lambda^* \to \R$ with the property that
$g_z=0$ whenever $p>p_\Ncal$.
Given $p_\Phi \ge 0$ (we take $p_\Phi=4$) and a sequence $\h_j>0$, we define
\begin{equation}
\lbeq{Phinorm}
    \|g\|_{\Phi_j(\h_j)}
    =
    \sup_{z\in \Lambda^*} \sup_{|a| \le p_\Phi}
    \h_j^{-p}
    L^{j|a|}|\nabla^a g_z|,
\end{equation}
where $|a|$ denotes the total number of discrete gradients applied by $\nabla^a$.

An important special case arises when we regard the field
$\varphi\in (\R^n)^\Lambda$ as a particular test function.
Then, with $\ell_j$ given by \refeq{elldefa},
\begin{equation}
\lbeq{phinorm}
    \|\varphi\|_{\Phi_j(\ell_j)}
    =
    \ell_j^{-1}
    \sup_{x\in \Lambda}
    \sup_{1 \le i \le n}
    \sup_{|\multia|  \le p_\Phi}
    L^{j|\multia|}
    |\nabla^{\multia} \varphi_x^i|.
\end{equation}
A local version
of \eqref{e:phinorm} is defined, for subsets
$X \subset \Lambda$, by
\begin{align}
\label{e:PhiXdef}
    \|\varphi\|_{\Phi_j(X,\ell_j)}
    &=
    \inf \{ \|\varphi -f\|_{\Phi_j(\ell_j)} :
    \text{$f \in (\R^n)^\Lambda$ such that $f_{x} = 0$
    $\forall x\in X$}\}.
\end{align}
Also, for $X$ small enough that it makes sense to define a linear function on $X$
(i.e., $X$ should not ``wrap around'' the torus), we define
\begin{align}
\label{e:PhitilXdef}
    \|\varphi\|_{\tilde\Phi_j(X,\ell_j)}
    &=
    \inf \{ \|\varphi -f\|_{\Phi_j(\ell_j)} :
    \text{$f \in (\R^n)^\Lambda$ such that $f|_X$ is a linear function}\}.
\end{align}

We may also regard the covariance $C_j$ as a particular case of a test function.
According to \refeq{cbetabd} and \refeq{elldefa},
\begin{align}
    \frac{L^{-(d-\alpha)j}}{(1+m^2L^{\alpha j})^2}
    &\le
    L^{-(d-\alpha)(j\wedge j_m)} L^{-(d+\alpha)(j-j_m)_+}
    =
    \ell_0^{-2} \ell_j^2 L^{-(\alpha-\alpha')(j-j_m)_+}
     .
\end{align}
Recall that $\alpha' < \frac 12 \alpha$
by \refeq{alphapbds}.
It follows from \refeq{scaling-estimate} that (with an $L$-dependent constant $c_L$)
\begin{equation}
\lbeq{CLbd1}
    \|C_{j}\|_{\Phi_j^+(\ell_j)}
    \le
    c_L
    \ell_0^{-2}   L^{-(\alpha-\alpha')(j-j_m)_+},
\end{equation}
where $\Phi^+$ refers to the norm
\refeq{Phinorm} with $p_\Phi$ replaced by $p_\Phi+d$ (a larger value could also have been
chosen).

In \cite{BS-rg-IE,BS-rg-step}, enhanced decay of the covariance beyond the mass
scale is exploited via a factor $\chi_j = \Omega^{-(j-j_m)_+}$ with $\Omega$ a fixed
constant often taken to equal $2$ (this factor is called $\chicCov_j$ in
\cite{BBS-phi4-log,ST-phi4,BSTW-clp} to avoid confusion with the susceptibility).
In \refeq{CLbd1}, beyond the mass scale
there is exponential decay with base $L$, which is better than base $2$ since we
take $L$ to be large.
We exploit this by setting, for a fixed  $a>0$
(which should not be confused with the multi-index used for spatial derivatives
in \refeq{Phinorm} and elsewhere),
\begin{equation}
\lbeq{chidef}
    \chiL_j=L^{- \frac 12 (\alpha-\alpha'-a)(j- j_m)_+}.
\end{equation}
Then, given $\ellconst \in (0,1]$, we can choose $\ell_0 \ge (c_L/\ellconst)^{1/2}$ to obtain,
for $j=1,\ldots,N$,
\begin{equation}
\lbeq{CLbd}
    \|C_{j}\|_{\Phi_j^+(\ell_j)}
    \le
    \ellconst \chicCov_j^2
.
\end{equation}
In \refeq{CLbd}, we have kept a factor $L^{-a(j-j_m)_+}$ in reserve.
The bound \refeq{CLbd} is a version of the requirement \cite[(1.73)]{BS-rg-IE}.

\begin{rk}
\label{rk:chiL}
For concreteness, for the case $j>j_m$ in \refeq{chidef}, we  fix $a>0$
according to $\alpha'+a = \frac 12 \alpha$, which is consistent
with our restriction
$\alpha' \in (0, \frac 12 \alpha)$  in \refeq{alphapbds}.
This concrete choice gives
\begin{equation}
\lbeq{chiLachoice}
    \chiL_j = L^{-\frac 14 \alpha(j-j_m)_+}.
\end{equation}
We have not attempted to obtain an optimal exponent beyond the mass scale.
Our choice is pragmatic: it is a choice of $\chiL_{j}$ for which we have proved
Theorem~\ref{thm:step-mr}.
\end{rk}

\subsubsection{Norms on \texorpdfstring{$\Ncal$}{N}}

For $z=((x_1,i_i),\ldots,(x_p,i_p)) \in \Lambda^*_p$, we define
$z!=p!$, and, for $F \in \Ncal$, we write
\begin{equation}
    F_z(\varphi) = \frac{\partial^p F(\varphi)}{\partial \varphi_{x_1}^{i_1} \cdots \partial \varphi_{x_p}^{i_p}}.
\end{equation}
Given $\varphi$, we define the pairing of $F \in \Ncal$ and a test function $g$ by
\begin{equation}
    \langle F,g\rangle_\varphi = \sum_{z \in \Lambda^*} \frac{1}{z!} F_z(\varphi)g_z.
\end{equation}
The $T_\varphi(\h)$-seminorm on $\Ncal$, which depends on the scale $j$,
is defined by
\begin{equation}
\lbeq{Tphidef}
    \|F\|_{T_{\varphi,j}(\h)} = \sup_{g : \|g\|_{\Phi_j(\h)}\le 1} | \langle F,g\rangle_\varphi |.
\end{equation}

Let $X \subset \Lambda$, $\varphi \in (\R^n)^{\Lambda}$, and, for $x\in\Lambda$, let
$B_{x}\in \Bcal_j$ be the unique block that contains $x$.
We define the
\emph{fluctuation-field regulator}
\begin{align}
\label{e:GPhidef}
    G_j(X,\varphi)
    =
    \prod_{x \in X} \exp
    \left(L^{-dj} \|\varphi\|_{\Phi_j (B_{x}^\Box,\ell_j )}^2 \right)
    ,
\end{align}
and
the \emph{large-field regulator}
\begin{align}
\label{e:9Gdef}
    \tilde G_j  (X,\varphi)
    =
    \prod_{x \in X}
    \exp \left(
    \frac 12 L^{-dj}\|\varphi\|_{\tilde\Phi_j (B_{x}^\Box,\ell_j)}^2
    \right)
    .
\end{align}
We use $\tilde G_j$ only for $j \le j_m$.

The two regulators serve as weights in \emph{regulator norms}.
For $Y \subset \Lambda$, let $\Ncal(Y)\subset \Ncal$ denote those
elements which are functions of $\varphi_x$ only for $x \in Y$.
Fix $\Gtilp \in (0,1]$ (appears as a power in \refeq{Gnormdef2}).
The regulator norms are defined,
for $F\in\Ncal(X^\Box)$,
by
\begin{align}
\label{e:Gnormdef1}
    \| F\|_{G_j(\ell_j)}
    &=
    \sup_{\varphi \in (\R^n)^\Lambda}
    \frac{\|F\|_{T_{\varphi,j}(\ell_j)}}{G_{j}(X,\varphi)}
    ,
\\
\label{e:Gnormdef2}
    \|F\|_{\tilde G_j^{\Gtilp}(h_j)}
    &=
    \sup_{\varphi \in (\R^n)^\Lambda}
    \frac{\|F \|_{T_{\varphi,j}(h_j)}}{\tilde{G}^{\Gtilp}_{j}(X,\varphi)}
    .
\end{align}
For the parameter $h_j$ in \refeq{Gnormdef2}, we
fix a (small) constant $k_0$,
recall the definition of $\gLfix$ in \refeq{sfix}, and
set
\begin{align}
\lbeq{hdef}
    h_j &
    = \frac{1}{\gLfix^{1/4}}k_0 L^{-\frac 12 (d-\alpha)j}
    = \frac{1}{\gLfix^{1/4}}\frac{k_0 }{\ell_0} \ell_j
    \quad\quad
    (j \le j_m)
    .
\end{align}
Since $\gLfix$ is of order $\epsilon$, $h_j$ is much larger than $\ell_j$.

\subsubsection{Norm on \texorpdfstring{$\Ccal\Kcal$}{CK}}

With $\chiL_j$ given by \refeq{chiLachoice}, we set
\begin{align}
\lbeq{epdVdef}
    \epdV &= \epdV_j(\h) =
    \begin{cases}
    \gLfix \chiL_j & (\h=\ell)
    \\
    \gLfix^{1/4} & (\h=h,\; j \le j_m)
    .
    \end{cases}
\end{align}

By Definition~\ref{def:Kspace},
an element $K\in\Ccal\Kcal$ is a collection of elements $K(X) \in \Ncal(X^\Box)$ for
polymers $X \in\Ccal$.  To control growth in the size of $X$, we fix
$r \in (0,\frac 14 2^{-d})$ and set $f_j(X) = r(|X|_j -2^d)_+$, where $|X|_j$ denotes
the number of $j$-blocks that comprise $X\in \Pcal_j$.  In particular, $f_j(X)=0$ if
$X\in \Scal_j$.
For $\Gcal_j = G_j(\ell_j)$ or $\Gcal_j = \tilde G_j^\Gtilp (h_j)$,
for $\epdV$ given by \refeq{epdVdef} with the choice dictated by $\ell_j$ vs $h_j$
in $\Gcal_j$, and for $K \in \Ccal\Kcal$,
we define a norm on $\Ccal\Kcal$ by
\begin{equation}
    \|K \|_{\Fcal_j(\Gcal_j)} = \sup_{X \in \Ccal_j} \left( \frac {1}{\epdV_j} \right)^{f(X)}
    \!\! \|K(X)\|_{\Gcal_j}
    ,
\end{equation}
and we let $\Fcal_j(\Gcal_j)$ consist of the elements of finite norm.  Let
\begin{equation}
\lbeq{gammajdef}
    \gamma_j =
    \begin{cases}
    \gLfix^{3/4} & ( j \le j_m)
    \\
    0 & (j>j_m).
    \end{cases}
\end{equation}
Finally, we define
\begin{equation}
\lbeq{Wcaldef}
    \|K\|_{\Wcal_j} = \max \left\{  \|K \|_{\Fcal(G)}, \gamma_j^3 \|K \|_{\Fcal(\tilde G)} \right\}.
\end{equation}
A difference here, compared to \cite{BS-rg-step,BBS-saw4-log,BBS-phi4-log}, is that
the $\Wcal$-norm is simply the $\Fcal(G)$-norm above the mass scale.
This innovation was first implemented in \cite{BSTW-clp}.
It is proved in \cite[Proposition~1.8]{BS-rg-step} that the vector space
$\Fcal(G) \cap \Fcal(\tilde G)$, with the $\Wcal$-norm, is a Banach space.

\subsection{The renormalisation group map}

We define a scale-dependent norm on $\Ucal \simeq \C^3$, for $U=g\tau^2+\nu\tau + u$,  by
\begin{equation}
\lbeq{Vcalnormdef}
    \|U\|_\Ucal = \max\{ |g| L^{\epsilon (j\wedge j_m)},
    |\nu|  L^{\alpha (j\wedge j_m)}, |u| L^{dj} \}.
\end{equation}
The appearance of the minimum $j\wedge j_m$ in two exponents reflects the fact that
$\tau^2$ and $\tau$ are relevant below the mass scale, but are irrelevant above the mass scale.
The norm on $\Ucal$ restricts to a norm on the subspace $\Vcal\simeq \C^2$ with $u=0$.
Given a constant $C_\DV >1$ (independent of $L$), let
\begin{align}
\lbeq{DVdef}
    \DV_j = \{ &V \in \Vcal :
    \|V\|_\Vcal \le C_\DV \gLfix,
    \;
    |{\rm Im} g| < \textstyle{\frac {1}{10}} {\rm Re} g, \;
    {\rm Re}g >  C_\DV^{-1} \gLfix L^{-\epsilon (j\wedge j_m)}
    \}.
\end{align}
Thus, $\DV_j$ requires $|g|\le C_\DV \gLfix L^{-\epsilon(j\wedge j_m)}$ and
$|\nu|\le C_\DV \gLfix L^{-\alpha(j\wedge j_m)}$,
while keeping $g$ away from zero
in a wedge about the positive real axis.

Given $C_\DV>1$, $\DVa>0$, $\delta >0$, $L>1$, and
$(m^2, \epsilon) \in  [0,\delta) \times (0,\delta)$, let
\begin{equation} \label{e:domRG}
  \domRG_j
  =
  \DV_j
  \times
  B_{\Wcal_j}(\DVa \chiL_j^3\gLfix^3)
  ,
\end{equation}
where $B_X(\rho)$ is the open ball of radius $\rho$ centred at the origin in the
Banach space $X$, $\gLfix$ is determined by $\epsilon$ and $L$ in \refeq{sfix},
and $\chiL_j$ is determined by $m^2$ in \refeq{chidef}.
The domain $\domRG_j$ is equipped
with the norm of $\Vcal \times \Wcal_j$.

To simplify the notation, we write the \emph{renormalisation group (RG) map}
as $(V,K) \mapsto (U_+,K_+)$, typically dropping subscripts $j$ and writing
$+$ in place of $j+1$.
The map depends on the mass parameter $m^2$ via the covariance of the expectation
$\Ex_{j+1}$ in \refeq{Kspace-objective}, but we
leave this dependence implicit.
The RG map
\begin{equation}
  \label{e:VKplusmap}
  U_+
  : \domRG \to \Ucal,
  \quad\quad
  K_+
  : \domRG  \to \Kspace_{+}(\Lambda),
\end{equation}
is such that $(V,K) \in \domRG_{j}$ determine
$U_+(V,K) = (\delta u_+, V_+)$ and $K_+=K_+(V,K)$,
with $I=I(V)$ and $I_+=I_+(V_+)$,
with the property that
\begin{equation}
\lbeq{EIK}
  \Ex_{+}\theta (I  \circ K)(\Lambda)
  = e^{-\delta u_+|\Lambda|}(I_{+}  \circ K_+)(\Lambda).
\end{equation}
The maps \refeq{VKplusmap} are defined in \cite{BS-rg-step}.
In addition, in \cite[Section~1.8.3]{BS-rg-step} there is a definition of
closely related maps also
on the infinite lattice $\volume=\Zd$ rather than on the torus $\volume=\Lambda_N$.

Let $\PT=\PT_j$ denote the map of Proposition~\ref{prop:PT}.
The map $U_+$ is given explicitly in \cite[(1.73)]{BS-rg-step} by
$U_+(V,K) = \PT(\hat V)$, where
\begin{equation}
\lbeq{Vhatdef}
    \hat V
    =
    V
    -
    Q(V,K),
    \quad  \quad
    Q(V,K) =
    \sum_{Y \in \Scal (\Lambda) : Y \supset B}
    \LT_{Y,B}   \frac{K (Y)}{I(V,Y)}.
\end{equation}
We use the map $R_+: \Vcal \times \Kcal \to \Ucal$, which is defined in \cite[(1.75)]{BS-rg-step} by
\begin{equation}
\lbeq{Rplushat}
    R_+(V,K) = \PT(\hat{V}) - \PT(V).
\end{equation}
Then, by definition,
\begin{equation}
\lbeq{RUPT}
  U_+(V,K) = \PT(\hat V) = \PT (V) + R_+(V,K) .
\end{equation}

For small $\delta >0$, we define the intervals
\begin{equation}
\lbeq{massint}
    \Iint_j = \begin{cases}
    [0,\delta] & j<N
    \\
    [\delta L^{-\alpha(N-1)},\delta] & j=N.
    \end{cases}
\end{equation}
The following theorem provides estimates for the maps $R_+,K_+$.
For its statement, we view $R_+,K_+$ as maps jointly on
$(V,K,m^2) \in \domRG  \times \Iint_+$ with $\domRG =\domRG(m^2)$.
The $L^{p,q}$-norm is the operator norm of a multi-linear operator
from $\Vcal^p \times \Wcal^q$ to $\Ucal_+$ or to $\Wcal_{+}$, for $R_+$
or $K_+$ respectively.
Note that the mass continuity statement only concerns scales below the mass scale,
in which case $\chiL_j=1$,
the norms on the spaces $\Ucal_j$ and  $\Wcal_j$ are independent of $m^2$, and there is no
$m^2$-dependence of the domain $\domRG_j$.

\begin{figure}
  \begin{center}
    \input{alpha_Kplusdom.tex}
  \end{center}
  \caption{
  The map $K_+$
  maps the ball of radius
      $4\CRG\gLfix^3$ to a ball of radius $\CRG\gLfix^3$.
    }
\label{fig:Kplusdom}
\end{figure}

\begin{theorem}
\label{thm:step-mr}
  Let $d =1,2,3$ and let $\volume = \Lambda_N$ or $\Zd$.
  Let $C_\DV$ and $L$ be sufficiently large, and let $p,q\in \N_0$.
  Let $0\le j<N(\volume)$.
  There exist $\CRG, C_{(p,q)}>0$ (depending on $L$), $\delta >0$,
  and $\kappa <1$,  such that,
  with the domain
  $\domRG$ defined using $\DVa =4 \CRG$, the maps
  \begin{equation}
  \lbeq{RKplusmaps}
    R_+:\domRG  \times \Iint_+  \to \Ucal_+,
    \quad
    K_+:\domRG  \times \Iint_+  \to \Wcal_{+}
  \end{equation}
  are analytic in $(V,K)$,
  and satisfy the estimates
  \begin{align}
\label{e:Rmain-g}
    \|D_V^p D_K^q R_+\|_{L^{p,q}}
    & \le
    \begin{cases}
    C_{(p,0)}
    \chiL_+ \gLfix^{3} & (p\ge 0,\, q=0)\\
    C_{(p,q)}
    \chiL_+^{-2} & (p\ge 0,\, q = 1,2)\\
    \rlap{$0$}\hspace{3.3cm}  & (p\ge 0,\, q \ge  3),
    \end{cases}
\\
\lbeq{DVKbd}
    \|D_{V}^pD_{K}^{q}K_+\|_{L^{p,q}}
    &\le
    \begin{cases}
    \CRG  \chiL_+^3  \gLfix^{3}
    &
    (p = 0,\, q=0)
    \\
    C_{(p,0)}  \chiL_+^3  \gLfix^{3-p}
    &
    (p \ge 0,\, q=0)
    \\
    \rlap{$\kappa$}\hspace{3.3cm}
    & (p=0,\, q=1)
    \\
    C_{(p,q)}  \gLfix^{-p}
    (
    \chiL_+
    \gLfix^{5/2}
    )^{1-q}
    &
    (p \ge 0,\, q \ge 1)
    .
    \end{cases}
  \end{align}
In addition, $R_+,K_+$, and every Fr\'echet derivative in $(V,K)$,
when applied as a multilinear map to directions $\dot{V}$ in $\Vcal^{p}$
and $\dot{K}$ in $\Wcal^{q}$, is jointly continuous in all arguments,
$V,K, \dot{V}, \dot{K}$,
as well as in $m^2 \in [0,L^{-\alpha j}]$.
\end{theorem}

The fact that $\kappa <1$ in \refeq{DVKbd} shows that the map $K_+$ is contractive
as a function of $K$, consistent with Figure~\ref{fig:Kplusdom}.
In fact, $\kappa$ is bounded by an inverse power of $L$, with the power depending
on whether the scale is above or below the mass scale;
the details are given in Sections~\ref{sec:kappabms}--\ref{sec:kappapms}.
A new feature in Theorem~\ref{thm:step-mr} is that the factor $\chi_+^{3/2}$ present in the
results of \cite{BS-rg-step}, which decays at an $L$-independent rate above the mass
scale, has been replaced by $\chiL_+^3$ which has better exponential decay with base $L$.
The utility of such a replacement was pointed out in \cite{BSTW-clp}, where it was an
important ingredient in the analysis of the finite-order correlation length.
We only use \refeq{Rmain-g}--\refeq{DVKbd} for $0 \le p+q\le 2$, and do not need higher-order
derivatives.

\subsection{Proof of Theorem~\ref{thm:step-mr}}
\label{sec:step-mr-pf}

Theorem~\ref{thm:step-mr} combines
\cite[Theorems~1.10, 1.11, 1.13]{BS-rg-step}
into a single statement (see also \cite[(1.61)]{BS-rg-step}).
To prove Theorem~\ref{thm:step-mr}, we apply the main result of \cite{BS-rg-step},
which in turn relies on \cite{BS-rg-IE}.
These two references focus on the 4-dimensional nearest-neighbour self-avoiding
walk, but they are more general than that.
In this section, we discuss the modifications required in our present setting,
which mainly occur above the mass scale.
 The bounds on $R_+$ in \refeq{Rmain-g}
have better powers of $\gLfix$ and worse powers of $\chiL_+$ compared to
\cite[(1.61)]{BS-rg-step}; this is discussed in Section~\ref{sec:Rbound}.
As a side remark, for scales above the mass scale it is possible to improve
the factor $\gLfix^{5/2}$ in the third case of \refeq{DVKbd} to $\gLfix$, now that we
take $\gamma_j=0$ in \refeq{gammajdef}
(we then only need the first inequality of \cite[(2.20)]{BS-rg-step}
and can take $A=r\epdV(\ell)$ there).

We assume familiarity with the methods of \cite{BS-rg-IE,BS-rg-step}.
This section can be skipped in a first reading; it is seldom referred to later in
the paper.

\subsubsection{Choice of regularity parameters}

\emph{Choice of $p_\Phi=4$.}
The parameter $p_\Phi$ is chosen to satisfy the restriction of
\cite[Proposition~1.12]{BS-rg-loc}, namely it must be greater than or
equal to
$d_{\rm min} - [\varphi]_j$, where $d_{\rm min}$ is
the least dimension of a monomial not in the range of $\LT$.
It can be verified from Table~\ref{tab:mondim}, \refeq{alphapbds}--\refeq{dimphi}, and
\refeq{Locrange} that, for $d=1,2,3$, and for both $j \le j_m$ and $j >j_m$,
all requirements are met by the choice $p_\Phi=4$.

\smallskip\noindent
\emph{Choice of $\tilde\Phi$-norm.}
We use $\tilde\Phi$ only
for $j \le j_m$, so we assume $j \le j_m$.
The determination that only linear functions $f$ are required in \refeq{PhitilXdef}
occurs as in \cite[Lemma~1.2]{BS-rg-IE}.  In our present context,
in \cite[Lemma~1.2]{BS-rg-IE} we have $d_+=\frac{d-\alpha}{2}+1$, and hence
the minimal monomial dimension $d_+'$ which exceeds $d_+$ is
$d_+'=[\nabla^2\varphi]=d_+ +1$.
Thus, in
\cite[(1.56)]{BS-rg-IE} the power of $L$ on the right-hand side becomes
\begin{equation}
    L^{-2d_+'} = L^{-d-(4-\alpha)},
\end{equation}
which suffices for the proof since $4-\alpha>2-\alpha>0$.

\smallskip\noindent
\emph{Choice of $p_\Ncal \ge 10$.}
The value of $p_\Ncal$ in \refeq{Ncaldef} remains the same as for $d=4$, namely
any $p_\Ncal \ge 10$.
This is determined by \cite[Lemma~2.4]{BS-rg-step}: the ratio $\ell_j/h_j$ here is
proportional to
$\gLfix^{1/4}$, and the one-fourth power plays the same role as
the one-fourth power $\ggen^{1/4}$ in \cite[Lemma~2.4]{BS-rg-step}.

\subsubsection{Simplified \texorpdfstring{$\Wcal$-norm}{W-norm} above the mass scale}

In \cite{BS-rg-IE}, estimates are given in terms of
norm pairs $(\|\cdot\|_j, \|\cdot\|_{j+1})$, which are either of the pairs
\begin{equation}
\label{e:np1}
    \|F\|_j = \|F\|_{G_j(\ell_j)}
    \quad \text{and} \quad
    \|F\|_{j+1} = \|F\|_{T_{0,j+1}(\ell_{j+1})},
\end{equation}
or
\begin{equation}
\label{e:np2}
    \|F\|_j = \|F\|_{\tilde{G}_j(h_j)}
    \quad \text{and} \quad
    \|F\|_{j+1} = \|F\|_{\tilde{G}_{j+1}^{\Gtilp}(h_{j+1})}.
\end{equation}
It was pointed out in \cite{BSTW-clp} that, for the nearest-neighbour model with $d=4$,
\emph{above the mass scale} it is possible to replace the two
norm pairs in \eqref{e:np1} and \eqref{e:np2} by the single new norm pair
\begin{equation}
\label{e:npmass}
    \|F\|_j = \|F\|_{G_j(\ell_j)}
    \quad \text{and} \quad
    \|F\|_{j+1} = \|F\|_{G_{j+1}(\ell_{j+1})},
\end{equation}
with $\ell_j$ given by a variant of \eqref{e:elldefa}.

That this is true also here is a consequence of the following lemma, which
is a slight adaptation of
\cite[Lemma~4.3]{BSTW-clp}.
As explained in \cite{BSTW-clp},
Lemma~\ref{lem:mart} does allow us to dispense with the $\tilde G$-norm
beyond the mass scale
and thus to set $\gamma_j=0$ for $j >j_m$ in the definition of the $\Wcal$-norm
in \refeq{Wcaldef}.

\begin{lemma}
\label{lem:mart}
Let $X \subset \Lambda$,  $j_m < j < N$, and
$t >0$.  If $L$ is sufficiently large (depending on $t$) then
\begin{equation}
\label{e:mart}
    G_{j}(X, \varphi)^{t}
    \le
    G_{j+1}(X, \varphi).
\end{equation}
\end{lemma}

\begin{proof}
Let $b \in \Bcal_j$, and let $B \in\Bcal_{j+1}$ with $b \subset B$.
By \eqref{e:GPhidef}, it suffices to show that
\begin{equation}
\lbeq{mart1}
t \|\varphi\|^2_{\Phi_j (b^\Box,\ell_j )}
\leq
L^{-d} \|\varphi\|^2_{\Phi_{j+1} (B^\Box,\ell_{j+1})}.
\end{equation}
In fact, since $\|\varphi\|_{\Phi_j (b^\Box,\ell_j )}
\leq \|\varphi\|_{\Phi_j (B^\Box,\ell_j )}$ by definition,
it suffices to prove \refeq{mart1} with $b$ replaced by $B$.
According to the definition of the norm in \eqref{e:PhiXdef},
to show this it suffices to prove that
\begin{equation}
\lbeq{martwant}
    t \|\varphi\|_{\Phi_j(\ell_j)}^2 \leq L^{-d} \|\varphi\|_{\Phi_{j+1}(\ell_{j+1})}^2,
\end{equation}
as then we can replace $\varphi$ by $\varphi -f$ in \refeq{martwant} and take the infimum.

By definition,
\begin{equation}
    \|\varphi\|_{\Phi_j(\ell_j)}
    \le
    \ell_j^{-1} \ell_{j+1}
    \sup_{x\in \Lambda} \sup_{1 \le i \le n}
    \sup_{|\multia| \leq p_\Phi}
    \ell_{j+1}^{-1}
    L^{(j+1) |\multia|}
    |\nabla^\multia \varphi_x^i|,
\end{equation}
with the inequality due to replacement of $L^{j |\multia|}$ on the left-hand
side by $L^{(j+1) |\multia|}$ on the right-hand side.
Since $\ell_j^{-1} \ell_{j+1} = L^{-[\varphi]_j} = L^{-\frac 12 (d+\alpha')}$
by \refeq{dimphi},
\begin{equation}
    \|\varphi\|_{\Phi_j(\ell_j)}
    \leq
    L^{-\frac 12 (d+\alpha')}
    \|\varphi\|_{\Phi_{j+1}(\ell_{j+1})},
\end{equation}
and hence
\begin{equation}
t \|\varphi\|_{\Phi_j(\ell_j)}^2
\leq t L^{-\alpha'} L^{-d}  \|\varphi\|^2_{\Phi_{j+1}(\ell_{j+1})}.
\end{equation}
Since $\alpha'>0$ by \refeq{alphapbds}, \refeq{martwant}
follows once $L$ is large enough that $tL^{-\alpha'}\leq 1$.
\end{proof}

\subsubsection{Small parameter \texorpdfstring{$\epV$}{epsilonV}}

A scale-dependent
small parameter $\epsilon_V$ controls the size of $V \in \Vcal$ for stability
estimates.  It is defined and discussed in detail in \cite[Section~1.3.3]{BS-rg-IE},
where it  is given by
\begin{equation}
    \epV = \epV(\h_j)
    =
    L^{dj} \left(  \|g\tau^2_x\|_{T_0(\h_j)} +  \|\nu\tau_x\|_{T_0(\h_j)} \right).
\end{equation}
We use two separate choices for $\h_j$, namely
$\h_j=\ell_j= \ell_0 L^{-\frac 12 (d-\alpha)j} L^{- \alphahat (j-j_m)_+}$
from \refeq{elldefa},
and for $j\le j_m$ also $\h_j=h_j=\gLfix^{-1/4}k_0L^{-\frac 12 (d-\alpha)j}$
from \refeq{hdef}.  Each of the choices defines a value for $\epV$.
Computation gives
\begin{equation}
\lbeq{epV}
    \epV \asymp
     |g|\, L^{dj} \h_j^4 + |\nu|\, L^{dj} \h_j^2 .
\end{equation}
The evaluation of the right-hand side is given next, below and above the mass scale.
Stability domains for $V$ and $\Vpt$ are then as discussed in \cite[Section~1.3.4]{BS-rg-IE}.
In particular, \cite[Proposition~1.5]{BS-rg-IE} applies in our present context.

\smallskip \noindent \emph{Below the mass scale.}
Let $V \in \DV_j$.
For $j \le j_m$, \refeq{epV} gives
\begin{equation}
\lbeq{epVbd}
    \epV(\h_j)
    \le
    \begin{cases}
    O(\gLfix) & (\h=\ell)
    \\
    O(k_0) & (\h=h).
    \end{cases}
\end{equation}
The powers of $L$ in \refeq{Vcalnormdef} are exactly those that cancel the exponential
growth due the relevant monomials $\tau,\tau^2$.
The small parameter $\epVbar=L^{dj}\|g\tau_x^2\|_{T_0(h_j)}$ of \cite[(1.81)]{BS-rg-IE}
obeys the important stability bound $\epVbar(h) \asymp k_0^4$,
as in \cite[(1.90)]{BS-rg-IE}.

\smallskip \noindent \emph{Above the mass scale.}
Let $V \in \DV_j$.
For $j>j_m$, we only use the case $\h_j=\ell_j$.  In this case, according to \refeq{elldefa},
$\ell_j=\ell_0 L^{-\frac 12 (d-\alpha)j_m} L^{-\frac 12 (d+\alpha')(j-j_m)}$,
and computation gives
\begin{equation}
\lbeq{epVabovejm}
    \epV(\ell_j) \le
    O(\gLfix) ( L^{-(d+2\alpha')(j- j_m)_+}
    +
    L^{-\alpha' (j- j_m)_+})
    =
    O(
    \gLfix L^{-\alpha' (j- j_m)_+}).
\end{equation}

\subsubsection{Small parameter \texorpdfstring{$\epdV$}{epsilonVbar}}
\label{sec:epdV}

Let $V \in \DV_j$, $b\in \Bcal_j$, $\Upt = \PT_j(V)= g_\pt \tau^2 + \nu_\pt \tau +\delta u_\pt$, and
\begin{equation}
    \delta V = \theta V - \Upt = (\theta V -V) + (V-\Upt).
\end{equation}
An essential feature of $\epdV$ is
that
$\|\delta V(b)\|_{T_0(\h\sqcup \hat\ell)}$ (norm at each scale $j$ and $j+1$)
should be bounded by an $L$-dependent multiple of
$\epdV$;
see \cite[Sections~3.3, 1.3.5]{BS-rg-IE}.
Here, $\hat\ell$ is defined by
\begin{equation}
\lbeq{ellhatdef}
   \hat\ell_j = \hat\ell_0 \ell_0 L^{-\frac 12 (d-\alpha)j}L^{-\alpha(j-j_m)_+}
    = \hat\ell_0 \ell_j L^{-\frac 12 (\alpha-\alpha')(j-j_m)_+},
\end{equation}
where $\hat\ell_0$ is a constant
chosen as indicated below \cite[(3.17)]{BS-rg-IE}.
The utility of $\hat\ell_j$ is that $\hat\ell_j^2$
gives a faithful measure of the decay of $C_{j+1;x,y}$ in \refeq{CLbd1}.
By \refeq{alphapbds}, $\alpha-\alpha'>\frac 12 \alpha >0$, so $\hat\ell_j$ is exponentially smaller than $\ell_j$, above the mass scale.

We argue next that the value $\epdV$ given in \refeq{epdVdef}
(with $\chiL_j = L^{-\frac 14 \alpha(j-j_m)_+}$ as in \refeq{chiLachoice})
does provide the required bound on $\|\delta V(b)\|_{T_0(\h\sqcup \hat\ell)}$.

\smallskip \noindent \emph{Below the mass scale.}
Consider first the case $j \le j_m$, for which $\hat\ell_j$ is a constant
multiple of $\ell_j$.
It is straightforward to estimate $V-\Upt$ using Proposition~\ref{prop:PT}
and Lemma~\ref{lem:wlims}, and, with minor bookkeeping changes,
the result of \cite[Lemma~3.4]{BS-rg-IE} applies with $\epdV$ given by the
two options in \refeq{epdVdef} for $j \le j_m$.
We illustrate this with some sample terms.

A linear term (in the coupling constants)
that arises in $V-\Upt$ is $g\eta'\tau$, and
\begin{equation}
\lbeq{Wickbd}
     \|g\eta'\tau(b)\|_{T_0(\h )}
    \le c\times
    \begin{cases}
    \gLfix L^{-\epsilon j}
    L^{-j(d-\alpha)}
    L^{jd}
    L^{-(d-\alpha)j}
    =
    \gLfix & (\h=\ell)
    \\
    \gLfix L^{-\epsilon j}
    L^{-(d-\alpha)j}
    L^{jd}
    \gLfix^{-1/2}
    L^{-j(d-\alpha)}
    =
    \gLfix^{1/2} & (\h=h).
    \end{cases}
\end{equation}
Another term in $V-\Upt$ is $\delta u_\pt$.  Its norm
on a block $b$ (for either choice of $\h$),
is simply $L^{dj}|\delta u_\pt|$.  In view of \refeq{uptlong}
and \refeq{betabd}, $L^{dj}|\delta u_\pt|$ is bounded above by $O(\gLfix)$.
It can be checked that  the right-hand side of \refeq{Wickbd} is an upper bound on
$\|V(b)-\Upt(b)\|_{T_0(\h)}$.

A typical term in $\theta V - V$ is $g(\zeta\cdot\varphi)|\varphi|^2$, whose norm on $b$ is
\begin{equation}
\lbeq{zetaphi3}
    \|g(\zeta\cdot \varphi)|\varphi|^2(b)\|_{T_0(\h\sqcup \hat\ell)} \le c |g| \hat\ell_j \h_j^3L^{dj}
    = c\left(\frac{\hat\ell_j}{\h_j}\right)|g| \h_j^4 L^{dj}
    \le
    \begin{cases}
    O(\gLfix)   & (\h=\ell)
    \\
    O(\hat\ell_j/h_j)
     & (\h=h),
    \end{cases}
\end{equation}
and since $\hat\ell_j/h_j = \ell_j/h_j \le O(\gLfix^{1/4})$, this gives
\begin{equation}
    \|g(\zeta\cdot \varphi)|\varphi|^2(b)\|_{T_0(\h\sqcup \hat\ell)}
    \le
    \begin{cases}
    O(\gLfix)  & (\h=\ell)
    \\
    O(\gLfix^{1/4})  & (\h=h).
    \end{cases}
\end{equation}

\smallskip \noindent \emph{Above the mass scale.}
For $j >j_m$, we only use $\h=\ell$, which now has improved decay.
Also, $\LT$ no longer extracts $\tau^2$, so $g_\pt =g$ as in \refeq{gpt2a}.
We indicate now that $\epdV$ provides an upper bound on the norm of $\delta V(b)$,
by verifying that the $\gLfix$ bound obtained below the mass scale can be
improved to $\epdV$.
In fact, $\epdV$ is a crude upper bound, but it is sufficient for our needs.
We again look only at typical terms, as we did below the mass scale.

The bound on the left-hand side of \refeq{Wickbd} now becomes
\begin{align}
    &\gLfix L^{-\epsilon j_m} L^{-(d-\alpha)j_m}L^{-(d+\alpha)(j-j_m)}
    L^{dj_m}L^{d(j-j_m)} L^{-(d-\alpha)j_m}L^{-(d+\alpha')(j-j_m)}
    \nnb &
    =
    \gLfix
    L^{-(d+\alpha+\alpha')(j-j_m)},
\end{align}
which is (much) better than $\epdV$.
The bound on the left-hand side of \refeq{zetaphi3} now becomes
\begin{equation}
    \hat\ell_j \ell_j^{-1} \gLfix L^{-\epsilon j_m} \ell_j^4 L^{dj}
    = O(\gLfix) L^{-\frac 12 (\alpha-\alpha')(j-j_m)} L^{-2(d+2\alpha' )(j-j_m)},
\end{equation}
which is again better than $\epdV$.
It is straightforward to verify the remaining estimates.
For a final example, the contribution to $L^{dj}\delta u_\pt$
(which occurs in $V-\Upt$) due to $L^{dj}\kappa_{g\nu}'g\nu$ is at most
(recall \refeq{kappadef} and \refeq{betabd})
\begin{equation}
    L^{dj} L^{(\alpha +\epsilon) j_m}L^{\alpha(j-j_m)} M_jL^{-dj}
    \gLfix L^{-\epsilon j_m} \gLfix L^{-\alpha j_m}
    \le \gLfix^2 L^{-\alpha(j-j_m)}.
\end{equation}

It is apparent from the above estimates that a smaller choice of $\epdV$ could be
obtained as an upper bound on the norm of $\delta V(b)$ above the mass scale.
We have made a choice of $\epdV$ that remains consistent with the requirements
of the crucial contraction, discussed next.

\subsubsection{Crucial contraction below the mass scale}
\label{sec:kappabms}

The \emph{crucial contraction} refers to the application of \cite[Proposition~5.5]{BS-rg-step}
in \cite[Lemma~5.6]{BS-rg-step}.
It produces the bound $\kappa<1$ on $D_KK_+$ in \refeq{DVKbd}, which is essential to
prevent the effect of $K$ from being magnified in $K_+$.
Below the mass scale,
the crucial contraction works the same way for both $\h=h$ and $\h=\ell$, since
each
scales the same way with $L$, namely $\h_{j+1}/\h_j = L^{-\frac 12 (d-\alpha)}$.
Thus, the gain
is the same under change of scale, for both norms, namely $\kappa = O(L^d \cgam)$
with $\cgam$ equal to
the reciprocal of $L$ raised to a power equal to the dimension of the
least irrelevant of the
symmetric irrelevant monomials.
Suppose that $j \le j_m$,
and recall Table~\ref{tab:mondim}.
We write $\tau = \frac 12 |\varphi|^2$ as in \refeq{tauxdef}.
The irrelevant monomials of smallest dimensions are:
\begin{align}
    [\tau^3]
    &=
    (d-\epsilon) + (d-\alpha),
    \quad
    [\nabla^2 \tau]
    =
    2+ (d-\alpha).
\end{align}
For  $d=1$ and $d=2$, $[\tau^3]$ is smaller, whereas $[\nabla^2\tau]$ is smaller for $d=3$.
Therefore, $\cgam$ of \cite[(5.32)]{BS-rg-step} is modified to become
\begin{equation}
\lbeq{kappa}
    \cgam =
    \begin{cases}
     L^{-[\tau^3]} & (d=1,2)
     \\
     L^{-[\nabla^2\tau]} & (d=3),
    \end{cases}
    \quad\quad
    L^d \cgam =
    \begin{cases}
    L^{-\frac 12 + \frac{3\epsilon}{2}} & (d=1)
    \\
    L^{-1 + \frac{3\epsilon}{2}} & (d=2)
    \\
    L^{-\frac 12+ \frac\epsilon 2} & (d=3).
    \end{cases}
\end{equation}
The factor $L^d$ multiplying $\cgam$ is the entropic factor arising in the transition from
\cite[(5.38)]{BS-rg-step} to \cite[(5.39)]{BS-rg-step}.
Below the mass scale, we can take $\kappa=O(L^d\cgam)$ to be given by the above formulas.
Since $\epsilon$ is as small as desired, $\kappa$ is of the order of an inverse power of $L$,
for scales $j \le j_m$.

\subsubsection{Crucial contraction  above the mass scale}
\label{sec:kappapms}

Suppose $j>j_m$.
As discussed in Section~\ref{sec:nr}, above the mass scale we use only  $\h=\ell$ and not $\h=h$.
The estimates we obtain here are not canonical ones, but they are sufficient.
A new feature, compared to \cite{BS-rg-step}, is that $\epdV$ now decays exponentially
with base $L$.  In fact, according to \refeq{epdVdef} and \refeq{chiLachoice},
\begin{equation}
\lbeq{epdVabovejm}
    \epdV = \gLfix L^{-\frac 14 \alpha(j-j_m)}.
\end{equation}
We must verify that \cite{BS-rg-step} does provide this exponential decay beyond
the mass scale.  This requires a certain consistency between the perturbative contribution
to $K_+$ and the crucial contraction, and we verify this consistency here.

\smallskip \noindent \emph{Perturbative contribution to $K$.}
The perturbative contribution to $K_+$ is the value $K_+(V,0)$ arising
from $K=0$.  For $d=4$, this is estimated in \cite[(2.10)]{BS-rg-step}, as
$\|K_+(V,0)\|_{\Fcal_+} \le O(\epdV_+^3)$ (with the value of $\epdV_+$ suitable for $d=4$).
With our current definition of the $\Wcal$ norm
in \refeq{Wcaldef}, this estimate translates as $\|K_+(V,0)\|_{\Wcal_+} \le O(\epdV_+^3)$.
This estimate relies on the fact that $\epdV$ provides a bound on the norm of $\delta V$.
We have verified this fact above, with $\epdV$ given by \refeq{epdVabovejm},
and can therefore conclude that in our present context
$\|K_+(V,0)\|_{\Wcal_+}$ is bounded above by an
$L$-dependent multiple of
\begin{equation}
    \epdV_+^3 = \gLfix^3 L^{-\frac 34 \alpha(j+1-j_m)}.
\end{equation}

\smallskip \noindent \emph{Crucial contraction above the mass scale.}
According to Table~\ref{tab:mondim}, all monomials except $1$ are irrelevant,
and
\begin{align}
    [\nabla^2 \tau] &
    = (d + \alpha')  + 2,
    \quad
    [\tau^2]
    = (d + \alpha') + (d + \alpha') .
\end{align}
We have kept $\tau$ in the range of $\LT$, despite its irrelevance above the mass scale.
Consequently, it is the least irrelevant monomial beyond $\tau$ that determines
the estimate for the crucial contraction.
For $d=2,3$, we have
$[\tau^2]>[\nabla^2\tau]$, so $\nabla^2 \tau$ is the least irrelevant
monomial after $\tau$.  For $d=1$, instead $\tau^2$ is the least irrelevant.
Thus \refeq{kappa} now becomes
\begin{equation}
\lbeq{kappajm}
    \cgam =
    \begin{cases}
    L^{-[\tau^2]} & (d=1)
    \\
    L^{-[\nabla^2\tau]} & (d=2,3),
    \end{cases}
    \quad\quad
    L^d \cgam =
    \begin{cases}
    L^{-(1+2\alpha')} & (d=1)
    \\
    L^{-(2+\alpha')} & (d=2,3).
    \end{cases}
\end{equation}
Thus we can take
\begin{equation}
\lbeq{kappa-above-jm-5}
    \kappa =
    O(L^d\cgam )
    =
    \begin{cases}
    L^{-(1+2\alpha')} & (d=1)
    \\
    L^{-(2+\alpha')} & (d=2,3).
    \end{cases}
\end{equation}
For future reference, we observe that since $\alpha = \frac 12 (d+\epsilon)$,
\begin{equation}
\lbeq{kappa-small-above-jm}
    \kappa   L^{\frac 34 \alpha}
    =
    \begin{cases}
    O(L^{-(\frac 58 +2\alpha' - \frac 38 \epsilon)}) & (d=1)
    \\
    O(L^{-(\frac 54 +\alpha'- \frac 38 \epsilon)}) & (d=2)
    \\
    O(L^{-(\frac {7}{8} +\alpha'- \frac 38 \epsilon)}) & (d=3),
    \end{cases}
\end{equation}
and the right-hand side is as small as desired (by taking $L$ large).

\smallskip \noindent \emph{Consistency of the above two effects.}
The perturbative and contractive effects come together in the estimate \cite[(2.30)]{BS-rg-step},
whose first inequality becomes, in our present setting,
\begin{equation}
    \|K_+\|_{\Wcal_+} \le c \vartheta_+^3 \gLfix^3 + O(\kappa)\|K\|_\Wcal.
\end{equation}
For $\|K\|_{\Wcal} \le 4\CRG \epdV^3$ (consistent with \refeq{domRG}), this gives
\begin{equation}
    \|K_+\|_{\Wcal_+} \le c \vartheta_+^3 \gLfix^3 + O(\kappa) \vartheta^3 \gLfix^3
    = (c + O(\kappa L^{\frac 34 \alpha})) \vartheta_+^3 \gLfix^3.
\end{equation}
By \refeq{kappa-small-above-jm}, the term $O(\kappa L^{\frac 34 \alpha})$ is as small
as desired, and we obtain the required estimate
$\|K_+\|_{\Wcal_+} \le 2c \vartheta_+^3 \gLfix^3$.
(A minor detail is that $\epdV_+$ in \cite[(2.24)]{BS-rg-step} should be replaced here
by $\epdV$; the improvement to $\epdV_+$ in the proof of
\cite[Theorem~2.2(i)]{BS-rg-step}, explained in \cite[Section~7]{BS-rg-step},
is not actually used in \cite[(2.24)]{BS-rg-step}.
The improvement was inconsequential in \cite{BS-rg-step} but here would cost a factor
$L^{\frac 14\alpha}$.)

\subsubsection{Bound on $R_+$}
\label{sec:Rbound}

We now discuss the proof of the bound on $R_+$ stated in \refeq{Rmain-g}.
The estimate \refeq{Rmain-g} is an estimate for $R_+$ as a map into a space of
polynomials measured with the $\Ucal$ norm.  Estimates on $R_+$ are more naturally
carried out when $R_+(B)$ (for a block $B \in \Bcal_+$) is measured with the $T_0$ norm.
We claim that,
under the hypotheses of Theorem~\ref{thm:step-mr}, and with
$R_+(B)$ as a map into a space
with norm $T_0$,
\begin{align}
\label{e:Rmain-g-improved}
    \|D_V^p D_K^q R_+(B)\|_{\to T_0}
    & \le
    \begin{cases}
    C_{(p,0)}
    \chiL_+^3 \gLfix^{3} & (p\ge 0, \, q=0)\\
    C_{(p,q)}
    & (p \ge 0, \, q=1,2)  \\
    \rlap{$0$}\hspace{3.3cm}  & (p\ge 0,\, q \ge  3).
    \end{cases}
\end{align}
Worse estimates than \refeq{Rmain-g-improved} are proved in \cite{BS-rg-step} using Cauchy
estimates.
The improved estimates are obtained using
explicit computation of the derivatives (in fact, Cauchy estimates could also be used with a larger
domain of analyticity to give the improvement in \cite[(1.61)]{BS-rg-step}).

The difference between \refeq{Rmain-g} and \refeq{Rmain-g-improved} occurs only above the mass scale.
To conclude \refeq{Rmain-g} from \refeq{Rmain-g-improved}, it suffices to show that
\begin{equation}
\lbeq{Rnormcomp}
    \|R_+\|_\Ucal \le O(L^{\alpha'(j-j_m)_+})\|R_+(B)\|_{T_0},
\end{equation}
since $L^{\alpha'(j-j_m)} \le \chiL^{-2}$ because $\alpha' < \frac 12 \alpha$.
(Below the mass scale, the $\Ucal$ and $T_0$ norms are comparable on polynomials of
the form $g\tau^2+\nu\tau+u$.)
It is the growth factor on the right-hand side of \refeq{Rnormcomp}
that creates the need for the $L$-dependent
factor $\chiL$ in our estimates for $R_+$ and $K_+$ above the mass scale.
The following lemma proves \refeq{Rnormcomp} and more.

\begin{lemma}
\label{lem:monnormcomp}
Let $F_1 = \nu \tau + u$ and $F_2=g\tau^2+\nu\tau$.  There are constants $c>0$ (independent
of $L$) and $c_L$ (depending on $L$) such that, for $B$ a block at the scale of
the norms,
\begin{equation}
    \|F_1\|_{\Ucal} \le c_L L^{\alpha'(j-j_m)_+}\|F_1(B)\|_{T_0},
    \quad
    \quad
    \|F_2(B)\|_{T_0} \le c L^{-\alpha'(j-j_m)_+} \|F_2\|_{\Ucal}.
\end{equation}
\end{lemma}

\begin{proof}
The $T_0$ norm of $F_i$ is equivalent to the sum of the norms of the monomials in $F_i$.
Also, using the definition of $\ell_j$ in \refeq{elldefa}, we have (with $L$-dependent
constants)
\begin{align}
    \|\tau(B)\|_{T_0} &\asymp L^{dj}\ell_j^2 \asymp L^{\alpha(j\wedge j_m)}L^{-\alpha'(j-j_m)_+},
    \\
    \|\tau^2(B)\|_{T_0} &\asymp L^{dj}\ell_j^4 \asymp L^{\epsilon(j\wedge j_m)}
    L^{-(2\alpha+2\alpha'-\epsilon)(j-j_m)_+}.
\end{align}
Therefore,
\begin{align}
    \frac{\|F_1\|_{\Ucal}}{\|F_1(B)\|_{T_0}}
    & \le
    \frac{|\nu| L^{\alpha (j \wedge j_m)} + |u|L^{dj}}{|\nu|\|\tau(B)\|_{T_0} + |u|L^{dj}}
    \le
    O(L^{\alpha'(j-j_m)_+} + 1) = O(L^{\alpha'(j-j_m)_+}).
\end{align}
The proof for $F_2$ is similar, with the $\tau$ term dominating the $\tau^2$ term.
The constant is independent of $L$ in the bound on $F_2$ because the $L$-dependence
of the $T_0$ norms arises as a power of $\ell_0$ (which is large depending on $L$),
and this goes in the helpful direction in the bound on $F_2$.
\end{proof}

For the proof of \eqref{e:Rmain-g-improved}, we first recall from \refeq{Rplushat}
that, by definition,
\begin{align}
\lbeq{RplusVQ-1}
    R_+(V,K) & = \PT(V-Q)-\PT(V) ,
\end{align}
with $Q=Q(V,K)$ given by \refeq{Vhatdef} and the map $\PT$ given by \refeq{Vptdef}.
The map $\PT$ is quadratic, and $Q$ is linear in $K$, so
$R_+$ is quadratic in $K$ and hence three or more $K$-derivatives
must vanish.   This proves the third case of \refeq{Rmain-g-improved}.

For the substantial cases of \refeq{Rmain-g-improved}, we must look into the definition
of $R_+$ more carefully.  For this, it is useful to extend the definitions of $W$ in
\refeq{WLTF} and $P$ in \refeq{PdefF} by defining
\begin{align}
    W_j(V_x,\tilde V_y)
    &=
    \frac 12 \sum_{x\in X} (1-\LT_{x}) F_{w_j}(V_x,\tilde V_y),
    \\
    P_j(V_x, \tilde V_y)
    &=
    \LT_x \left(
    \Ex_{C_{j+1}}\theta W_j(V_x,\tilde V_y)
    + \frac 12
    F_{C_{j+1}}
    ( \Ex_{C_{j+1}}\theta V_x, \Ex_{C_{j+1}}\theta \tilde  V_y)
    \right).
\end{align}
Changes are needed in estimates on $W$ and $P$ above the mass scale, compared to \cite{BS-rg-IE}:
(1) now $\tau$ and $\tau^2$ are irrelevant monomials, and thus $V$ no
longer satisfies a hypothesis needed to apply
\cite[Proposition~4.10]{BS-rg-IE} to bound $W$, and (2)
in the $\epdV^2$ bound required for $F,W,P$ as in
\cite[Proposition~4.1]{BS-rg-IE}, we now need $\epdV$ to include a factor
$L^{-\frac 14 \alpha (j-j_m)}$.
The following lemma gives more than is needed, and implies the required
$\epdV^2$ bounds.

\begin{lemma}
\label{lem:W1}
There exists $c>0$ such that, for
$j_m < j \le N$, $B_j \in \Bcal_j$, large $L$, and $V, \tilde V \in \Ucal$,
\begin{align}
\lbeq{Wlembd1}
     \sum_{x \in B_j} \sum_{y \in \Lambda}\|F(V_x,\tilde V_y)\|_{T_0(\ell_j)}
    &\le
    c L^{-(\alpha+\alpha')(j-j_m)_+}\|V\|_{\Ucal_j}\|\tilde V\|_{\Ucal_j},
    \\
    \sum_{x \in B_j} \sum_{y \in \Lambda} \|W_j(V_x,\tilde V_y) \|_{T_{0}(\ell_j)}
    &\le
    c (c /L)^{(\alpha+\alpha') (j-j_m)_+} \|V\|_{\Ucal_j} \|\tilde V\|_{\Ucal_j},
    \\
    \sum_{x \in B_j} \sum_{y \in \Lambda}\|P(V_x,\tilde V_y)\|_{T_0(\ell_j)}
    &\le
    c (c/L)^{(\alpha+\alpha')(j-j_m)_+}\|V\|_{\Ucal_j}\|\tilde V\|_{\Ucal_j}.
\end{align}
\end{lemma}

\begin{proof}
We may assume without loss of generality that $V$ and $\tilde V$ have no constant
term, since such terms make no contribution to $F,W,P$.

By \cite[Lemma~4.7]{BS-rg-IE},
\begin{equation}
\lbeq{FVV}
    \sum_{x \in B_j} \sum_{y\in \Lambda} \|F_{C_{j}} ( V_x,\tilde V_y )\|_{T_0}
    \le
    O(L^{2dj}) \|C_j\|_{\Phi_j} \|V_x\|_{T_{0,j}} \|\tilde V_y\|_{T_{0,j}} .
\end{equation}
By \refeq{CLbd1} and our choice of $\ell_0$ above \refeq{CLbd},
$\|C_j\|_{\Phi_j} \le L^{-(\alpha-\alpha' )(j-j_m)}$.
By Lemma~\ref{lem:monnormcomp}, $\|V_x\|_{T_0} \le L^{-dj} L^{-\alpha'(j-j_m)_+} \|V\|_\Ucal$,
and the desired bound on $F$ follows.

For the bound on $W$,
\cite[Proposition~4.10]{BS-rg-IE} applies below the mass scale, so we only consider
scales $j >j_m$.
In this case, the $\Ucal$ norm is scale independent,
and we adapt the proof of \cite[Proposition~4.10]{BS-rg-IE},
as follows.
Let
\begin{equation}
\lbeq{WAj}
    A_j = \sum_{x\in B_j}\sum_{y \in \Lambda} \|W_j(V_x,\tilde V_y)\|_{T_0(\ell_j)}.
\end{equation}
We prove, by induction on $j$, that
$A_j \le c  (c /L)^{(\alpha+\alpha') (j-j_m)}\|V\|_\Ucal \|\tilde V\|_\Ucal$,
with $c$ to be determined during the proof.
The base case $j=j_m$ holds by
\cite[Proposition~4.10]{BS-rg-IE}, which does apply until the mass scale.
We assume that the induction hypothesis holds for $j-1$, and prove that it holds also for $j$.
Recall from \cite[Lemma~4.6]{BBS-rg-pt} that
\begin{equation}
\label{e:WWFj1}
    W_{j} (V_x,\tilde V_y)
    =
    (1-\LT_{x})
    \Big(
    e^{\Lcal_{j}}
    W_{j-1} (e^{-\Lcal_{j}} V_x , e^{-\Lcal_{j}} \tilde V_y)
    +
    \frac{1}{2} F_{C_{j}} ( V_x,\tilde V_y )
    \Big)
    .
\end{equation}

We first consider the term $\frac{1}{2} (1-\LT_{x}) F_{C_{j}} ( V_x,\tilde V_y )$.
The operator $1-\LT_x$ is bounded as an operator on $T_{0}(\ell_j)$, as in \cite[(4.33)]{BS-rg-IE},
and our estimate on $F$ shows that there exists $a>0$ (independent of $L$) such that
\begin{align}
    \frac 12 \sum_{x \in B_j} \sum_{y\in \Lambda} \|(1-\LT_x)F_{C_{j}} ( V_x,\tilde V_y )\|_{T_0(\ell_j)}
    &\le
    a
    L^{-(\alpha+\alpha') (j-j_m)}
    \|V\|_\Ucal \|\tilde V\|_\Ucal
    .
\end{align}

As an operator from $T_0(\ell_{j-1})$ to $T_0(\ell_j)$,
from a small extension of \cite[Proposition~4.8]{BS-rg-IE}
(to identify $d'$) we find that $1-\LT_x$ acts here as a contraction
whose operator norm is at most a multiple of $L^{-d'}$, where $d'$ is the dimension of
the least irrelevant monomial that is not in the domain of $\LT_x$.
As in \refeq{kappajm}, $d'=[\tau^2]=2+2\alpha'$ for $d=1$,
and $d'=[\nabla^2 \tau]= d+2+\alpha'$ for $d=2,3$.
As in \cite[(4.21)]{BS-rg-IE}, the operator $e^{\Lcal_j}$ has bounded norm as an operator
on $T_0(\ell_{j-1})$.  This leads, for some $b>0$ (independent of $L$), to
\begin{align}
    & \sum_{x \in B_j} \sum_{y\in \Lambda}
    \|(1-\LT_{x})
    e^{\Lcal_{j}}
    W_{j-1} (e^{-\Lcal_{j}} V_x , e^{-\Lcal_{j}} \tilde V_y)\|_{T_0(\ell_j)}
    \nnb & \qquad \le
    b
    L^{-d'}
    \sum_{x \in B_j} \sum_{y\in \Lambda}\|W_{j-1} ( V_x ,  \tilde V_y)\|_{T_0(\ell_{j-1})}
    \nnb & \qquad \le
    b
    L^{-d'} L^d A_{j-1}
    \le
    bL^{-(d'-d)} c   (c/L)^{(\alpha+\alpha') (j-1-j_m)}
    \|V\|_\Ucal \|\tilde V\|_\Ucal
    ,
\end{align}
where the last inequality uses the induction hypothesis.

After assembling the above estimates, and assuming that $c \ge 1$,
we find that
\begin{align}
    A_j & \le
    c
    \|V\|_\Ucal \|\tilde V\|_\Ucal
    \left(
    b  L^{-(d'-d-\alpha-\alpha')}
    +
    a c^{-1}
    \right)
    (c/L)^{(\alpha+\alpha') (j-j_m)}
    .
\end{align}
It suffices if the sum in parentheses is at most $1$.
We have
\begin{equation}
    d'-d-\alpha-\alpha' =
    \begin{cases}
    1+\alpha'-\alpha & (d=1)
    \\
    2-\alpha & (d=2,3),
    \end{cases}
\end{equation}
which is positive in all cases $d=1,2,3$.  Thus,
it is sufficient to choose $c\ge 2a$ so that $a c^{-1} \le \frac 12$, and
$L$ large enough that $b  L^{-(d'-d-\alpha-\alpha')} \le \frac 12$.

Finally, the bound on $P$ follows from the bounds on $F,W$ as in \cite[Proposition~4.1]{BS-rg-IE}.
\end{proof}

\begin{proof}[Proof of  \refeq{Rmain-g-improved}]
It is of no concern if
constants in estimates here are $L$-dependent, so the distinction between $\chiL$
and $\chiL_+$ is unimportant.
 We use $\prec$ in this
proof to denote bounds with (omitted) $L$-dependent constants.  Also, to simplify
the notation, we omit $B$ in $T_0$ norms such as $\|R_+(B)\|_{T_0}$.
The case $q \ge 3$ has been discussed already.

By definition,
\begin{align}
\lbeq{RplusVQ}
    R_+ & = \PT(V-Q)-\PT(V) =
    -\Ex\theta Q + 2P(V,Q) - P(Q,Q).
\end{align}
For the case $p=q=0$, we apply Lemma~\ref{lem:W1}, use the fact that
$\Ex\theta$ is a bounded operator on polynomials of bounded degree (with respect to the $T_0$ norm)
by \cite[(4.21)]{BS-rg-IE}, and then apply Lemma~\ref{lem:monnormcomp}.  This  gives
\begin{align}
    \|R_+(B)\|_{T_0} & \prec
    \| Q \|_{T_0} +  \|P(V,Q)\|_{T_0} + \|P(Q,Q)\|_{T_0}
    \nnb
    & \prec
    \|Q\|_{T_0} + (c/L)^{(\alpha+\alpha')(j-j_m)_+} (\|V\|_\Ucal + \|Q\|_\Ucal)\|Q\|_\Ucal
    .
\end{align}
The assumption $\|K\|_{T_0} \prec \vartheta_j^3 \gLfix^3$
implies that $\|Q\|_{T_0} \prec \vartheta_j^3 \gLfix^3$.  Also, $\|V\|_\Ucal \prec \gLfix$
by assumption.
We apply Lemma~\ref{lem:monnormcomp},
as well as
\begin{equation}
\lbeq{pKalphap}
    L^{\alpha'(j-j_m)_+} \chiL_+^3  \le
    L^{-(\frac 34 \alpha - \alpha')(j+1-j_m)_+}
    \le
    L^{-\frac 14 \alpha (j+1-j_m)_+}
    = \chiL_+
\end{equation}
(since $\alpha'  < \frac 12 \alpha$ by \refeq{alphapbds}),
to obtain
\begin{align}
\lbeq{Qratio}
    \|Q\|_{\Ucal}
    &
    \prec
    L^{\alpha'(j-j_m)_+} \|Q(B)\|_{T_0(\ell_+)}
    \prec
    L^{\alpha'(j-j_m)_+}\vartheta_j^3 \gLfix^3
    \prec
    \vartheta_+ \gLfix^3
    .
\end{align}
Since $\alpha + \alpha' \geq \frac 34 \alpha$, we obtain the desired estimate
$\|R_+(B)\|_{T_0} \prec \vartheta_+^3 \gLfix^3$, and the
proof is complete for the case $p=q=0$.

For $K$ derivatives, let $\dot Q = D_K(Q,\dot K)$.  Then
\begin{equation}
\lbeq{RplusKderivs}
    D_KR_+(\dot K) = -\Ex \theta \dot Q + 2P(V,\dot Q) - 2P(\dot Q,Q),
    \quad
    D_KR_+(\dot K,\ddot K) =  - 2P(\dot Q,\ddot Q),
\end{equation}
and estimates with the $T_0$ norm, like those used previously, give the desired
result when $p=0$ and $q=1,2$.

The proof of the cases with $p >0$ are similar, but involve more calculation.
We consider in detail only some representative cases.
For $D_VR_+(\dot V)$ with $\|\dot V\|_\Ucal \le 1$,
one term is $2P(\dot V,Q)$, which has $T_0$ norm bounded by
\begin{equation}
    (c/L)^{(\alpha + \alpha')(j-j_m)_+}\|\dot V\|_\Ucal L^{\alpha'(j-j_m)_+} \|Q\|_{T_0}
    \prec
    (c/L)^{(\alpha + \alpha')(j-j_m)_+} \vartheta \gLfix^3,
\end{equation}
which is more than small enough.
By \refeq{Vhatdef}, $Q$ is a sum of terms of the form $\Loc_{Y,B}\frac{K(Y)}{I(V,Y)}$.
Another term arises from differentiation of $I$ in $Q$ in $I(V,Y) = e^{-V(Y)}\prod_{B' \in \Bcal(Y)}
(1+W(V,B'))$ (recall \refeq{Idef}).
When the exponential is differentiated inside the term $\Ex\theta Q$,
we are led to estimate
\begin{equation}
    \|\Ex\theta \Loc_{Y,B} K(Y) I(V,Y)^{-1}\dot V\|_{T_0}
    \prec \|K(Y)\|_{T_0} \|\dot V\|_{T_0} \le   \|K(Y)\|_{T_0} \|\dot V\|_{\Ucal}
    \le   \|K(Y)\|_{T_0} \prec \chiL^3 \gLfix^3,
\end{equation}
as desired.  Differentiation of the $W$ factors produces a smaller result,
as does differentiation of $Q$ in either of the last two terms of \refeq{RplusVQ}.
Higher-order $V$ derivatives, and mixed $V,K$ derivatives, can be handled similarly.
\end{proof}

\subsubsection{Mass continuity}

The fundamental ingredient in the proof of mass continuity in \cite{BS-rg-step},
that requires attention here,
is \cite[Proposition~B.2]{BS-rg-step}.
In our present setting,
this ingredient becomes the statement
that, for each $j<N(\volume)$, the map $m^2 \mapsto C_j$ is a continuous map from $[0,L^{-\alpha j}]$
to the unit ball $B_1(\Phi_j(\ell_j))$.
Note that, for the interval $m^2\in [0,L^{-\alpha j}]$, $\ell_j$ is independent of $m^2$
and the space $\Phi_j(\ell_j)$ is therefore also independent of $m^2$.
By \refeq{CLbd}, $C_j$ does map into $B_1(\Phi_j)$.
The continuity is a consequence of the formula \refeq{Cjintegral} for $C_{j;x,y}(m^2)$,
the upper bound $cL^{-(d-2+|a|)(j-1)}$ on $|\nabla^a C_{j;x,y}|$ (uniform in $m^2,x,y$)
given in \refeq{scaling-estimate-new}, and the dominated convergence theorem to take
an $m^2$ limit under the integral in \refeq{Cjintegral} to see that
$\lim_{\mgen^2 \to m^2}\|C_j(\mgen^2)-C_j(m^2)\|_{\Phi_j} =0$.

\section{Global renormalisation group flow}
\label{sec:rgflow}

Theorem~\ref{thm:step-mr} controls one step of the renormalisation group map,
and the map can be iterated over multiple scales $j$ as long as $(V_j,K_j)$ remains in the
domain $\domRG_j$.  However, the fact that $V_j=g_j\tau^2+\nu_j\tau$ consists
of a sum of two relevant monomials indicates that, without precise tuning, the
domain will soon be exited.
In this section, we construct a global renormalisation group flow for all
scales $j=0,1,\ldots, N$, by tuning the
initial value $\nu_0=\nu_0(m^2)$ to a critical value.
The main effort lies in constructing a flow that
exists for scales up to the mass scale $j_m$; this construction is done in
Section~\ref{sec:PCMI}.
Beyond the mass scale,
the renormalisation group map simplifies dramatically:
the map $\PT_j$ is close to the identity map due to the exponential decay
given by the factor $M_j$
(recall \refeq{Mjbd})
in the bound on the coefficients appearing in Proposition~\ref{prop:PT},
and similar exponential decay occurs in the estimates for $K_j$
due to the appearance of $\chiL_{j+1}$ in Theorem~\ref{thm:step-mr}.
The flow beyond the mass scale is discussed in Section~\ref{sec:flowpastjm}.

\subsection{Flow equations}

Let $R_+$ of \refeq{Rplushat} be given by $R_+ = r_g \tau^2 + r_\nu \tau + r_u$.
By \refeq{RUPT} and \refeq{gpt2a}--\refeq{nupta}, the flow equations for $g,\nu$ are,
for $j< j_m$,
\begin{align}
  g_+
  &
  =
  g - \beta' g^{2} - 4g \delta[\nu w^{(1)}]  + r_{g,j}
  ,
\label{e:gptr}
  \\
\label{e:nuptr}
  \nu_+
  &=
  \nu
  +  \eta' (g + 4g\nu w^{(1)})
  - \xi' g^{2}
  -\gamhat \beta' \nu g
  - \delta[\nu^{2} w^{(1)}]
  + r_{\nu,j}.
\end{align}
Also, the map $K_+$ advances $K$ from scale $j$ to scale $j+1$.
Since $\tau^2$ is not in the range of $\LT$ for scales above $j_m$,
the flow of $g_j$ simply stops at $j_m$, with $g_j=g_{j_m}$ for $j \ge j_m$.
In particular, $r_{g,j}=0$
for $j \ge j_m$.

Theorem~\ref{thm:step-mr} provides the following estimates for $R_+$ and $K_+$,
assuming $(V,K) \in \domRG$:
\begin{align}
\lbeq{RKbds}
    \|R_+\|_{\Ucal_+} & \le C_{(0,0)} \chiL_+ \gLfix^3,
    \quad\quad
    \|K_+\|_{\Wcal_+}  \le \CRG  \chiL_+^3 \gLfix^3,
    \\
\lbeq{DRbds}
    \|D_VR_+\| & \le C_{(1,0)}  \chiL_+ \gLfix^2,
    \quad \quad \|D_KR_+\| \le C_{(0,1)} \chiL_+^{-2},
    \\
\lbeq{DKbds}
     \|D_VK_+\| & \le C_{(1,0)} \chiL_+^3 \gLfix^2,\quad\quad  \|D_KK_+\|  \le \kappa  < 1.
\end{align}
In particular, the remainders $r_g$ and $r_\nu$ obey, for general
$j$ including $j > j_m$,
\begin{align}
    r_{g,j} & \le
    L^{-\epsilon (j\wedge j_m)}\|R_{j+1}\|
    \le
    \1_{j < j_m} C_{(0,0)} L^{-\epsilon j} \gLfix^3
    ,
    \\
\lbeq{rnubd}
    r_{\nu,j} &
    \le
    L^{-\alpha (j\wedge j_m)}\|R_{j+1}\|
    \le
    C_{(0,0)} L^{-\alpha (j\wedge j_m)} \chiL_{j+1}  \gLfix^3
     .
\end{align}

Recall that the variables $g,\nu$ are related to $\gL,\mu$ via
\refeq{munu} and \refeq{gch-def}--\refeq{nuhatdef}.
In preparation for a rewriting of the flow equations, we
write $\dgL_j  = \gLfix - \gL_j$ as in \refeq{ydef}, and define
\begin{align}
    c_\epsilon &= 2-L^\epsilon \sim 1-\epsilon \log L,
\\
\lbeq{rhomudef}
    \rho_{\mu,j}
    &=
    -L^\alpha \left( \gamhat \beta_j \mu_j(\gLfix - \dgL_j) + \newxi_j (\gLfix - \dgL_j )^2 \right)
    .
\end{align}
In particular, $\rho_{\mu,j}$ is second order.
We rewrite $K_+$ in terms of the variables $\mu,\dgL$ as $\Kch_+(\mu,\dgL,K) = K_+(g,\nu,K)$,
with $(g,\nu)$ determined by $(\mu,\dgL)$ via the map $T^{-1}$
(see Proposition~\ref{prop:transformation}) and \refeq{munu}.

\begin{lemma}
\label{lem:yflow}
For $j < j_m$, the flow equations written in terms of $\mu$ and $\dgL$ are:
\begin{align}
\lbeq{musm1}
    \mu_{j+1} &= L^{\alpha}\mu_j
    +\rho_{\mu,j}+ r_{\mu,j},
    \\
\lbeq{sigsm1}
    \dgL_{j+1} & = c_\epsilon\dgL_j +  aL^{\epsilon}\dgL_j^2
    +(\beta^:_j-a)L^{\epsilon} (\bar\gL - y_{j})^2 + r_{y,j},
    \\
\lbeq{Ksm1}
    K_{j+1} & = \Kch_{j+1} (\mu_j,\dgL_j,K_j).
\end{align}
Suppose that $(g_j,\nu_j,K_j) \in \domRG_j$.
Then for $r_*=r_\mu,r_\gL$, and for derivatives $D = D_\mu,D_\gL$,
\begin{align}
\lbeq{rKbds}
    |r_{*,j}| & \le
    O( \gLfix^3) ,
    \quad
    |Dr_{*,j}|  \le O( \gLfix^2), \quad \|D_Kr_{*,j}\| \le O(1) ,
    \quad
    \|D \Kch_{j+1}\| \le O(\gLfix^2).
\end{align}
\end{lemma}

\begin{proof}
This follows from Proposition~\ref{prop:transformation},
\refeq{sigsm0}--\refeq{musm0}, \refeq{RUPT}, and the bounds
\refeq{RKbds}--\refeq{DKbds}.
\end{proof}

The role of the mass scale is especially prominent in the $y$
flow \refeq{sigsm1},
where the important coefficient $c_\epsilon \sim 1-\epsilon \log L <1$ appears.
This coefficient
$c_\epsilon $ is responsible for contraction of the sequence $y_j$.
Apart from transient effects,
below the mass scale $\beta^:_j-a$ is small, but above the mass scale it is essentially
$-a$, and the third term on the right-hand side of \refeq{sigsm1} begins to play an important role.
In particular, with $\beta^:_j$ and $r_{y,j}$ set equal to zero, the derivative of the
right-hand side of \refeq{sigsm1} with respect to $y_j$ (at $y_j=0$) becomes
$c_\epsilon + 2L^\epsilon a\gLfix \sim 1+\epsilon \log L >1$ (using $a\gLfix \sim \epsilon
\log L$ by \refeq{asfix}).
For this reason, we only use \refeq{sigsm1} for scales $j < j_m$.

In Section~\ref{sec:PCMI}, we construct
an $m$-dependent  flow
$(g_j(m^2),\nu_j(m^2),K_j(m^2))_{j \le j_m}$, which lies in $\domRG_j$ for each $j \le j_m$.
In particular, in Corollary~\ref{cor:mu0} ,
we determine a \emph{critical} initial value $\mu_0=\mu_0(m^2)$
which is responsible for ensuring that the flow remains in $\domRG_j$.
By \refeq{munu} and \refeq{nuhatdef} (and since $w_0=0$ by definition),
this corresponds to a critical value $\nu^c_0$ given by
\begin{equation}
\lbeq{nu0}
    \nu^c_0(m^2) = \mu_0(m^2) -  g\eta_{\ge 0}(m^2)
    =
    \mu_0(m^2) - g(n+2)C_{00}(m^2).
\end{equation}
Our rough point of view is that,
below the mass scale, the RG map is only weakly dependent on $m^2$, in the sense that
\begin{equation}
\lbeq{approxflow}
    (\gL_j(m^2),\mu_j(m^2),K_j(m^2))
    \approx
    (\gL_j(0),\mu_j(0),K_j(0))
    \quad \text{for $j\le j_m$}
    ,
\end{equation}
whereas, above the mass scale, the RG map is approximately the identity map, with $R_+$ and $K_+$
negligible, and
\begin{align}
\lbeq{pms}
    (g_j(m^2),\nu_j(m^2),K_j(m^2))
    &\approx
    (g_{j_m}(m^2),\nu_{j_m}(m^2),0)
    \approx
    (g_{j_m}(0),\nu_{j_m}(0),0)
    \quad \text{for $j> j_m$}
    .
\end{align}
Also, there is no need for $\nu_{j_m}$ to be tuned to any special value in order
to continue its flow beyond $j_m$.  Note that we use different variables in
\refeq{approxflow}--\refeq{pms}.  This is because the variables $(\gL_j,\mu_j)$ lose their
relevance beyond the mass scale, and $(g_j,\nu_j)$ are restored as the natural variables.

Figure~\ref{fig:RGflow} gives a schematic depiction of the dynamical system.  It shows the
unstable Gaussian and stable non-Gaussian fixed points (of the renormalisation group map),
and the stable and unstable manifolds for the massless theory.
The ``fixed points'' G and NG are not literally fixed points, due to the non-autonomous nature
of the RG map,
but we use this terminology nevertheless.
Flows are illustrated with initial masses
$m_1^2>m_2^2>0$, up to the mass scale.
For $m^2>0$,
there is some flexibility in the choice of an initial value $\mu_0(m^2)$ that permits
the flow to continue until the mass scale.
We exploit this by choosing $\mu_0(m^2)$ so that at the mass scale $\mu_{j_m}(m^2)$
is approximately zero; the precise construction is in Section~\ref{sec:PCMI}.
However, for $m^2=0$, there is no flexibility:  the initial condition must be precisely
tuned in order to iterate the RG map infinitely often,
and the unique value $\mu_0(0)$ ultimately determines the critical value
$\nu_c$ (see \refeq{nuc}).

Two distinct notions of ``fixed point'' occur.
One is the fixed point of the RG map, discussed above.
A distinct notion is a fixed point of a map $T:X\to X$ for a Banach space $X$,
i.e., a solution $x\in X$ to $Tx=x$.  This second type of fixed point plays an important
role in the construction of $\mu_0(m^2)$.

\begin{figure}
\begin{center}
\includegraphics[scale = 0.5]{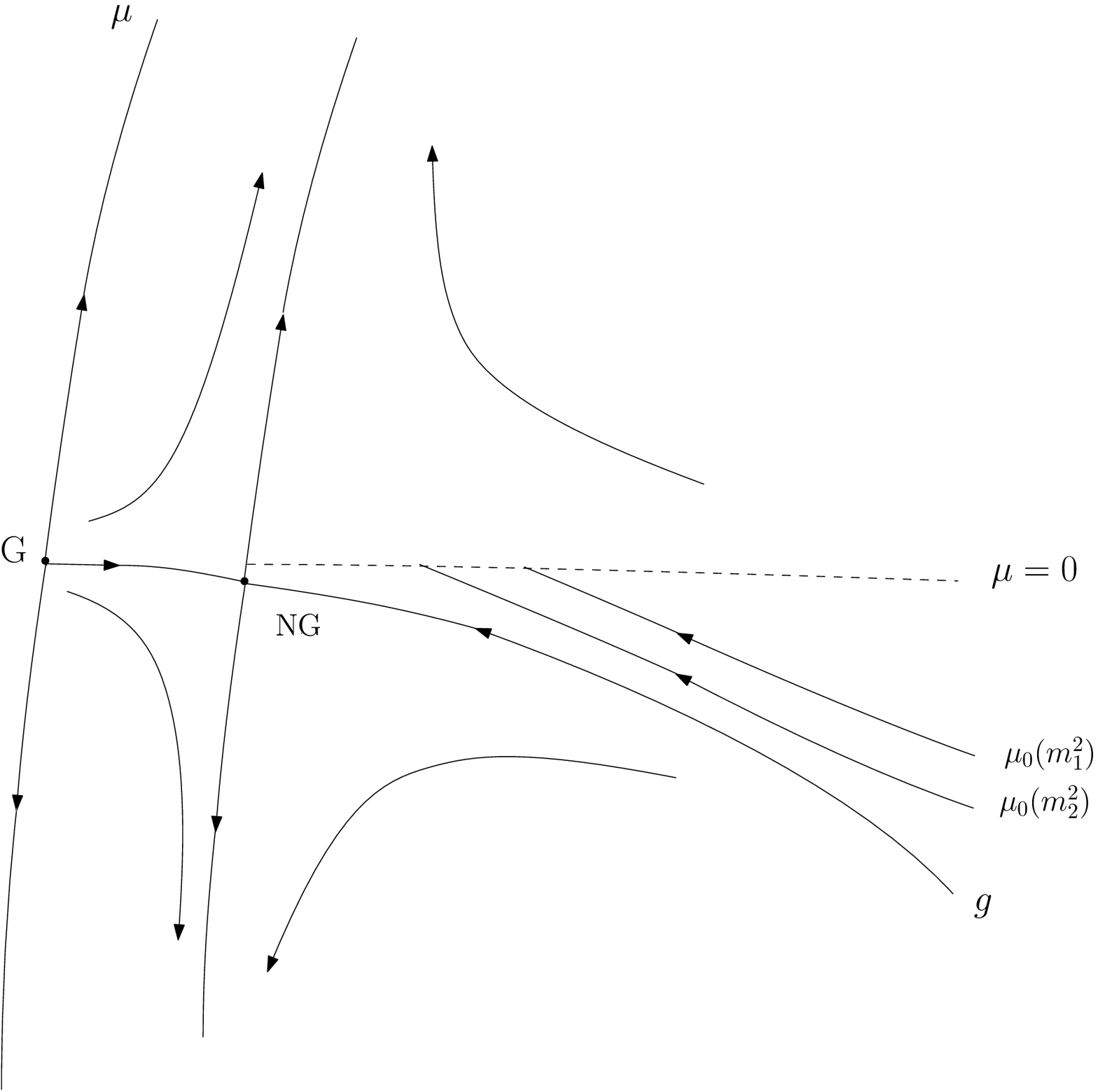}
\end{center}
\caption{
Schematic depiction of the flow until the mass scale, for two masses $m_1^2 > m_2^2>0$.
For $m^2 = 0$, the flow lies on the stable manifold and
flows to the non-Gaussian RG fixed point NG.}
\label{fig:RGflow}
\end{figure}

\subsection{Flow until mass scale}
\label{sec:PCMI}

Throughout this section, our focus is on scales $j\le j_m$.
We permit $m^2=0$, in which case  $j \le j_m$ means $j< j_0=\infty$.
We take $L$ large, then choose $\epsilon$ (hence $\gLfix$) small depending on $L$,
and often do this in the following without explicit mention.

The construction of the global flow is via the identification of a fixed point for
a map $T$ on a certain Banach space, which we introduce in this section.
We define the space $X$ and map $T$ using the \emph{infinite-volume}
 version ($\volume = \Zd$) of the maps
$K_+,R_+$.  In this way we avoid dependence on the volume parameter $N$, for the map $T$ and
its fixed point.  On the other hand, the fixed point for the infinite-volume flow immediately
produces a flow for the finite torus $\volume =\Lambda_N$ over all scales $0 \le j \le N$.
This is because the infinite-volume RG map obeys the \emph{same} estimates as the finite-volume
RG map in Theorem~\ref{thm:step-mr}, so the finite-volume RG map can be iterated
over as many scales as the infinite-volume RG map,
from the same initial condition.

\subsubsection{The Banach space \texorpdfstring{$X$}{X}}

Let $\Wcal_j^*$ denote the Banach space with
norm \refeq{Wcaldef} defined for the infinite volume
$\volume = \Zd$ and for $m^2=0$.
Let $E_j=\R$ and let $F_j=\R\times\Wcal_j^*$, so $\mu_j \in E_j$ and $(y_j,K_j) \in F_j$.
Let $X_0=E_0$, $X_j=E_j\times F_j$ for $j > 0$.  The Banach space of interest is
$X=\oplus_{j=0}^{\infty} X_j$.
An initial condition specifies $(y_0,K_0) \in F_0$, and $F_0$ is not part of the space $X$.
We only need the case $K_0=\1_\varnothing$, which we assume in the following.
The initial condition $y_0$ is determined by the parameter $g$
appearing in the statements of Theorems~\ref{thm:suscept}--\ref{thm:sh},
via $y_0=\gLfix - g_0=\gLfix -g$,
with $\gLfix$ the perturbative fixed point given in \refeq{sfix}.

The norm on $X$ is defined in terms of weights $\w$, by
\begin{equation}
\lbeq{Xnorm}
    \|(\mu,y,K)\|
    =
    \max \left\{
    \sup_{j \ge 0}
    \w_{\mu}^{-1}|\mu_j|, \;
    \sup_{j \ge 1}
    \w_{y,j}^{-1}|y_j|, \;
    \sup_{j \ge 1}
    \w_{K}^{-1}\|K_j\|_{\Wcal_j^*}
    \right\}.
\end{equation}
The weights are
\begin{equation}
\lbeq{Xweights}
    \w_{\mu}  = \sigma\gLfix^2,
    \quad
    \w_{y,j}= \omega_j\gLfix,
    \quad
    \w_{K}= \Kweight\gLfix^3,
\end{equation}
where
$\sigma,\omega_j,\Kweight>0$ are chosen as follows.
With $\newxi_j=\newxi_j(m^2)$ given by \refeq{newxidef}, we define $\sigma$ by
\begin{equation}
\lbeq{muweight}
    \sigma = 5\Pi,
    \quad\quad
    \Pi = \sup_{m^2 \in [0,1]} \sup_{j\ge 0} |\newxi_j|< \infty .
\end{equation}
With the constant $\CRG$ of Theorem~\ref{thm:step-mr},
we set
\begin{equation}
    \Kweight=  \CRG.
\end{equation}
To define the weight $\omega_j$,
let $J_L$ and $b_L$ be given by Lemma~\ref{lem:beta-am}.
Let $\zeta = 1-64b_L \gLfix <1$,
$\omega= \frac{1}{32}$, and
\begin{equation}
\lbeq{omegajdef}
    \omega_j = \omega \zeta^{(J_L-j)_+}
    .
\end{equation}
By definition, $\omega_j=\omega$ when $j \ge J_L$, and, for all $j$,
\begin{equation}
    \omega_0 = \omega \zeta^{J_L} \le
    \omega_j \le \omega .
\end{equation}
Since $\epsilon$ can be chosen small depending on $L$, we have
$\zeta^{J_L} = 1-O(\epsilon)$, so $\omega_j$ remains within order $\epsilon$ of $\omega$
for all $j$.
It is via this choice of weight that we deal with transient lattice effects near scale $0$;
this avoids an analysis as in \cite[Theorem~6.3]{MS08}.

We write $B_1=B_1(X)$ for the closed unit ball in $X$.
The assumption that $x=(\mu,\dgL,K) \in B_1$ implies in particular that
$|\mu_j| \le \sigma \gLfix^2$, that $|\gL_j| = |\gLfix -\dgL_j |\le (1+\omega)\gLfix
= \frac{33}{32} \gLfix$, and that $\|K_j\| \le \CRG \gLfix^3\le 4\CRG \gLfix^3$.
Therefore,
\begin{equation}
\lbeq{B1domRG}
    x=(\mu,\dgL,K) \in B_1
    \quad
    \Rightarrow
    \quad
    (g_j,\nu_j,K_j) \in \domRG_j \;\;\text{for all $j$}.
\end{equation}

By definition, the space $X$, including its norm, does not depend $m^2$.
Moreover, for $m^2>0$, the space $\Wcal_j$ defined with mass $m^2$ is identical to
the space $\Wcal_j$ defined with mass zero as long as $j \le j_m$.
This follows from the fact that
the massive and massless versions of the parameters $\ell_j, h_j, \epdV_j,
\gamma_j$ are identical below the mass scale defined for $m^2>0$
(the parameters are defined in \refeq{elldefa}, \refeq{hdef}, \refeq{epdVdef}, \refeq{gammajdef}).
Thus, although we have defined $X$ in terms of $\Wcal_j^*$ with zero mass, it
would be equivalent when $m^2>0$ to use $\Wcal_j$ defined with $m^2$ instead, as long
as $j \le j_m$.

\subsubsection{The map \texorpdfstring{$T$}{T}}

Next, we define a map $T:X \to X$.

We are interested in  $m^2
\in [0,\delta]$ with $\delta$ small, so we may assume that
$j_m \ge j_\delta > 2J_L$.
In addition to the transient lattice effects handled via the weights \refeq{omegajdef},
there are also transient effects near scale $j_m$ when $m^2>0$
(effect of nonzero mass).  We define $T$ in such a manner that avoids dealing with
the latter effects, whose treatment is postponed.

Let
\begin{equation}
\lbeq{Mextdef}
    \Mext
    =
    \{
    (\mgen^2,m^2)\in (0,\infty) \times [0,\delta]
    :
    j_{\mgen} \le j_m - (J_L+2)
    \}.
\end{equation}
We regard $\Mext$ as a metric space, with metric induced from $\R^2$.
Note that $\mgen^2=0$ is excluded from $\Mext$.  For $m^2>0$, the point $(\mgen^2,m^2)$ lies
in $\Mext$ if and only if $\mgen^2 \in [m^2_-,\infty)$, where $m^2_-$ (a function of
$m^2$) is the least
value of $u^2$ for which $j_u=j_m - (J_L+2)$.

Fix $(\mgen^2,m^2)\in\Mext$ or $(\mgen^2,m^2)=(0,0)$.
Given $(y_0,K_0)\in F_0$ with $|y_0| \le \omega_0 \gLfix$ and $K_0=\1_\varnothing$,
we define a map $\hat T: B_1 \to X$,
with $\hat T=(\hat T^{(\mu)},\hat T^{(y)},\hat T^{(K)})$, by
setting $(\hat T x)_j=0$ if $j >j_{\mgen}$, whereas for $1\le j \le j_{\mgen}$ we define
\begin{align}
\lbeq{Tlamdef}
    (\hat T^{(\mu)} x)_j &=
    L^{-\alpha}(\mu_{j+1}
    -\rho_{\mu,j}-r_{\mu,j}) ,
    \\
\lbeq{Tydef}
    (\hat T^{(y)} x)_j &=
    c_\epsilon y_{j-1} + aL^{\epsilon}y_{j-1}^2
    +(\beta^:_{j-1}-a)L^{\epsilon} (\bar\gL - y_{j-1})^2 + r_{\dgL,j-1},
    \\
\lbeq{TKdef}
    (\hat T^{(K)}x)_j & = \Kch_j(\mu_{j-1},y_{j-1},K_{j-1}).
\end{align}
On the right-hand side, $\rho_{\mu,j},\beta^:_{j-1},r_{\mu,j} , r_{\dgL,j-1}, \Kch_j$
are defined with mass $m^2$ ($r_{\mu,j} , r_{\dgL,j-1}$ are determined by the map $R_+$).
By \refeq{B1domRG}, the hypotheses of Theorem~\ref{thm:step-mr}
are satisfied and $r_\mu,r_y,\Kch$ are well-defined on the right-hand side of \refeq{Tlamdef}--\refeq{TKdef}.

Observe that a fixed point of $\hat T$, i.e., a solution to $\hat Tx=x$,
 defines a flow satisfying \refeq{musm1}--\refeq{Ksm1}
up to scale $j_{\mgen}$,
with initial condition given by $(y_0,\1_\varnothing)$, and with final condition
$\mu_{j_{\mgen}+1}=0$ when $\mgen^2>0$.
When $\mgen^2=m^2=0$,  no final condition is imposed.

We desire continuity on $\Mext$, but
$\hat T$ jumps as $\mgen$ varies through values where $j_{\mgen}$
jumps.  Recall from \refeq{jmdef} that $\jgen = \lceil f_{\mgen}\rceil$, and let
\begin{equation}
    \delta_{\mgen} = \jgen - f_{\mgen} = \lceil f_{\mgen}\rceil - f_{\mgen} \in [0,1),
\end{equation}
which is a sawtooth function of $\mgen$.
We smooth out the jump in $\hat T$ at
scale $\jgen$ to get continuity in $\mgen$.
This is done by setting $T_j=\hat T_j$ for all $j\neq j_{\mgen}+1$,
and instead of $\hat T_{j_{\mgen}+1}=0$, we define
\begin{align}
\lbeq{Tlamdefm}
    (T^{(\mu)} x)_{\jgen+1}
    &=
    (1-\delta_{\mgen})
    L^{-\alpha}(\mu_{j_{\mgen}+2}-\rho_{\mu,\jgen+1}-r_{\mu,\jgen+1}) ,
    \\
\lbeq{Tydefm}
    (T^{(y)} x)_{\jgen+1}
    &=
    (1-\delta_{\mgen})
    \left(
    c_\epsilon y_{\jgen} + aL^{\epsilon}y_{\jgen}^2
    +(\beta^:_{\jgen}-a)L^{\epsilon} (\bar\gL - y_{\jgen})^2 + r_{\dgL,\jgen}
    \right)
    ,
    \\
\lbeq{TKdefm}
    (T^{(K)}x)_{\jgen+1} &
    =
    (1-\delta_{\mgen})
    \Kch_{\jgen+1}(\mu_{\jgen},y_{\jgen},K_{\jgen})
    .
\end{align}
For $(\mgen^2,m^2)=(0,0)$, we have
$j_0=\infty$ and $T_j=\hat T_j$ for all $j$, but we do not consider continuity of
$T$ at $(\mgen^2,m^2)=(0,0)$.
The following lemma provides the continuity statement that we need for the map $T$.

\begin{lemma}
\label{lem:Tcont}
For each $x\in B_1(X)$,  $Tx$ is a continuous function of $(\mgen^2,m^2) \in \Mext$ into $X$.
\end{lemma}

\begin{proof}
It suffices to prove the continuity of $(Tx)_j$ for each $j$.
The continuity of $(Tx)_j$ in $m^2\in [0,\delta]$ is not the difficulty.
It follows from the continuity of
$\beta^:,\beta,\newxi$ (the latter two are in $\rho_\mu$) due to Lemma~\ref{lem:wlims},
together with the continuity in $m^2$ provided by Theorem~\ref{thm:step-mr} for the remainders
$r_*$ and the map $K_+$.  (The continuity provided by Theorem~\ref{thm:step-mr}
is for $m^2 \le [0,L^{-\alpha j}]$, so $j \le j_m$,
and this is satisfied here since the nonzero
$T_j$ have $j \le j_{\mgen} + 1 < j_m$.)

Discontinuities in $\mgen^2$ can only occur at values $\mgen_*$ where $\jgen$ makes its
jumps, namely values where $\delta_{\mgen_*}=0$.  By definition, $\delta_{\mgen} \downarrow 0$
as $\mgen \downarrow \mgen_*$, whereas $\delta_{\mgen} \uparrow 1$ as $\mgen \uparrow \mgen_*$,
and $\jgen$ increases its value from $j_{\mgen_*}$
to $j_{\mgen_*}+1$ as $\mgen$ decreases through $\mgen_*$.
The effect of the factors $(1-\delta_{\jgen})$ in \refeq{Tlamdefm}--\refeq{TKdefm}
is to continuously acquire the
terms that occur discontinuously in $\hat T$ when $\mgen$ varies through values where $\jgen$
jumps.  To see this, consider a small neighbourhood $N \ni \mgen_*$,
containing no other point
of discontinuity.  Let $\mgen \in N$.
If $\mgen > \mgen_*$, then $\jgen = j_{\mgen_*}$, and taking the limit
as $\mgen \downarrow \mgen_*$ in \refeq{Tlamdefm}--\refeq{TKdefm} gives
\begin{align}
\lbeq{Tlamdefm8}
    (T^{(\mu)} x)_{j_{\mgen_*}+1}
    &=
    L^{-\alpha}( \mu_{j_{\mgen_*}+2} - \rho_{\mu,j_{\mgen_*+1}} - r_{\mu,j_{\mgen_*}+1})
     ,
    \\
\lbeq{Tydefm8}
    (T^{(y)} x)_{j_{\mgen_*}+1}
    &=
    c_\epsilon y_{j_{\mgen_*}} + aL^{\epsilon}y_{j_{\mgen_*}}^2
    +(\beta^:_{j_{\mgen_*}}-a)L^{\epsilon} (\bar\gL - y_{j_{\mgen_*}})^2 + r_{\dgL,j_{\mgen_*}}
    ,
    \\
\lbeq{TKdefm8}
    (T^{(K)}x)_{j_{\mgen_*}+1} &
    =
    \Kch_{j_{\mgen_*}+1}(\mu_{j_{\mgen_*}},y_{j_{\mgen_*}},K_{j_{\mgen_*}})
    .
\end{align}
For $\mgen<\mgen_*$, we have $j_{\mgen_*}+1=j_{\mgen}$, so $(Tx)_{j_{\mgen}}$
is given by  \refeq{Tlamdef}--\refeq{TKdef}, and its limit as $\mgen \uparrow \mgen_*$
agrees with the limit $\mgen \downarrow \mgen_*$.
In addition, if $\mgen < \mgen_*$ then $\jgen = j_{\mgen_*}+1$, and taking the limit
as $\mgen \uparrow \mgen_*$ in \refeq{Tlamdefm}--\refeq{TKdefm} gives
\begin{align}
\lbeq{Tlamdefm9}
    (T^{(\mu)} x)_{j_{\mgen_*}+2}
    &=
    0 ,
    \\
\lbeq{Tydefm9}
    (T^{(y)} x)_{j_{\mgen_*}+2}
    &=
    0
    ,
    \\
\lbeq{TKdefm9}
    (T^{(K)}x)_{j_{\mgen_*}+2} &
    =
    0
    .
\end{align}
This shows the required continuity of $(Tx)_j$ as $\mgen$ varies through $j_{\mgen_*}$.
\end{proof}

\subsubsection{Contractivity of \texorpdfstring{$T$}{T}}

The dynamical system defined by $T$ is hyperbolic, in contrast to the more
difficult non-hyperbolic
system for $d=4$ analysed in \cite{BBS-rg-flow}.
Our analysis of $T$ is inspired by  \cite{BMS03,MS08}, as well as by the
stable manifold theorem of \cite[Theorem~2.16]{Bryd09}.
It bears resemblance to the analysis used for the massless case in
\cite{BDH98,BMS03,MS08}, e.g., \cite[Section~6]{BMS03}, but there are
differences.  In particular, we consider the
massive case.  Also, \cite{BDH98,BMS03}
work in the continuum where lattice effects are absent.
We do not need to deal separately with lattice effects
at small scales, as was done via an additional application of the implicit function
theorem in \cite[Theorem~6.3]{MS08}.  Instead, lattice transients are handled
via our choice of weights $\omega_j$.

The following theorem proves that $T$ is a contraction.
As usual, we fix $L$ large enough and then choose $\epsilon$
small enough depending on $L$.
Continuity of $T$ in $\Mext$ is not needed for Theorem~\ref{thm:pcmi-new1}, but is used
in Corollary~\ref{cor:mu0}.

\begin{theorem}
\label{thm:pcmi-new1}
Let $(\mgen^2,m^2) \in \Mext$, or $\mgen^2=m^2=0$.
For every
initial condition $(y_0,\1_\varnothing)\in F_0$ with $|y_0|\le \omega_0\gLfix$, we have
$T:B_1\to B_1$, and
there exists $c \in (0,1)$ (depending on $\epsilon, L$, independent of $(\mgen^2,m^2)$)
such that
$\|DT\|\le c$ on $B_1$.
\end{theorem}

\begin{proof}
Recall the definitions of the weights in \refeq{Xweights}--\refeq{omegajdef}.
Suppose $x \in B_1$.
Then, for all $j$,
\begin{equation}
    |\mu_j| \le \sigma\gLfix^2, \quad |y_j| \le \omega_j \gLfix \le \omega \gLfix,
    \quad \|K_j\| \le \Kweight \gLfix^3.
\end{equation}
By \refeq{B1domRG}, $(g_j,\nu_j,K_j) \in \domRG_j$ for all $j$.

\smallskip \noindent
\emph{Bound on $T$.}
We verify that $T:B_1\to B_1$.
First, we have (by Theorem~\ref{thm:step-mr} for $T^{(K)}$)
\begin{align}
\lbeq{Tmubd}
    |(T^{(\mu)} x)_j|
    &\le
    L^{-\alpha}\sigma \gLfix^2
    + \Pi \gLfix^2
    + O(\gLfix^3)
    \le \sigma\gLfix^2,
    \\
\lbeq{TKbd}
    |(T^{(K)} x)_j|
    &\le
    \CRG\gLfix^3
    = \Kweight \gLfix^3
    .
\end{align}
For the more delicate component, we start with
\begin{align}
\lbeq{Tybd0}
    |(T^{(y)} x)_j|
    &\le
    \omega_j \gLfix \frac{\omega_{j-1} }{\omega_{j} }
    \Big(c_\epsilon  + aL^{\epsilon}\omega_{j-1}\gLfix
    + \omega_{j-1}^{-1} |\beta_j^:-a| L^\epsilon (1+\omega_{j-1})^2\gLfix +O(\gLfix^2)
    \Big).
\end{align}
Recall from \refeq{asfix} that $a\gLfix \sim \epsilon \log L$, and that
$\zeta$ in \refeq{omegajdef} is given by
$\zeta=1-64b_L\gLfix$.
By Lemma~\ref{lem:beta-am}, for scales $j \le J_L$
we have
$|\beta_j^:-a| \le b_L$, whereas for $j_L \le j \le j_m-J_L$ we have
$|\beta_j^:-a| \le \frac{a}{64}$.  In the former case,
\refeq{Tybd0} gives
\begin{align}
    |(T^{(y)} x)_j|
    &\le
    \omega_j \gLfix
    \zeta
    \Big(1 - \epsilon\log L(1 - \tfrac{1}{31})
    +
    33 b_L  (\tfrac{17}{16})^2)\gLfix
    \Big)
    \nnb
    & \le
    \omega_j \gLfix
    \zeta
    (1+ \tfrac{3}{4} 64b_L\gLfix )
    \le
    \omega_j \gLfix
    .
\end{align}
For larger scales, we have instead that
\begin{align}
    |(T^{(y)} x)_j|
    &\le
     \omega_j \gLfix
    \Big(1 - \epsilon\log L(1 - \tfrac{1}{31}
    - 33 \tfrac{1}{64} (\tfrac{17}{16})^2)
    \Big) \le
    \omega_j \gLfix.
\end{align}
Thus, in either case,
\begin{align}
\lbeq{Tybd1}
    |(T^{(y)} x)_j|
    &
    \le \omega_j \gLfix,
\end{align}
and with \refeq{Tmubd}--\refeq{TKbd} and \refeq{Tybd1} we conclude that $T:B_1\to B_1$.

\smallskip \noindent
\emph{Bound on $DT$.}
To prove that there exists $c<1$ such that $\|DT(x)\|\le c$ for $x\in B_1$,
it suffices to prove that $\|D_\mu T^{(*)}(x)\|+\|D_y T^{(*)}(x)\|+\|D_K T^{(*)}(x)\|
\le c$ for $*=\mu,y,K$.  This is a consequence of the following estimates.
The crucial step is \refeq{DKK}, which is provided by
Theorem~\ref{thm:step-mr}.
The delicate step which requires attention here is \refeq{DTybd1}, and its proof is
similar to that of \refeq{Tybd1}.
Differentiation of $T$ gives (recall $\Kweight=\CRG$ and \refeq{rKbds})
\begin{align}
    \|D_\mu T^{(\mu)}(x)\| & =
    \sup_j
    \frac{\sigma\gLfix^2}{\sigma\gLfix^2}
    |D_\mu (T^{(\mu)}(x))_j|
    \le
    L^{-\alpha}+O(\gLfix) \le \tfrac 12 ,
    \\
    \|D_y T^{(\mu)}(x)\| & =
    \sup_j
    \frac{\omega_{j-1}\gLfix}{\sigma\gLfix^2}
    |D_y (T^{(\mu)}(x))_j|
    \nnb & \le \frac{\omega\gLfix}{\sigma\gLfix^2}
    (4\Pi \gLfix +O(\gLfix^2))
    = \omega  \left( \tfrac 45 + O(\gLfix) \right) \le \tfrac{1}{32},
    \\
    \|D_K T^{(\mu)}(x)\| & \le
    \frac{\Kweight\gLfix^3}{\sigma\gLfix^2}
    O(1)
    \le  O(\gLfix) \le \tfrac 14,
\\
    \|D_\mu T^{(K)}(x)\| & \le
    \frac{\sigma\gLfix^2}{\Kweight\gLfix^3} O(\gLfix^2)
     \le  O(\gLfix)   \le \tfrac 14,
    \\
    \|D_y T^{(K)}(x)\| & \le
    \sup_j
    \frac{\omega_{j-1}\gLfix}{\Kweight\gLfix^3}
    \CRG\gLfix^2
    =  \omega
    = \tfrac{1}{32} ,
    \\
\lbeq{DKK}
    \|D_K T^{(K)}(x)\| & \le
    \frac{\Kweight\gLfix^3}{\Kweight\gLfix^3}
    \kappa = \kappa
     \le \tfrac 14
    ,
\\
\lbeq{Dmuy}
    \|D_\mu T^{(y)}(x)\| & \le
    \sup_j
    \frac{\sigma\gLfix^2}{\omega_j\gLfix} O(\gLfix^2)
    \le O(\gLfix^3)  ,
    \\
\lbeq{DKy}
    \|D_K T^{(y)}(x)\| & \le
    \sup_j
    \frac{\Kweight\gLfix^3}{\omega_j\gLfix}
    O(1)
    \le O(\gLfix^2)  .
\end{align}
Finally, again using  $a\gLfix \sim \epsilon \log L$, we have
\begin{align}
\lbeq{DTybd1}
    \|D_y T^{(y)}(x)\| & =
    \sup_j
    \frac{\omega_{j-1}\gLfix}{\omega_{j}\gLfix}
    |D_y (T^{(y)}(x))_j|
    \nnb &
    \le
    \sup_j
    \frac{\omega_{j-1}}{\omega_{j}}
    \left(
    c_\epsilon +2aL^{\epsilon}\omega_{j-1}\gLfix
    +|\beta^:_{j-1}-a| L^{\epsilon} 2(1+\omega_{j-1})\gLfix  + O(\gLfix^2)\right)
    \nnb &
    \le
    \sup_j
    \frac{\omega_{j-1}}{\omega_{j}}
    \left(
    1+\epsilon \log L (-1 + \tfrac{1}{8})
    +|\beta^:_{j-1}-a|2\tfrac{17}{16}\gLfix   \right)
    .
\end{align}
Suppose first that $j \le J_L$, so that
$|\beta^:_{j-1}-a| \le b_L$ by \refeq{betajmdiff}.
With  $\zeta = 1-64b_L \gLfix$, we see that the argument of the supremum is bounded above by
\begin{align}
\lbeq{Dyy1}
    \zeta
    \left(
    1+ \tfrac{17}{8}b_L \gLfix   \right)
    \le
    1- 32b_L\gLfix
    .
\end{align}
On the other hand, if $J_L<j<j_m-J_L$, then we have
$|\beta^:_{j-1}-a| \le \frac{a}{64}$ by Lemma~\ref{lem:beta-a0}, and hence (since $\omega_{j-1}
= \omega_j$ by definition) the bound becomes
\begin{align}
\lbeq{Dyy2}
    1+\epsilon \log L \Big( -1 + \tfrac{1}{8}
    +\tfrac{1}{16} \Big)
    \le
    1- \tfrac 34 \epsilon \log L
    .
\end{align}
Each of \refeq{Dyy1}--\refeq{Dyy2} remains less than $1$ after addition of the bounds
in \refeq{Dmuy}--\refeq{DKy}, and the proof is complete.
\end{proof}

\begin{rk}
\label{rk:g0}
\emph{Restriction on $g$.}
By definition,
$\omega_0 = \omega (1-O(\epsilon))$.  In particular, $\omega_0 \ge \half\omega = \frac{1}{64}$.
The restriction $|y_0|
\le \omega_0\gLfix$ in  Theorem~\ref{thm:pcmi-new1} is therefore satisfied if
$g=g_0=\gLfix - y_0$ obeys
\begin{equation}
    |g -\gLfix| \le \tfrac{1}{64} \gLfix
    .
\end{equation}
This restriction on $g$ is incorporated into
the statements of Theorems~\ref{thm:suscept}--\ref{thm:sh}.  It is clear that the constant
$\frac{1}{64}$ could be improved.
\end{rk}

\subsubsection{Flow until mass scale and construction of critical initial value}

Let $|y_0|\le \omega_0\gLfix$, and let
$(\mgen^2,m^2) \in \Mext$ with $\Mext$ defined in \refeq{Mextdef},
or let $(\mgen^2,m^2)=(0,0)$.
Together with the contraction mapping principle, Theorem~\ref{thm:pcmi-new1} implies that
$T$ has a unique fixed point $x^c \in B_1$, i.e., $Tx^c=x^c$.  This fixed point provides a
solution to the flow equations \refeq{musm1}--\refeq{Ksm1} that maintains
$(g_j,\nu_j,K_j)$ in the RG domain $\domRG_j$ for all
scales $0 \le j \le j_{\mgen}$.
For the case $\mgen^2=m^2=0$, this is a flow on all scales.  In either case, we write the
$\mu$-component of the initial value $x_0^c$ as
\begin{equation}
    \mu_0(\mgen^2,m^2) = (x_0^c)^{(\mu)}(\mgen^2,m^2).
\end{equation}

For $m^2>0$, the minimal value for $\mgen^2$ in $\Mext$ is
such that $j_{\mgen}=j_m-(J_L+2)$, so in this case the flow does not quite reach $j_m$,
but fails to do so by only a bounded ($L$-dependent) number of scales.
As in the discussion below \refeq{Mextdef}, we write this minimal value of $\mgen^2$
as $m^2_-=m^2_-(m^2)$.  Then we define the \emph{critical initial value}
\begin{equation}
\lbeq{mu0m2}
    \mu_0(m^2) = \mu_0(m_-^2,m^2) \quad (m^2>0),
\end{equation}
i.e., $\mu_0(m^2)$ is the $\mu$-component of $x_0^c$ for $T$ defined with
$(\mgen^2,m^2)=(m^2_-,m^2) \in \Mext$.  For $\mgen^2=m^2=0$,  which is also permitted
in Theorem~\ref{thm:pcmi-new1}, we denote the critical
initial value by $\mu_0(0)=\mu_0(0,0)$.
The following corollary to Theorem~\ref{thm:pcmi-new1} gives continuity of
$\mu_0(m^2)$ in $m^2\in (0,\delta]$.

\begin{cor}
\label{cor:mu0}
The function $\mu_0(\mgen^2,m^2)$ is continuous in $(\mgen^2,m^2)\in \Mext$.
In particular,  $\mu_0(m^2)$ is continuous in $m^2 \in (0,\delta]$.
\end{cor}

\begin{proof}
By Lemma~\ref{lem:Tcont}, $Tx$ is
jointly continuous in
$(\mgen^2,m^2) \in \Mext$, for each $x \in B_1$, and $T$ is uniformly contractive.
By the version of the contraction mapping principle given in
\cite[Corollary~4, p.230]{LS14}, the fixed point of $T$ is continuous in
$(\mgen^2,m^2) \in \Mext$.  In particular, so is $\mu_0(\mgen^2,m^2)$, and therefore
$\mu_0(m^2)$ is continuous in $m^2 \in (0,\delta]$.
\end{proof}

The next corollary makes the important extension of Corollary~\ref{cor:mu0}
to include right-continuity at $m^2=0$.

\begin{cor}
\label{cor:mu00}
The limit $ \lim_{m^2 \downarrow 0} \mu_0(m^2)$ exists and equals $\mu_0(0)$
(the critical initial value for the case $\mgen^2=m^2=0$).
\end{cor}

\begin{proof}
Fix $y_0$ with $|y_0|\le \omega_0\gLfix$,
fix a sequence $m' \downarrow 0$, and let
$x_0'= (\mu_0(m'),y_0,\1_\varnothing)\in E\times F_0$.
Since $\mu_0(m')$ remains in a
bounded subset of $\R$,
it has a limit point $\mu_0^*$.
It suffices to show that $\mu_0^*=\mu_0(0)$, for any
sequence $m'$.

Let $x_0^*=(\mu_0^*,y_0,\1_\varnothing)$.
We use $x_0^*$ as the initial condition for the flow equations
\refeq{musm1}--\refeq{Ksm1}, and we solve those flow equations inductively,
to produce $x_j^*$, for as
long as this remains in the closed unit ball $B_1(X_j)$ (with norm on $X_j$
given by \refeq{Xnorm}
with the suprema over $j$ omitted).
On the other hand, with initial condition $x_0'$ the fixed point
solves the equations for $j \le j_{m'_-}$, with $x_j'\in B_{1}(X_j)$.
Given any fixed $j$, eventually $j < j_{m'_-}$ as $m' \downarrow 0$.
By the continuity of the RG map \refeq{VKplusmap} at $m^2=0$ (recall Theorem~\ref{thm:step-mr}),
we know that $x_j'$ converges to
$x_j^{*}$, which must remain in the closed ball $B_1(X_j)$.
This produces a sequence $x_j^*$ for all $j<\infty$, which is a solution of the
zero-mass flow equations for all $j$, and hence a fixed point of the
zero-mass $T$.  This fixed point is unique, and $x_0^*=(\mu_0(0),y_0,\1_\varnothing)$.
Therefore, $\mu_0^*=\mu_0(0)$, and the proof is complete.
\end{proof}

\begin{rk}
The RG fixed point NG in Figure~\ref{fig:RGflow} will correspond to the limits
\begin{align}
    \gL_\infty & =\gL_\infty(0)=\lim_{j \to \infty} \gL_j(0),
    \quad
    \mu_\infty = \mu_\infty(0)=\lim_{j \to \infty} \mu_j(0)
\end{align}
for the massless flow.
We do not prove existence of these limits, and we do not need or use it,
but it would be of interest to explore this further.
\end{rk}

The space $X$ is defined in terms of the massless infinite-volume RG map, and in particular
involves the norm $\Wcal^*$.
However, since we only consider scales below the mass scale, for $m^2>0$ the bounds
for all massless norms are identical to those for the massive case.  Also, since
Theorem~\ref{thm:step-mr} provides the same estimates for either the finite- or infinite-volume
RG maps, the infinite-volume flow also gives rise to a finite-volume flow by iterating the
finite-volume RG map from the initial condition given by the fixed point of $T$.
Thus there is no distinction between existence of finite- or infinite-volume flows up to
the mass scale.

For $m^2>0$, the fixed point of the map $T$ produces a flow $x_j \in B_1(X_j)$
for scales $j \le j_{m_-}=j_m-J_L$.  We wish to extend this to a flow for all scales $j <\infty$.
The following lemma does a small portion of this, by extending to scales $j \le j_m$.
The full extension is provided by Theorem~\ref{thm:pastjm}.

\begin{lemma}
\label{lem:tojm}
Let $m^2 \in (0,\delta]$, and let $(x_j)_{j \le j_{m}-J_L}$ be the RG flow produced by
Theorem~\ref{thm:pcmi-new1}.  This flow extends to a flow $(\mu_j,y_j,K_j)_{j \le j_m}$
for all $j \le j_m$, with $|\mu_j| \le c_L \gLfix^2$, $|y_j| \le \frac{1}{31}\gLfix$,
and $\|K_j\|_{\Wcal_j} \le \CRG \gLfix^3$, where $c_L$ is an $L$-dependent constant.
\end{lemma}

\begin{proof}
We solve the flow equations \refeq{musm1}--\refeq{Ksm1} forward until scale $j_m$,
starting from scale $j_{m}-J_L$.  This can be done as long as $(g_j,\nu_j, K_j)\in \domRG_j$.
The number of scales to be advanced is $J_L$, which is independent of $m^2$ and $\epsilon$.
At each step, the bound
on $\mu_j$ deteriorates by an $L$-dependent factor.
In fact, by \refeq{musm1} and \refeq{muweight}, if $|\mu_j| \le t\gLfix^2$ then
(we may assume $t \ge 1$)
\begin{equation}
    |\mu_{j+1}| \le L^\alpha(t\gLfix^2 + \Pi 4\gLfix^2)
    \le L^\alpha(1 + 4\Pi )t \gLfix^2.
\end{equation}
For $y_j$, by Lemma~\ref{lem:beta-am} we have
$|\beta_j^:-a| \le b_L$.
From \refeq{sigsm1}, we obtain
\begin{equation}
    |y_{j+1}| \le |y_j| (1+aL^\epsilon |y_j|) + 4b_L \gLfix^2.
\end{equation}
Thus a bound $|y_j| \le t\gLfix$ yields a bound $|y_{j+1}| \le (1+c_L'\gLfix)t\gLfix$,
so the deterioration is $1+c_L'\gLfix$ per scale.
The accumulation of these deteriorations, which multiply over scales, is
some constant $c_L = (L^\alpha(1+4\Pi))^{J_L}$ for $\mu$ and $(1+c_L'\gLfix)^{J_L}
\le 1+c_L''\gLfix$ for $y$.  These accumulated effects cannot move $(g_j,\nu_j)$ out of
$\DV_j$, since $|y_j|$ remains less than $(1+O(\gLfix))\frac{1}{32}\gLfix$, and
$\mu_j$ is an $O(\gLfix^2)$ adjustment to the leading term $-\eta_{\ge j}\gLfix$ in $\muhat_j$
(recall \refeq{Tinvlin}).  Theorem~\ref{thm:step-mr} then guarantees that
\refeq{musm1}--\refeq{Ksm1} can indeed by iterated forward until scale $j_m$, as required.
\end{proof}

\subsection{Flow beyond mass scale}
\label{sec:flowpastjm}

Theorem~\ref{thm:pcmi-new1} and
Lemma~\ref{lem:tojm} produce a flow that exists up to the mass scale.
Beyond the mass scale, there is exponential decay in the flow equations which makes
it possible to obtain a solution by forward iteration without further tuning.  This is accomplished in the
next theorem.
Its hypothesis that $m^2 \ge L^{-\alpha(N-1)}$ ensures that $C_{N,N}$
obeys a bound $L^{-(N-1)(d-\alpha)}$, by \refeq{CNNbd}.
To study the flow past $j_m$, we extend the definition of $\muhat_j$, given in \refeq{munu}
for $j \le j_m$, and define a corresponding remainder term $r_{\muhat,j}$, by
\begin{equation}
\lbeq{muhatpastjm}
    \muhat_j = L^{\alpha j_m}\nu_j, \quad r_{\muhat,j}= L^{\alpha j_m}r_{\nu,j},
    \quad (j \ge j_m).
\end{equation}
By definition, $\ghat_{j_m}=L^{\epsilon j_m}g_{j_m}$.
After some algebra, the flow equation for $\nu_j$ given in \refeq{nuptr} can be equivalently
written, for $j \ge j_m$, as
\begin{align}
\lbeq{nuflow-pastjm}
    \muhat_{j+1} - \muhat_j
    & =
    L^{-(d-\alpha)(j-j_m)}\eta_j\ghat_{j_m}(1+4\muhat_j \bar w_j^{(1)}) -
    L^{-(\alpha-2\epsilon)(j-j_m)}\xi_j \ghat_{j_m}^2
    -\gamhat \beta_j \ghat_{j_m}\muhat_j
    \nnb
    & \quad
    - \Big[ \muhat_j^2 L^{-\alpha j_m}C_{j+1}^{(1)} + 2L^{-(d-\alpha)(j-j_m)}\eta_j\muhat_j \ghat_{j_m} \bar w_{j+1}^{(1)}
    \\ \nonumber
    & \quad\quad\quad
    + (L^{-(d-\alpha)(j-j_m)}\eta_j\ghat_{j_m})^2 \bar w_{j+1}^{(1)}\Big]
    + r_{\muhat,j}.
\end{align}

\begin{theorem}
\label{thm:pastjm}
Let $m^2 \in [L^{-\alpha(N-1)},\delta]$ and $g \in [\frac{63}{64}\gLfix,\frac{65}{64}\gLfix]$.
With initial condition at $j=j_m$ produced by Lemma~\ref{lem:tojm},
the flow equations for $(g_j,\nu_j,K_j)$ can be solved
forward (inductively) to scale $N$, and produce a sequence which remains in the
domain $\domRG_j$ for each $j$.  For $g_j$, we have simply $g_j=g_{j_m}$ for all $j \ge j_m$
and the limit $g_\infty=\lim_{N \to \infty} g_N = g_{j_m}\in [\frac{15}{16}L^{-\epsilon j_m}\gLfix,
\frac{17}{16}L^{-\epsilon j_m}\gLfix]$ exists.
For $\nu_j$, the limit $\nu_\infty = \lim_{N \to \infty} \nu_N=O(L^{-\alpha j_m} \gLfix)$ exists,
and is attained uniformly on compact subsets of $m^2\in(0,\delta]$.  Moreover, $(\nu_j)_{j \le N}$
is independent of the volume in the sense that it is identical up to scale $N$ with the
sequence $(\nu_j)_{j \le N'}$ for any $N'>N$.
\end{theorem}

\begin{proof}
We first consider the initial conditions.
By Lemma~\ref{lem:tojm}, $|\dgL_{j_m}| \le \omega \gLfix = \frac{1}{31} \gLfix$, and
$\mu_{j_m} = O(\gLfix^2)$.  Therefore $\gL_{j_m}=\gLfix-\dgL_{j_m} \in
[\frac{30}{31} \gLfix, \frac{32}{31}\gLfix ]$.  By \refeq{Tinvlin},
$\ghat_{j_m} \in [\frac{15}{16} \gLfix, \frac{17}{16}\gLfix ]$
and $\muhat_{j_m} = \mu_{j_m} - \eta_{\ge j_m}s_{j_m} +O(\gL^2)
= - \eta_{\ge j_m}s_{j_m} +O(\gL^2)$.
By \refeq{etagedef} and Lemma~\ref{lem:wlims}, $|\eta_{\ge j_m}| = O(1)$ uniformly
in $L$.  Therefore,
\begin{equation}
    \ghat_{j_m} \in [\tfrac{15}{16}  \gLfix, \tfrac{17}{16} \gLfix ],
    \quad\quad
    |\muhat_{j_m}| \le c_0   \gLfix,
\end{equation}
with $L$-independent $c_0$.

Let $j > j_m$.
Then
$g_j = g_{j_m}$, since the flow of $g$ is stopped at the mass scale, by definition.

What needs to be verified is that the forward flow keeps $\nu_j$ in the domain $\DV_j$,
i.e., that $|\muhat_j|$ remains bounded by an $L$-independent multiple of
$\gLfix$.  As long as this happens, the forward flow of $K_j$ is
given by Theorem~\ref{thm:step-mr} and the remainder due to $R_j$ in the flow of $\muhat_j$
remains bounded as $r_{\muhat,j}\le C_{(0,0)} \chiL_j \gLfix^3$
by \refeq{rnubd}.  The flow of $\muhat_j$ is given by \refeq{nuflow-pastjm}.
An important term is the first term,
$L^{-(d-\alpha)(j-j_m)}\eta_j\ghat_{j_m}$, which by \refeq{Greeknoprime}, \refeq{Mjbd}
and Lemma~\ref{lem:wlims} obeys (with $L$-independent
constant $c$)
\begin{equation}
    L^{-(d-\alpha)(j-j_m)}|\eta_j|\ghat_{j_m}
    \le
    c L^{-(d-\alpha)(j-j_m)} L^{-2\alpha(j-j_m)}  \gLfix
    =
    cL^{-(d+\alpha)(j-j_m)} \gLfix
    .
\end{equation}
Under the assumption that the flow remains in the domain up to scale $j$, we see
from the above inequalities, and Lemma~\ref{lem:wlims},
that there is a $z>0$ such that, with $e_j=O(L^{-z(j-j_m)})$,
\begin{align}
\lbeq{nudiff}
    |\muhat_{j+1} - \muhat_j|
    &\le
    c L^{-(d+\alpha)(j-j_m)} \gLfix
    + e_j   \gLfix^2
    + C_{(0,0)}  \chiL_j \gLfix^3
    \le
    ce_j   \gLfix
    ,
\end{align}
and hence
\begin{equation}
    |\nu_{j+1}| \le |\nu_{j_m}| + \sum_{i=j_m}^{j} |\nu_{i+1}-\nu_i|
    \le (c_0+2c) L^{-\alpha j_m} \gLfix.
\end{equation}
Thus $\nu_j$ remains in the domain, since we may choose $C_{\DV}$ to be greater than
the $L$-independent constant $c_0+2c$.
Also, it follows from \refeq{nudiff} that $\nu_j$ is a Cauchy sequence and
hence its limit $\nu_\infty$ exists.  For $m^2$ bounded away from $0$, the limit
is attained uniformly, as claimed.
Finally, the independence of the volume for the sequence $\nu_j$ follows exactly
as in \cite[Section~8.1]{BBS-saw4-log}.
\end{proof}

According to Theorem~\ref{thm:pastjm}, the limiting value $g_{j_m}$ of the sequence $g_j$
is of order $L^{-\epsilon j_m}\gLfix \asymp m^{2\epsilon/\alpha} \epsilon$.
By \refeq{bubm}, this is the same order as the reciprocal of the free bubble diagram $B_{m^2}$.
For the nearest-neighbour 4-dimensional model, the relationship with the bubble diagram
is made explicitly in \cite[Lemma~8.5]{BBS-saw4-log}.

\section{Proof of Theorem~\ref{thm:suscept}}
\label{sec:pfsuscept}

In this section, we prove Theorem~\ref{thm:suscept},
The proof relies on the analysis of the global renormalisation group flow developed
in Section~\ref{sec:flowanalysis}.
In Section~\ref{sec:diffineq}, we use the results of
Section~\ref{sec:flowanalysis} to study the susceptibility, and in particular to
prove a differential inequality that it satisfies.  This is used in
Section~\ref{sec:pf-suscept} to complete the proof of Theorem~\ref{thm:suscept}.

Throughout this section, we work with the global flow equations determined
by Theorem~\ref{thm:pastjm}.  For scales below the mass scale, we express estimates
in terms of the
variables $(\gL_j,\mu_j)$ rather than $(g_j,\nu_j)$.

\subsection{Analysis of flow equations for \texorpdfstring{$g,\nu$}{g,nu}}
\label{sec:flowanalysis}

We write $e^{O(\gLfix^2 j)}$ to mean that, for
some positive $\delta$ and $c$, and for all $\epsilon \in (0,\delta]$ and $j \ge 1$,
\begin{equation}
    e^{-c\gLfix^2 j}
    \le
    e^{O(\gLfix^2 j)}
    \le
    e^{c\gLfix^2 j}.
\end{equation}
Recall the definitions of $\gamhat$ and $\beta_k$ in \refeq{gamhatdef}
and \refeq{Greeknoprime}.
With $s_0=g$, and for small $m^2>0$,
we define $P_j = P_j(m^2,g)$ by
\begin{equation}
\lbeq{Pproddef}
  P_j =
  \begin{cases}
  \prod_{k=0}^{j-1}(1-\gamhat \beta_k\gL_k) & ( j \le j_m)
  \\
  P_{j_m} & ( j>j_m).
  \end{cases}
\end{equation}

\begin{lemma}
\label{lem:gprod}
For $m^2 \in (0,\delta]$,
$g \in [\frac{63}{64}\gLfix,\frac{65}{64}\gLfix]$,
and for
$j \le j_m$,
\begin{align}
\lbeq{Piprod}
    P_j
    & =
    (1+O(\gLfix))\left( \frac{ L^{-\epsilon j}\gL_{j}}{\gL_{0}} \right)^{\gamhat}
    e^{O(\gLfix^2 j)}
    =
    (1+O(\gLfix))\left( \frac{g_{j}}{g_{0}} \right)^{\gamhat}
    e^{O(\gLfix^2 j)}
    .
\end{align}
\end{lemma}

\begin{proof}
Since $j \le j_m$, we have $\chiL_j=1$.
The second equality in \refeq{Piprod}
follows from the first, together with the fact that $g_j=L^{-\epsilon j}\gL_j
(1+O(\gLfix))$ by \refeq{munu} and \refeq{gch-def}.

For the first equality, we recall the definition of $\beta_k^:$ in \refeq{betaW}, and write
\begin{align}
    1-\gamhat\beta_k \gL_k
    & =
    (1-\gamhat\beta^:_k \gL_k)
    (1+  f_k),
    \qquad
    f_k= \frac{\gamhat(\beta^:_k -\beta_k)\gL_k}{1-\gamhat\beta^:_k \gL_k}.
\end{align}
By \refeq{betabd}, $\beta_k$ and $\beta_k^:$ are $O(1)$,
and by Lemma~\ref{lem:tojm}, $\gL_k = \gLfix - \dgL_k \asymp \gLfix$.
By Lemma~\ref{lem:betadiff}, $f_k$ is therefore summable, and hence
$\prod_{k=0}^{j-1}(1+f_k) = (1+O(\gLfix))$.
For the factor $(1-\gamhat\beta^:_k \gL_k)$, we use
\begin{equation} \label{e:gchrec}
  \gL_{k+1}
  = L^\epsilon (1-\beta^:_k \gL_k) \gL_k  + r_{\gL,k},
    \quad\quad
    r_{\gL,k} = O(  \gLfix^3 ),
\end{equation}
which is the flow equation for $\gL_k$ (recall \refeq{gLbar}---this is \refeq{sigsm1}
written in terms of $\gL$ rather than $y$).
Therefore, by Taylor's theorem, there
exists $\delta_k = O(\gLfix^2)$ such that
\begin{align}
\lbeq{Taygam}
  (1 - \gamhat  \beta_k^:\gL_k)
  &=
  (1 -  \beta_k^:\gL_k)^{\gamhat} (1+ \delta_k)
  =
  \left(\frac{\gL_{k+1}-r_{\gL,k}}{L^\epsilon \gL_{k}} \right)^{\gamhat} (1+ \delta_k)
  \nnb &
  =
  \left(\frac{\gL_{k+1}}{L^\epsilon \gL_{k}} \right)^{\gamhat} (1+ O(\gLfix^2))
  .
\end{align}
It follows that
\begin{align}
    \prod_{k=0}^{j-1} (1-\gamhat  \beta_k\gL_k)
    & =
    (1+O(\gLfix))
    \left( \frac{L^{-\epsilon j} \gL_{j}}{\gL_{0}} \right)^{\gamhat}
    \prod_{k=0}^{j-1} (1+ O(\gLfix^2))
    \nnb & =
    (1+O(\gLfix))
    \left( \frac{ L^{-\epsilon j}\gL_{j}}{\gL_{0}} \right)^{\gamhat}
    e^{O(\gLfix^2 j)}
    ,
\end{align}
and the proof is complete.
\end{proof}

For a function $f = f(m^2,g,\nu_0)$, we write
\begin{align}
\lbeq{fp}
  f' &= \ddp{}{\nu_0} f(m^2,g, \nu_0^c ),
  \quad
  f'' = \ddp{^2}{\nu_0^2} f(m^2,g, \nu_0^c ),
\end{align}
with $\nu_0^c=\nu^c_0(m^2,g)$ the critical value given by \refeq{nu0} (with $\mu_0(m^2)$
from \refeq{mu0m2}).

\begin{lemma}
\label{lem:gzmuprime}
For $m^2 \in (0,\delta]$, $g \in [\frac{63}{64}\gLfix,\frac{65}{64}\gLfix]$,
and for $j \le j_m$,
  \begin{equation} \label{e:mugzprime}
    \mu_j'
    =
    L^{\alpha j} P_j
    e^{O(\gLfix^2  j)}
    ,
    \quad
    \gL_j', \|K_j'\|_{\Wcal_j }, \|R_j'\|_{\Ucal_j } = O(\mu_j' \gLfix^{2}  )
    ,
  \end{equation}
  \begin{equation} \label{e:gzmuprime2pf1}
    \mu_j'', \gL_j'',  \|K_j''\|_{\Wcal_j},  \|R_j''\|_{\Ucal_j} = O( (\mu_j')^2 \gLfix)
    .
  \end{equation}
\end{lemma}

\begin{proof}
For the proof of \refeq{mugzprime}, we first recall the flow equations
\begin{align}
    \gL_{j+1} &= L^\epsilon (1-\beta^:_j \gL_j)\gL_j + r_{\gL,j},
    \\
    \mu_{j+1} & = L^\alpha (1-\gamhat \beta_j \gL_j) \mu_j  - L^\alpha \newxi_j \gL_j^2
    + r_{\mu,j}.
\end{align}
The equation for $\gL$ is \refeq{gchrec}, and the equation for $\mu$ is \refeq{musm1}.
Differentiation gives
\begin{align}
\lbeq{gLflowdiff}
    \gL_{j+1}' &= L^\epsilon (1-2\beta^:_j \gL_j)\gL_j' + r_{\gL,j}',
    \\
\lbeq{muflowdiff}
    \mu_{j+1}' & = L^\alpha (1-\gamhat \beta_j \gL_j) \mu_j'
    - L^\alpha (2 \newxi_j \gL_j + \gamhat \beta_j \mu_j)\gL_j'
    + r_{\mu,j}'.
\end{align}
We set $\Sigma_{-1}=0$, and, for $j \ge 0$, define
$\Sigma_j = \Sigma_j(m^2,g)$ by
\begin{equation}
  \label{e:muPiSig}
  \mu_j' = L^{\alpha j} P_j e^{\Sigma_j}
  .
\end{equation}
We make the inductive assumption that
there are constants $c,M_1,M_2>0$ such that
\begin{equation} \label{e:induct1}
  |\Sigma_{j}-\Sigma_{j-1}| \leq c(M_1+M_2)  \gLfix^2,
  \quad
  |\gL_j'| \leq M_1  \gLfix^2 \mu_j',
  \quad
  \|K_j'\|_{\Wcal_j} \leq M_2  \gLfix^2 \mu_j'.
\end{equation}
The constants are determined in the proof, with $M_1 \gg M_2 \gg 1$.
Since $(\gL_0',\mu_0',K_0') = (0,1,0)$, and $P_0=1$ (empty product),
the assumption \eqref{e:induct1} holds for $j=0$, with $\Sigma_0=0$.

To advance the induction, we first note that it follows from the change of variables
given by \refeq{munu} and \refeq{gch-def}--\refeq{nuhatdef},
together with the definition of the $\Ucal_j$ norm in \refeq{Vcalnormdef}, that
\begin{equation}
\lbeq{induct2}
    \|V_j'\|_{\Ucal_j} \le O(\gL_j' +\mu_j').
\end{equation}
Also,
by the chain rule, with $F = R_{j+1}$ or $F=K_{j+1}$,
\begin{equation} \label{e:chain}
  F'(V_j,K_j) = D_{V}F(V_j,K_j)V_j' + D_KF(V_j,K_j)K_j'.
\end{equation}
By the estimates of Theorem~\ref{thm:step-mr}
and by \eqref{e:induct2},
this gives
\begin{align}
  \|D_VF(V_j,K_j)V_j'\| &\leq O(\gLfix^2)(M_1 \gLfix^2+1)\mu'_j
  \leq O( \gLfix^2) \mu_j' ,
  \\
  \|D_KR_{j+1}(V_j,K_j)K_j'\| &\leq O(M_2) \gLfix^2 \mu_j',
  \\
  \|D_KK_{j+1}(V_j,K_j)K_j'\| &\leq M_2 \gLfix^2 \mu_j' ,
\lbeq{chainDKp}
\end{align}
where the norms on the left-hand sides are those for the appropriate
$\Ucal,\Wcal$ spaces; for simplicity we use the weaker $K_{j+1}$ bounds also for $R_{j+1}$.
With $M_2 \gg 1$, this implies
\begin{equation} \label{e:phiprime}
  \|R_{j+1}'(V_j,K_j)\| \leq O(M_2)  \gLfix^2 \mu_j',
  \quad
  \|K_{j+1}'(V_j,K_j)\| \leq 2M_2\gLfix^2 \mu_j'
  .
\end{equation}

To advance the induction for $\mu'$, we use \refeq{muflowdiff},
as well as the bounds $|\gL_j|\le O(\gLfix)$ and
$|\mu_j| \le O(\gLfix^2)$ from Lemma~\ref{lem:tojm}, to see that
\begin{align} \label{e:muchO}
  \mu_{j+1}'
  &=
  L^\alpha (1-\gamhat\beta_j \gL_j)\mu_j'
  + O(M_1\gLfix+M_2)  \gLfix^2 \mu_j'
  \nnb
  &
  =
  L^\alpha (1-\gamhat\beta_j \gL_j)\mu_j'(1+O(M_2)  \gLfix^2)
  .
\end{align}
This enables us to advance
the induction for $\mu'$, namely the first estimate of \eqref{e:induct1}.
This implies, in particular, that $2M_2\gLfix^2 \mu_j' \leq M_2 \gLfix^2 \mu_{j+1}'$
(for large $L$),
and combined with the second inequality of \refeq{phiprime} this advances the
induction for $K'$.

For  $\gL'$,
we first observe that $P_j \le 2P_{j+1}$ by \refeq{Pproddef}, and that
$e^{\Sigma_j } \le 2 e^{\Sigma_{j+1}}$ since we have advanced
the  first bound of \refeq{induct1}.
We choose $M_1 \gg M_2$ and see from \refeq{gLflowdiff} that
\begin{align}
  |\gL_{j+1}'|
  &\leq (L^\epsilon (1+O(\gLfix))M_1 +O(M_2))  \gLfix^2 \mu_j'
  \nnb
  &= (L^\epsilon (1+O(\gLfix))M_1 +O(M_2))  \gLfix^2 L^{\alpha j} P_j e^{\Sigma_j}
  \nnb
  &\le
  (L^\epsilon (1+O(\gLfix))M_1 +O(M_2))
  \gLfix^2 L^{\alpha (j+1)} L^{-\alpha} 2  P_{j+1} 2e^{\Sigma_{j+1}}
  \nnb
  &\le
  M_1   \gLfix^2 \mu_{j+1}'.
\end{align}
In the last step we used $4L^\epsilon L^{-\alpha}\le \frac 12$ for large $L$.
This advances the induction for $\gL'$, and completes the proof of \refeq{mugzprime}.

Next, we prove \refeq{gzmuprime2pf1}.
Differentiation of \refeq{gLflowdiff}--\refeq{muflowdiff} gives
\begin{align}
\lbeq{gLpp}
    \gL_{j+1}'' &= L^\epsilon (1-2\beta^:_j \gL_j)\gL_j''
    -2L^\epsilon \beta^:_j (\gL_j')^2 + r_{\gL,j}'',
    \\
\lbeq{mupp}
    \mu_{j+1}'' & = L^\alpha (1-\gamhat \beta_j \gL_j) \mu_j''
    - L^\alpha (2 \newxi_j \gL_j + \gamhat \beta_j \mu_j)\gL_j''
    -2L^\alpha \gamhat \beta_j \gL_j' \mu_j'
    - 2L^\alpha   \newxi_j (\gL_j')^2
    + r_{\mu,j}''.
\end{align}
The proof is again by induction, with the induction hypothesis that there exist
$N_1,N_2 > 0$ such that
\begin{equation} \label{e:induct1-bis}
  |\mu_j''|, |\gL_j''| \leq N_1  (\mu_j')^2 \gLfix, \quad
  \|K_j''\|_{\Wcal_j} \leq N_2  (\mu_j')^2 \gLfix.
\end{equation}
The constants are chosen in the proof, with $N_1\gg N_2 \gg 1$.
For $j=0$, \eqref{e:induct1-bis}  holds trivially since the three left-hand sides
are $0$.

With $F$ equal to either $R_{j+1}$ or $K_{j+1}$, the chain rule gives
\begin{align}
  F''(V_j,K_j)
  &=
  D_{V}F(V_j,K_j)V_j'' + D_KF(V_j,K_j)K_j''
  + D_{V}^2F(V_j,K_j)V_j'V_j'
  \nnb &\quad
  + D_K^2F(V_j,K_j)K_j'K_j'
  + 2D_{V}D_KF(V_j,K_j)V_j'K_j'
\end{align}
(here $D_{V}D_KF(V,K)AB$ denotes the second derivative of $F$ with derivative in the variable
$V$ taken in direction $A$ and derivative in $K$ taken in direction $B$).
From an examination of the change of variables
(see \refeq{munu} and \refeq{gch-def}--\refeq{nuhatdef}, detailed calculations
are made in the proof of Corollary~\ref{cor:nug}),
together with \refeq{mugzprime},
we obtain
\begin{equation}
    \|V_j''\|_{\Ucal_j} \le O(\gL_j''+\mu_j'') +O((\gL_j'+\mu_j')^2)
    \le
    O(\gL_j''+\mu_j'' +(\mu_j')^2).
\end{equation}
With the norms the appropriate ones involving the $\Ucal,\Wcal$ spaces,
it follows from \eqref{e:Rmain-g}--\eqref{e:DVKbd} for $R_+,K_+$,
\eqref{e:induct1-bis}, and \refeq{mugzprime}, that
\begin{align}
  \|D_VF(V_j,K_j)V_j''\| &\leq O(\gLfix^2) (\mu_j')^2,
  \\
  \|D_V^2F(V_j,K_j)V_j'V_j'\| &\leq O( \gLfix) (\mu_j')^2,
  \\
  \|D_VD_KF(V_j,K_j)V_j'K_j'\| &\leq
  O(\gLfix^{-1}) \mu_j' ( \mu_j' \gLfix^2 ),
  \\
  \|D_K^2F(V_j,K_j)K_j'K_j'\| &\leq
  O( \gLfix^{-5/2}) (\mu_j' \gLfix^2 )^2
  = O((\mu_j')^2 \gLfix^{3/2} ),
  \\
  \|D_KR_{j+1}(V_j,K_j)K_j''\| &\leq O(N_2)(\mu_j')^2 \gLfix,
  \\
  \|D_KK_{j+1}(V_j,K_j)K_j''\| &\leq N_2 (\mu_j')^2 \gLfix
  .
\end{align}
This implies, for $N_2 \gg 1$,
\begin{equation} \label{e:phiprime-bis}
  \|R_{j+1}''(V_j,K_j)\| \leq O(N_2)  (\mu_j')^2 \gLfix,
  \quad
  \|K_{j+1}''(V_j,K_j)\| \leq 2N_2  (\mu_j')^2 \gLfix.
\end{equation}
Since $\mu_j' = L^{-\alpha}\mu_{j+1}'(1+O(\gLfix))$ by \refeq{muPiSig}--\refeq{induct1},
for large $L$
we have
\begin{align} \label{e:gmuprime2}
  (\mu'_j)^2
  \leq 2 L^{-2\alpha} (\mu_{j+1}')^2
  &\leq
  \frac{1}{2L^\alpha} (\mu_{j+1}')^2
  .
\end{align}
The second bound of \eqref{e:phiprime-bis}
and \eqref{e:gmuprime2} advance the induction for $K_j''$ (with room to spare due to $L^\alpha$
in the upper bound \refeq{gmuprime2}).

To advance the induction for $\gLfix''$, we use \eqref{e:gLpp}, and use \refeq{mugzprime}
and the induction hypothesis \refeq{induct1-bis} to estimate
the first and second derivatives.  With \eqref{e:phiprime}, this leads to
\begin{align}
  |\gL_{j+1}''|
  &\leq  (L^\epsilon(1+O(\gLfix))N_1+O(N_2)) (\mu_j')^2 \gLfix
  \nnb
  &\leq \frac{1}{2}(N_1+O(N_2))  (\mu_{j+1}')^2 \gLfix
  \leq N_1  (\mu_{j+1}')^2 \gLfix,
\end{align}
by \eqref{e:gmuprime2} for the second inequality,
and using $N_1 \gg N_2$ in the last inequality.
The argument for $\mu_j''$ is analogous
(the factor $L^\alpha$ in \refeq{mupp} is bounded using  $L^{\alpha}$
on the right-hand side of \refeq{gmuprime2}).
This completes the proof.
\end{proof}

\begin{cor}
\label{cor:nug}
For $m^2 \in (0,\delta]$,
$g \in [\frac{63}{64}\gLfix,\frac{65}{64}\gLfix]$,
and  $j \le j_m$,
\begin{equation}
    \ghat_j' = O(\muhat_j' \gLfix),
    \quad
    r_{\ghat,j}', r_{\muhat,j}'  = O(\muhat_j' \gLfix^2 ) ,\quad
    \ghat_j'', r_{\ghat,j}'', r_{\muhat,j}''   = O((\muhat_j')^2 \gLfix ),
\end{equation}
\begin{equation}
    \muhat_j' =
    L^{\alpha j} P_je^{O(\gLfix^2 j)}(1+O(\gLfix)),
    \quad
    \muhat_j'' = (\muhat_j')^2(-2\bar w_j^{(1)}+O(\gLfix)).
\end{equation}
\end{cor}

\begin{proof}
We differentiate each of the change of variables formulas \refeq{gch-def}--\refeq{nuhatdef},
and obtain
\begin{align}
    \ghat' & = \frac{1}{1+4\bar w^{(1)}(\muhat + 2\eta_{\ge j} \ghat )}
    (s'-4\ghat\muhat' \bar w^{(1)})
    = (1+O(\gLfix))(s'+O(\gLfix)\muhat'),
    \\
    \muhat' &= \frac{1}{1+2 \bar w^{(1)}(\muhat+2\eta_{\ge j} \ghat ) }
    (\mu'-\eta_{\ge j}\ghat'(1+4\muhat \bar w^{(1)}))
    = (1+O(\gLfix))(\mu' + O(\ghat')).
\end{align}
With the bounds on $\mu',s'$ from Lemma~\ref{lem:gzmuprime},
the above equations lead to the desired bounds on $\ghat',\muhat'$.
Similarly,
\begin{align}
    \ghat'' & = \frac{1}{1+4\bar w^{(1)}(\muhat + 2\eta_{\ge j} \ghat )}
    (s''-8\ghat'\muhat' \bar w^{(1)}
    -8 \eta_{\ge j} (\ghat')^2 \bar w^{(1)}
    - 4\ghat\muhat'' \bar w^{(1)})
    ,
    \\
    \muhat'' & = \frac{1}{1+2 \bar w^{(1)}(\muhat+2\eta_{\ge j} \ghat ) }
    (\mu''-8\eta_{\ge j}\ghat'\muhat'\bar w^{(1)} - 2(\muhat')^2 \bar w^{(1)}
    -\eta_{\ge j} \ghat''(1+4\muhat \bar w^{(1)}))
    .
\end{align}
With the bounds on $\mu'',s''$ from Lemma~\ref{lem:gzmuprime},
this leads to the desired bounds on $\ghat',\ghat'',\muhat',\muhat''$.
Finally, the bounds on the derivatives
of the remainders $r_{\ghat},r_{\muhat}$ follow from the bounds on the derivatives of $R_+$ in
Lemma~\ref{lem:gzmuprime}.
\end{proof}

The next lemma
extends the estimates of Corollary~\ref{cor:nug} beyond the mass scale.
This is straightforward, due to the
exponential decay of coefficients in the flow equations beyond the mass scale.
The lemma does not make a statement about $g_j$, because $g_j=g_{j_m}$
and $r_{g,j}=0$ for all $j >j_m$.
According to \refeq{muhatpastjm}, $\muhat_j = L^{\alpha j_m}\nu_j$ for $j \ge j_m$.

\begin{lemma}
\label{lem:nup}
For $m^2 \in [L^{-\alpha(N-1)},\delta]$,
$g \in [\frac{63}{64}\gLfix,\frac{65}{64}\gLfix]$,
$j_m \le j \le N$, and with $u=1$ for $R_j$ and $u=3$ for $K_j$,
\begin{alignat}{2}
\lbeq{abovejm}
    \muhat_j &= O( \gLfix),
    &&\|K_j\|_{\Wcal_j}, \|R_j\|_{\Ucal_j} \le O(\chicCov_j^u \gLfix^3),
\\
\lbeq{abovejmp}
    \muhat_j' &= \muhat_{j_m}' (1+O(\gLfix)), \quad
    &&\|K_j'\|_{\Wcal_j}, \|R_j'\|_{\Ucal_j} \le O(\chicCov_j^u \muhat_{j_m}'\gLfix^2),
\\
\lbeq{abovejmpp}
    \muhat_j'' &= (\muhat_{j_m}')^2(-2\bar w_j^{(1)}+O(\gLfix)),
    \quad
    &&\|K_j''\|_{\Wcal_j}, \|R_j''\|_{\Ucal_j} \le O(\chicCov_j^u (\muhat_{j_m}')^2\gLfix)
     .
\end{alignat}
Also, the limit $\nu'_\infty = \lim_{N \to \infty}\nu_N'$ exists and is
attained uniformly on compact subsets of $m^2\in (0,\delta]$.
\end{lemma}

\begin{proof}
The bounds
on $\muhat_j,K_j,R_j$ in \refeq{abovejm}
follow directly from Theorems~\ref{thm:pastjm} and \ref{thm:step-mr}.
The proof of the other items combines elements of the proofs of Theorem~\ref{thm:pastjm}
and Lemma~\ref{lem:gzmuprime}.

To simplify the notation, we define $\hat\eta_j = L^{-(d-\alpha)(j-j_m)}\eta_j$,
$\hat\xi_j = L^{-(\alpha-2\epsilon)(j-j_m)}\xi_j$, and $\bar C_{j+1}^{(1)}=L^{-\alpha j_m}
C_{j+1}^{(1)}$.
For the first derivatives,
differentiation of the flow equation \refeq{nuflow-pastjm} for $\muhat_j$ gives
\begin{align}
\lbeq{muhatp-pastjm}
    \muhat_{j+1}' - \muhat_j'
    & =
    \hat\eta_j\ghat_{j_m}'(1+4\muhat_j \bar w_j^{(1)})
    +\hat\eta_j\ghat_{j_m}4\muhat_j' \bar w_j^{(1)}
    -
    2(\hat\xi_j +\hat\eta_j^2 \bar w_{j+1}^{(1)})\ghat_{j_m}\ghat_{j_m}'
    \nnb & \quad
    -(\gamhat \beta_j + 2\hat\eta_j\bar w_{j+1}^{(1)})
    (\ghat_{j_m}'\muhat_j + \ghat_{j_m}\muhat_j')
    - 2\muhat_j\muhat_j' \bar C_{j+1}^{(1)}
    + r_{\muhat,j}'.
\end{align}
The initial conditions are given by $\ghat_{j_m},\muhat_{j_m}=O( \gLfix)$, and, by
Corollary~\ref{cor:nug},
\begin{align}
    \muhat_{j_m}' & = L^{\alpha j_m}P_{j_m} e^{O(\gLfix^2 j_m)}(1+O(\gLfix)),
    \quad
    \ghat_{j_m}' = O(\muhat_{j_m}' \gLfix).
\end{align}
Using estimates already established (including the exponential decay of coefficients
provided by Lemma~\ref{lem:wlims}), we see from \refeq{muhatp-pastjm} that
there exist $A',z>0$, such that
\begin{align}
\lbeq{muprec}
    |\muhat_{j+1}' - \muhat_j'|
    & \le
    A'(\muhat_{j_m}' + |\muhat_{j}'|)\gLfix L^{-z(j-j_m)}
    + |r_{\muhat,j}'|.
\end{align}
We make the inductive hypothesis that
\begin{equation}
    |\muhat_j'-\muhat_{j_m}'| \le A_1\muhat_{j_m}' \gLfix,
    \quad
    \|K_j'\|_{\Wcal_j} \le A_2 \chicCov_j^3 \muhat_{j_m}' \gLfix^2
    ,
\end{equation}
with $A_1,A_2$ to be determined.
This is satisfied for $j=j_m$.
Application of the chain rule, as in \refeq{induct2}--\refeq{chainDKp}
but now retaining $\kappa$ from Theorem~\ref{thm:step-mr} in \refeq{chainDKp}, gives the desired
estimate on $R_{j+1}'$, as well as (for some $c>0$)
\begin{equation}
    \|K_{j+1}'\|_{\Wcal_{j+1}}
    \le
    c\chicCov_{j+1}^3\muhat_{j_m}' \gLfix^2 + \kappa A_2 \chicCov_j^3 \muhat_{j_m}' \gLfix^2.
\end{equation}
The induction for $K'$ can be advanced once we know that
\begin{equation}
\lbeq{kappa-demand}
    c\chicCov_{j+1}^3  + \kappa A_2 \chicCov_j^3  \le A_2 \chicCov_{j+1}^3 ,
    \quad
    {\rm i.e.,} \quad cA_2^{-1} +\kappa   L^{\frac 34 \alpha} \le 1.
\end{equation}
This last inequality is satisfied with $A_2=2c$, since $\kappa   L^{\frac 34 \alpha}
\le \frac 12$ by \refeq{kappa-small-above-jm}.  The bound on $R_{j+1}'$ implies that
$|r_{\muhat,j}'| = O(\chicCov_{j+1}
\muhat_{j_m}'\gLfix^2)$, and thus the last term on the
right-hand side of \refeq{muprec} is smaller by a factor $\gLfix$ than
its first term.  It follows that, with $A_1=5A'$,
\begin{align}
    |\muhat_{j+1}' - \muhat_{j_m}'|
    & \le
    4A'\muhat_{j_m}' \gLfix \sum_{k=j_m}^j L^{-z(k-j_m)} \le A_1 \muhat_{j_m}' \gLfix,
\end{align}
which advances the induction and completes the proof of \refeq{abovejmp}.
By \refeq{muprec}, $|\muhat_{j+k}-\muhat_j|$ is bounded by $O(\muhat_{j_m}' \gLfix L^{-z(j-j_m)})$,
so $\muhat_j'$ is a Cauchy sequence, hence convergent.
The convergence is uniform on compact subsets of $m^2$, since then $j_m$ remains
bounded, and hence the same is true for
the convergence of $\nu_j'$ to its limit.

The analysis of the second derivative is similar, and we only sketch the proof.
Inspection of the derivative of \refeq{muhatp-pastjm},
together with estimates already established, leads to
\begin{align}
\lbeq{muhatpp}
    |\muhat_{j+1}'' - \muhat_j'' +2(\muhat_{j_m}')^2  \bar C_{j+1}^{(1)}|
    & \le
    A''((\muhat_{j_m}')^2 + |\muhat_{j}''|)\gLfix L^{-z(j-j_m)}
    + |r_{\muhat,j}''|.
\end{align}
We make the induction hypothesis that there are constants $A_3,A_4$ such that
\begin{equation}
    |\mu_j''+2\bar w_j^{(1)}(\mu_{j_m}')^2| \le A_3 (\mu_{j_m}')^2 \gLfix,
    \quad
    \|K_j''\|_{\Wcal_j}  \le A_4 \chicCov_j^3 (\muhat_{j_m}')^2\gLfix.
\end{equation}
The induction hypothesis leads to the conclusion that
\begin{equation}
    |r_{\muhat,j}''| \le O(A_4) \chicCov_j (\muhat_{j_m}')^2\gLfix.
\end{equation}
For $A_4\gg 1$, the induction hypothesis for $K_j''$ can be advanced, as in the
proof of Lemma~\ref{lem:gzmuprime} and again using $\kappa$ as in the previous paragraph.
From \refeq{muhatpp}, we obtain
\begin{align}
\lbeq{muhatpp1}
    |\muhat_{j+1}'' - \muhat_j'' +2(\muhat_{j_m}')^2  \bar C_{j+1}^{(1)}|
    & \le
    O(A''+A_4)(\mu_{j_m}')^2 \gLfix L^{-z(j-j_m)}
    .
\end{align}
We replace $j$ by $k$ in the above inequality, and sum over $k$ from $j_m$ to $j$.
By definition, $\bar w_j^{(1)} = \sum_{k=1}^j\bar C_{k}^{(1)}$.
This gives
\begin{align}
\lbeq{muhatpp2}
    |\muhat_{j+1}'' - \muhat_{j_m}'' +2(\muhat_{j_m}')^2  (\bar w_{j+1}^{(1)}-\bar w_{j_m}^{(1)})|
    & \le
    O(A''+A_4)(\mu_{j_m}')^2 \gLfix
    .
\end{align}
Since $\mu_{j_m}'' = (\mu_{j_m}')^2(-2\bar w_{j_m}^{(1)} + O(\gLfix))$, there is a
cancellation on the left-hand side, and the induction hypothesis for $\mu_j''$ can
be advanced once we choose $A_3\gg A_4$.  This completes the proof.
\end{proof}

For the next lemma, we write $W_N =W_N( V_N(\Lambda), V_N(\Lambda))$,  $W_N'  = \ddp{}{\nu_0} W_N$,
and $W_N''  = \frac{\partial^{2}}{\partial \nu_0^2} W_N$.

\begin{lemma}
\label{lem:WKNp}
For $m^2 \in [L^{-\alpha(N-1)},\delta]$ and
$g \in [\frac{63}{64}\gLfix,\frac{65}{64}\gLfix]$,
\begin{align}
\lbeq{WNbd}
    \|W_N\| _{T_{0,N}} &\leq O((c/L)^{(\alpha+\alpha') (N-j_m)} \gLfix^2),
    \\
\lbeq{WNpbd}
    \|W_N'\| _{T_{0,N}} &\leq O((c/L)^{(\alpha+\alpha') (N-j_m)} \mu_{j_m}' \gLfix),
    \\
\lbeq{WNppbd}
    \|W_N''\| _{T_{0,N}} &\leq O((c/L)^{(\alpha+\alpha') (N-j_m)} (\mu_{j_m}')^2)
    .
\end{align}
\end{lemma}

\begin{proof}
The bound \refeq{WNbd} follows from Lemma~\ref{lem:W1}.
By definition, $ W_N(V,\tilde V)$ is bilinear in $(V,\tilde V)$, so differentiation gives
\begin{equation}
\lbeq{WNp}
  W_N'
  =
  W_{N}(V_N,V_N')
  + W_{N}(V_N',V_N).
\end{equation}
We obtain a bound on the terms in \refeq{WNp} by
multiplying the bound on $W_N(V_N,V_N)$ by an upper bound on the ratio of the
coefficients of $V_N'$ and $V_N$.
This gives
\refeq{WNpbd},
and \refeq{WNppbd} follows similarly.
\end{proof}

In the following lemma, $W_N(0)$ and $K_N(0)$ denotes evaluation at $\varphi=0$.  Also, $K_N$ is
evaluated on the unique nonempty polymer in $\Pcal_N$, namely the torus $\Lambda$, though
this is not made explicit in the notation.
The test function $\1:\Lambda \to \R^n$ is defined (for $n \ge 1$)
by $\1_x=(1,0,\ldots,0)$ for $x\in\Lambda$.

\begin{cor}
\label{cor:WK}
For $m^2 \in [L^{-\alpha(N-1)},\delta]$ and
$g \in [\frac{63}{64}\gLfix,\frac{65}{64}\gLfix]$,
\begin{align}
\lbeq{W12}
    &|W_N(0)| \le O((c/L)^{(\alpha+\alpha') (N-j_m)} \gLfix^2),
    \quad
    |W_N'(0)| \le O((c/L)^{(\alpha+\alpha') (N-j_m)} \mu_{j_m}' \gLfix),
\\
\lbeq{W3}
    &|W_N''(0)| \le O((c/L)^{(\alpha+\alpha') (N-j_m)} (\mu_{j_m}')^2 ),
\\
\lbeq{W4}
  &
  |D^2 W_{N}(0;\1,\1)|
  \leq
    O(L^{Nd}(c/L)^{\alpha (N-j_m)}  \gLfix^2 m^2)
  ,
\\
\lbeq{W5}
  &
  |D^2 W_{N}'(0;\1,\1)|
  \leq
    O(L^{Nd}(c/L)^{\alpha (N-j_m)} \mu_{j_m}' \gLfix  m^2)
  ,
\\
\lbeq{K123}
    &|K_N(0)| \le O(\chiL_N^3 \gLfix^3),
    \quad
    |K_N'(0)| \le O(\chiL_N^3 \mu_{j_m}' \gLfix^2),
    \quad
    |K_N''(0)| \le O(\chiL_N^3 (\mu_{j_m}')^2 \gLfix)
    ,
\\
\lbeq{K45}
  &
  |D^2 K_{N}(0;\1,\1)|
  \leq
    O( \chiL_N L^{Nd}  \gLfix^3  m^2)
  ,
\quad
  |D^2 K_{N}'(0;\1,\1)|
  \leq
    O(\chiL_N
    L^{Nd}
    \mu_{j_m}' \gLfix^2  m^2)
  .
\end{align}
\end{cor}

\begin{proof}
The bounds \refeq{W12}--\refeq{W3} follow immediately from Lemma~\ref{lem:WKNp}
and $|F(0)| \le \|F\|_{T_{0,N}}$.
By definition of the $T_\varphi$-seminorm in \refeq{Tphidef},
\begin{equation}
  |D^2F(0; f, f)|
  \leq 2 \|F\|_{T_{0,N}} \|f\|_{\Phi_N}^2.
\end{equation}
By \refeq{Phinorm} and \refeq{elldefa},
the norm of the constant test function $\1 \in \Phi_N$ is
\begin{equation}
    \|\1\|_{\Phi_N} = \ell_N^{-1} \sup_x |\1_x| = \ell_N^{-1}
    = \ell_0^{-1} L^{dN/2}L^{-\alpha j_m/2} L^{\alpha' (N-j_m)/2} .
\end{equation}
Therefore, by Lemma~\ref{lem:WKNp}, and
since $L^{-\alpha j_m}=O(m^2)$,
\begin{align}
    |D^2 W_{N}(0;\1,\1)|
    & \le
    O(c^{(\alpha+\alpha') (N-j_m)} \gLfix^2) L^{dN} m^2 L^{-\alpha (N-j_m)}
    .
\end{align}
This gives \refeq{W4} (with a new $c$),
and \refeq{W5} follows similarly from Lemma~\ref{lem:WKNp}.

The bounds on $K$ follow similarly from the norm estimates on $K$ and
its derivatives in Lemma~\ref{lem:nup}.
The factor $\chiL_N$ arises from $\chiL_N^3L^{\alpha'(N-j_m)} \le \chiL_N$,
as in \refeq{pKalphap}.
\end{proof}

\subsection{Susceptibility and its derivative}
\label{sec:diffineq}

Recall from \refeq{chichihat} that
\begin{equation}
\lbeq{chichihat-bis}
     \chi_N(g,\nu_0 + m ^2)
     =
     \hat\chi_N (m^2,g,\nu_0)
     .
\end{equation}
We begin with an elementary formula for $\hat\chi$.
Recall that $Z_N$ is defined in \refeq{ZNdef}, and the constant test function
$\1$ is as in Corollary~\ref{cor:WK}.

\begin{lemma}
\label{lem:comsq}
For $n \ge 0$, $m^2>0$, $g>0$, $\nu_0 \in \R$,
\begin{align}
  \hat\chi_N (m^2,g,\nu_0)
  &=
  \frac{1}{m^{2}} + \frac{1}{m^4} \frac{1}{|\Lambda_N|}  \frac{D^2Z_N(0; \1,\1)}{Z_N(0)}.
\end{align}
\end{lemma}

\begin{proof}
For simplicity we restrict attention to $n \ge 1$, as
$n=0$ requires merely notational changes.
Given a test function $J: \Lambda \to \R$, we write
  $(J,\varphi) = \sum_{x\in \Lambda} J_x \varphi_x^1$.
By \refeq{chihatdef} and symmetry,
\begin{align}
  \hat\chi_{N}
  &=
  \hat\chi_{N}(m^2, g, \nu_0)
  =
  \frac{1}{|\Lambda_N|}
  \frac{\Ex_{C}((\1,\varphi)^2 Z_0)}{Z_N(0)},
\end{align}
with $C = (-\Delta_\Lambda^{\alpha/2}+m^2)^{-1}$, and
$Z_0 =e^{-V_0(\Lambda)}$ as in \refeq{V0Z0}.
 (If $n=0$ then $Z_N(0)=1$.)
We define
$\Sigma_N: \R^\Lambda
\to \R$ by
\begin{align}
\lbeq{SigmaNdef}
  \Sigma_N(J)
  & =
  \Ex_C(e^{(J,\varphi)} Z_0).
\end{align}
Then
\begin{equation}
\label{e:chibarG}
  \hat\chi_{N}
  = \frac{1}{|\Lambda_N|} \frac{D^2\Sigma_N(0; \1,\1)}{Z_N(0)}.
\end{equation}

In \refeq{SigmaNdef}, we combine the exponential arising from the expectation
with the exponential containing the test function,
and complete the square to obtain
\begin{align}
\lbeq{fieldtrans}
    -\frac 12 (\varphi, C^{-1} \varphi) + (\varphi, J)
    &=
    -\frac 12 (\varphi -CJ, C^{-1}(\varphi -CJ)) + \frac 12 (J,CJ).
\end{align}
Then, by a change of variables,
\begin{align} \label{e:Gamma-sq}
  \Sigma_N(J)
  &=
  e^{(J,C J)} \Ex_C(Z_0(\varphi+C J))
  =
  e^{(J,C J)} Z_N(C J).
\end{align}
We differentiate \eqref{e:Gamma-sq}, and use the fact that  $C\1=m^{-2}\1$ by \refeq{C1}.
This
leads to $D^2\Sigma_N(0; \1,\1)=m^{-2}|\Lambda_N|Z_N(0)+m^{-4}D^2Z_N(0;\1,\1)$, and hence
\begin{align}
  \hat\chi_N
  &= \frac{1}{m^{2}} + \frac{1}{m^4} \frac{1}{|\Lambda_N|}  \frac{D^2Z_N(0; \1,\1)}{Z_N(0)},
\end{align}
and the proof is complete.
\end{proof}

The value of $\nu_0$ is arbitrary in Lemma~\ref{lem:comsq}, but now we fix
$\nu_0$ to be the critical value $\nu_0=\nu_0^c(m^2)$ of \refeq{nu0},
determined by Corollary~\ref{cor:mu0}
for $m^2 \in (0,\delta]$.
We write $\hat\chi_N'$ for the derivative of $\hat\chi_N$ with respect
to $\nu_0$ with $m^2,g$ held fixed, and evaluated at the critical $\nu_0^c(m^2)$.
By \refeq{chichihat-bis},
this is equal to the partial derivative of $\chi_N$ with respect to $\nu$,
evaluated at $\nu_0^c(m^2)+m^2$.
Recall from Theorem~\ref{thm:pastjm} and Lemma~\ref{lem:nup}
that $\nu_\infty$ and $\nu'_\infty$ are given by
the limits $\nu_\infty=\lim_{N \to \infty}\nu_N$ and $\nu'_\infty =\lim_{N \to \infty}\nu_N'$.
In the next proposition, we fix $m^2 \in (0,\delta]$, and consider the infinite volume limit of
$\hat\chi_N(m^2,\nu_0^c(m^2)) = \chi_N(\nu_0^c(m^2)+m^2)$,
together with its $\nu_0$-derivative.

\begin{prop}
\label{prop:chihat}
For $n \ge 0$, $m^2 \in (0,\delta]$, and
$g \in [\frac{63}{64}\gLfix,\frac{65}{64}\gLfix]$,
the limits
$\hat\chi=\lim_{N \to \infty} \hat\chi_N(m^2,\nu_0^c(m^2))$ and
$\hat\chi'=\lim_{N \to \infty} \hat\chi_N'(m^2,\nu_0^c(m^2))$ exist and are given by
\begin{align}
\lbeq{chihatlim}
    \hat\chi &=
    \frac{1}{m^2}  - \frac{\nu_\infty}{m^4}
    =
    \frac{1}{m^2} \left(1 +O(\gLfix) \right)  ,
    \\
\lbeq{chihatplim}
    \hat\chi' &= -\frac{\nu'_\infty}{m^4}
    \asymp
    -\frac{1}{m^4} m^{2\gamhat \epsilon/\alpha + O(\epsilon^2)} .
\end{align}
\end{prop}

\begin{proof}
By \refeq{ZIK}, $Z_N(\Lambda) = e^{-u_N|\Lambda|}(I_N(\Lambda) + K_N(\Lambda))$,
since the only polymers
at scale $N$ are $\varnothing,\Lambda$.
By Lemma~\ref{lem:comsq},
\begin{align}
  \hat\chi_N
\label{e:chibarm-bisIK}
  &=
  \frac{1}{m^{2}} +  \frac{1}{m^4}\frac{1}{|\Lambda|}\frac{1}{1+K_N(0)}
  \left(  D^2I_N(0; \1,\1)
  +   D^2K_N(0; \1,\1) \right).
\end{align}
Since $I_N(\Lambda) = e^{-V_N(\Lambda)}(1+W_N(\Lambda))$,
and since $DV_N(\Lambda,0;\1)= DW_N(\Lambda,0;\1)=0$ (because $V_N$ and $W_N$ are
even polynomials in $\varphi$),
\begin{equation}
\lbeq{D2Icalc}
   D^2 I_N(\Lambda; 0; \1,\1)
  =  D^2 e^{-V_N}(\Lambda;0; \1,\1)
  + D^2W_N(\Lambda; 0; \1,\1)
  .
\end{equation}
Also, since
$V_N(\Lambda)= \sum_{x\in\Lambda}
(\tfrac 14 g_N |\varphi_x|^4 + \tfrac 12 \nu_N |\varphi_x|^2)$ by \refeq{Vjdef},
\begin{equation}
  \label{e:D2eVN}
   D^2 e^{-V_N}(\Lambda;0; \1,\1)
  = - \nu_N |\Lambda|
  .
\end{equation}
This gives the identity
\begin{align}
  \label{e:chibarm-bis2}
  \hat\chi_N
  &
    =
  \frac{1}{m^{2}}
  +
   \frac{1}{m^4}
  \frac{A_N}{1+K_N(0)}
  ,
\end{align}
with
\begin{align}
    A_N &= -   \nu_N
  +
  \frac{1}{|\Lambda|} \left(
   D^2W_N(0; \1,\1)
  +  D^2K_N(0; \1,\1)
  \right)
   .
\end{align}
By Theorem~\ref{thm:pastjm}  and Corollary~\ref{cor:WK},
$A_N \to -\nu_\infty = O(m^2\gLfix)$,
and \refeq{chihatlim} follows from \refeq{chibarm-bis2}.

Differentiation of \eqref{e:chibarm-bis2} with respect to $\nu_0$,
followed by Corollary~\ref{cor:WK} and Lemma~\ref{lem:nup}, similarly gives
\begin{align}
\lbeq{chihatp}
   \hat\chi_N'
  & =
  \frac{1}{m^4}
  \left(
  \frac{A_N'}{1+K_N(0)}-\frac{A_NK_N'(0)}{(1+K_N(0))^2}
  \right)
  \to
  \frac{- \nu_\infty'}{m^4}
  .
\end{align}
Finally, it follows from Lemma~\ref{lem:nup},
Corollary~\ref{cor:nug},  Lemma~\ref{lem:gprod}, and Theorem~\ref{thm:pastjm}, that
\begin{align}
    \nu_\infty' & \asymp  P_{j_m} e^{O(\epsilon^2 j_m)}
    \asymp
    \left( \frac{g_{j_m}}{g} \right)^{\gamhat} e^{O(\epsilon^2 j_m)}
    \asymp
    L^{- \gamhat\epsilon j_m}  e^{O(\epsilon^2 j_m)}
    \asymp
    m^{2\gamhat\epsilon/\alpha +O(\epsilon^2)},
\end{align}
which establishes the asymptotic relation in \refeq{chihatplim}.

The convergence of $\nu_N'$ to $\nu_\infty'$ is uniform on
compact subsets of $m^2 \in (0,\delta)$ by Lemma~\ref{lem:nup}.  Using this,
it can be verified that
the convergence of $\hat\chi_N'$ to its limiting value
is uniform on compact
subsets of $m^2 \in (0,\delta)$.
Therefore the limit and derivative can be interchanged,
and $\chi'$ is in fact the derivative of $\chi$.
\end{proof}

\begin{rk}
\label{rk:tauLoc}
We chose to extract $\tau$ from $K$ using $\LT$ after the mass scale, even though $\tau$
is then an irrelevant monomial.
The reason for this choice is that $\chi$ receives a contribution from a $\varphi_0\varphi_{x}$
term in $K$.  Since we have extracted terms of this type
from $K$, their important contribution to the susceptibility has already been made and
what remains in $K_N$ goes to zero as $N \to \infty$, as in the proof
of Proposition~\ref{prop:chihat}.
\end{rk}

Let $\nu^*=\nu^*(m^2) = \nu_0^c(m^2)+m^2$.
By \refeq{chichihat-bis}, $\frac{\partial}{\partial \nu}\chi_N(g,\nu^*) =
\hat\chi_N'(m^2,g,\nu_0^c)$.
By Proposition~\ref{prop:chihat}, there is a constant $c$ such that
\begin{equation}
    c^{-1}m^{-4+ 2\frac{\gamhat}{\alpha}\epsilon +c\epsilon^2}
    \le -\frac{\partial\chi}{\partial\nu}(\nu^*)
    \le
    c m^{-4+2\gamhat\epsilon/\alpha -c\epsilon^2}
    .
\end{equation}
Since $\chi(\nu^*) \asymp m^{-2}$ by \refeq{chihatlim}, it is natural to write the above as
\begin{equation}
\lbeq{chi1-bis}
    -\chi(\nu^*)^{-2+ \epsilon \gamhat/\alpha + O(\epsilon^2)}
    \frac{\partial\chi}{\partial\nu}(\nu^*)
    \asymp 1.
\end{equation}
In particular,
\begin{equation}
\lbeq{chipneg}
    \frac{\partial\chi}{\partial\nu}(\nu^*)<0.
\end{equation}

\subsection{Proof of Theorem~\ref{thm:suscept}}
\label{sec:pf-suscept}

According to \refeq{nu0},
\begin{equation}
\lbeq{nuceta}
    \nu_0^c(m^2) = \mu_0(m^2) -  \eta_{\ge 0}(m^2) g
    =
    \mu_0(m^2) - (n+2)C_{00}(m^2)g.
\end{equation}
The limit $\mu_0(0)= \lim_{m^2 \downarrow 0} \mu_0(m^2)$ exists by Corollary~\ref{cor:mu00},
and is $O(\gLfix^2)$ by Theorem~\ref{thm:pcmi-new1}.
We define
\begin{equation}
\lbeq{nuc}
    \nu_c
    =  \lim_{m^2 \downarrow 0} \nu^*(m^2)
    = \mu_0(0) - (n+2)C_{00}(0)g.
\end{equation}
The following theorem identifies  $\nu_c$ as the critical value.

\begin{theorem}
\label{thm:criticalpt}
For $n \ge 0$ and $g \in [\frac{63}{64}\gLfix,\frac{65}{64}\gLfix]$,
the susceptibility $\chi(g,\nu)$ diverges to infinity as $\nu \downarrow \nu_c$,
and $\nu_c$ obeys the asymptotic formula
$\nu_c =  - (n+2)C_{00}(0)g(1+O(g))$
claimed in \refeq{nucasy}.
\end{theorem}

\begin{proof}
The asymptotic formula \refeq{nucasy} follows from
\refeq{nuc} and the observation, made above, that $\mu_0(0)=O(\gLfix^2)$.
To see that the susceptibility diverges as $\nu \downarrow \nu_c$,
we argue as follows.
Let
\begin{equation} \label{e:sfNdef}
    {\sf N} = \{ \nu^*(m^2) : m^2 \in [0,\delta]\},
    \quad
    {\sf N}_+ = \{ \nu^*(m^2) : m^2 \in (0,\delta]\}.
\end{equation}
Since $\nu^*: [0,\delta] \to \R$ is continuous by Corollaries~\ref{cor:mu0}--\ref{cor:mu00},
and since continuous functions
map compact connected sets to compact connected sets, ${\sf N}$ is a closed interval.
It is not
possible that ${\sf N}$ consists of a single point.  Indeed, by \refeq{chihatlim},
for $m^2 \in (0,\delta]$,
\begin{equation}
\lbeq{chim-bis}
  \chi(\nu^*(m^2))
    = \frac{1}{m^2}(1+O(\epsilon)).
\end{equation}
The right-hand side is not constant in $m^2$, so the left-hand side cannot be constant,
and hence ${\sf N}$ cannot consist of a single point.
Therefore, for some $x_c$, ${\sf N}=[x_c,x_c+\eta]$ with $\eta>0$.
By
\refeq{chim-bis}, $\chi(\nu^*(m^2)) < \infty$ for $m^2>0$ whereas
$\chi(\nu^*(m^2)) \to \infty$ as $m^2 \downarrow 0$.  We have not proved that
$\chi(\nu^*(m^2))$ increases as $m^2$ decreases.
However, we do know from \refeq{chipneg} that $\chi'(\nu)<0$ for each $\nu \in {\sf N}_+$,
so $\chi$ is strictly monotone decreasing in $\nu\in {\sf N}_+$.  Therefore,
the only point in ${\sf N}$ at which $\chi$ can be infinite is $x_c$, and we must
have $\nu^*(m^2) \to x_c$ as $m^2 \downarrow 0$.  It follows from  \refeq{nuc} that
$x_c=\nu_c$, and it also follows that $\chi(\nu) \uparrow \infty$ as $\nu\downarrow\nu_c$.
This completes the proof.
\end{proof}

\begin{proof}[Proof of Theorem~\ref{thm:suscept}]
It remains to prove \refeq{chigam}.
Fix $\nu>\lambda>\nu_c$ with $\nu-\nu_c$ small.
Integration of \refeq{chi1-bis} over the interval $[\lambda,\nu]$ gives
\begin{equation}
    \chi(\nu)^{-1+ \epsilon \gamhat/\alpha +O(\epsilon^2)}
    -
    \chi(\lambda)^{-1+ \epsilon \gamhat/\alpha + O(\epsilon^2)}
    \asymp
    \nu- \lambda.
\end{equation}
Since $\chi(\lambda) \uparrow \infty$ as $\lambda \downarrow \nu_c$, this gives
\begin{equation}
    \chi (\nu) \asymp (\nu-\nu_c)^{-1/(1- \epsilon \gamhat/\alpha +O(\epsilon^2))}
    \asymp
    (\nu-\nu_c)^{-(1+ \epsilon \gamhat/\alpha + O(\epsilon^2))},
\end{equation}
and the proof is complete.
\end{proof}

\section{Proof of Theorem~\ref{thm:sh}}
\label{sec:pfsh}

The coupling constant $u_j$ plays no role in the analysis of the susceptibility,
as it cancels in numerator and denominator in the formula for $\hat\chi_N$
in Lemma~\ref{lem:comsq}.  However, for the specific heat it is fundamental, and
we begin in Section~\ref{sec:uflow} with an analysis of the flow of $u_j$.
The proof of Theorem~\ref{thm:sh}
is then given in Section~\ref{sec:sh}.
We only consider $n \ge 1$ in this section.

\subsection{Analysis of flow equation for  \texorpdfstring{$u$}{u}}
\label{sec:uflow}

\begin{lemma} \label{lem:u2p}
Let $n \ge 1$, $m^2 \in (0,\delta]$,
$g \in [\frac{63}{64}\gLfix,\frac{65}{64}\gLfix]$,
and $\nu_0=\nu_0^c(m^2)$.
For $l=0,1,2$, the limits $u^{(l)}_\infty =  \lim_{N\to\infty} u_N^{(l)}$
exist, are attained uniformly on compact subsets of $m^2\in (0,\delta]$, and
\begin{alignat}{2}
\label{e:uNprime2lim}
    -u_\infty'' &\asymp
    m^{-2\frac\epsilon\alpha \frac{4-n}{n+8} + O(\epsilon^2)}
    \quad
    && (n <4),
    \nnb
    -u_\infty''&\le O(m^{- O(\epsilon^2)}) && (n=4),
    \\ \nonumber
    -u_\infty'' &\asymp1 && (n>4).
\end{alignat}
\end{lemma}

\begin{proof}
We give the proof only for $u_\infty''$; existence of the limits $u_\infty,u_\infty'$ is similar.
Recall from \refeq{uptlong} that
\begin{equation}
\lbeq{uptlong-reminder}
  \delta u_\pt =    \kappa_g' g + \kappa_{\nu}' \nu
  - \kappa_{g\nu}' g\nu
  - \kappa_{gg}' g^2 - \kappa_{\nu\nu}' \nu^2.
\end{equation}
The primes in \refeq{uptlong-reminder} occur in the unscaled
coefficients defined in \refeq{ugreeks2};
they are not derivatives with respect to $\nu_0$, whereas primes on $u,g,\nu$ do denote
derivatives.
By \refeq{RUPT},
and as discussed above \refeq{Kspace-objective},
$u_{j+1}-u_j = \delta u_\pt + r_{u,j}$.
From  \refeq{uptlong-reminder}, we obtain
\begin{equation}
    u_N
    = \sum_{j=0}^{N-1}
    \Big( \kappa_{g,j}' g_j + \kappa_{\nu,j}' \nu_j
    - \kappa_{g\nu}' g_j\nu_j
    - \kappa_{gg,j}' g_j^2 - \kappa_{\nu\nu,j}' \nu_j^2
     + r_{u,j}
    \Big).
\end{equation}
In terms of the rescaled variables $\ghat_j=L^{\epsilon (j\wedge j_m)}g_j$,
$\muhat_j = L^{\alpha(j\wedge j_m)}\nu_j$, and the rescaled coefficients given in
\refeq{kappadef}, this becomes
\begin{equation}
\begin{aligned}
    u_N &= \sum_{j=0}^{N-1}
    \Big( \kappa_{g,j} \ghat_j + \kappa_{\nu,j} L^{\alpha(j-j_m)_+} \muhat_j
    - \kappa_{g\nu}  L^{\alpha(j-j_m)_+} \ghat_j\muhat_j
    \\ & \quad\quad
    - \kappa_{gg,j} L^{2\epsilon(j-j_m)_+} \ghat_j^2 - \kappa_{\nu\nu,j}\muhat_j^2
     + r_{u,j}
    \Big).
\end{aligned}
\end{equation}
Therefore,
\begin{align}
\begin{aligned}
\lbeq{uNpp}
    u_N'' &= \sum_{j=0}^{N-1}
    \Big( \kappa_{g,j}\ghat_j'' + \kappa_{\nu,j} L^{\alpha(j-j_m)_+} \muhat_j''
   -\kappa_{g\nu}  L^{\alpha(j-j_m)_+}  (\ghat_j''\muhat_j + 2\ghat_j'\muhat_j' + \ghat_j\muhat_j'')
    \\ & \quad\quad\quad
     - 2\kappa_{gg,j} L^{2\epsilon(j-j_m)_+}  (\ghat_j\ghat_j''+(\ghat_j')^2)
     -2\kappa_{\nu\nu,j} (\muhat_j\muhat_j''+(\muhat_j')^2)
     + r_{u,j}''
    \Big).
\end{aligned}
\end{align}

By Lemma~\ref{lem:wlims} and \refeq{Mjbd}, $\kappa_{*,j} \le O(L^{-2\alpha(j-j_m)_+}L^{-dj})$.
Fix $m^2 \in (0,\delta]$, and let $N$ be large enough
that $m^2 \in [L^{-\alpha(N-1)},\delta]$.
By Theorem~\ref{thm:pastjm} and Lemma~\ref{lem:nup}, $\ghat_j = O(\gLfix)$ and
$\muhat_j=O(\gLfix)$ for all $j$.  In fact, the flow of $\ghat$ stops at the mass scale.
By Corollary~\ref{cor:nug} and Lemma~\ref{lem:nup},
for all $0\le j<N$,
\begin{alignat}{2}
\lbeq{umup}
    \muhat_j' &=
    L^{\alpha (j\wedge j_m)} P_{j\wedge j_m}e^{O(\gLfix^2 (j\wedge j_m))}(1+O(\gLfix)),
    \qquad
    &&
    \ghat_j' = O(\muhat_j'\gLfix),
    \\
\lbeq{umupp}
    \muhat_j'' &= (\muhat_j')^2(-2\bar w_j^{(1)} +O(\gLfix))
    ,
    &&
    \ghat_j'' =  O((\muhat_j')^2\gLfix)
    .
\end{alignat}
By Lemmas~\ref{lem:gzmuprime} and \ref{lem:nup},
\begin{equation}
    \quad
    r_{u,j}'' = O(L^{-d j}\chicCov_j (\muhat_j')^2 \gLfix ).
\end{equation}
Each term in \refeq{uNpp} contains a factor $\gLfix$, except
$\kappa_{\nu,j} L^{\alpha(j-j_m)_+} \muhat_j''
-2\kappa_{\nu\nu,j}  (\muhat_j')^2$,
and these two terms enjoy a cancellation.  In fact,
according to \refeq{ugreeks2}, and \refeq{umup}--\refeq{umupp}, they are equal to
\begin{align}
\lbeq{kappacancel}
    \kappa_{\nu}'\nu_j'' - 2\kappa_{\nu\nu,j}'(\nu_j')^2
    & =
    - \tfrac{1}{2}n\delta[w^{(2)}] (\nu_j')^2
    +
    \tfrac{1}{2} nC(\nu_j'' + 2w_j^{(1)}(\nu_j')^2)
    \nnb & =
    - \tfrac{1}{2}n\delta[w^{(2)}] (\nu_j')^2
    +
    O(\gLfix L^{\alpha(j\wedge j_m)})C (\nu_j')^2
    .
\end{align}

We write $A_N$ for the contribution to $-u_N''$ due to $j \le j_m$,
and $B_N$ for the contribution due to $j_{m} <j <N$.
By use of the above estimates, and since $2\alpha-d=\epsilon$, we find that, with $\tilde P_j= P_je^{O(\gLfix^2 j)}$,
\begin{equation}
\lbeq{ANsum}
    A_N = \sum_{j=0}^{j_m}
    \left(
    \tfrac{1}{2}n\delta[w^{(2)}] \tilde P_j^2
    +
    O(L^{\epsilon j} \tilde P_j^2 \gLfix)
    \right)
    .
\end{equation}
By Lemma~\ref{lem:gprod}, $\tilde P_j \asymp L^{-\epsilon j \gamhat}e^{O(\gLfix^2 j)}$.
Also, by definition and by Lemma~\ref{lem:wlims},
$0 \le \delta_j [w^{(2)}] =\frac{1}{n+8}\beta_j' \le O(  L^{\epsilon j})$ for $j \le j_m$.
We conclude that
\begin{equation}
\lbeq{ANdef}
    A_N
    \le
    \sum_{j=0}^{j_m} O(L^{\epsilon j(1-2 \gamhat+ O(\epsilon ))}).
\end{equation}
Recall from \refeq{gamhatdef} that $\gamhat = \frac{n+2}{n+8}$.
The sign of $1-2\gamhat$ is important:
\begin{equation}
    1 - 2\hat \gamma   = \frac{4-n}{n+8}
    \;\;\;
    \begin{cases}
    > 0 & (n<4)
    \\
    =0 & (n=4)
    \\
    < 0 & (n>4).
    \end{cases}
\end{equation}
This gives
\begin{align}
    A_N
    & \le
    c\times
    \begin{cases}
    L^{\epsilon j_m(1-2 \gamhat + O(\epsilon ))}
    & (n<4)
    \\
    L^{O(\epsilon^2 j_m)} & (n=4)
    \\
    1 & (n>4),
    \end{cases}
\end{align}
which in turn gives
\begin{align}
\lbeq{ANbd}
    A_N
    &
    \le
    c\times
    \begin{cases}
    m^{-2\frac{\epsilon}{\alpha} \frac{4-n}{n+8} + O(\epsilon^2 ))}
    & (n<4)
    \\
    m^{-O(\epsilon^2 )} & (n=4)
    \\
    1 & (n>4).
    \end{cases}
\end{align}
We also need a lower bound on $A_N$ for $n \neq 4$.
By \refeq{xidef}, \refeq{Greeknoprime}, and Lemmas~\ref{lem:betadiff}--\ref{lem:beta-am},
apart from a bounded number of scales near $0$ and near $j_m$, for $j \le j_m$ we have
$\delta_j[w^{(2)}] \asymp L^{\epsilon j}$.  By restricting the sum
in \refeq{ANsum} to avoid those few scales, the desired lower bound on $A_N$ follows
similarly.

Next, we estimate the contribution $B_N$ to the sum \refeq{uNpp}
due to scales $j_m<j<N$.  It is straightforward to obtain an upper bound
by using the additional exponential decay in
$\kappa_{*,j} \le O(L^{-2\alpha(j-j_m)_+}L^{-dj})$.
The result is that $B_N$ also obeys the upper bound \refeq{ANbd}, and consequently so
does $-u_N''$.
For a lower bound on $B_N$ for $n\neq 4$,
as in \refeq{ANsum}, but now taking into account the exponential decay above
the mass scale, there is a $z>0$ such that
\begin{equation}
\lbeq{BNlbd}
    B_N = \sum_{j=j_m+1}^{N-1}
    \left(
    \tfrac{1}{2}n\delta[w^{(2)}] \tilde P_j^2
    +
    O(L^{-z(j-j_m)} L^{\epsilon j_m} \tilde P_{j_m}^2 \gLfix)
    \right)
    \ge
    -
    O(\gLfix)
    L^{\epsilon j_m} \tilde P_{j_m}^2
    ,
\end{equation}
where the (nonnegative) $\delta[w^{(2)}]$ term has been discarded in the lower bound.
This is of the same order in $j_m$ as the upper bound on $A_N$, but it contains
an additional factor $\gLfix$, so it cannot spoil the lower bound provided by $A_N$.

We finally show that $u_N \to u_\infty$
uniformly in $m^2 \in [\delta_0, \delta]$, for any $\delta_0 \in (0,\delta]$.
It suffices to show that this holds for the restriction of \eqref{e:uNpp}
to $j \geq j_{\delta_0}$.
Then the summands are uniformly bounded by $O(\chiL_j)$,
which decays exponentially, and the claim follows.  Uniform convergence of the
derivatives is similar.
This completes the proof.
\end{proof}

\subsection{Proof of Theorem~\ref{thm:sh}}
\label{sec:sh}

\begin{proof}[Proof of Theorem~\ref{thm:sh}]
Let $n \ge 1$.
By the definition in \refeq{pf}, and by moving part of the quadratic term
$\frac 12 \nu |\varphi_x|^2$ into the covariance as in \refeq{g0gnu0nu}--\refeq{ExF},
the partition function is
\begin{equation}
\begin{aligned}
    Z_{g,\nu,N}
    &= \int_{(\R^n)^{\Lambda_N}} e^{-\sum_{x\in\Lambda}
  (\frac{1}{4} g |\varphi_x|^4 + \frac 12 \nu |\varphi_x|^2
  +
  \frac 12
   \varphi_x \cdot ((-\Delta)^{\alpha/2} \varphi)_x )} d\varphi
   \\ & =  Z_{0,m^2,N}
    \Ex_C Z_0
    =  Z_{0,m^2,N} Z_N(0),
\label{e:Zpress}
\end{aligned}
\end{equation}
where $m^2>0$ is arbitrary, $C=((-\Delta)^{\alpha/2}+m^2)^{-1}$, and
$Z_{0,m^2,N}$  cancels the normalisation of the Gaussian measure $\Ex_C$.
The finite-volume pressure is
  $p_N(g,\nu) = |\Lambda_N|^{-1} \log Z_{g,\nu,N}$.

We have seen in the proof of Theorem~\ref{thm:criticalpt} that the
set ${\sf N}_+ = \{ \nu^*(m^2) : m^2 \in (0,\delta]\}$ is a non-trivial interval
${\sf I}_+=(\nu_c,\nu_c+\eta]$.  Given $\nu\in {\sf I}_+$, we can therefore
choose $\mgen^2=\mgen^2(\nu)>0$ so that $\nu=\nu^*(\mgen^2)=\nu_0(\mgen^2) + \mgen^2$.
We take $m^2=\mgen^2$ in \refeq{Zpress} (also in $Z_N$), and then take the logarithm, to
obtain
\begin{align}
\lbeq{pNuN0}
    p_N(g,\nu) & = p_N(0,\mgen^2) + |\Lambda_N|^{-1} \log Z_N(0).
\end{align}
As in the proof of Proposition~\ref{prop:chihat},
$Z_N(\Lambda) = e^{-u_N|\Lambda|}(I_N(\Lambda) + K_N(\Lambda))$,
so evaluation at $\varphi=0$ gives
\begin{align}
\lbeq{pNuN}
    p_N(g,\nu)
    & = p_N(0,\mgen^2) -u_N + |\Lambda_N|^{-1} \log (1+W_N(\Lambda;0)+K_N(\Lambda;0)).
\end{align}
By Corollary~\ref{cor:WK}, and with $N$ large enough that $\mgen^2 \ge L^{-\alpha(N-1)}$,
\begin{align}
\lbeq{pNpp}
    \frac{\partial^2 p_N}{\partial \nu^2}
    & =
    -u_N''
    +
     |\Lambda_N|^{-1} O(|W_N'(\Lambda;0)|^2+|K_N'(\Lambda;0)|^2 + |W_N''(\Lambda;0)|+|K_N''(\Lambda;0)|)
     \nnb
     & = -u_N'' +  O(L^{-dN}\chiL_N^3 (\mu_{j_m}')^2 )
     .
\end{align}
As in \refeq{fp}, derivatives on the right-hand side are partial derivatives with respect
to $\nu_0$, evaluated at $(\mgen^2,g,\nu_0^c(\mgen^2))$, now with $\mgen^2=\mgen^2(\nu)$.

Let $t>0$ be given by $\nu= \nu_c+t$.
Since
$\mgen^{-2}\asymp\chi$ by Proposition~\ref{prop:chihat} and \refeq{chichihat-bis},
and since $\chi \asymp t^{-(1+ O(\epsilon ))}$ by Theorem~\ref{thm:suscept},
it follows from Lemma~\ref{lem:u2p} that
\begin{equation}
    \lim_{N\to\infty}\frac{\partial^2 p_N}{\partial \nu^2} = -u_\infty'' \;
    \begin{cases}
    \asymp
    t^{- \frac{4-n}{n+8} \frac \epsilon\alpha +O(\epsilon^2)}
    & (n <4)
    \\
    \le O(t^{- O(\epsilon^2)})  & (n=4)
    \\
    \asymp
    1 & (n>4).
    \end{cases}
\end{equation}
This completes the proof, except we have not yet shown that
$\frac{\partial^2 p}{\partial \nu^2}=\lim_{N\to\infty}\frac{\partial^2 p_N}{\partial \nu^2}$.

For this last detail, we see from \refeq{pNuN} that
$\lim_{N\to\infty} p_N(g,\nu) = p(g,\nu)$, for $\nu\in (\nu_c,\nu_c+\delta)$.
It suffices to show that
the derivatives $p_N',p_N''$
converge compactly (uniformly on compact subsets)
in $\nu \in {\sf N}_+$  to limiting functions,
as this implies that $p'=\lim_{N\to\infty}p_N'$ and $p''=\lim_{N\to\infty}p_N''$.
We establish the compact convergence for $p_N''$, and $p_N'$ is similar.
First we claim that the
right-hand side of \eqref{e:pNpp} converges compactly in
$m^2 \in (0,\delta]$.  We know this for $u_N''$ by Lemma~\ref{lem:u2p}.
The bounds of Corollary~\ref{cor:WK}
hold
uniformly on $[L^{-\alpha N},\delta]$, and thus uniformly on compact subsets of
$m^2 \in (0,\delta)$,
for sufficiently large $N$ (depending on the subset).
They all
converge compactly to $0$ as $N\to\infty$.
To translate this into compact convergence in $\nu \in {\sf N}_+$, let
$I \subset {\sf N}_+$ be a compact interval,
and let $J$ be the closure of its image
under $\mgen^2$.  It is impossible that $0 \in J$.
To see this, we observe that
since $m^2 \mapsto \nu = \nu_0^c(m^2)+m^2$
is continuous with $\nu \downarrow \nu_c$ as $m^2\downarrow 0$, if $0$ were in $J$
then $\nu_c$ would have to be a limit point of $I$, which is not possible.
Thus compact convergence on $m^2$-intervals implies compact convergence on $\nu$-intervals.
\end{proof}

\section{Estimates on covariance decomposition}
\label{sec:Cbound}

In this section, we prove the necessary estimates on the covariance decomposition for
the fractional Laplacian, together with estimates on $\beta_j$ and other coefficients
in the flow equations.
Namely, we prove
Proposition~\ref{prop:Cbound} and Lemmas~\ref{lem:wlims}--\ref{lem:betadiff}.

\subsection{Proof of Proposition~\ref{prop:Cbound}}

According to \refeq{CGamint}, the covariance decomposition for the fractional Laplacian
involves terms
\begin{equation}
\lbeq{CGamintz}
    C_{j;0,x} =
    \int_0^\infty  \Gamma_{j;0,x}(s) \; \rho^{(\alpha/2)}(s,m^2) \, ds,
\end{equation}
with $\Gamma_j(s)$ a term in the decomposition $\Gamma(s)=(-\Delta_{\Zd}+s)^{-1}
=\sum_{j=1}^\infty \Gamma_j(s)$.
This requires control of $\Gamma_j(s)$ for all $s\in (0,\infty)$.
The following proposition is an extension of Proposition~\ref{prop:Cbound},
which does not have the restriction $m^2 \le \bar m^2$,
and which includes the estimate \refeq{dCdm-new} giving regularity in $m^2$.
Relaxation of the restriction $m^2 \le \bar m^2$ leads
to an additional term in the estimate \refeq{scaling-estimate-new},
compared to \refeq{scaling-estimate}.

\begin{prop}
\label{prop:Cbound-new}
Let $d \ge 1$, $\alpha \in (0,2 \wedge d)$,
$L \ge 2$,
$m^2 \in [0,\infty)$,
and let $a$ be a multi-index with $|\multia| \le \bar a$.
Let $j \ge 1$ for $\Zd$, and let
$1 \le j <N$ for $\Lambda_N$.
The covariance
$C_j=C_j(m^2)$ has range $\frac 12 L^j$,  i.e.,
$C_{j;x,y}=0$ if $|x-y|\ge \frac 12 L^j$,
$C_{j;x,y}$ is continuous in
$m^2 \in [0,\infty)$, and, for any $p' \ge 0$,
\begin{equation}
\lbeq{scaling-estimate-new}
    | \nabla^\multia  C_{j;x,y}|
    \le
    c L^{-(d-\alpha+|\multia|)(j-1)}
    \left(
    \frac{1}{1+m^4 L^{2\alpha (j-1)}} + \frac{1}{1+m^2 L^{p' (j-1)}}
    \right)
    ,
\end{equation}
where $\nabla^\multia$ can act on either $x$ or $y$ or both.
For $m^2 \in (0,\bar m^2]$,
\begin{align}
\lbeq{CNNbd-new}
    |\nabla^\multia  C_{N,N;x,y}|
    &\le
    c
    L^{-(d-\alpha+|\multia|)(N-1)}
    \frac{1}{(m^2L^{\alpha (N-1)})^{2}}.
\end{align}
Let $d=1,2,3$.  For $m^2L^{\alpha(j-1)} \in (0, 1]$, for  $\alpha \in (\half,1)$ when $d=1$,
and for $\alpha \in (1,2)$ when $d=2,3$,
\begin{align}
\lbeq{dCdm-new}
    \Big|\frac{\partial}{\partial m^2} \nabla^\multia  C_{j;x,y} \Big|
    & \le
        c
    L^{(\epsilon - |\multia|)(j-1)}
    \times
    \begin{cases}
    (m^2L^{\alpha(j-1)})^{-(2-1/\alpha)} & (d=1)
    \\
    (m^2L^{\alpha(j-1)})^{-(2-2/\alpha)}|\log (m^2L^{\alpha(j-1)})| & (d=2)
    \\
    (m^2L^{\alpha(j-1)})^{-(2-2/\alpha)} & (d=3).
    \end{cases}
\end{align}
The constant $c$ may depend on $\bar a, p'$, but does not depend on
$m^2, L,j,N$.
\end{prop}

\begin{rk}
\label{rk:Mitt16}
A version of \refeq{scaling-estimate-new} appears in \cite{Mitt16}.
Our proof has the same starting point as the one in \cite{Mitt16}, but
an incorrect estimate was applied in \cite{Mitt16} (subsequently
corrected in an Erratum, see also \cite{Mitt17}).  We give a self-contained proof
here.  In particular:
\begin{enumerate}
\item
It is incorrectly claimed in \cite[(3.4)]{Mitt16} that
$|\Gamma_{j;x,y}(s)| \le c_\Gamma L^{-(j-1)(d-2)} e^{-s^{1/2}L^{j-1}}$ for $j >1$
(the same claim occurs in \cite{BGM04,BM12}), and this claim is used in proofs in \cite{Mitt16}.
An indication of the problem can be seen from the fact
that the decay of the full covariance $\Gamma_{0,x}(s)$ is slower, namely
$e^{-m_0(s)|x|_\infty}$ with $\cosh (m_0(s))= 1 + \frac 12 s$, i.e.,
$m_0 \sim \log s$ as $s \to \infty$ \cite[Theorem~A.2]{MS93}.
This is consistent with the random walk representation
 (cf.~\refeq{srw}), since for large $m^2$ the dominant
contribution to $\Gamma_{0,x}(s)$ will arise from the shortest possible walk, which has
weight $O(1+s)^{-\|x\|_1}$.
\item
The last term of \refeq{scaling-estimate-new} is absent in \cite{Mitt16},
yet a term of order $m^{-2}$ must be present.
This can be seen from the random walk representation for $C_{0,x}$,
which for large $m^2$ and $x=0$ is dominated by the zero-step walk, which contributes
$O(1+m^2)^{-1}$.  Our needs concern small $m^2$, for which the second
term in \refeq{scaling-estimate-new} is dominated by the first.
However, the second term can dominate, e.g., for large $m^2$ and $j=1$.
\item
The constant $c_p$ in \cite[(3.4--3.5)]{Mitt16} is stated to be independent of $L$ for $d=2$,
but the best bound we are aware of is $O(\log L)$;
see \cite[Proposition~3.3.1]{BBS-brief}.  Use of such a bound
spoils the proof of the $L$-independence of $c_{p,\alpha}$ in \cite[(1.17)]{Mitt16}
(see, however, the Erratum \cite{Mitt16}).
Our argument below does
prove $L$-independence of $c$ in Proposition~\ref{prop:Cbound-new} for all dimensions $d$.
\end{enumerate}
\end{rk}

We base our analysis on \cite{BBS-brief}, which in turn is based on \cite{Baue13a}.
By\cite[Proposition~3.3.1]{BBS-brief},
for any multi-index $a$ and for any $p \ge 0$ (as large as desired), we have
\begin{equation}
\lbeq{Gamjbd}
    |\nabla^{a} \Gamma_{j;x,y}(s)|
    \leq c_{\Gamma} \frac{1}{2d+s} \left(1+\frac{sL^{2(j-1)}}{2d+s} \right)^{-p}
    L^{-(j-1)(d-2 + |a|)},
\end{equation}
where the constant $c_{\Gamma}$ depends on $a,p$,
is independent of $L$ for $d>2$, but contains a factor
$\log L$ for $d=2$ and a factor $L^{2-d}$ for $d<2$.
To avoid having the $L$-dependence of $c_\Gamma$ enter into the constant $c$ of
\refeq{scaling-estimate-new}, we do not apply \refeq{Gamjbd} directly, but proceed
instead as follows.

Let  $J_j=[\half L^{j-1},\half L^j]$ for $j \ge 1$.
From the proof of \cite[Proposition~3.3.1]{BBS-brief}, we have
\begin{equation}
\lbeq{Gamw}
    \Gamma_{j;0,x}(s) = \int_{J_j} w(t,x;s) \frac{dt}{t}
    + \1_{j=1} \int_0^{\frac 12} w(t,x;s) \frac{dt}{t},
\end{equation}
and hence, by \refeq{CGamintz},
\begin{equation}
\lbeq{Cjintegral}
    C_{j;0,x}(m^2) = \int_0^\infty \, ds \rho(s,m^2) \int_{J_j} \frac{dt}{t} w(t,x;s)
    +
    \1_{j=1} \int_0^\infty \, ds \rho(s,m^2) \int_0^{\frac 12} \frac{dt}{t} w(t,x;s)
    .
\end{equation}
The function $w$ obeys the estimates of
\cite[Lemma~3.3.6]{BBS-brief}, namely
(with $L$-independent constant depending on $a,p$)
\begin{equation}
    \left| \nabla^a w(t,x;s)\right|
    \le
    c_0 \frac{1}{1+s} \frac{1}{(1+ \frac{st^2}{1+s})^p} (t^2 \wedge t^{-(d-2+|a|)}),
\end{equation}
and this implies that
\begin{equation}
\lbeq{wbounds}
    \left| \nabla^a w(t,x;s)\right|
    \prec
    \begin{cases}
    \frac{t^2}{1+s} & (t \le 1)
    \\
    \frac{1}{(1+ st^2)^p} \frac{1}{t^{d-2+|a|}}
     & (t \ge \half, \, s \le 1)
    \\
    \frac{1}{s}
    \frac{1}{t^{2p}}  \frac{1}{t^{d-2+|a|}}
    & (t \ge \half, \, s \ge 1)
    .
    \end{cases}
\end{equation}
In the above inequality, the notation $f \prec g$ means that $f \le cg$ with a
constant $c$ whose value is unimportant.
We continue to use this notation throughout this section.
The specification $t \ge \frac 12$ is for later convenience and the form of the
bound remains the same for $t \ge t_0$ for any fixed $t_0>0$.

We also need estimates on $\rho^{(\alpha/2)}$ (recall \refeq{rhodef}).  We write
$\rho=\rho^{(\alpha/2)}$, $\beta = \alpha/2$, and  $A=m^2$.
Then
\begin{align}
\lbeq{rho4}
    0 \le \rho(s,A) \prec \frac{s^{\beta}}{(s^\beta + A)^2},
    \quad
    \Big| \frac{\partial}{\partial A}\rho(s,A) \Big|
    \prec \frac{s^{\beta}}{(s^\beta + A)^3}.
\end{align}
The first bound is elementary; a proof is given in \cite{Mitt16}.
For the second bound, by definition,
\begin{align}
    \frac{\partial}{\partial A} \rho (s,A)
    & =
    - \frac{\sin\pi\beta}{\pi}
    \frac{s^{\beta}(2A + 2s^{\beta} \cos \pi\beta)}
    {(s^{2\beta} + A^2 +2As^{\beta}\cos \pi\beta)^2}
    ,
\end{align}
and hence, using a bound analogous to the first one in \refeq{rho4}, we have
\begin{align}
\lbeq{drhobd}
    \Big|\frac{\partial}{\partial A}  \rho (s,A)\Big|
    & \le
    \frac{s^{\beta}(s^{\beta}+A  )}
    {(s^{2\beta} + A^2 +2As^{\beta}\cos \pi\beta)^2}
    \prec
    \frac{s^{\beta}}{(s^{\beta}+A  )^3}.
\end{align}

The next two lemmas concern elementary integrals that enter into the analysis.
Let
\begin{equation}
    I(d,s) = \int_1^\infty \frac{dt}{t} \frac{1}{1+st^2} t^{2-d}.
\end{equation}

\begin{lemma}
\label{lem:Ids}
For $s \ge 1$,  $I(d,s) \asymp s^{-1}$.
For $s \le 1$,
\begin{equation}
    I(d,s) \asymp
    \begin{cases}
    1 & ( d > 2)
    \\
    \log s^{-1} & (d=2)
    \\
    s^{-1/2} & (d=1).
    \end{cases}
\end{equation}
\end{lemma}

\begin{proof}
The statement for $s \ge 1$ is immediate after using $\frac{1}{1+st^2} \asymp \frac{1}{st^2}$
(for $t \ge 1$).
Suppose that $s \le 1$.
For $d>2$, we have $I(d,1) \le I(d,s) \le I(d,0) < \infty$.  For $d=1,2$,
with $\tau = s^{1/2}t$, we have
\begin{align}
    I(d,s) & = s^{-1+d/2} \int_{s^{1/2}}^\infty
    \frac{d\tau}{\tau} \frac{1}{1+\tau^2} \tau^{2-d}.
\end{align}
The integral converges at $\infty$,   diverges logarithmically at $0$ for $d=2$,
and converges at $0$ for $d=1$.
\end{proof}

For $A,q \ge 0$ and $\gamma,\beta\in \R$, we define the integrals (possibly infinite)
\begin{align}
\lbeq{I1def}
    I_1(\gamma,\beta,q,A) & =
    \int_0^1 ds \frac{s^{\gamma}}{(s^\beta + A)^{2+q}},
    \\
\lbeq{I1log}
    I_1(\log,\beta,q,A) & =
    \int_0^1 ds \frac{s^{\beta}}{(s^\beta + A)^{2+q}} \log s^{-1},
    \\
    I_2(\gamma,\beta,q,A) & =
    \int_1^\infty \, ds \frac{s^{\gamma}}{(s^\beta + A)^{2+q}} .
\end{align}
The use of ``$\gamma=\log$'' as an argument on the left-hand side of \refeq{I1log} is a notational
convenience to indicate the right-hand side, which is like the $\gamma=\beta$
case of \refeq{I1def} but with an additional logarithmic factor.
The next lemma examines the behaviour of these integrals as $A \downarrow 0$ and $A \to \infty$,
for various ranges of $q,\gamma,\beta$.

\begin{lemma}
\label{lem:I12}
Let $r=(2+q)-(\gamma+1)/\beta$, with $\gamma=\beta$ for the integral \refeq{I1log}.
Then
\begin{align}
    I_1(\gamma,\beta,q,A) &
    \begin{cases}
    \asymp1 & (A \le 1, \, r<0)
    \\
    \asymp A^{-r} & (A \le 1, \, r>0, \, \gamma \neq \log)
    \\
    \asymp A^{-r}\log A^{-1} & (A \le 1, \, r>0, \, \gamma = \log)
    \\
    \prec A^{-(2+q)} & (A \ge 1, \, \gamma > -1),
    \end{cases}
    \\
    I_2(\gamma,\beta,q,A) &
    \begin{cases}
    \asymp 1 & (A \le 1, \, r>0)
    \\
    \asymp A^{-r} & (A \ge 1, \, r>0).
    \end{cases}
\end{align}
\end{lemma}

\begin{proof}
We first consider $I_1$.  We give the proof for \refeq{I1def}; the case \refeq{I1log}
follows similarly.
For $r<0$ and $A \le 1$, the integral converges when $A=0$, so $I_1 \asymp 1$.
For $r>0$ and $A \le 1$, since the $\sigma$-integral converges
at infinity we obtain (set $s=\sigma A^{1/\beta}$)
\begin{align}
    I_1 & = A^{-r} \int_0^{A^{-1/\beta}} d\sigma
    \frac{\sigma^{\gamma}}{(\sigma^\beta + 1)^{2+q}} \asymp A^{-r}.
\end{align}
For $I_1$ and $A \ge 1$, the estimate follows from the inequality
$(\sigma^\beta + A)^{-2-q} \le A^{-2-q}$.

For $I_2$, we assume $r>0$.
For $A \le 1$, the integral converges if $A=0$, and the result follows.
For $A \ge 1$, using the same change of variables as above, we now obtain
\begin{align}
    I_2 & = A^{-r} \int_{A^{-1/\beta}}^\infty d\sigma
    \frac{\sigma^{\gamma}}{(\sigma^\beta + 1)^{2+q}} \asymp A^{-r},
\end{align}
since the integral converges at zero.
\end{proof}

\begin{proof}[Proof of Proposition~\ref{prop:Cbound-new}]
The fact that $C_j$ has range $\frac 12 L^j$ follows immediately from
\refeq{Cjfinran} and \refeq{CGamintz}.
The continuity in $m^2$ claimed for $C_{j;x,y}$ then follows
from \refeq{CGamintz} and the dominated convergence theorem.

Assuming \refeq{scaling-estimate-new},
we obtain \refeq{CNNbd-new} easily, as follows.
By definition,
\begin{align}
    C_{N,N;x,y}
    &=
    \sum_{z\in \Zd}\sum_{j=N}^\infty
     C_{j;x,y+zL^N} ,
\end{align}
Since we assume $m^2 \le \bar m^2$, the second term on the right-hand side
of \refeq{scaling-estimate-new} is dominated by the first term.
Therefore, by the finite-range property of $C_j$,
\begin{align}
    |\nabla^a C_{N,N;x,y}|
    &\prec
    \sum_{j=N}^\infty L^{d(j-N)} L^{-(j-1)(d-\alpha+|a|)} (1+ m^4 L^{2\alpha(j-1)})^{-1}
    \nnb & \le
    m^{-4} L^{-d (N-1)}
    \sum_{j=N}^\infty   L^{-(j-1)(\alpha+|a|)}
    \prec m^{-4}
    L^{-(N-1)(d+\alpha+|a|)} ,
\end{align}
as required.

The substantial part remains, which is to prove \refeq{scaling-estimate-new}
and \refeq{dCdm-new}.
We consider these together, with $q \in \{0,1\}$ denoting the number of $m^2$-derivatives.
Now we change notation,
and write $\rho^{(q)}$ for the
$q^{\rm th}$ derivative of $\rho$ with respect to $A$, for $q=0,1$.
To begin, consider the special term
\begin{equation}
    S_0 = \int_0^\infty \, ds \rho^{(q)}(s,A) \int_0^{\frac 12} \frac{dt}{t} w(t,x;s)
\end{equation}
that occurs in \refeq{Cjintegral} (or its $A$-derivative) for $j=1$.
According to \refeq{wbounds}--\refeq{rho4},
\begin{align}
    |S_0| & \prec
    \int_0^\infty \, ds \frac{s^{\beta}}{(s^\beta + A)^{2+q}} \frac{1}{1+s}
    \prec
    I_1(\beta,\beta,q,A)+I_2(\beta-1,\beta,q,A)
    .
\end{align}

The typical term in \refeq{Cjintegral} is
\begin{align}
    T_j= \int_0^\infty \, ds \rho^{(q)}(s,A) \int_{J_j} \frac{dt}{t} w(t,x;s)
    & =
    \int_0^\infty \, ds \rho^{(q)}(s,A) \int_{J_1} \frac{dt}{t} w(tL^{j-1},x;s).
\end{align}
We decompose the $s$-integral as $\int_0^\infty = \int_0^{L^{-2(j-1)}} +  \int_{L^{-2(j-1)}}^1
+\int_1^\infty$, and write this decomposition as
\begin{align}
    T_j= T_{j,1}+T_{j,2}+T_{j,3}.
\end{align}
Let
\begin{equation}
    A_j=AL^{\alpha(j-1)}=m^2L^{\alpha(j-1)}, \qquad z=2\beta(1+q)=\alpha(1+q).
\end{equation}

For $T_{j,1}$, we apply \refeq{wbounds}
(with $p=1$) and Lemma~\ref{lem:Ids} to see that
\begin{align}
    T_{j,1} & \prec
    L^{-(d-2+|a|)(j-1)}
    \int_0^{L^{-2(j-1)}}
    ds \rho^{(q)}(s,A) \int_{\frac 12}^{\frac 12 L}
    \frac{dt}{t}
    \frac{1}{(1+sL^{2(j-1)}t^2)^p} \frac{1}{t^{d-2+|a|}}
    \nnb & \le
    L^{-(d-2+|a|)(j-1)}
    \int_0^{L^{-2(j-1)}}
    ds \rho^{(q)}(s,A) I(d,sL^{2(j-1)})
    \nnb & \prec
    L^{-(d-z+|a|)(j-1)}
    \int_0^{1}
    d\sigma \frac{\sigma^\beta}{(\sigma^\beta + A_j)^{2+q}} I(d,\sigma)
    \nnb & \prec
    L^{-(d-z+|a|)(j-1)} I_1(\gamma_d,\beta,q,A_j)
    ,
\lbeq{T1bd}
\end{align}
where $\gamma_1=\beta - \half$, $\gamma_2 = \log$, $\gamma_d=\beta$ for $d>2$.

For $T_{j,2}$, we proceed as above, but do not put $p=1$, to obtain
\begin{align}
    T_{j,2} & \prec
    L^{-(d-z+|a|)(j-1)}
    \int_1^{L^{2(j-1)}}
    d\sigma \frac{\sigma^\beta}{(\sigma^\beta + A_j)^{2+q}}
    \int_{\frac 12}^{\frac 12 L}
    \frac{dt}{t}
    \frac{1}{(1+\sigma t^2)^p} \frac{1}{t^{d-2+|a|}}
    \nnb & \le
    L^{-(d-z+|a|)(j-1)}
    \int_1^{\infty}
    d\sigma \frac{\sigma^\beta}{(\sigma^\beta + A_j)^{2+q}}
    \frac{1}{\sigma^p}
    \int_{\frac 12}^{\infty}
    \frac{dt}{t}
    \frac{1}{t^{2p}} \frac{1}{t^{d-2+|a|}}
    \nnb & \prec
    L^{-(d-z+|a|)(j-1)}
    I_2(\beta-p,\beta,q,A_j)
    .
\lbeq{T2bd}
\end{align}

Finally, for arbitrary $p$ and for $p'=2p+z-2$, we use the last case of \refeq{wbounds}
to obtain
\begin{align}
    T_{j,3} & \prec
    L^{-(d-2+|a|)(j-1)}
    \int_1^{\infty}
    ds \rho^{(q)}(s,A) \int_{\frac 12}^{\frac 12 L}
    \frac{dt}{t}
    \frac{1}{s}\frac{1}{(L^{2(j-1)}t^2)^p} \frac{1}{t^{d-2+|a|}}
    \nnb & \prec
    L^{-(d-2+|a|)(j-1)} L^{-2p(j-1)}
    \int_1^{\infty}
    ds \rho^{(q)}(s,A)
    \frac{1}{s}
    \nnb & =
    L^{-(d-z+|a|)(j-1)} L^{-p'(j-1)}
    I_2(\beta-1,\beta,q,A)
    .
\lbeq{T3bd}
\end{align}

Thus, for $j=1$, since $I_2$ is increasing in its first argument $\gamma$,
whereas $I_1$ is decreasing, we obtain
\begin{align}
    \Big| \frac{\partial^q}{\partial A^q} \nabla^a C_{1;0,x}(A) \Big|
    & \prec
    S_0+T_{1,1}+T_{1,2}+T_{1,3}
    \nnb & \prec
    I_1(\beta,\beta,q,A)
    + I_1(\gamma_d,\beta,q,A)
    \nnb & \quad
    + I_2(\beta-p,\beta,q,A)+I_2(\beta-1,\beta,q,A)
    \nnb & \prec
    I_1(\gamma_d,\beta,q,A)
    +I_2(\beta-1,\beta,q,A).
\end{align}
For the $I_2$ term, we have $r=2+q-(\beta-1+1)/\beta = 1+q>0$.
For the $I_1$ term, we have
\begin{align}
\lbeq{rddef}
    r_d & = \begin{cases}
    2+q - (\beta + \frac 12)/\beta
    = 1 - \frac{1}{2\beta} + q
    & (d=1)
    \\
    2+q - (\beta+1)/\beta  =  1 - \frac{1}{\beta} + q & (d=2,3),
    \end{cases}
\end{align}
so $r_1 <0$ if $q=0$ and $r_1>0$ if $q=1$ (for $\half < \alpha < 1$),
and the same inequalities hold for $d\ge 2$ (for $1 < \alpha < 2$).
Therefore, with the abbreviation $\lambda_d= \lambda_d(A)$ defined by
$\lambda_1=\lambda_3=1$ and $\lambda_2= \log A^{-1}$,
\begin{align}
\lbeq{badterm}
    \Big| \frac{\partial^q}{\partial A^q} \nabla^a C_{1;0x}  \Big|
    & \prec
    \frac{1}{1+A^{2+q}} +
    \begin{cases}
    \frac{1}{1+A^{1+q}} & (q=0)
    \\
    \frac{\lambda_d(A)}{A^{r_d}}\1_{A \le 1} & (q=1).
    \end{cases}
\end{align}
The relevance of the second term when $q=1$ is its divergence as $A \downarrow 0$.
This proves the $j=1$ case of \refeq{scaling-estimate-new} and \refeq{dCdm-new}.

For $j \ge 2$, we have instead (with freedom to choose $p$ and hence $p'$ large)
\begin{align}
\lbeq{Cj2}
    \Big| \frac{\partial^q}{\partial A^q} \nabla^a C_{j;0,x}  \Big|
    & \prec
    L^{-(d-z +|a|)(j-1)}
    \Big(
    I_1(\gamma_d,\beta,q,A_j)
    +
    I_2(\beta-p,\beta,q,A_j)
    \nnb & \qquad\qquad\qquad\qquad\qquad +
    L^{-p'(j-1)} I_2(\beta-1,\beta,q,A)
    \Big)
    .
\end{align}
By Lemma~\ref{lem:I12}, with $r_d$ given by \refeq{rddef},
\begin{align}
    I_1(\gamma_d,\beta,q,A_j)
    &
    \prec
    \frac{1}{1+A_j^{2+q}} +
    A_j^{-r_d} \lambda_d(A_j) \1_{A_j \le 1} \1_{q=1},
    \\
    I_2(\beta-p,\beta,q,A_j)
    &
    \prec
    \frac{1}{1+A_j^{1+q+(p-1)/\beta}} .
\end{align}
For $q=0$ this simplifies to
\begin{align}
\lbeq{Cjq0}
    \left|  \nabla^a C_{j;0,x}  \right|
    & \prec
    L^{-(d-2\beta +|a|)(j-1)}
    \Big(
    \frac{1}{1+A_j^2}
    +
    \frac{1}{1+A L^{p'(j-1)}}
    \Big)
    ,
\end{align}
which proves \refeq{scaling-estimate-new} for $j \ge 2$.
For $q=1$, we consider only $d=1,2,3$.  We have $z=4\beta=2\alpha$
and $d-2\alpha = -\epsilon$, and
\begin{align}
\lbeq{Cjqnot0}
    \Big| \frac{\partial}{\partial A} \nabla^a C_{j;0,x}  \Big|
    & \prec
    L^{(\epsilon -|a|)(j-1)}
    \Big(
    \frac{1}{1+A_j^3}
    +
    \frac{1}{1+A^2 L^{p'(j-1)}}
    +
    \frac{\lambda_d(A_j)}{A_j^{r_d}} \1_{A_j \le 1}
    \Big)
    .
\end{align}
For $A_j \le 1$, the above gives
\begin{align}
\lbeq{Cjq0xx}
    \Big| \frac{\partial}{\partial A} \nabla^a C_{j;0,x}  \Big|
    & \prec
    L^{(\epsilon -|a|)(j-1)}
    \frac{\lambda_d(A_j)}{A_j^{r_d}}
    =
    L^{(\epsilon -|a|)(j-1)}
    \begin{cases}
    A_j^{-(2-1/\alpha)} & (d=1)
    \\
    A_j^{-(2-2/\alpha)}\log A^{-1} & (d=2)
     \\
    A_j^{-(2-2/\alpha)} & (d=3).
    \end{cases}
\end{align}
This proves \refeq{dCdm-new} for $j \ge 2$, and completes the proof.
\end{proof}

\subsection{Proof of Lemma~\ref{lem:wlims}}
\label{sec:Greekpfs}

\begin{proof}[Proof of Lemma~\ref{lem:wlims}]
The claimed continuity in $m^2$
is a consequence of the definitions together with the continuity of $C_j$ given
by Proposition~\ref{prop:Cbound}.
Thus it suffices to prove the estimates.
The proof is based on the proof of
\cite[Lemma~6.2]{BBS-rg-pt}.
Due to the assumption that $m^2 \le \bar m^2$, the last term on the right-hand side
of \refeq{scaling-estimate-new} can be ignored since it can be dominated by the first term.

\smallskip \noindent \emph{Note:} in this proof, constants implied by $\prec$
may depend on $L$, except in \refeq{etaMj}.

\smallskip\noindent
\emph{Bound on $\eta_j, \eta_{\ge j}$.}
It follows immediately from the definitions in \refeq{etadef}
and \refeq{Greeknoprime}, together with the bound \refeq{scaling-estimate}
on the covariance, that
\begin{equation}
\lbeq{etaMj}
  \eta_j = (n+2)L^{(d-\alpha)j} C_{j+1;0,0}
  \prec M_j ,
\end{equation}
with a constant that is \emph{independent} of $L$.  The desired bound on $\eta_{\ge j}$ then follows
as well.

\smallskip\noindent
\emph{Bound on $w_j^{(1)}$.}
By definition, and by \refeq{Mjbd},
\begin{equation}
  |w_j^{(1)}|
  \le \sum_x \sum_{k=1}^j |C_{k;0,x}|
  \prec \sum_{k=1}^j   L^{dk}M_kL^{-(d-\alpha)k}
  \prec
  L^{\alpha (j\wedge j_m)}.
\end{equation}

\smallskip\noindent
\emph{Bound on $\beta_j,\beta_j^:$.}
By definition, $\beta_j'$ is proportional to
\begin{equation}
\lbeq{w2diff}
  w^{(2)}_+ - w^{(2)} = 2(wC)^{(1)}+C^{(2)}
  \le
  2 \sum_x C_{j+1;0,x} \sum_{k=1}^{j+1} C_{k;0,x}.
\end{equation}
Therefore, using the finite range of $C_k$,
\begin{align}
    \beta_j' & \prec
     M_j L^{-(d-\alpha)j}  \sum_{k=0}^{j} L^{dk}M_k L^{-(d-\alpha)k}
  \prec
  M_j L^{-(d-\alpha)j} L^{\alpha(j\wedge j_m)}
  \le M_j L^{\epsilon (j\wedge j_m)}.
\end{align}
This proves \refeq{betabd} for $\beta_j=L^{-\epsilon (j\wedge j_m)}\beta_j'$.
The bound on $\beta_j^:$ then follows from the bounds on $\eta_{\ge j}, \bar{w}_j^{(1)}$.

For \refeq{dbetam}, we restrict to $m^2L^{\alpha j} \in (0,1]$.
Let $r_1=2-1/\alpha$ and $r_2=r_3=2-2/\alpha$.
We differentiate the middle member of \refeq{w2diff} using the product rule, and apply
\refeq{dCdm-new} and $M_j \le 1$.  For $d=1,3$, this gives
\begin{equation}
\begin{aligned}
    \Big| \frac{\partial \beta_j}{\partial m^2} \Big|
    & \prec
    \frac{1}{(m^2 L^{\alpha j})^r}\sum_{k=0}^j L^{dk}L^{-(d-\alpha)k}
    +
    L^{-\epsilon j}L^{-(d-\alpha)j}
    \sum_{k=0}^{j} L^{dk} L^{\epsilon k}\frac{1}{(m^2 L^{\alpha k})^r}
    \\ &
    \prec
    L^{ \alpha j}\frac{1}{(m^2 L^{\alpha j})^r},
\end{aligned}
\end{equation}
as stated in \refeq{dbetam}.  For $d=2$, there is an additional logarithmic factor
due to the logarithmic factor in \refeq{dCdm-new}.

\smallskip\noindent
\emph{Bound on $\xi_j$.}
The third term in the formula for $\xi$ in \refeq{xidef} is a multiple of
\begin{equation}
    \beta_j'\eta_j' =  L^{-(\alpha-2\epsilon)j}\beta_j\eta_j \prec M_j^2 L^{-(\alpha-2\epsilon)j}.
\end{equation}
The remaining terms in \refeq{xidef} are proportional to
\begin{equation}
\begin{aligned}
\label{e:delta-w3-sum}
    &\big(w_{j+1}^{(3)} - w_{j}^{(3)}\big)  - 3w_j^{(2)}C_{j+1;0,0}
    \\
    \quad\quad\quad&=
    3 \left((w_{j}^{2}C_{j+1})^{(1)}  - w_j^{(2)}C_{j+1;0,0}\right) +
    3 (w_{j}C_{j+1}^{2})^{(1)} +
    C_{j+1}^{(3)}
    .
\end{aligned}
\end{equation}
We use $\epsilon=2\alpha-d$ to obtain
\begin{align}
  |C_{j+1}^{(3)}|
  \le
  \sum_{y} |C_{j+1;0,x}^3|
  \prec
  M_j^3 L^{dj} L^{-3(d-\alpha)j}
  \prec
  M_j^3 L^{-(\alpha-2\epsilon)j} .
\end{align}
Similarly,
\begin{align}
  |(w_{j}C_{j+1}^{2})^{(1)}|
  &\prec
  M_j^2 L^{-2(d-\alpha)j}
  \sum_{k=1}^j L^{dk} L^{-(d-\alpha)k}
  \prec
  M_j^2  L^{-(\alpha-2\epsilon)j}
  .
\end{align}
Finally, we write
$ w_{j,x}^{2} = \sum_{k=0}^{j-1} \delta_k[w_{x}^{2}]$
with $\delta_k[w_x^2] = w_{k+1,x}^2 - w_{k,x}^2$, so that
\begin{align}
  \big(w_{j}^{2} (C_{j+1}-C_{j+1;0,0}) \big)^{(1)}
  =
  \sum_{k=0}^{j-1}
  \sum_{x}
  \delta_k[w_{x}^{2}]
  (C_{j+1;0,x} - C_{j+1;0,0}).
\end{align}
The identity (which follows from $w_{-x}^2=w_x^2$)
\begin{equation}
  \sum_{x} \sum_{i=1}^d \delta_k[w_x^{2}] x_i (\nabla^{e_i} C)_0
  = - \sum_x \delta_k[w_x^{2}] x_i (\nabla^{e_i} C)_0 = 0,
\end{equation}
and the bounds
\begin{align}
  \big|C_{j+1;0,x}-C_{j+1;0,0}-\sum_{i=1}^d x_i (\nabla^{e_i}C)_0\big|
  &
  \prec
  |x|^2 \|\nabla^2 C_{j+1}\|_\infty
  \nnb
  &\prec
  |x|^2    M_j L^{-(d-\alpha)j}L^{-2j}
  ,
\end{align}
\begin{equation}
  \sum_x \delta_k[w_{x}^{2}]|x|^2
  \prec
  L^{2k} \sum_x \delta_k[w_x^{2}]
  \prec
  L^{2k}\beta_k'
  \prec
  L^{2k}L^{\epsilon k}
  ,
\end{equation}
then imply that
\begin{equation}
  \Big|\big(w_{j}^{2} (C_{j+1}-C_{j+1;0,0}) \big)^{(1)}\Big|
  \prec
  M_j L^{-(d-\alpha)j}L^{-2j}
  \sum_{k=0}^{j-1}
  L^{(2+\epsilon)k}
  \prec
  M_j L^{-(\alpha-2\epsilon)j} .
\end{equation}
This gives the desired bound on $\xi_j$.

\smallskip \noindent \emph{Bound on $\newxi$.}
This follows from the definition in \refeq{newxidef} together with the estimates
obtained above for $\xi_j,\eta_{\ge j}, \beta_j$.

\smallskip \noindent \emph{Bound on the $\kappa$'s.}
The bounds for $\kappa_{g}$ and $\kappa_{\nu}$  follow from
  the above estimates.

\smallskip \noindent \emph{Bound on $\kappa_{gg}$.}
  It suffices to prove that
  \begin{equation} \label{e:ug2bd}
      \delta[w^{(4)}] - 4 C_{0,0}w^{(3)}
      - 6 C_{0,0}^2 w^{(2)}
      \prec
      M_jL^{-(d-2\epsilon)j}.
  \end{equation}
    The left-hand side of \eqref{e:ug2bd} is equal to
  \begin{equation} \label{e:ug2bd2}
    4 \sum_x w_x^3 (C_{0,x}-C_{0,0}
    )
    + 6 \sum_x w_x^2 (C_{0,x}^2-C_{0,0}^2) + 4\sum_x w_x C_{0,x}^3 + \sum_x C_{0,x}^4.
  \end{equation}
  By discrete Taylor approximation (and symmetry), in the first term we can replace
  $C_{0,x}-C_{0,0}$ by $O(|x|^2 \|\nabla^2 C\|_\infty)$.  Therefore,
  \begin{align}
  \lbeq{discTay}
  \sum_x w_x^3 |C_{0,x}-C_{0,0}
  |
  &\prec \sum_x w_x^3 |x|^2 \|\nabla^2 C\|_\infty
  \nnb
  &\prec
  M_jL^{-(2+d-\alpha)j} \sum_{j \geq i\geq l\geq m} \sum_{x} C_{i;0,x}C_{l;0,x}C_{m;0,x} |x|^2
  \nnb
  &\prec
  M_jL^{-(2+d-\alpha)j}
  \sum_{j \geq i \geq l} L^{-(d-\alpha)i} L^{-(d-\alpha)l} L^{(2+\alpha)l}
  \nnb
  &\prec
  M_jL^{-(2+d-\alpha)j} \sum_{j \geq i} L^{(2-2d+3\alpha)i}
  \prec
  M_jL^{-(d-2\epsilon)j}.
\end{align}
Similarly,
\begin{align}
  \sum_x w_x^2 |C_{0,x}^2-C_{0,0}^2|
  &\prec
  M_j L^{-(2+2d-2\alpha)j} \sum_x w_x^2 |x|^2
  \nnb &
  \prec
  M_jL^{-(2+2d-2\alpha)j} L^{(2-d+2\alpha)j}
  \prec
  M_jL^{-(d-2\epsilon)j}.
\end{align}
Up to a  factor,
the last two terms in \eqref{e:ug2bd2} are bounded by
$M_jL^{\alpha j} L^{-3(d-\alpha)j}= M_jL^{-(d-2\epsilon)j}$
and $M_jL^{dj}L^{-4(d-\alpha)j} = M_jL^{-(d-2\epsilon)j}$, as claimed.

\smallskip \noindent \emph{Bound on $\kappa_{\nu\nu}$.}
By definition, and arguing as above, we see that $|\kappa_{\nu\nu}'|$ is proportional to
\begin{align}
    |\delta [w^{(2)}] - 2Cw^{(1)}|
    &=
    \Big| 2\sum_x w_x (C_{0,x} - C_{0,0})  +\sum_x C_{0,x}^2
    \Big|
    \nnb & \prec
    M_j L^{-(2+d-\alpha)j} \sum_{k=1}^j \sum_x |x|^2 M_k L^{-(d-\alpha)k}
    +L^{dj} M_j^2 L^{-2(d-\alpha)j}
    \nnb &
    \prec
    M_j L^{-(2+d-\alpha)j} \sum_{k=1}^j    L^{(2+\alpha)k} L^{-2\alpha(k-j_m)_+}
    + M_j L^{-dj}  L^{2\alpha(j\wedge j_m)}
    \nnb &
    \prec
    M_j L^{-dj}L^{2\alpha (j\wedge j_m)}   ,
\end{align}
as required.
For the case $j >j_m$, we used
\begin{align}
    \sum_{k=1}^j    L^{(2+\alpha)k} L^{-2\alpha(k-j_m)_+}
    &\prec
    L^{(2+\alpha)j_m}
    +
    L^{(2+\alpha)j_m}
    \sum_{k={j_m+1}}^j L^{(2-\alpha)(k-j_m)}
    \nnb &
    \prec
    L^{(2+\alpha)j_m} L^{(2-\alpha)(j-j_m)}
    =
    L^{(2-\alpha)j}L^{2\alpha j_m}.
\end{align}

\smallskip \noindent \emph{Bound on $\kappa_{g\nu}$.}
This follows from a combination of the bounds on $C$ and $\kappa_{\nu\nu}$, and
completes the proof.
\end{proof}

\subsection{Self-similarity of the covariance decomposition}

This section
concerns the asymptotic self-similarity of the covariance decomposition.
We recall that
there is a function $\bar w$ such that $w$ of \refeq{Gamw} obeys
\begin{equation}
\lbeq{wbarw}
    w(t,x;s) = (c/t)^{d-2} \bar w(cx/t;st^2) + O(t^{-(d-1)}(1+st^2)^{-p}),
\end{equation}
with the error estimate valid for any $p\ge 0$ and uniform in bounded $s$,
and in particular for $s \le 1$
(see \cite[(1.37)--(1.38)]{Baue13a}; $w$ is called $\phi^*$ in \cite{BBS-rg-pt,Baue13a},
and $\bar w$ is $\bar \phi$ of \cite[(3.17)]{Baue13a}).
For any $p \ge 0$, the function $\bar w$ obeys (by \cite[(1.34), (1.38)]{Baue13a})
\begin{equation}
\lbeq{wbarbd}
    \bar w(cx/t;st^2) \le O(1+st^2)^{-p}.
\end{equation}
It is shown in \cite{Baue13a} that
\begin{equation}
\lbeq{barwvarphi}
    \bar w(y,m^2) = \int_{\Rd} \varphi(\sqrt{|\xi|^2+m^2}) e^{iy\cdot \xi}d\xi,
\end{equation}
where $\varphi$ is a nonnegative function.

We define a smooth function $c_0: \R^d \times [0,\infty)
\to \R$, with compact support in $\R^d$,  by
\begin{align}
\lbeq{c0wm}
    c_0(x,m^2) & =
    \int_0^\infty d\sigma \rho(\sigma,m^2)
    \int_{\frac 12 L^{-1}}^{\frac 12}
    \frac{d\tau}{\tau}
    (c/\tau)^{d-2} \bar w (cx/\tau ,\sigma \tau^2)
    .
\end{align}
By \refeq{c0wm} and \refeq{barwvarphi}, the spatial Fourier transform of $c_0(\cdot ,m^2)$ is
\begin{align}
\lbeq{c0hat}
    \hat c_0(\xi,m^2) & =
    \int_0^\infty d\sigma \rho(\sigma,m^2)
    \int_{\frac 12 L^{-1}}^{\frac 12}d\tau
    \frac{\tau}{c^2}
    \varphi(\tau\sqrt{c^{-2}|\xi|^2+\sigma})
    .
\end{align}
Consequently,
$\hat c_0(\xi,m^2)$ is nonnegative.
The following lemma is a version of \cite[(1.14)]{Mitt16}.

\begin{lemma}
\label{lem:contlim}
Let $d\ge 1$, $\alpha \in (0, 2\wedge d)$, and $m^2 \in [0,\bar m^2]$.
As $j \to \infty$,
\begin{equation}
\lbeq{c0-lem}
    C_{j;0,x}(m^2) = L^{-(d-\alpha)j}
    \Big( c_0(L^{-j}x,m^2L^{\alpha j})
    + O( L^{-j}) \Big),
\end{equation}
with the constant in the
error estimate uniform in $x \in \Zd$, but possibly $\bar m^2$-
and $L$-dependent.
\end{lemma}

\begin{proof}
Since we are interested in the limit $j \to \infty$, we assume that $j \ge 2$
to avoid the special case of $C_1$.
By \refeq{Cjintegral},
\begin{align}
\lbeq{contlim0}
    C_{j;0,x}(m^2) & =
     \int_0^\infty ds \rho(s,m^2)
     \int_{J_j}
     \frac{dt}{t}
     w(t,x;s).
\end{align}
The analysis of $T_{j,3}$ in the proof of Proposition~\ref{prop:Cbound-new} shows
that the contribution to the $s$-integral from $s \ge 1$ can be absorbed into the error
term in \refeq{c0-lem}, and similar estimates show that the same is true for the contribution to the
right-hand side of \refeq{c0-lem} from the portion of the integral \refeq{c0wm} due to
$s \ge 1$.  The error estimate in \refeq{wbarw} is uniform in $s \le 1$, and
thus it suffices to prove that
\begin{equation}
    \int_0^1 ds \rho(s,m^2)
     \int_{J_j}
     \frac{dt}{t}
     \frac{1}{t^{d-1}}\frac{1}{(1+st^2)^p} = O(L^{-(d-\alpha-1)j}).
\end{equation}
This follows from estimates like those used previously, using $\rho(s,m^2) \prec s^{-\alpha/2}$.
\end{proof}

The next lemma is used in the proof of Lemma~\ref{lem:betadiff}.

\begin{lemma}
\label{lem:dc0bd}
Let $d \ge 1$, $\alpha \in (0,2 \wedge d)$.
There exists $z>0$ such that
for $x\in \Rd$ and $0 < A < A'$,
\begin{equation}
    |c_0(x,A)-c_0(x,A')|
    \prec
    \1_{A \le 1} (A')^{z} + \1_{A' > 1}A^{-2}.
\end{equation}
\end{lemma}

\begin{proof}
By \refeq{c0wm} and the Fundamental Theorem of Calculus,
\begin{align}
    c_0(x,A)-c_0(x,A') & =
    \int_0^\infty d\sigma
    \int_A^{A'} da
    \frac{\partial \rho(\sigma,a)}{\partial a}
    \int_{\frac 12 L^{-1}}^{\frac 12}
    \frac{d\tau}{\tau}
    (c/\tau)^{d-2} \bar w (cx/\tau ,\sigma \tau^2)
    .
\end{align}
By \refeq{drhobd} and \refeq{wbarbd}, with arbitrary $p \ge 0$ and with $\beta = \alpha/2$,
\begin{align}
    |c_0(x,A)-c_0(x,A')| & \prec
    \int_A^{A'} da
    \int_0^\infty d\sigma
    \frac{\sigma^\beta}{(\sigma^\beta + a)^3}
    \int_{\frac 12 L^{-1}}^{\frac 12}
    \frac{d\tau}{\tau}
    \frac{1}{\tau^{d-2}} \frac{1}{(1+\sigma \tau^2)^p}
    .
\end{align}
We decompose the $\sigma$-integral as $\int_0^\infty = \int_0^1 + \int_1^\infty$.
This leads to
\begin{align}
    |c_0(x,A)-c_0(x,A')| & \prec
    \int_A^{A'} da
    \int_0^1 d\sigma
    \frac{\sigma^\beta}{(\sigma^\beta + a)^3}
    +
    \int_A^{A'} da
    \int_1^\infty d\sigma
    \frac{\sigma^\beta}{(\sigma^\beta + a)^3}\frac{1}{\sigma^p}
    \nnb & =
    \int_A^{A'} da \Big( I_1(\beta,\beta,1,a) + I_2(\beta-p,\beta,1,a)\Big)
    .
\end{align}
By Lemma~\ref{lem:I12}, the integrand on the right-hand side is bounded by a multiple
of $(\1_{a \le 1}a^{-(2-1/\beta)} + \1_{a \ge 1}a^{-3}) + (\1_{a \le 1}+\1_{a \ge 1} a^{-p'})$,
with $p'$ as large as desired.  The contribution from $I_1$ is dominant for both large and
small $a$, and it is bounded by $\1_{A \le 1}a^{-(2-1/\beta)} + \1_{A' > 1}a^{-3}$.
Integration of this upper bound then gives the desired result, with $z=-1+1/\beta>0$.
\end{proof}

\subsection{Proof of Lemmas~\ref{lem:beta-a0} and \ref{lem:betadiff}}
\label{sec:betapfs}

\begin{proof}[Proof of Lemma~\ref{lem:beta-a0}]
We adapt the proof of
\cite[Lemma~6.3(a)]{BBS-rg-pt}.
Let $m^2=0$. For $F,G:\Zd \to \R$, we write $(F,G) = \sum_{x\in\Zd} F_xG_x$.
By definition,
  \begin{align}
    \lbeq{betajm0}
    \beta_j & = (8+n)L^{-\epsilon j}(w_{j+1}^{(2)} -w_j^{(2)})
    = (8+n)L^{-\epsilon j}\left( (C_{j+1}, C_{j+1}) + 2(w_j, C_{j+1})  \right)
    \nnb
    & =
    (8+n)L^{-\epsilon j}\left( (C_{j+1}, C_{j+1}) + 2\sum_{k=1}^k(C_k, C_{j+1})  \right)
    .
  \end{align}
With $c_0(x,m^2)$ from \refeq{c0wm}, let
$c_0(x)=c_0(x,0)$.
Let $c_k(x) = L^{-{(d-\alpha)k}}c_0(L^{-k}x)$ and $p=d-\alpha + 1$.
By Lemma~\ref{lem:contlim},
  \begin{equation} \label{eq:app-C-approx}
    C_{k;0,x} = c_k(x) + O(L^{-pk}).
  \end{equation}
We write
  $\la f,g\ra = \int_{\R^d} fg \; dx$ for $f,g:\R^d\to\R$.

We claim that
  \begin{equation} \label{eq:Cc-approx}
    (C_{k}, C_{k+l}) = L^{\epsilon k} \la c_0, c_l \ra +
    O(L^{\epsilon k} L^{- k} L^{-(d-\alpha)l}).
  \end{equation}
To see this, let $R_{k,x} = C_{k;0,x} - c_{k}(x)$. Then
  \begin{equation}
    (C_{k}, C_{k+l})
    = (c_{k},c_{k+l})
    + (c_{k}, R_{k+l}) + (c_{k+l}, R_k) + (R_k, R_{k+l}).
  \end{equation}
  Riemann sum approximation gives
  \begin{align}
    (c_{k},c_{k+l}) - L^{\epsilon k} \la c_{0},c_{l}\ra
    &= L^{\epsilon k}
    \left(L^{-dk} \sum_{y\in L^{-k}\Z^d} c_{0}(y) c_l(y) - \int_{\R^d} c_{0}(y) c_l(y) \; dy \right)
    \nonumber\\
    &= L^{\epsilon k}O(L^{-k}) \|\nabla(c_{0} c_l)\|_{L^\infty}
    =
    O(L^{\epsilon k} L^{-k}L^{-(d-\alpha)l})
    .
  \end{align}
  For the remaining terms, we use the fact that the supports of
  $C_k$ and $R_k$ are $O(L^{dk})$ to see that
  \begin{alignat}{2}
    (c_{k},R_{k+l})
    &\leq O(L^{dk}) \|c_{k}\|_{L^\infty(\Rd)} \|R_{k+l}\|_{L^\infty(\Zd)}
    &&\leq O(L^{ \epsilon k}L^{-k}  L^{-pl}),
    \\
    (c_{k+l},R_{k})
    &\leq O(L^{dk}) \|c_{k+l}\|_{L^\infty(\Rd)} \|R_{k}\|_{L^\infty(\Zd)}
    &&\leq O(L^{ \epsilon k} L^{- k } L^{-(d-\alpha)l} ),
    \\
    (R_k, R_{k+l})
    &\leq O(L^{dk}) \|R_k\|_{L^\infty(\Zd)}\|R_{k+l}\|_{L^\infty(\Zd)}
    &&\leq O(L^{ \epsilon k}L^{-2 k}L^{-pl}),
  \end{alignat}
  and \eqref{eq:Cc-approx} follows.

  From \eqref{eq:Cc-approx}, we obtain
  \begin{align} \label{eq:w2-v2-diff}
    \sum_{k=1}^{j} (C_k,C_{j+1})
    &= \sum_{k=1}^{j}
     L^{\epsilon k}
     \la c_0,c_{j+1-k} \ra + \sum_{k=1}^{j}
      L^{\epsilon k} O(L^{-k -(d-\alpha)(j-k)})
    \nonumber\\ &
    = L^{\epsilon (j+1)}
    \left( \sum_{k=1}^{j} L^{-\epsilon k} \la c_0,c_{k} \ra +  O(L^{-(\alpha\wedge 1)j}) \right)
    ,
    \\
    (C_{j+1},C_{j+1})
    &= L^{\epsilon (j+1)}\left( \la c_0,c_{0} \ra + O(L^{-pj}) \right)
    .
  \end{align}
  With \refeq{betajm0}, this gives
  \begin{align}
    \beta_j
    &
    = (8+n)L^{\epsilon}
    \left( \la c_0,c_0\ra + 2\sum_{k = 1}^j L^{-\epsilon k} \la c_0, c_k \ra
    + O(L^{-(\alpha\wedge 1)j})
    \right)
    .
  \end{align}
  Since $c_0$ has support of order $1$, $|\la c_0, c_k \ra| \le O(L^{-(d-\alpha)k})$,
  and hence
  \begin{equation}
  \lbeq{c0ck}
    \sum_{k = j}^\infty L^{-\epsilon k} |\la c_0, c_k \ra|
    \le
     \sum_{k = j}^\infty O(L^{-\alpha k})= O(L^{-\alpha j}).
  \end{equation}
  Thus we have obtained
  \begin{equation}
    \lbeq{betajinf}
    \beta_j =
    a
    + O(L^{-(\alpha \wedge 1)j}),
  \end{equation}
with
\begin{equation}
\lbeq{adef}
    a = L^{\epsilon}
    (8+n)\left( \la c_0,c_0\ra + 2\sum_{k=1}^\infty L^{-\epsilon k} \la c_0, c_k \ra
    \right).
\end{equation}

By \refeq{c0ck}, the sum in \refeq{adef} converges.
Also, it follows from the Parseval equality, together with the nonnegativity of the
Fourier transform $\hat c_0$, that each inner product on the right-hand side of
\refeq{adef} is nonnegative, with the first term strictly positive.
\end{proof}

\begin{proof}[Proof of Lemma~\ref{lem:betadiff}]
Let $j \le j_m$.  By \refeq{betaW}, the triangle inequality, and \refeq{betabd}
\begin{align}
    |\beta^:_j - \beta_j|
    & \prec
    |\eta_{\ge j}- \eta_{\ge j+1}|\, | \bar w^{(1)}_j|
    +
    |\eta_{\ge j+1}|\, |\bar w^{(1)}_j-\bar{w}_{j+1}^{(1)}|
    \nnb
    & \prec
    |\eta_{\ge j}- \eta_{\ge j+1}|
    +
      |\bar w^{(1)}_j-\bar{w}_{j+1}^{(1)}|
    .
\end{align}
It suffices to prove that there exists $z>0$ such that,
uniformly in $m^2 \in [0,\bar m^2]$ and $j \le j_m$,
\begin{align}
\lbeq{etawsuff}
    | \eta_{\ge j} - \eta_{\ge j+1}| & \prec L^{-zj}+L^{-z(j_m-j)},
    \qquad
    | \bar w^{(1)}_j - \bar w^{(1)}_{j+1}|  \prec L^{-zj}+L^{-z(j_m-j)}.
\end{align}

We write $f_j=\frac{1}{n+2}(\eta_{\ge j}- \eta_{\ge j+1})$.
By definition of $\eta_{\ge j}$
in \refeq{etagedef},
\begin{align}
    f_j
    &=
    L^{(d-\alpha)j} \sum_{i=j}^\infty C_{i+1:0,0}(m^2)
    -
    L^{(d-\alpha)(j+1)} \sum_{i=j+1}^\infty C_{i+1:0,0}(m^2)
    \nnb & =
    L^{(d-\alpha)j} \sum_{i=j}^\infty \left( C_{i+1:0,0}(m^2) - L^{d-\alpha}C_{i+2;0,0}\right)
    .
\end{align}
Let $q_i=c_0(0,m^2L^{\alpha i})$.
By Lemma~\ref{lem:contlim},
\begin{align}
    f_j
    &=
    L^{(d-\alpha)j} \sum_{i=j}^\infty
    L^{-(d-\alpha)(i+1)} ( q_i - q_{i+1})
    +O(L^{-j}).
\end{align}
By definition of the mass scale $j_m$, $m^2L^{\alpha i} \asymp L^{\alpha(i-j_m)}$.
By Lemma~\ref{lem:dc0bd},
\begin{equation}
    |q_i - q_{i+1}|
    \prec
    \1_{i \le j_m} L^{-z\alpha (j_m-i) }
    +
    \1_{i +1 > j_m} L^{-2 \alpha (i-j_m) }.
\end{equation}
Therefore, with $z$ reduced if necessary to ensure that $z\alpha < d-\alpha$, say $z \le \frac 12$,
\begin{align}
    |f_j| & \prec
    L^{(d-\alpha)j}\sum_{i=j}^{j_m}
    L^{-(d-\alpha)i} L^{-z\alpha (j_m-i)}
    +
    L^{(d-\alpha)j}
    \sum_{i=j_m}^\infty L^{-(d-\alpha)i} L^{-2 \alpha (i-j_m) }
    +L^{-j}
    \nnb & =
    L^{-z\alpha (j_m-j)} \sum_{i=j}^{j_m}
    L^{-(d-\alpha-z\alpha)(i-j)}
    +
     L^{-(d-\alpha)(j_m-j)}
    \sum_{i=j_m}^\infty L^{-(d+\alpha)(i-j_m)}
    +
    L^{-j}
    \nnb & \prec
    L^{-z\alpha (j_m-j)}+L^{-j}.
\end{align}
This proves the first estimate of \refeq{etawsuff}, after a redefinition of $z$.

By definition,
\begin{align}
\lbeq{wdiff}
    \bar w^{(1)}_{j+1} - \bar w^{(1)}_{j}
    & =
    L^{-\alpha(j+1)}\sum_{i=1}^{j+1} C_i^{(1)} - L^{-\alpha j}\sum_{i=1}^{j} C_i^{(1)}
    .
\end{align}
Let $Q_i = \int_{\Rd}c_0(y,m^2L^{\alpha i})dy$.
By Lemma~\ref{lem:contlim}, and by Riemann sum approximation,
\begin{align}
\lbeq{C1Q}
    C_i^{(1)}
    &=
    L^{-(d-\alpha)i}\Big( \sum_{x} c_0(xL^{-i},m^2L^{\alpha i}) +L^{di}O(L^{-i}) \Big)
    = L^{\alpha i} (Q_i + O(L^{-i})).
\end{align}
Then, for some $z'>0$,
\begin{align}
    |\bar w^{(1)}_{j+1} - \bar w^{(1)}_{j}|
    & \prec
    \Big|
    L^{-\alpha j}\sum_{i=1}^{j+1} L^{\alpha (i-1)} Q_i
    - L^{-\alpha j}\sum_{i=1}^{j} L^{\alpha i} Q_i \Big|
    +L^{-z'j}
    \nnb
    & =
    \Big|
    L^{-\alpha j}\sum_{i=0}^{j} L^{\alpha i} Q_{i+1}
    - L^{-\alpha j}\sum_{i=1}^{j} L^{\alpha i} Q_i
    \Big|
    +L^{-z'j}
    \nnb
    & \le
    L^{-\alpha j}|Q_1 |+
    L^{-\alpha j}\sum_{i=1}^{j} L^{\alpha i} |Q_{i+1}  -  Q_i|
    +L^{-z'j}
    \nnb
    & \prec
    L^{-\alpha j}\sum_{i=1}^{j} L^{\alpha i} |Q_{i+1}  -  Q_i|
    +L^{-z'j}
    .
\end{align}
Since $m^2L^{\alpha i} \asymp L^{-\alpha(j_m-i)}$, it follows from Lemma~\ref{lem:dc0bd} that
\begin{align}
    |Q_{i+1}-Q_i| \prec L^{-z\alpha(j_m-i)}.
\end{align}
Therefore,
\begin{align}
    |\bar w^{(1)}_{j+1} - \bar w^{(1)}_{j}|
    & \prec
    L^{-\alpha j}\sum_{i=1}^{j} L^{\alpha i} L^{-z\alpha(j_m-i)}
    +L^{-z'j}
    \nnb & =
    L^{-z\alpha(j_m-j)} \sum_{i=1}^{j} L^{-(\alpha +z\alpha)(j-i)} +L^{-z'j}
    \prec
    L^{-z\alpha(j_m-j)}  +L^{-z'j}
    .
\end{align}
This gives the second estimate of \refeq{etawsuff}, and the proof is complete.
\end{proof}

\section{Supersymmetry and \texorpdfstring{$n=0$}{n=0}}
\label{sec:saw}

In this section, we indicate how the weakly self-avoiding walk can be
represented as a supersymmetric field theory.  It is this representation
that leads to an interpretation as the $n=0$ case.
Nothing in this section is used in our analysis for $n \ge 1$.

\subsection{Infinite volume limit}

With $E^{N}$ the expectation for the Markov Chain on the torus with generator
$-(-\Delta_{\Lambda_N})^{\alpha/2}$, as in Section~\ref{sec:rwtorus}, let
  $c_{N,T} = E^{N}_0(e^{-gI_T})$.
The torus susceptibility is
\begin{equation}
  \chi_N(\nu)
  = \int_0^\infty c_{N,T} e^{-\nu T} \; dT
  .
\end{equation}
By the Cauchy--Schwarz inequality,
$T= \sum_{x\in \Lambda_N} L_T^x \le (|\Lambda_N|I_T)^{1/2}$, and hence
\begin{equation}
  \chi_N(\nu)
  \leq
  \int_0^\infty e^{-gT^2/|\Lambda_N|} e^{-\nu T} \; dT < \infty
  \quad
  \text{for all $\nu \in \R$.}
\end{equation}
The following lemma, which is
an adaptation of
\cite[Lemma~2.1]{BBS-saw4-log},
shows that $\chi$ is the limit of $\chi_N$.

\begin{lemma} \label{lem:suscept-finvol}
  Let $d\ge 1$.
  For all $\nu \in \R$, $\chi_N(\nu)$ is non-decreasing
  in $N$, and  $\chi(\nu)=\lim_{N\to\infty}\chi_N(\nu)$
  (with $\chi(\nu)=\infty$ for $\nu \le \nu_c$).
  The functions
  $\chi_N$ and $\chi$ are analytic on $\{\nu \in \C :{\mathrm{Re}}\nu > \nu_c\}$,
  and $\chi_N$ and all its derivatives
  converge uniformly on compact subsets of ${\mathrm{Re}}\nu > \nu_c$ to $\chi$ and its derivatives.
\end{lemma}

\begin{proof}
  Let $c_T = E_0(e^{-gI_T})$.
  We will show that
  \begin{equation}
  \lbeq{ctmon}
    c_{N,T} \leq c_{N+1,T} \leq c_T
    ,\qquad
    \lim_{N \to \infty} c_{N,T} =c_T.
  \end{equation}
  The monotone convergence theorem then implies that
  \begin{equation}
    \chi(\nu)
    = \int_0^\infty \lim_{N\to\infty} c_{N,T} e^{-\nu T} \; dT
    = \lim_{N\to\infty} \chi_N(\nu)
    \quad  \text{for $\nu \in \R$}
  \end{equation}
  (both sides are finite if and only if $\nu > \nu_c$).
  Also,
  since
  $|c_{N,T} e^{-\nu T}| \le c_{N,T} e^{-({\rm Re}\nu) T} \le
  c_{T} e^{-({\rm Re}\nu) T}$,
  it follows from the
  dominated convergence theorem that
  \begin{equation}
  \chi(\nu) = \lim_{N\to\infty}\chi_N(\nu)
  \quad
  \text{for ${\rm Re} \nu > \nu_c$}.
  \end{equation}
  The analyticity of $\chi$ and $\chi_N$ follows from analyticity of
  Laplace transforms,
  and the desired compact convergence of $\chi_N$ and all its derivatives
  then follows from Montel's theorem. % \cite[(9.12.1),(9.13.1)]{Dieu69}.

  It remains to prove \refeq{ctmon}.
  Given a walk $X$ on $\Z^d$ starting at $0$, we denote by $X^N$ the corresponding
  walk on $\Lambda_N$, defined by the coupling discussed in Section~\ref{sec:rwtorus}.
  We denote the local time of a walk $X$ up to time $T$ by $L^x_T(X) =
  \int_0^T \1_{X(S)=x} \; dS$, and similarly the intersection local
  time by $I_T(X)$.
  Given $X$ and a positive integer $N$,
  \begin{equation}
  \begin{aligned}
    I_T(X^{N+1}) & = \sum_{x\in\Lambda_{N+1}} \left(L^{x}_T(X^{N+1}) \right)^2
    = \sum_{x\in\Lambda_N} \sum_{y\in\Z^d:  \|y\|_\infty < L}
    \left(L^{x+yL^N}_T(X^{N+1}) \right)^2
    \\ &
    \leq \sum_{x\in\Lambda_N}
    \left(\sum_{y\in\Z^d:  \|y\|_\infty < L} L^{x+yL^N}_T(X^{N+1}) \right)^2
    = \sum_{x\in\Lambda_N}
    \left(L^{x}_T(X^{N})\right)^2
    = I_T(X^N),
  \end{aligned}
  \end{equation}
  and hence
  \begin{equation} \label{e:EI-cond}
    e^{-gI_{T}(X^N)}   \leq e^{-gI_{T}(X^{N+1})}.
  \end{equation}
  Now we take the expectation over $X$ to obtain the first inequality
  of \refeq{ctmon}.
  This shows monotonicity in $N$ of $c_{N,T}$.
  Also, since $X^N$ can only have more intersections than $X$,
  we have $I_T(X^N) \ge I_T(X)$ for any walk $X$ on
  $\Zd$ and for any $N$.  This implies that $c_{N,T} \le c_{N+1,T} \le c_T$.

  Finally, for the convergence of $c_{T,N}$ to $c_T$,
  a crude estimate suffices.
  Walks which do not reach distance $\frac 12 L^N$ from the origin in time $T$
  do not contribute to the difference
  $c_{T,N}-c_T$.  Let $R_T^N$ denote the event that $X$ on $\Zd$
  reaches such a distance.  Then $|c_{T,N}-c_T| \le 2P(R_T^N)$.
  Let $F_n^N$ be the event that a walk $X$ on $\Zd$ reaches distance $\frac 12 L^N$
  within its first $n$ steps.
  Since the number of steps taken by time $T$ has a Poisson($2dT$) distribution,
  \begin{align}
    P(R_T^N) &= \sum_{n=0}^\infty e^{-2dT}\frac{(2dT)^n}{n!} P(F_n^N).
  \end{align}
  When $F_n^N$ occurs, at least one of the $n$ steps must extend over a distance of at least
  $r=\frac 1n \frac 12 L^N$.  By a union bound and Lemma~\ref{lem:fracLapdecay}, this has probability at most
  $kn r^{-\alpha}$ for some $k>0$.
  Therefore,
  \begin{align}
    P(R_T^N) &\le
    \sum_{n=0}^\infty e^{-2dT}\frac{(2dT)^n}{n!} kn \left(\frac{2n}{L^N} \right)^\alpha
    =
    a_T L^{-\alpha N} ,
  \end{align}
  where the constant is $a_T=k2^{\alpha} EY^{1+\alpha}$ with
  $Y \sim {\rm Poisson}(2dT)$.
  The upper bound goes to zero as $N \to\infty$, so
  $\lim_{N\to\infty}c_{N,T} = c_T$.  This proves the second item in \refeq{ctmon},
  and completes the proof.
\end{proof}

\subsection{Supersymmetric representation}
\label{sec:ir}

Although the result we need is contained in \cite[Proposition~2.7]{BIS09},
we present some details here to make our account more self-contained.
We follow the analysis of
\cite[Appendix~A]{ST-phi4}.
We apply the next lemma with $Q=-(-\Delta_\Lambda)^{\alpha/2}$.

\begin{lemma}
\label{lem:intrepLHS}
Let $X$ be a Markov chain on $\Lambda=\Lambda_N$ with generator $Q$,
local time $L_T^x$, and with expectation $E_x^N$ for the process started at $x \in \Lambda$.
Let $D$ be a complex diagonal matrix with entries $d_u$ with ${\rm Re}\,d_u > 0$
for all $u\in\Lambda$.  Then, for $x,y \in \Lambda$,
\begin{equation}
\lbeq{DJV}
    (-Q + D)^{-1}_{xy}
    =
    \int_{0}^\infty E_x^N \big[e^{-\sum_{u\in\Lambda} d_u L_T^u } \1_{X(T)= y}\big]  dT
    .
\end{equation}
\end{lemma}

\begin{proof}
Let $H$ denote the
diagonal part of $-Q$ (diagonal elements $h_u$),
and let $J=H+Q$ denote the off-diagonal part
of $Q$.  Both $H$ and $J$ have non-negative entries.
On the right-hand side of \refeq{DJV},
we regard $X$ as a discrete time random walk $Y$ with
independent $\operatorname{Exp}(h_{u})$ holding times $(\sigma_i)_{i\geq 0}$
and transition probabilities $h_{u}^{-1}J_{uv}$ (as discussed above \refeq{Pxy}).
We set $\gamma_j = \sum_{i = 0}^j \sigma_i$, write $\Wcal_{xy}^n$ for the set
of walks $x=x_0,x_1,\ldots,x_n=y$ with $x_i\in \Lambda$, and condition on
$Y\in \mathcal{W}_{xy}^n$ to obtain
\begin{equation}
\begin{aligned}
\int_{0}^\infty E_x^N &\left[e^{-\sum_{u\in\Lambda} d_u L_T^u} \1_{X(T) = y}\right] dT
\\
&= \sum_{n=0}^\infty \sum_{Y \in \mathcal{W}_{xy}^n} (H^{-1}J)^Y
E \left[e^{-\sum_{j=0}^{n-1} d_{Y_j} \sigma_j} \int_{\gamma_{n-1}}^{\gamma_n} e^{-d_{Y_n} (T - \gamma_{n-1}) } dT \right]
\\
&= \sum_{n=0}^\infty \sum_{Y \in \mathcal{W}_{xy}^n} (H^{-1}J)^Y
E \left[\left(e^{-\sum_{j=0}^{n-1} d_{Y_j} \sigma_j}\right) \frac{-1}{d_{Y_n}}
\left(e^{-d_{Y_n} \sigma_{n}} - 1\right)\right].
\end{aligned}
\end{equation}
Here $(H^{-1}J)^Y=  \prod_{j=0}^{n-1} (h_{Y_j}^{-1}J_{Y_{j-1}Y_j})$.
Since the $\sigma_i$ are i.i.d., the expectation factors
into a product of $n+1$ expectations that can each be evaluated explicitly,
with the result that
\begin{equation}
\lbeq{expfactors}
\begin{aligned}
\int_0^\infty E_x^N &\left[e^{-\sum_{u\in\Lambda} d_u L_T^u} \1_{X(T) = y}\right] dT \\
&= \sum_{n=0}^\infty \sum_{Y \in \mathcal{W}_{xy}^n} (H^{-1}J)^Y
\left(\prod_{j=0}^{n-1}\frac{h_{Y_j}}{h_{Y_j} + d_{Y_j}}\right)
 \left(\frac{-1}{d_{Y_n}}\right)
\left(\frac{h_{Y_j}}{h_{Y_j} + d_{Y_n}} - 1\right) \\
&= \sum_{n=0}^\infty \sum_{Y \in \mathcal{W}_{xy}^n}J^Y
\prod_{j=0}^{n}\frac{1}{h_{Y_j}+ d_{Y_j}}.
\end{aligned}
\end{equation}

On the other hand, for the left-hand side of \refeq{DJV} we set $U=H+D$,
and note that $(-Q + D)^{-1}$ is given by the Neumann series
\begin{equation}
(-Q + D)^{-1}
= (U - J)^{-1} = \bigg(U(I - U^{-1}J)\bigg)^{-1}
= \sum_{n=0}^{\infty} \big(U^{-1}J\big)^n U^{-1}.
\lbeq{srw}
\end{equation}
The $xy$ element of the right-hand side is
equal to \refeq{expfactors},
and the proof is complete.
\end{proof}

The complex Gaussian probability
measure on $\C^\Lambda$ with covariance $C$ is defined by
\begin{equation}
    d\mu_C =
    \frac{\det A}{(2\pi i)^M}  e^{-\phi A\bar\phi} d\bar\phi d\phi,
\end{equation}
where $A=C^{-1}$, and $d\bar\phi d\phi$ is the Lebesgue measure
$d\bar\phi_1 d\phi_1 \cdots d\bar\phi_\Lambda d\phi_\Lambda$
(see, e.g., \cite[Lemma~2.1]{BIS09} for a proof that this measure is properly normalised).
In particular, $\int \bar\phi_a\phi_b d\mu_C = C_{ab}$.

In terms of the complex
boson field $\phi,\bar\phi$ and conjugate fermion fields $\psi,\bar\psi$ introduced in
\cite[Section~3]{BBS-saw4-log},
for $x \in \Lambda$
we define the differential form
\begin{equation}
    \tau_x
    = \phi_x \bar\phi_x
    + \psi_x  \wedge \bar\psi_x
    .
\end{equation}
The fermion field is given by the 1-forms $\psi_x = \frac{1}{\sqrt{2\pi i}}d\phi_x$,
$\psib_x = \frac{1}{\sqrt{2\pi i}}d\phib_x$,
and $\wedge$ denotes
the wedge product; we drop the wedge from the notation subsequently with the
understanding that forms are always multiplied using this anti-commutative product.
Let
\begin{equation}
\lbeq{SAaction}
    S_A = \sum_{x\in \Lambda}
    \phi_x A_{xy} \bar\phi_y
    + \psi_x  A_{xy} \bar\psi_y.
\end{equation}
Then
\begin{equation}
\label{e:intrefRHS}
    C_{xy}
    =
    \int e^{-S_A}
    \phib_{x}\phi_{y}
    ,
\end{equation}
where the right-hand side is defined and the identity proved
in
\cite[Section~2.10]{BS-rg-norm}.

The space $\Ncal$ used in the renormalisation group analysis
is an algebra of even differential forms (see \cite[Section~3]{BBS-saw4-log}).
An element $F\in\Ncal$
can be written as
\begin{equation} \label{e:psipsib}
  F=
  \sum_{k=0}^{2|\Lambda|}
  \sum_{s,t: s + t=2k}
  \sum_{x_1,\ldots,x_s\in \Lambda}
  \sum_{y_1,\ldots, y_t \in \Lambda} F_{x,y}
  \psi^x \bar\psi^{y}
  ,
\end{equation}
where $x=(x_1,\ldots,x_s)$, $y=(y_1,\ldots,y_t)$,
$\psi^x = \psi_{x_1}\cdots\psi_{x_s}$,
$\psib^y = \psib_{y_1}\cdots\psib_{y_t}$, and where each $F_{x,y}$
(including the degenerate case $s=t=0$) is a function of
$(\phi,\phib)$.
We require
that elements of $\Ncal$ are such that the coefficients $F_{x,y}$
are in $C^{p_\Ncal}$, with $p_\Ncal=10$ (any larger choice would also suffice).
The Gaussian superexpectation of a differential form $F$ is defined by
\begin{equation}
\lbeq{superex}
    \Ex_C F = \int e^{-S_A} F.
\end{equation}

The following
supersymmetric representation goes back to \cite{BM91},
with antecedents in the physics literature  \cite{McKa80,PS80,Lutt83}.

\begin{prop}
\label{prop:intrep1a}
Let  $N<\infty$, $g>0$, $\nu \in \R$,  $m^2>0$, $A= (-\Delta_\Lambda)^{\alpha/2} + m^2$,
and $C=A^{-1}$.
Let  $\nu_0 = \nu-m^2$
and
$ V_{0} (\Lambda) = \sum_{u\in\Lambda} ( g \tau_u^2 + \nu_{0} \tau_u )$.
Then, for $x,y \in \Lambda$,
\begin{equation}
\label{e:intrep1a}
\begin{aligned}
    \int_{0}^\infty E_x^N \big[e^{- gI_T} \1_{X(T)= y}\big] e^{-\nu T } dT
    &=
    \Ex_C \left( \bar\phi_x\phi_y e^{-V_0(\Lambda)} \right)
    .
\end{aligned}
\end{equation}
\end{prop}

\begin{proof}
We define $f : \R^{\Lambda_N} \to \R$ by
\begin{equation}
\label{e:intrep1}
    f(\rho) = e^{-\sum_{u\in\Lambda_N} \big( g\rho_u^2 + \nu_0 \rho_u \big)}
    \quad \quad
    (\rho \in \R^{\Lambda_N}).
\end{equation}
Since $\sum_{u\in\Lambda}L_T^u= T$,
\begin{equation}
    \int_{0}^\infty E_x^N \big[e^{- gI_T} \1_{X(T)= y}\big] e^{-\nu T } dT
    =
    \int_{0}^\infty E_x^N \big[f(L_T) \1_{X(T)= y}\big] e^{-m^2 T} dT.
\end{equation}
On the other hand,
\begin{equation}
    \Ex_C \left( \bar\phi_x\phi_y e^{-V_0(\Lambda)} \right)
    =
    \int
    e^{-S_A}
    e^{-V_0(\Lambda)} \bar{\phi}_x  \phi_y
    =
    \int
    e^{-S_A}
    f(\tau)
    \bar{\phi}_x  \phi_y .
\end{equation}
We write $f$ in terms of its Fourier transform $\hat f$ as
\begin{equation}
    f(\rho) =
    \int_{\R^{\Lambda_N}} e^{-i \sum_{u\in\Lambda} r_u \rho_u} \hat f(r)\, dr.
\end{equation}
With an appropriate argument to justify interchanges of integration
(done carefully in \cite{BIS09}), it therefore
suffices to show that for all $r_u \in \R$,
\begin{equation}
\label{e:intrepF}
    \int
    e^{-S_A}
    e^{-\sum_{u\in\Lambda} i r_u\tau_u }
    \phib_{x} \phi_{y}
    =
    \int_{0}^{\infty} E_x^N \big[e^{-\sum_{u\in\Lambda} ir_u L_T^u } \1_{X(T)= y}\big]
    e^{-m^2 T}
    dT
    .
\end{equation}
Let $V$ be the diagonal matrix with entries $m^2+ir_u$.
The integral on the left-hand side of \eqref{e:intrepF} is
$((-\Delta_\Lambda)^{\alpha/2} + V)^{-1}_{xy}$ by \refeq{intrefRHS} (with $A$ replaced
by $A+ir$).
By Lemma~\ref{lem:intrepLHS} (with $d_u=ir_u+m^2$),
the right-hand side of \eqref{e:intrepF} is therefore
equal to the left-hand side, and the proof is complete.
\end{proof}

By definition and by Proposition~\ref{prop:intrep1a},
the finite volume susceptibility $\chi_N(g,\nu)$ is given by
\begin{equation}
\lbeq{chiNn0}
    \chi_N(g,\nu)
    =
    \sum_{x \in \Lambda_N}
    \Ex_C \left( \bar\phi_0\phi_x e^{-V_0(\Lambda)} \right),
\end{equation}
with $m^2>0$ and $C= ((-\Delta_\Lambda)^{\alpha/2} + m^2)^{-1}$.
Therefore, by Lemma~\ref{lem:suscept-finvol},
\begin{equation}
\lbeq{supersym}
    \chi(g,\nu) = \lim_{N \to \infty} \chi_N(g,\nu)
    = \lim_{N \to \infty} \sum_{x \in \Lambda_N}
    \Ex_C \left( \bar\phi_0\phi_x e^{-V_0(\Lambda)} \right).
\end{equation}
The mass parameter $m^2$ is introduced here solely to ensure existence of the
inverse defining $C$ on the torus.  It cancels between $C$ and $\nu_0$,
and the right-hand side of \refeq{supersym} is in fact independent of $m^2$.
The identity \refeq{supersym} gives an exact representation of the susceptibility
as the infinite volume limit of the susceptibility of a supersymmetric
field theory, and provides the starting point for the renormalisation group analysis
for the weakly self-avoiding walk.

%%%%%%%%%%%%%%%%%%%%%%%%%%%%%%%%%%%%%%%%%%%%%%%%%%%%%%%%%%%%%%%%%%%%%%
\section*{Acknowledgements}
This work is built on a foundation that was established in collaboration with
David Brydges and Roland Bauerschmidt.  It would not exist without that collaboration and
without their influence,
and I am grateful for their substantial assistance and advice which helped bring this
work to completion.
I thank
Pronob Mitter for correspondence concerning \cite{Mitt16},
Mathav Murugan for advice concerning
the literature on the fractional Laplacian,
and Martin Lohmann and
Benjamin Wallace for discussions and helpful comments on preliminary versions of this paper.
This work was supported in part by NSERC of Canada.

%\bibliography{../bibdef/bib}
%\bibliographystyle{plain}

\end{document}